\documentclass[12pt]{article}
\usepackage{setspace}
\usepackage[utf8]{inputenc}
\usepackage[english]{babel}  
\usepackage[colorlinks=true, linkcolor=blue, citecolor = blue]{hyperref} \usepackage[ruled,vlined]{algorithm2e}    
\usepackage{url}  
\usepackage{booktabs}       
\usepackage{amsfonts}       
\usepackage{nicefrac}       
\usepackage{arxiv}
\usepackage{microtype}
\usepackage[top=1 in,bottom=1in, left=1 in, right=1 in]{geometry}
\usepackage{graphicx}
\usepackage{subfig}
\usepackage{rotating}
\usepackage{caption} 

\usepackage{float}
\usepackage{wrapfig}
\usepackage{amsmath}
\usepackage{amssymb}
\usepackage{amsthm}
\usepackage[round]{natbib}
\usepackage{tabularx}
\usepackage[export]{adjustbox}
\usepackage{multicol}
\usepackage{color}
\usepackage{multirow}
\usepackage[thinlines]{easytable}
\usepackage{multirow}
\usepackage[table,xcdraw]{xcolor}
\hypersetup{
    colorlinks=true,
    linkcolor=blue,
    filecolor=blue,      
    urlcolor=blue,
}
\usepackage{tcolorbox}
\usepackage{fancyhdr}
\newcommand\independent{\protect\mathpalette{\protect\independenT}{\perp}}
\def\independenT#1#2{\mathrel{\rlap{$#1#2$}\mkern2mu{#1#2}}}

\newtheorem{theorem}{Theorem}[section]

\title{\textbf{A Doubly Robust Framework for Addressing Outcome-Dependent Selection Bias in Multi-Cohort EHR Studies}}
\author
{RITOBAN KUNDU$^{1*}$ ,
XU SHI$^1$, MICHAEL KLEINSASSER$^1$, LARS G. FRITSCHE $^1$,\\ \textbf{MAXWELL SALVATORE$^2$, BHRAMAR MUKHERJEE$^{3,4}$}  \\
$^{1}$Department of Biostatistics, University of Michigan, Ann Arbor, Michigan, U.S.A. \\
$^{2}$Biomedical and Translational Informatics Laboratory, Perelman School of Medicine,\\
University of Pennsylvania, Philadelphia, Pennsylvania, U.S.A.\\
$^{3}$Department of Biostatistics, Yale School of Public Health, New Haven, Connecticut, U.S.A.\\
$^{4}$Department of Epidemiology (Chronic Disease), Yale School of Public Health, New Haven, Connecticut, U.S.A.\\
$^{*}$Corresponding Author}

\begin{document}
\maketitle

\centerline{\Large \bf Abstract}
\vspace{0.1in}
\noindent
Selection bias can hinder accurate estimation of association parameters in binary disease risk models using non-probability samples like electronic health records (EHRs). The issue is compounded when participants are recruited from multiple clinics/centers with varying selection mechanisms that may depend on the disease/outcome of interest. Traditional inverse-probability-weighted (IPW) methods, based on constructed parametric selection models, often struggle with misspecifications when selection mechanisms vary across cohorts. This paper introduces a new Joint Augmented Inverse Probability Weighted (JAIPW) method, which integrates individual-level data from multiple cohorts collected under potentially outcome-dependent selection mechanisms, with data from an external probability sample. JAIPW offers double robustness by incorporating a flexible auxiliary score model to address potential misspecifications in the selection models. We outline the asymptotic properties of the JAIPW estimator, and our simulations reveal that JAIPW achieves up to six times lower relative bias and five times lower root mean square error (RMSE) compared to the best performing joint IPW methods under scenarios with misspecified selection models. Applying JAIPW to the Michigan Genomics Initiative (MGI), a multi-clinic EHR-linked biobank, combined with external national probability samples, resulted in cancer-sex association estimates  closely aligned with national benchmark estimates. We also analyzed the association between cancer and polygenic risk scores (PRS) in MGI to illustrate a situation where the exposure variable is not measured in the external probability sample.

\vspace{0.2in}


\section{Introduction}
\label{sec:intro}
Electronic Health Records (EHRs) are vital for clinical care, biomedical research, and managing population health. They support various tasks, such as analyzing disease-exposure associations, studying treatment patterns, and tracking patient outcomes. However, observational studies using EHR data often struggle with validity and generalizability due to biases such as selection bias, information bias, and confounding \citep{beesley2020analytic}. These systematic biases become more amplified in large datasets, a phenomenon termed the ``curse of large n" \citep{kaplan2014big,bradley2021unrepresentative}. Selection bias, in particular \citep{haneuse2016general,beesley2022statistical}, poses significant challenges in non-probability samples, such as the reported healthy volunteer bias in the UK Biobank, which skewed genetic association analyses \citep{schoeler2023participation}. Recent work by \citet{salvatore2024weight} tackled selection bias across multiple biobanks using weighting methods.\\

\noindent
\citet{kundu2024framework} posit a framework for analyzing selection bias using Selection Directed Acyclic Graphs (DAGs) and proposed inverse probability weighted (IPW) methods that rely on external data from the target population. However, their approach assumes homogeneous sampling mechanism in the internal, non-probability EHR sample, which is an untenable assumption when recruitment is done through specialized clinics with distinct selection strategies. Heterogeneity in multi-center EHR data can result from various potential sources such as varying covariate distributions, differences in clinical and coding practices, differences in underlying patient population or through center-specific data harmonization choices. It often goes unrecognized in practice that heterogeneity can also result from differing recruitment mechanisms. While the problem of differing covariate distributions \citep{fu2020assessment} and clinical practices \citep{leese2023clinical} have been studied elsewhere, this paper specifically focuses on EHR heterogeneity introduced by different sampling mechanisms.\\

\noindent
Our work is motivated by the Michigan Genomics Initiative (MGI) \citep{zawistowski2021michigan}, a longitudinal repository at the University of Michigan that integrates EHRs with multi-modal data including germline genetics. Since 2012, MGI has recruited over 100,000 individuals through at least six major specialized clinics: MGI Anesthesiology (BB), MIPACT (Michigan Predictive Activity and Clinical Trajectories), MEND (Metabolism, Endocrinology, and Diabetes), MHB (Mental Health BioBank), PROMPT (PROviding Mental health Precision Treatment), and MY PART (Michigan and You – Partnering to Advance Research Together, a diverse cohort that oversamples Black, Latino/a or LatinX, Middle Eastern, and North African populations), each sub-cohort giving rise to different participant profiles as a result of different patient selection mechanisms. This is common practice in hospital-based cohort building where some clinicians agree to enroll participants to a given biobanking effort. Our goal in the current paper is to derive principled inferential techniques that account for such sampling heterogeneity. In this example, our target population of interest is the United States adult population.\\

\noindent
In Section \ref{sec:jointipw1}, we extend the existing IPW methods to multi-cohort settings with heterogeneous selection mechanisms. However, IPW estimators are susceptible to model misspecification, arising from either incorrect functional forms or unmeasured selection variables. This challenge persists even in single-center studies \citep{kundu2024framework}, underscoring the need for robust methodologies to mitigate biases arising from such misspecification. In the causal inference literature, IPW estimators for Average Treatment Effect (ATE) are augmented with a prediction model for the outcome variable to improve efficiency and robustness, resulting in the widely known class of augmented inverse probability weighted (AIPW) estimators \citep{robins1994estimation, scharfstein1999adjusting} which are Doubly Robust (DR). \citet{chen2020doubly} and \citet{yang2020doubly} focused on developing DR methods for estimating population means in presence of selection bias by integrating a single non-probability sample with individual-level data arising from external probability samples. Under a causal DR framework in EHR, \citet{du2024doubly} additionally studied the effect of selection bias for estimating average treatment effects.\\

\noindent
Unlike standard AIPW contexts, where the primary goal is to estimate population means or treatment effects, our study focuses on estimating disease model association parameters while accounting for the influence of the disease or outcome on selection, an issue largely unaddressed in existing literature. In Section \ref{sec:aipws}, we introduce a novel DR estimator tailored for estimating disease model parameters under outcome-dependent selection within a single cohort. The challenge of model misspecification becomes even more pronounced when multiple sub-cohorts with distinct selection mechanisms are combined together. To address this, in Section \ref{sec:jaipw} we extend the DR framework to the Joint Augmented Inverse Probability Weighting (JAIPW) method, which ensures double robustness across multiple outcome-dependent selection mechanisms. We further propose a flexible cross-fitted non-parametric version of JAIPW for scenarios where the auxiliary score model may be highly complex. We derive the asymptotic distributions and variance estimators for JAIPW estimates.  Extensive simulations (Section \ref{sec:simu}) and two applications to MGI data (Section \ref{sec:data}) demonstrate the superior performance of JAIPW compared to joint IPW methods.

\section{Problem Setup and Notations}\label{sec:setup}
We focus on estimating the association between a binary disease indicator \( D \) and a set of covariates \( \boldsymbol{Z} \) (\( p = \text{dim}(\boldsymbol{Z}) \)) in a target population of size \( N \). For example, in the context of MGI data, the target population of interest is the US adult population. We analyze data from \( K \) internal non-probability EHR samples (cohorts) drawn from the same target population, where each cohort is represented by a binary selection indicator \( S_k \) for \( k = 1, 2, \ldots, K \). We assume that individual-level data can be shared across the cohorts. Figure \ref{fig:popumulti} illustrates the population structure with $K=3$.  Within each cohort, the variables \( \boldsymbol Z_{1k} \) are those that only appear in the disease model, while \( \boldsymbol Z_{2k} \) are shared by both the disease and selection models. Across all cohorts, the full covariate set is consistently specified as \( \boldsymbol{Z} = (1, \boldsymbol{Z}_{1k}, \boldsymbol{Z}_{2k}) \), where the intercept term is explicitly included for ease of notation. The corresponding parameter vector is given by \( \boldsymbol{\theta}_{\boldsymbol{Z}}\) which is of dimension $(p+1)$,  where $p = \text{dim}(\boldsymbol Z_{1k}) + \text{dim}(\boldsymbol Z_{2k})$. For example, suppose the goal is to estimate the effects of depression and diabetes on cancer using data from two cohorts. Cohort 1 primarily consists of diabetic patients, while Cohort 2 focuses on mental health patients. In this scenario, cancer (\(D\)) is the outcome variable, and the covariates (\(\boldsymbol{Z}\)) include depression and diabetes. In Cohort 1, $\boldsymbol{Z}_{11}$ represents depression, and $\boldsymbol{Z}_{21}$ represents diabetes. Conversely, in Cohort 2, $\boldsymbol{Z}_{12}$ represents diabetes, and $\boldsymbol{Z}_{22}$ represents depression. The primary disease model in the target population is expressed by:
\begin{equation}
    \text{logit}\{P(D=1|\boldsymbol Z)\}=\boldsymbol\theta_{\boldsymbol Z}'\boldsymbol Z \cdot \label{eq:eq1}
\end{equation}
Additionally, \( \boldsymbol W_k \) denotes the variables appearing exclusively in the selection model for cohort \( k \). We also account for the possibility that the disease indicator \( D \) may influence the selection process. The cohort-specific selection probability is defined as,
\begin{equation*}
    P(S_k=1|\boldsymbol X_k) = \pi_k(\boldsymbol X_k), \hspace{0.2cm} \boldsymbol X_k =(D,\boldsymbol Z_{2k}, \boldsymbol W_k) ,
\end{equation*}
where $\boldsymbol X_k$ includes all the variables affecting selection indicator $S_k$ directly. The following conditions are imposed on the selection mechanisms, collectively denoted as \textbf{Condition C1}:
\begin{itemize}
\item \textbf{Condition C1.1:} The target population of interest is considered a finite random sample drawn from a hypothetical infinite super-population, similar to the framework considered in \citet{liu2025superpopulation}.
\item \textbf{Condition C1.2:} For each cohort $k \in \{1, \dots, K\}$ and every individual in the target population, the selection probabilities are bounded away from zero and one, i.e., $0 < \pi_{k}(\boldsymbol{X}_{k}) < 1$.
\item \textbf{Condition C1.3:} For any two distinct cohorts $j, k \in \{1, \dots, K\}$ (where $j \neq k$), the selection indicators are independent given the covariates: $S_j \independent S_k \mid \boldsymbol{X}_j, \boldsymbol{X}_k$."
\end{itemize}
Condition C1.3 may be violated when, for example, there exists common unmeasured variables such as access to healthcare that cannot be included in $\boldsymbol X_k$ as it may not be available in the EHR, but is likely to influence selection across cohorts. Let \( S^{\text{mult}} = \max(S_1, S_2, \ldots, S_K) \) denote the composite selection indicator which captures overall selection into the combined study sample. Mathematically, our main goal is to estimate the association parameters of the conditional distribution \( D|\boldsymbol Z \) in equation \eqref{eq:eq1}, but we can only directly assess \( D|\boldsymbol Z, S^{\text{mult}} = 1 \) based on available data. As shown by \citet{beesley2022statistical}, the true model parameters and those from the naively fitted model (conditional on \( S^{\text{mult}} = 1 \)) are related as follows:
\begin{equation}
 \text{logit}\{P(D=1|\boldsymbol Z,S^{\text{mult}}=1)\}=\boldsymbol\theta_{\boldsymbol Z}'\boldsymbol Z+ \text{log}\{r(\boldsymbol Z)\},  \label{eq:eq2}
\end{equation}
where \( r(\boldsymbol Z) = \frac{P(S^{\text{mult}}=1|D=1,\boldsymbol Z)}{P(S^{\text{mult}}=1|D=0,\boldsymbol Z)} \) represents how the disease model covariates \( \boldsymbol Z \) affect the selection mechanism. We present this identity to explicitly demonstrate that naive logistic regression yields biased estimates in most scenarios. Figure \ref{fig:dagmulti} illustrates the relationships among the disease and selection model variables through a selection DAG. While we considered the most general and complex DAG setup, in some cohorts, certain arrows may be absent and strengths of associations may vary. It is important to note that this diagram is not a causal DAG in the traditional intervention sense. Rather, they follow the selection DAG framework described in prior work \citep{kundu2024framework}, depicting a joint association model that captures the structural dependencies between disease, selection, and the covariates driving each of these mechanisms. We assume that if a directed path from $D$ to $S_k$ exists (either directly, $D \to S_k$, or indirectly, $D\rightarrow \boldsymbol W_k\rightarrow S_k$), then $D$ must occur temporally before $S_k$. If the association between $D$ and $S_k$ arises solely from backdoor paths between $D$ and $S_k$, no temporal ordering restriction is imposed. Unweighted logistic regression, as described in equation \eqref{eq:eq2}, typically leads to biased estimates of \(\boldsymbol \theta_{\boldsymbol Z} \). Therefore in the next section we propose IPW and AIPW methods to reduce selection bias in the context of multi-cohort analysis. Even if the disease model is not logistic, all the methods proposed in this paper remain valid, as they only require the model to be specified through a valid estimating equation derived from a parametric model.
\section{Methods}
\subsection{IPW Methods} \label{sec:jointipw}
First, for each individual \( i \) in the target population, we define the selection probability into combined sample as \(\pi(\boldsymbol{X}^{\text{mult}}_i ) = P(S^{\text{mult}}_i = 1 \mid \boldsymbol{X}^{\text{mult}}_i )\), which represents the probability of being selected into at least one of the $K$ cohorts, where  \(\boldsymbol{X}^{\text{mult}}_i = \boldsymbol{X}_{1i} \cup \boldsymbol{X}_{2i} \cup \ldots \cup \boldsymbol{X}_{Ki}\). Using Condition C1.3 we obtain,
\begin{align}
     \pi(\boldsymbol{X}^{\text{mult}}_i)&=P(\text{max}(S_{1i},S_{2i},..,S_{Ki})=1|\boldsymbol{X}^{\text{mult}}_i )=1-P[\text{max}(S_{1i},S_{2i},..,S_{Ki})=0|\boldsymbol{X}^{\text{mult}}_i ] \nonumber\\
     & = 1-[1-\pi_1(\boldsymbol X_{1i})][1-\pi_2(\boldsymbol X_{2i})]...[1-\pi_K(\boldsymbol X_{Ki})]\cdot \label{eq:eq3}
\end{align} 
\noindent
With known cohort-specific selection propensity functions, \( \pi_1(.), \pi_2(.), \ldots, \pi_K(.) \), we estimate the parameters \( \boldsymbol{\theta}_{\boldsymbol{Z}} \) using a weighted logistic score equation:
\begin{equation}
     \frac{1}{N}\sum_{i=1}^{N}\frac{S^{\text{mult}}_i}{\pi(\boldsymbol{X}^{\text{mult}}_i)}\left\{D_i\boldsymbol Z_i-\frac{e^{\boldsymbol\theta_{\boldsymbol Z}'\boldsymbol Z_i}}{(1+e^{\boldsymbol\theta_{\boldsymbol Z}'\boldsymbol Z_i})} \boldsymbol Z_i\right\}=\mathbf{0}.\label{eq:eq4}
\end{equation}
\noindent
The consistency of the estimated parameters \( \widehat{\boldsymbol{\theta}}_{\boldsymbol{Z}} \), derived from equation \eqref{eq:eq4}, assuming known selection probabilities \( \pi(\boldsymbol{X}^{\text{mult}}_i) \), is demonstrated in Supplementary Section S2.1. Supplementary Sections S2.2 and S2.3 provide the asymptotic distribution of \( \widehat{\boldsymbol{\theta}}_{\boldsymbol{Z}} \) and a consistent estimator for the asymptotic variance, with detailed proofs respectively.

\subsection{Extension of existing IPW methods for multiple non-probability samples}\label{sec:jointipw1}
For non-probability samples, selection probabilities are typically unknown, necessitating the use of two-step joint IPW approaches. These approaches typically begin with estimating selection probabilities \(\widehat{\pi}(\boldsymbol{X}^{\text{mult}})\). Estimating each cohort's individual selection probabilities \(\widehat{\pi}_1(\boldsymbol{X}_1)\), \(\widehat{\pi}_2(\boldsymbol{X}_2)\), ..., \(\widehat{\pi}_K(\boldsymbol{X}_K)\) enables the construction of the overall joint selection probability \(\widehat{\pi}(\boldsymbol{X}^{\text{mult}})\) as defined in equation \eqref{eq:eq3}. Disease model parameters \(\boldsymbol{\theta}_{\boldsymbol{Z}}\) are then estimated using the weighted score equation \eqref{eq:eq4} where the true selection probabilities \(\pi(\boldsymbol{X}^{\text{mult}})\) are replaced with estimated \(\widehat{\pi}(\boldsymbol{X}^{\text{mult}})\).\\

\noindent
To estimate cohort-specific selection models, we use selection propensity estimation techniques from \citet{kundu2024framework}, including Pseudolikelihood (PL) \citep{chen2020doubly}, Simplex Regression (SR) \citep{beesley2022statistical}, Post-Stratification (PS) \citep{holt1979post}, and Calibration (CL) \citep{wu2003optimal}, extended to multi-cohort contexts and denoted by JPL, JSR, JPS, and JCL. The first two methods rely on access to individual-level data on selection variables from an external probability sample representing the target population, while the latter two rely on summary-level statistics of the selection variables. In our data analysis in Section \ref{sec:data}, our target population of interest is the 2018 adult US population $(\text{Age}\geq 18)$. This definition aligns with our use of the
NHANES (National Health and Nutrition Examination Survey) 2017-18  data as an individual-level external probability sample, which provides representative data on the selection variables for this population. NHANES is a complex multi-stage survey that provides a representative sample of the U.S. civilian, non-institutionalized population. It includes extensive data on health and nutrition, physical examinations, laboratory tests, demographic and socioeconomic information, and responses to health-related questionnaires. For summary level statistics we used SEER (Surveillance, Epidemiology, and End Results), the U.S. Census, and the CDC (Centers for Disease Control). These external data sources are depicted in Figure \ref{fig:popumulti}.\\

\noindent
While the details of each joint IPW method are provided in Supplementary Sections S3, S4, S5 and S6, we illustrate our extension approach for one of them, namely JPL here. For JPL, we first specify a parametric form for \(\pi_k(\boldsymbol{X}_k)\) as \(\pi_k(\boldsymbol{X}_k, \boldsymbol{\alpha}_k) = \frac{e^{\boldsymbol{\alpha}_k' \boldsymbol{X}_k}}{1 + e^{\boldsymbol{\alpha}_k' \boldsymbol{X}_k}}\), for \(k = 1, 2, \ldots, K\), where $\boldsymbol \alpha_k$ parametrize the selection model function $\pi_k(.)$. Selection parameters \(\boldsymbol{\alpha}_k\) are estimated using the pseudolikelihood estimating equation, $$\frac{1}{N} \sum_{i=1}^N S_{ki}\boldsymbol{X}_{ki} - \frac{1}{N} \sum_{i=1}^N \left(\frac{S_{\text{ext},i}}{\pi_{\text{ext},i}}\right) \pi_k(\boldsymbol{X}_{ki}, \boldsymbol{\alpha}_k) \boldsymbol{X}_{ki} = \mathbf{0},$$
where $S_{\text{ext}}$ and $\pi_{\text{ext}}$ denote the selection indicator variable and known sampling/design probabilities for the external probability sample. The second term of this estimating equation is estimated using the individual-level data from the external data. The Newton-Raphson method is used to solve for $\boldsymbol{\alpha}_k$ in this equation. The internal selection probabilities \(\pi_k(\boldsymbol{X}_{ki}, \boldsymbol{\alpha}_k)\) are then estimated by plugging in estimated \(\widehat{\boldsymbol{\alpha}}_k\) for all the cohorts. Using these probabilities and equation \eqref{eq:eq3}, we calculate the joint selection probability \(\widehat{\pi}(\boldsymbol{X}^{\text{mult}})\), which is used in the IPW equation \eqref{eq:eq4} to estimate \(\widehat{\boldsymbol{\theta}}_{\boldsymbol{Z}}\). Proofs of consistency, asymptotic distribution, and variance estimator for \(\widehat{\boldsymbol{\theta}}_{\boldsymbol{Z}}\) under multiple selection mechanisms are detailed in Supplementary Section S3. Descriptions and asymptotic theory for the other three joint IPW methods are provided in Supplementary Sections S4, S5, and S6. However, consistent estimation of disease model parameters using all of these joint IPW methods requires the correct specification of the combined selection model \(\pi(\boldsymbol{X}^{\text{mult}})\), which is a challenging task in practice. Next, we propose a Doubly Robust method to estimate disease model parameters while ensuring robustness against misspecified selection models. 

\vspace{-0.5cm}
\subsection{AIPW method for a single study}\label{sec:aipws}
We first propose our DR method for a single study with selection indicator \( S \), disease covariates \( \boldsymbol{Z} = (1,\boldsymbol{Z}_1, \boldsymbol{Z}_2) \) with corresponding $\boldsymbol \theta_{\boldsymbol Z}=(\theta_0,\boldsymbol \theta_1,\boldsymbol \theta_2)$ and selection variables \( \boldsymbol{X} \), where \( \boldsymbol{X} = (D, \boldsymbol{Z}_2, \boldsymbol{W}) \). The standard DR estimating equation integrates a selection propensity score component with an outcome prediction model (auxiliary score model). However, when the outcome variable \( D \) directly influences the selection mechanism, the conventional outcome prediction model, typically constructed based on the conditional distribution \( P(D \mid \boldsymbol{X}) \)—is misspecified, as \( P(D = 1 \mid \boldsymbol{X}) \neq P(D = 1 \mid \boldsymbol{X}, S = 1) \). To overcome this issue, we identify disease model covariates that do not directly affect selection (\(\boldsymbol{Z}_1\)) and construct an auxiliary score model that captures the relationship between \( \boldsymbol{Z}_1 \) and \( \boldsymbol{X} \), incorporating information from an external probability sample to improve robustness. Formally, we impose the following condition:\\

\noindent
\textbf{Condition C2:} There exists a non-empty subset of $\boldsymbol Z$, denoted $\boldsymbol{Z}_{1}$, such that $\boldsymbol{Z}_1 \independent S \mid \boldsymbol{X}$. Moreover, for all individuals in the target population, the positivity condition, $0<P(S=1|\boldsymbol X)<1$ holds. Individual-level data on \(\boldsymbol{X}\) is available in an external probability sample.\\

\noindent
This is a Missing At Random (MAR) assumption \citep{rubin1976inference,little2019statistical}, where $\boldsymbol{Z}_1$ is effectively ``missing" for individuals not selected into the sample (i.e., when $S=0$). This is analogous to the scenario in \citet{chen2020doubly}, where the outcome variable was assumed to be MAR in the internal sample. In our framework, because selection $S$ may depend on the outcome $D$, we require this subset $\boldsymbol{Z}_{1}$ to meet the MAR condition. The construction of our DR estimating equation, therefore, follows a similar approach to AIPW methods developed for handling missing data. The positivity condition in Condition C2 directly follows from Condition C1.2 for $K=1$.
\par
Since we have established the MAR framework, we formally specify the auxiliary score model using an approach similar to that proposed in \citet{williamson2012doubly}. Specifically, we project the disease model score function onto the space spanned by \(\boldsymbol{X}\). The disease model score function is: $
\mathcal{U}(\boldsymbol{\theta}_{\boldsymbol{Z}}) = D \boldsymbol Z - \frac{e^{\boldsymbol{\theta}_{\boldsymbol{Z}}' \boldsymbol{Z}}}{1 + e^{ \boldsymbol{\theta}_{\boldsymbol{Z}}' \boldsymbol{Z}}}\boldsymbol{Z}$. The auxiliary score model is denoted by this projection, \(\boldsymbol{f}(\boldsymbol{X}, \boldsymbol\theta_{\boldsymbol{Z}}):=\mathbb{E}(\mathcal{U}(\boldsymbol{\theta}_{\boldsymbol{Z}}) \mid \boldsymbol{X}, S=1) = \mathbb{E}(\mathcal{U}(\boldsymbol{\theta}_{\boldsymbol{Z}}) \mid \boldsymbol{X}) \). The first equality holds due to  \(\boldsymbol{Z}_{1} \independent S \mid \boldsymbol{X}\), from Condition C2. The MAR DR estimating equation in this context is given by:
\begin{equation}
\frac{1}{N} \sum_{i=1}^N \frac{S_i}{\pi(\boldsymbol{X}_i)} (\mathcal{U}_i(\boldsymbol{\theta}_{\boldsymbol{Z}}) - \boldsymbol f(\boldsymbol{X}_i, \boldsymbol{\theta}_{\boldsymbol{Z}})) + \frac{1}{N} \sum_{i=1}^N \boldsymbol f(\boldsymbol{X}_i, \boldsymbol{\theta}_{\boldsymbol{Z}}) = \boldsymbol 0 \cdot \label{eq:eq5}
\end{equation}
Since the second term cannot be estimated directly from available data, we replace it with \(\frac{1}{N} \sum_{i=1}^N \frac{S_{\text{ext},i}}{\pi_{\text{ext},i}} \boldsymbol{f}_i(\boldsymbol{X}_i, \boldsymbol\theta_{\boldsymbol{Z}}) \), which is possible due to availability of $\boldsymbol X$ in the external probability sample as specified in Condition C2. The proposed DR estimating equation is:
\begin{equation}
\frac{1}{N} \sum_{i=1}^N \frac{S_i}{\pi(\boldsymbol{X}_i)}  (\mathcal{U}_i(\boldsymbol{\theta}_{\boldsymbol{Z}}) - \boldsymbol f(\boldsymbol{X}_i, \boldsymbol{\theta}_{\boldsymbol{Z}})) + \frac{1}{N} \sum_{i=1}^N \frac{S_{\text{ext},i}}{\pi_{\text{ext},i}} \boldsymbol f(\boldsymbol{X}_i, \boldsymbol{\theta}_{\boldsymbol{Z}}) = \boldsymbol 0\cdot \label{eq:eq6}
\end{equation}
It is important to note that the availability of \(\boldsymbol{Z}_{1}\) in the external probability sample is not required. For instance, in the second data example in Section \ref{sec:prsresult}, \(\boldsymbol{Z}_{1}\) corresponds to Polygenic Risk Scores (PRS), which is available in the MGI data but not in the NHANES data. Even if \(\boldsymbol{Z}_{1}\) were available in NHANES, which has a smaller sample size compared to the MGI, integrating both the datasets using an auxiliary score model would improve estimation efficiency, as demonstrated in the first data example in Section \ref{sec:cansex}. Next we discuss estimation of nuisance components in equation \eqref{eq:eq6}, namely $\pi(.)$ and $\boldsymbol f(.)$.\\

\noindent
\textbf{Nuisance Parameter Estimation:} The forms of the selection weight function $\pi(\boldsymbol X)$ and the auxiliary score function $\boldsymbol f(\boldsymbol X, \theta_{\boldsymbol Z})$ are rarely known a priori. Consequently, we need to estimate these functions.\\

\noindent
\textbf{Estimation of the Selection Model:} We parameterize the selection model as \(\pi(\boldsymbol{X}_i, \boldsymbol{\alpha})\) and estimate the selection probabilities using one of the four IPW methods, namely PL, SR, PS, or CL for a single cohort, outlined in Section \ref{sec:jointipw1}. Let $\boldsymbol{q}(\boldsymbol{\alpha}, \boldsymbol{X}_i, S_i, S_{\text{ext},i}, \pi_{\text{ext},i})$ represent the selection estimating function using any of these IPW methods. For example in case of pseudolikelihood, $\boldsymbol{q}(\boldsymbol{\alpha}, \boldsymbol{X}_i, S_i, S_{\text{ext},i}, \pi_{\text{ext},i})=\left[S_{i}\boldsymbol{X}_{i} - \left(\frac{S_{\text{ext},i}}{\pi_{\text{ext},i}}\right) \pi(\boldsymbol{X}_{i}, \boldsymbol{\alpha}) \boldsymbol{X}_{i}\right]$. The general estimating equation for selection model estimation is then given by:
\begin{equation}
\frac{1}{N} \sum_{i=1}^N \boldsymbol{q}(\boldsymbol{\alpha}, \boldsymbol{X}_i, S_i, S_{\text{ext},i}, \pi_{\text{ext},i}) = \boldsymbol{0}\cdot \label{eq:eq7}
\end{equation}
\noindent
\textbf{Estimation of the Auxiliary Score Model - parametric approach:}
One can express:
\begin{equation}
\boldsymbol{f}(\boldsymbol{X}, \boldsymbol{\theta}_{\boldsymbol{Z}}) = \mathbb{E}(\mathcal{U}(\boldsymbol{\theta}_{\boldsymbol{Z}}) \mid \boldsymbol{X}) = \int_{\boldsymbol{Z}_{1}} \mathcal{U}(\boldsymbol{\theta}_{\boldsymbol{Z}}) \boldsymbol f(\boldsymbol{Z}_{1} \mid \boldsymbol{X}) d\boldsymbol{Z}_{1}. \label{eq:eq8}
\end{equation}
\noindent
In this approach, we assume a parametric distribution for \(\boldsymbol{Z}_{1} \mid \boldsymbol{X}\), parameterized by \(\boldsymbol{\gamma}\). The parameter \(\boldsymbol{\gamma}\) is estimated using maximum likelihood estimation (MLE) based on this specified distribution. Let $\boldsymbol{p}(\boldsymbol{\gamma}, \boldsymbol{X}_i, S_i, \boldsymbol{Z}_{1i})$ denote the score function corresponding to the maximum likelihood estimator (MLE) of $\boldsymbol{\gamma}$. The MLE score equation is given by:
\begin{equation}
\frac{1}{N} \sum_{i=1}^N \boldsymbol{p}(\boldsymbol{\gamma}, \boldsymbol{X}_i, S_i, \boldsymbol{Z}_{1i}) = \boldsymbol{0}. \label{eq:eq9}
\end{equation}
Once \(\boldsymbol{\gamma}\) is estimated, we compute \(\mathbb{E}(\mathcal{U}(\boldsymbol{\theta}_{\boldsymbol{Z}}) \mid \boldsymbol{X})\) using equation \eqref{eq:eq8}. If integration proves overly complex due to high dimensionality, nonlinearity, or other factors, alternative strategies for approximation can be used. These may include numerical quadrature approximation or Monte Carlo methods. Evaluating the integral in equation \eqref{eq:eq8} can be computationally challenging, especially with multivariate \(\boldsymbol{Z}_{1}\), and specifying the conditional distribution of \(\boldsymbol{Z}_{1} \mid \boldsymbol{X}\) parametrically can be difficult. Therefore, we offer an alternative way to model \(\boldsymbol{f}(\boldsymbol{X}, \boldsymbol{\theta}_{\boldsymbol{Z}})\) directly using assumption-lean flexible methods.\\

\noindent
\textbf{Non-parametric/Machine Learning approach for estimation of the auxiliary score model:} For flexible modelling of auxiliary score function, we rewrite 
\begin{equation}
    \mathcal{U}(\boldsymbol \theta_{\boldsymbol Z})=D(1,\boldsymbol Z_{1},\boldsymbol Z_{2})-\frac{e^{\theta_0+\boldsymbol \theta_{1} '\boldsymbol Z_{1}+\boldsymbol \theta_{2}' \boldsymbol Z_{2}}}{1+e^{\theta_0+\boldsymbol \theta_{1}' \boldsymbol Z_{1}+\boldsymbol \theta_{2} '\boldsymbol Z_{2}}} (1,\boldsymbol Z_{1},\boldsymbol Z_{2}),\label{eq:eq10}
\end{equation}
We can write the auxiliary score model as,
\begin{align}
    &\mathbb{E}(\mathcal{U}(\boldsymbol \theta_{\boldsymbol Z})|\boldsymbol X)=\boldsymbol f(\boldsymbol X, \boldsymbol \theta_{\boldsymbol Z})=[f_2(\boldsymbol X,\boldsymbol \theta_{\boldsymbol Z}),\boldsymbol f_1(\boldsymbol X,\boldsymbol \theta_{\boldsymbol Z}),f_2(\boldsymbol X,\boldsymbol \theta_{\boldsymbol Z}) \boldsymbol Z_{2}],\label{eq:eq11}\quad \text{where,}\\
    &\boldsymbol f_1(\boldsymbol X,\boldsymbol \theta_{\boldsymbol Z})=\mathbb{E}\left[\left.\left(D \boldsymbol Z_{1}-D\boldsymbol Z_{1}\frac{e^{\theta_0+\boldsymbol \theta_{1}' \boldsymbol Z_{1}+\boldsymbol \theta_{2} '\boldsymbol Z_{2}}}{1+e^{\theta_0+\boldsymbol \theta_{1}' \boldsymbol Z_{1}+\boldsymbol \theta_{2}' \boldsymbol Z_{2}}}\right)\right|\boldsymbol X\right],\quad
    f_2(\boldsymbol X,\boldsymbol \theta_{\boldsymbol Z})=\mathbb{E}\left[\left.\left(D-D\frac{e^{\theta_0+\boldsymbol \theta_{1}' \boldsymbol Z_{1}+\boldsymbol \theta_{2} '\boldsymbol Z_{2}}}{1+e^{\theta_0+\boldsymbol \theta_{1}' \boldsymbol Z_{1}+\boldsymbol \theta_{2}' \boldsymbol Z_{2}}}\right)\right|\boldsymbol X\right]\cdot\nonumber
\end{align}
For modeling both \( \boldsymbol{f}_1(\cdot) \) and \( \boldsymbol{f}_2(\cdot) \), it is natural to use non-parametric or machine learning algorithms, such as Random Forests or XGBoost due to their flexibility in capturing nonlinear relationships and interactions. Such flexible modeling substantially reduces the risk of incorrect functional specification of the conditional expectation \( \mathbb{E}[U(\boldsymbol{\theta}_{\boldsymbol{Z}}) \mid \boldsymbol{X}] \). Since we adopt machine learning estimators, we follow a sample splitting cross-fitted strategy to estimate the combined nuisance parameter function consisting of the selection propensity and auxiliary score functions based on the Double Machine Learning (DML) framework of \citet{chernozhukov2018double} to ensure valid inference, avoid overfitting, and circumvent Donsker conditions. Specifically, we use the DML2 estimator from \citet{chernozhukov2018double} which employs sample splitting and cross-fitting.\\

\noindent
Since we provide methods for estimating the auxiliary score function using both parametric and non-parametric approaches, we use \(\widehat{f}(\boldsymbol{X}, \boldsymbol{\theta}_{\boldsymbol{Z}})\) to denote the estimated auxiliary score function in a general form. In the parametric case, it is given by \(\widehat{\boldsymbol{f}}(\boldsymbol{X}, \boldsymbol{\theta}_{\boldsymbol{Z}}) = \boldsymbol{f}(\boldsymbol{X}, \widehat{\boldsymbol{\gamma}}, \boldsymbol{\theta}_{\boldsymbol{Z}})\), where \(\widehat{\boldsymbol{\gamma}}\) is estimated from equation \eqref{eq:eq9}. For flexible, non-parametric methods, \(\widehat{\boldsymbol{f}}(.)\) is estimated using equation \eqref{eq:eq11}. After estimating the two nuisance components, equation \eqref{eq:eq6} can be expressed as:
\begin{equation}
    \frac{1}{N} \sum_{i=1}^N \boldsymbol g(\boldsymbol \theta_{\boldsymbol Z},\widehat{\boldsymbol \alpha},\widehat{\boldsymbol f})=\boldsymbol 0, \quad \text{where,}\quad \boldsymbol g(\boldsymbol \theta_{\boldsymbol Z},\boldsymbol \alpha,\boldsymbol f)=\frac{S}{\pi(\boldsymbol X,\boldsymbol \alpha)} 
\left[\mathcal{\boldsymbol U}(\boldsymbol \theta_{\boldsymbol Z}) - \boldsymbol f(\boldsymbol X,\boldsymbol \theta_{\boldsymbol Z})\right] 
+ \frac{S_{\text{ext}}}{\pi_{\text{ext}}} f(\boldsymbol X,\boldsymbol \theta_{\boldsymbol Z})\cdot\label{eq:eq12}
\end{equation}
\subsubsection{Double Robustness Property}\label{sec:drsingle}
The following theorem establishes the double robustness property of the proposed estimating equation \eqref{eq:eq6}. Correct specification of the propensity score model implies that \( P(S=1 \mid \boldsymbol{X}) = \pi(\boldsymbol{X}, \boldsymbol{\alpha}^{*}) \), where \(\boldsymbol{\alpha}^{*}\) represents the in-probability limit of \(\widehat{\boldsymbol{\alpha}}\), estimated from equation \eqref{eq:eq8}. Similarly, correct specification of the parametric auxiliary score model requires that \(\mathbb{E}(\mathcal{U}(\boldsymbol{\theta}_{\boldsymbol{Z}}) \mid \boldsymbol{X}) = \boldsymbol{f}(\boldsymbol{X}, \boldsymbol{\theta}_{\boldsymbol{Z}}) = \boldsymbol{f}^*(\boldsymbol{X}, \boldsymbol{\theta}_{\boldsymbol{Z}})=\boldsymbol{f}(\boldsymbol{X}, \boldsymbol{\theta}_{\boldsymbol{Z}}, \boldsymbol{\gamma}^{*})\), where \(\boldsymbol{\gamma}^{*}\) is the in-probability limit of \(\widehat{\boldsymbol{\gamma}}\), estimated from equation \eqref{eq:eq9}. For flexible, non-parametric approaches, correct specification of the auxiliary score model implies that \(\mathbb{E}(\mathcal{U}(\boldsymbol{\theta}_{\boldsymbol{Z}}) \mid \boldsymbol{X}) = \boldsymbol{f}(\boldsymbol{X}, \boldsymbol{\theta}_{\boldsymbol{Z}}) = \boldsymbol{f}^*(\boldsymbol{X}, \boldsymbol{\theta}_{\boldsymbol{Z}})\), where \(\boldsymbol{f}^*(\boldsymbol{X}, \boldsymbol{\theta}_{\boldsymbol{Z}})\) denotes the limit in $L_2$ norm in probability of the estimated function \(\widehat{\boldsymbol{f}}(\boldsymbol{X}, \boldsymbol{\theta}_{\boldsymbol{Z}})\), obtained from equation \eqref{eq:eq11}.
\begin{theorem}
      Under Conditions C1 and C2 and regularity assumptions \textit{A1}, \textit{A2} for $K=1$, in Supplementary Section S1 and assuming either the selection propensity  model, specified by $\pi(\boldsymbol{X}, \boldsymbol{\alpha}^{*})$  or the auxiliary score model specified by $\boldsymbol{f}^*(\boldsymbol{X}, \boldsymbol{\theta}_{\boldsymbol{Z}})$ , or both, are correctly specified, \(\widehat{\boldsymbol{\theta}}_{\boldsymbol{Z}}\) estimated using equation \eqref{eq:eq6} is consistent for \(\boldsymbol{\theta}^*_{\boldsymbol{Z}}\) as \(N \rightarrow \infty\).
\end{theorem}
\noindent
The proof of this theorem is presented in Supplementary Section S7.1.
In the next section, we extend the proposed method to the multi-cohort setting. 
\subsection{JAIPW for multiple cohorts}\label{sec:jaipw}
In this section, we extend the proposed AIPW estimator from the previous section to the framework introduced in Section \ref{sec:setup}, accommodating \( K \) cohorts. This extension accounts for cohort-specific selection mechanisms while ensuring robustness through the incorporation of auxiliary score model across multiple cohorts. We impose the condition:\\

\noindent
\textbf{Condition C3}: There exists a non-empty subset of $\boldsymbol Z$, denoted $\boldsymbol{Z}_{1\cap}$, such that $\boldsymbol{Z}_{1\cap} \independent S^{\text{mult}} \mid \boldsymbol{X}^{\text{mult}}$. Moreover, for all individuals in the target population, the positivity condition $0<P(S^{\text{mult}}=1|\boldsymbol X^{\text{mult}})<1$ holds. We also assume individual-level data on $\boldsymbol{X}^{\text{mult}}$ is available from an external probability sample.\\

\noindent
This condition requires a non-empty subset $\boldsymbol{Z}_{1\cap}$ that is conditionally independent of the combined selection indicator $S^{\text{mult}}$ given $\boldsymbol{X}^{\text{mult}}$. This is a MAR assumption, where $\boldsymbol{Z}_{1\cap}$ is effectively ``missing" for individuals not selected into the combined sample. The positivity condition in Condition C3 directly follows from Condition C1.2.
In this scenario, we define the auxiliary score component of the DR estimating equation by projecting the disease model score function onto the space spanned by \(\boldsymbol{X}^{\text{mult}} \). This projection is denoted as \(
\mathbb{E}(\mathcal{U}(\boldsymbol{\theta}_{\boldsymbol{Z}}) \mid \boldsymbol{X}^{\text{mult}}, S^{\text{mult}}=1) = \mathbb{E}(\mathcal{U}(\boldsymbol{\theta}_{\boldsymbol{Z}}) \mid \boldsymbol{X}^{\text{mult}}) = \boldsymbol{f}(\boldsymbol{X}^{\text{mult}}, \boldsymbol{\theta}_{\boldsymbol{Z}}).
\) 
The first equality directly follows from Condition C3. The final DR estimating equation for the multi-cohort case is obtained with a slight modification of equation \eqref{eq:eq6}, given by:  
\begin{equation}
\frac{1}{N} \sum_{i=1}^N \frac{S^{\text{mult}}_i}{\pi(\boldsymbol{X}^{\text{mult}}_i)}  (\mathcal{U}_i(\boldsymbol{\theta}_{\boldsymbol{Z}}) - \boldsymbol f(\boldsymbol{X}^{\text{mult}}_i, \boldsymbol{\theta}_{\boldsymbol{Z}})) + \frac{1}{N} \sum_{i=1}^N \frac{S_{\text{ext},i}}{\pi_{\text{ext},i}} \boldsymbol f(\boldsymbol{X}^{\text{mult}}_i, \boldsymbol{\theta}_{\boldsymbol{Z}}) = \boldsymbol 0\cdot \label{eq:eq13}
\end{equation}
\textbf{Nuisance Parameter Estimation :} For estimating the selection components of the DR estimating equation, we utilize one of the four joint IPW methods, namely JPL, JSR, JPS, or JCL to estimate $\boldsymbol{\alpha}^{\text{mult}}$, using,
\begin{equation}
\frac{1}{N} \sum_{i=1}^N \boldsymbol{q}(\boldsymbol{\alpha}^{\text{mult}}, \boldsymbol{X}^{\text{mult}}_i, S^{\text{mult}}_i, S_{\text{ext},i}, \pi_{\text{ext},i}) = \boldsymbol{0}, \label{eq:eq14}
\end{equation}  
where \(\boldsymbol{\alpha}^{\text{mult}} = (\boldsymbol{\alpha}_1, \boldsymbol{\alpha}_2, \dots, \boldsymbol{\alpha}_K)\) represents the selection model parameters across cohorts.
\par
For parametric estimation of the auxiliary score, we specify a parametric distribution for \(\boldsymbol{Z}_{1\cap} \mid \boldsymbol{X}^{\text{mult}}\), parameterized by \(\boldsymbol{\gamma}^{\text{mult}}\). The parameter \(\boldsymbol{\gamma}^{\text{mult}}\) is estimated using the MLE score function derived from the assumed parametric distribution of \(\boldsymbol{Z}_{1\cap} \mid \boldsymbol{X}^{\text{mult}}\):  
\begin{equation}
\frac{1}{N} \sum_{i=1}^N \boldsymbol{p}(\boldsymbol{\gamma}^{\text{mult}}, \boldsymbol{X}^{\text{mult}}_i, S^{\text{mult}}_i, \boldsymbol{Z}_{1\cap,i}) = \boldsymbol{0}\cdot \label{eq:eq15}
\end{equation}  
For the non-parametric version of auxiliary score estimation, the framework remains identical to the single-cohort case, except that \(\boldsymbol{f}_1(.)\) and \(f_2(.)\) in equation \eqref{eq:eq11} are now defined as:  
\begin{align*}
    \boldsymbol f_1(\boldsymbol X^{\text{mult}},\boldsymbol \theta_{\boldsymbol Z})&=\mathbb{E}\left[\left.\left(D \boldsymbol Z_{1\cap}-D\boldsymbol Z_{1\cap}\frac{e^{\theta_0+\boldsymbol \theta_{\boldsymbol Z_{1\cap}}' \boldsymbol Z_{1\cap}+\boldsymbol \theta_{\boldsymbol Z_{-1\cap}} '\boldsymbol Z_{-1\cap}}}{1+e^{\theta_0+\boldsymbol \theta_{\boldsymbol Z_{1\cap}}' \boldsymbol Z_{1\cap}+\boldsymbol \theta_{\boldsymbol Z_{-1\cap}}' \boldsymbol Z_{-1\cap}}}\right)\right|\boldsymbol X^{\text{mult}}\right],\\
    f_2(\boldsymbol X^{\text{mult}},\boldsymbol \theta_{\boldsymbol Z})&=\mathbb{E}\left[\left.\left(D-D\frac{e^{\theta_0+\boldsymbol \theta_{\boldsymbol Z_{1\cap}}' \boldsymbol Z_{1\cap}+\boldsymbol \theta_{\boldsymbol Z_{-1\cap}}' \boldsymbol Z_{-1\cap}}}{1+e^{\theta_0+\boldsymbol \theta_{\boldsymbol Z_{1\cap}}' \boldsymbol Z_{1\cap}+\boldsymbol \theta_{\boldsymbol Z_{-1\cap}}' \boldsymbol Z_{-1\cap}}}\right)\right|\boldsymbol X^{\text{mult}}\right],
\end{align*}
where \(\boldsymbol{Z}_{-1\cap} = \boldsymbol{Z} / \boldsymbol{Z}_{1\cap}\). In this case also, we use cross fitting technique to estimate the nuisance parameters to avoid Donsker conditions. The following theorem establishes the double robustness property of the proposed estimating equation \eqref{eq:eq13}.
\begin{theorem}
    Under Conditions C1 and C3 and regularity assumptions \textit{A1} and \textit{A2} in Supplementary Section S1 and assuming either the combined selection propensity  model specified by $\pi(\boldsymbol X^{\text{
mult}},\boldsymbol \alpha^{\text{
mult}*})$ or the auxiliary score model specified by $\boldsymbol f^*(\boldsymbol X^{\text{
mult}}, \boldsymbol{\theta}_{\boldsymbol{Z}})$, or both, is correctly specified, \(\widehat{\boldsymbol{\theta}}_{\boldsymbol{Z}}\) estimated using equation \eqref{eq:eq13}  is consistent for \(\boldsymbol{\theta}^*_{\boldsymbol{Z}}\) as \(N \rightarrow \infty\).
\end{theorem}
\noindent
The proof of this theorem is presented in Supplementary Section S7.2.
\subsubsection{Asymptotic Distribution and Variance Estimation}
In this section, we present the asymptotic properties of the JAIPW estimator \(\widehat{\boldsymbol{\theta}}_{\boldsymbol{Z}}\), including its consistency, asymptotic distribution, and a consistent estimator of its asymptotic variance.\\

\noindent
\textbf{Parametric Estimation in the Auxiliary Score Model:}  
We first establish the asymptotic distribution of \(\widehat{\boldsymbol{\theta}}_{\boldsymbol{Z}}\) under parametric estimation of the nuisance components. 
\begin{theorem}
     Under Conditions C1 and C3 and regularity assumptions \textit{A1}, \textit{A2} and \textit{A}7 in the Supplementary Section S1, the asymptotic distribution of $\widehat{\boldsymbol \theta}_{\boldsymbol Z}$ using JAIPW with parametric nuisance functions' estimator is given by
     \begin{align*}
         & \sqrt{N}(\widehat{\boldsymbol{\theta}}_{\boldsymbol{Z}} - \boldsymbol{\theta}^*_{\boldsymbol{Z}}) \xrightarrow{d} \mathcal{N}(\mathbf{0}, V), \quad \text{where,}\\
         & V = (G_{\boldsymbol{\theta}^*_{\boldsymbol{Z}}})^{-1} \, \mathbb{E} \left[\left( \boldsymbol{g}^{\text{mult}*} + G_{\boldsymbol{\alpha}^{\text{mult}*}, \boldsymbol{\gamma}^{\text{mult}*}} \boldsymbol{\Psi} \right) \left(\boldsymbol{g}^{\text{mult}*} + G_{\boldsymbol{\alpha}^{\text{mult}*}, \boldsymbol{\gamma}^{\text{mult}*}} \boldsymbol{\Psi} \right)'\right] \left( G_{\boldsymbol{\theta}^*_{\boldsymbol{Z}}} \right)^{-1},\\
         & \boldsymbol g=\frac{S^{\text{mult}}}{\pi(\boldsymbol X^{\text{mult}},\boldsymbol \alpha^{\text{mult}*})} 
\left[\mathcal{\boldsymbol U}(\boldsymbol \theta_{\boldsymbol Z}^*) - \boldsymbol f(\boldsymbol X^{\text{mult}},\boldsymbol \theta_{\boldsymbol Z}^*,\boldsymbol \gamma^{\text{mult}*})\right] 
+ \frac{S_{\text{ext}}}{\pi_{\text{ext}}} f(\boldsymbol X^{\text{mult}},\boldsymbol \theta_{\boldsymbol Z}^*,\boldsymbol \gamma^{\text{mult}*})\cdot
     \end{align*}
\end{theorem}
\noindent
The matrices \( G_{\boldsymbol \theta_{\boldsymbol Z}^*} \) and \( G_{\boldsymbol \alpha^{\text{mult}*},\boldsymbol \gamma^{\text{mult}*}} \) are defined as the expectations of the first-order gradients of $\boldsymbol g$ with respect to \( \boldsymbol \theta_{\boldsymbol Z} \) and \( (\boldsymbol \alpha^{\text{mult}} ,\boldsymbol \gamma^{\text{mult}}) \) respectively evaluated at their true values. The quantity \( \boldsymbol \Psi \) is given by \( \boldsymbol \Psi = -H^{-1} \boldsymbol h \), evaluated at the true values of the nuisance parameters. Here, \( \boldsymbol h \) consists of the joint estimating equations for the nuisance parameters \( (\boldsymbol \alpha^{\text{mult}},\boldsymbol \gamma^{\text{mult}} )\), while \( H \) is the Jacobian matrix of \( \boldsymbol h \) with respect to those parameters. The proof and explicit expressions for all components are provided in Supplementary Section S7.3. The corresponding empirical variance estimator for \(V\) is discussed in Supplementary Section S7.4. Notably, this theorem holds even if either the selection or auxiliary score model is correctly specified, due to the inclusion of the nuisance contribution in \(V\).\\

\noindent
\textbf{Flexible Estimation via Machine Learning:}  
For most parametric models especially maximum likelihood estimation, the rate of convergence for estimators is $O(N^{-1/2})$, then we do not require to impose any conditions of the rates of convergences of the nuisance parameter estimation. In case of non-parametric nuisance parameter estimation, the use of cross-fitting allows us to circumvent the restrictive Donsker class conditions. However, we need to impose specific convergence rate constraints on the non-parametric nuisance parameter estimators. The rate constraints are formally specified in Assumptions \textit{A}9 and \textit{A}10  in Supplementary Section S1. These assumptions are drawn primarily from \citet{chernozhukov2018double} and have been tailored to our specific setting. In smooth problems, these assumptions essentially translate to the requirement that the nuisance parameters are estimated at a rate of $o(N^{-1/4})$.
\begin{theorem}
    Let \(c_0, c_1 > 0\), \(q > 2\), and let \(\delta_N, \Delta_N \to 0\) be sequences such that \(\delta_N \geq N^{-1/2 + 1/q} \log N\), and \(N^{-1/2} \log N \leq \tau_N \leq \delta_N\) for all \(N \geq 1\). Let \(L \geq 2\) denote the number of cross-fitting folds. Under conditions C1 and C3 and regularity assumptions, \textit{A1}, \textit{A2}, \textit{A}9 and \textit{A}10  in Supplementary Section S1, and assuming correct specification of both the selection and auxiliary score models, then for $\widehat{\boldsymbol{\theta}}_{\boldsymbol{Z}}$ obtained using JAIPW with non-parametric auxiliary score model estimator is given by,
\begin{align*}
    &\sqrt{N}(\widehat{\boldsymbol{\theta}}_{\boldsymbol{Z}} - \boldsymbol{\theta}^*_{\boldsymbol{Z}}) \xrightarrow{d} \mathcal{N}(\mathbf{0}, V) \quad \text{where,} \quad V = (G_{\boldsymbol{\theta}^*_{\boldsymbol{Z}}})^{-1} \, \mathbb{E}[\boldsymbol{g}^{\text{mult}*} \boldsymbol{g}^{\text{mult}*'}] \, (G_{\boldsymbol{\theta}^*_{\boldsymbol{Z}}})^{-1},\\
    & \boldsymbol{g}^{\text{mult}*}=\boldsymbol g^{\text{mult}}(\boldsymbol \theta_{\boldsymbol Z}^*,\boldsymbol \alpha^{\text{mult}*},\boldsymbol f^*)=\frac{S^{\text{mult}}}{\pi(\boldsymbol X^{\text{mult}},\boldsymbol \alpha^{\text{mult}*})} 
\left[\mathcal{\boldsymbol U}(\boldsymbol \theta_{\boldsymbol Z}^*) - \boldsymbol f^*(\boldsymbol X^{\text{mult}},\boldsymbol \theta^*_{\boldsymbol Z})\right] 
+ \frac{S_{\text{ext}}}{\pi_{\text{ext}}} f^*(\boldsymbol X^{\text{mult}},\boldsymbol \theta^*_{\boldsymbol Z})\cdot
\end{align*}
\end{theorem}
\noindent
The proof is presented in the Supplementary Section S7.5.\\

\noindent
Under correct specification of both nuisance parameter models, the property of Neyman orthogonality is satisfied, ensuring that the influence of first-order perturbations in the nuisance estimators vanishes. Specifically, the Gâteaux derivative of the estimating function \(\boldsymbol{g}^{\text{mult}}\) with respect to the nuisance parameters is zero in expectation when evaluated at their true values. In the parametric setting, this condition corresponds to \(G_{\boldsymbol{\alpha}^{\text{mult}*}, \boldsymbol{\gamma}^{\text{mult}*}} = 0\), effectively eliminating the nuisance term from the variance expression.
Intuitively, Neyman orthogonality arises from the structure of the estimating equation \(\boldsymbol{g}\), where its constituent components are constructed to offset errors introduced by nuisance parameter estimation. The auxiliary function \(\boldsymbol{f}\) captures the conditional expectation of the unweighted score, yielding a residual with mean zero. Simultaneously, the internal and external sample weights are designed to align with a common target population. Consequently, small deviations in either the selection model \(\pi\) or the auxiliary function \(\boldsymbol{f}\) do not impact the first-order behavior of the estimating function, thereby preserving valid inference.
As a result, when both nuisance models are correctly specified, the empirical estimator of the asymptotic variance \(V\), constructed via cross-fitting, is consistent. However, if either the selection model or the auxiliary score model is misspecified, the orthogonality property fails to hold, and the resulting variance estimator may be biased. The details of this variance estimator are presented in Supplementary Section S7.6.

\section{Simulation Study}\label{sec:simu}
\subsection{Basic Setup}
We simulated $K=3$ cohorts from a target population of 50,000 individuals. Using three disease model covariates $\boldsymbol{Z} = (Z_1, Z_2, Z_3)$, the outcome $D$ was generated via the logistic model: $D \mid Z_1, Z_2, Z_3 \sim \text{Ber}\left(\frac{e^{\theta_0 + \theta_1 Z_1 + \theta_2 Z_2 + \theta_3 Z_3}}{1 + e^{\theta_0 + \theta_1 Z_1 + \theta_2 Z_2 + \theta_3 Z_3}}\right)$, with \(\theta_0 = -2\), \(\theta_1 = 0.35\), \(\theta_2 = 0.45\), and \(\theta_3 = 0.25\). $\boldsymbol Z_{1k},\boldsymbol Z_{2k}$ varied by cohort, e.g., $\boldsymbol Z_{11}=Z_1, \boldsymbol Z_{12}=(Z_1,Z_2), \boldsymbol Z_{13}=(Z_1,Z_3)$, $\boldsymbol Z_{21}=(Z_2,Z_3), \boldsymbol Z_{22}=Z_3, \boldsymbol Z_{23}=Z_2$, and \( W_1,W_2,W_3 \) were simulated as conditional on \( \boldsymbol{Z} \) and \( D \). We examined two simulation scenarios differing in the complexity of the selection model \(\pi(\boldsymbol{X}^{\text{mult}})\). The first scenario included only main effects of the selection variables, while the second introduced interaction terms in the selection models. We evaluated both correct and incorrect specifications of the selection and auxiliary score models. In both scenarios, \(S_1\) depended on \(Z_2, Z_3, W_1\), and \(D\); \(S_2\) on \(Z_3, W_2\), and \(D\); and \(S_3\) on \(Z_2\) and \(W_3\). In the second scenario, interaction terms were added between the selection variables. We implemented unweighted logistic regression (with and without cohort-specific intercepts), all joint IPW methods, and the proposed DR JAIPW with a flexible machine-learning based auxiliary score model. Further simulation details are available in Supplementary Section S8.1.

\subsection{Assessing Robustness under Different Degrees of Model Misspecification}
\textbf{Selection Model:} For the IPW methods (JPL, JSR and JCL), we estimated the selection model under four simulation scenarios: (1) all cohort models correctly specified, (2) one model misspecified (Cohort 3), (3) two models misspecified (Cohorts 2 and 3), and (4) all models misspecified. For JAIPW, we only considered either all selection  models correctly specified or all the selection models misspecified. In Setup 1 (main effects only, no interactions), misspecification involved ignoring \(D\) in the first two cohorts and \(Z_2\) in the third. In Setup 2, all interaction terms were ignored. For JPS, we used two approaches: assuming the exact joint distribution of discretized selection variables or approximating joint probabilities using the product of conditionals. The details of JPS are provided in Supplementary Section S8.2. \\
\noindent
\textbf{Auxiliary Score Model}: In both simulation scenarios, \(\boldsymbol Z_{1\cap} = Z_1\). To estimate the auxiliary score component, we used XGBoost. This method was fitted with and without the inclusion of the disease indicator \(D\), representing correct and incorrect specifications of the auxiliary score model.
\subsection{Evaluation Metrics for Comparing Methods}
Across $R=500$ simulation replications, we compared the unweighted logistic regression (with and without cohort-specific intercepts), joint IPW methods, and JAIPW across all setups. Performance for $\theta_1, \theta_2, \text{and } \theta_3$ was assessed using Relative Percentage Bias (RBP), defined as $[\frac{1}{R}\sum_{r=1}^R (\widehat{\theta}_r - \theta)] / \theta \times 100$, and Relative Mean Squared Error (RMSE). The RMSE was calculated relative to the naive unweighted estimator ($\widehat{\theta}_{\text{naive}}$) as $\sum_{r=1}^R (\widehat{\theta}_r - \theta)^2 / \sum_{r=1}^R (\widehat{\theta}_{\text{naive},r} - \theta)^2$. Additionally, we assessed the accuracy of the variance estimators by examining Coverage Probabilities (CP) and comparing the Estimated Standard Error (ESE) defined as the average of the proposed standard errors against the Monte Carlo Standard Error (MCSE) of the estimates across simulation runs. Moreover, Monte Carlo standard errors for each metric were calculated to assess simulation precision.

\subsection{Simulation Results}
\subsubsection{Parameter Estimation}
Tables~\ref{tab:table_r1} and~\ref{tab:table_r2} summarize the parameter estimation performance across both setups.\\

\noindent
\textbf{Unweighted Methods:} Under both setups, the unweighted logistic regression without cohort-specific intercepts exhibited substantial relative biases, exceeding -30\% for both \(\theta_2\) and \(\theta_3\). Including cohort-specific intercepts, further inflated both bias and RMSE, particularly for \(\theta_2\).
\par
\textbf{Individual-Level IPW Methods:}  
Under Setup 1 with correctly specified selection models, both JPL and JSR performed well with low bias and RMSE. For example, JPL exhibited bias percentages around $\pm 1$\%, and JSR achieved similar performance with RMSEs not exceeding 0.05. However, under Setup 2 with interaction terms, JPL exhibited increased variance. For instance, the RMSE for \(\theta_1\) under JPL rose to 0.92 compared to 0.28 under JSR. This increased variance is attributable to the inclusion of interaction terms in JPL. As more selection models were misspecified, both bias and RMSE increased for both methods. Under all models' misspecification in Setup 2, RMSEs for \(\theta_1\) and \(\theta_3\) were 1.62 and 1.79 for JPL, compared to 1.39 and 1.39 for JSR, respectively.
\par
\textbf{Summary-Level IPW Methods:}  
Under Setup 1, JPS with exact joint distributions showed modest bias (-7\%) and RMSE (0.15). JPS with approximate joint probabilities, yielded higher bias (-21\%) and RMSE (0.32). Under Setup 2, both methods showed increased bias and RMSE. JCL followed a pattern similar to JSR and JPL. 
\par
\textbf{JAIPW:} Under Setup 1, JAIPW demonstrated very low bias and RMSE when either all selection models were correctly specified or the auxiliary score model was correctly specified. The maximum observed bias was around -2\%, with most bias percentages falling below $\pm 1$\%. Corresponding RMSEs were all below 0.1, indicating high estimation accuracy. However, when both nuisance parameter models were misspecified, bias and RMSE increased notably. Under Setup 2, a similar trend was observed, with slightly elevated RMSEs due to the added complexity from interaction terms. For example, the RMSE for \(\theta_1\) increased modestly to 0.39. Notably, under scenarios where the selection model was misspecified but the auxiliary score model was correctly specified, JAIPW achieved up to six times lower relative bias and five times lower root mean square error (RMSE) compared to the best-performing joint IPW methods, highlighting its superior robustness and efficiency. 
\par
\textbf{Variance Estimation:} 
Under both simulation setups, the unweighted method performed poorly, yielding near-zero coverage probabilities due to its extremely high bias. The joint IPW methods (JPL, JSR, and JCL) achieved coverage probabilities above 90\% when all selection models were correctly specified. However, their coverage dropped to near zero under misspecification of the selection models for all the cohorts. In contrast, the JAIPW method demonstrated its double-robustness property. In Setup 1, it maintained the nominal 95\% coverage rate even when only one of the two nuisance models (either the overall selection model or the auxiliary score model) was correctly specified. A similar pattern was observed in Setup 2. In order to assess the validity of the proposed variance estimators, we compared the average standard errors across simulation runs with the Monte Carlo standard errors of the point estimates across simulation runs. From Supplementary Tables S1 and S2, we observed that these values were closely aligned in most cases, confirming the validity of the proposed variance estimators.
\par
\textbf{Summary Takeaways:}
Our simulations reveal that unweighted logistic regression performs poorly under selection bias, exhibiting high bias and RMSE. Joint individual-level IPW methods such as JPL and JSR perform well under correctly specified selection models in Setup 1, but show increased variance and RMSE in Setup 2 with interaction terms—especially for JPL. JCL performs comparably to JSR and JPL. While JPL, JSR, and JCL outperform JPS when selection models are correct. In contrast, JAIPW consistently achieves low bias and RMSE whenever either the selection or auxiliary model is correctly specified, demonstrating its double robustness. Across both setups, JAIPW achieves up to six times lower bias and five times lower RMSE than the best-performing joint IPW methods under selection model misspecification. It also maintains approximately 95\% coverage when either nuisance model is correctly specified in Setup 1, with similar performance in Setup 2 except when only the selection model is correct, where coverage drops to around 90\%. This drop likely reflects instability in JPL under interactions and highlights that the DML-based variance estimator is theoretically valid only under correct specification of both nuisance models, necessitating bootstrap inference otherwise. In Tables \ref{tab:table_r1} and \ref{tab:table_r2}, the Monte Carlo standard errors for each performance metric were observed to be very small relative to the magnitude of the estimate, indicating that the mean estimates across simulation runs are precise and the number of replications is adequate.

\section{Data Example - Multi clinic recruitment in MGI}\label{sec:data}

\subsection{Introduction}
In this section, we use data from MGI, a large biobank that links biosamples to EHRs at the University of Michigan. Since 2012, MGI has recruited over 100,000 participants across six specialized clinics: MGI Anesthesiology, MIPACT, MEND, MHB, PROMPT, and MY PART. We conduct two analyses using MGI data. Our target population of interest for both data examples is the 2018 United States adult population. This definition is based on our use of the NHANES 2017-18 as an external probability sample, which provides representative data for our selection variables. The post-stratification and calibration weights were calculated using age-specific and marginal statistics from SEER, the U.S. Census, and the CDC. In the first analysis, we investigate the association between cancer (\(D\)) and biological sex (\(Z_1\)), both unadjusted and adjusted for age (\(Z_2\)). The unadjusted log-odds ratio is compared to benchmark SEER data (2008-2016), which show lower cancer risk for women than men, with marginal log-odds ratios ranging from -0.24 to -0.07.
In the second analysis, we estimate the age-adjusted association (\(Z_2\)) between cancer (\(D\)) and the Polygenic Risk Scores (PRS) for having any cancer based on summary data from the UK Biobank (\(Z_1\)) (\url{https://www.pgscatalog.org/score/PGS000356/}). Unlike in Analysis 1, no external reference data (such as SEER) is available to compare method performance for estimating the association of PRS with cancer. However, this example illustrates a scenario where external level information on \(Z_1\) (cancer PRS in this case) is not available in NHANES.
\subsection{Data Descriptive Statistics}
\textbf{Analysis 1:} For this analysis, we used data from four MGI cohorts (MGI Anesthesiology, MIPACT, MEND, MHB) collected up to August 2022. The MY PART and PROMPT cohorts were excluded due to small sample sizes, leaving a total of 80,371 participants: 69,294 from MGI Anesthesiology, 5,821 from MIPACT, 3,303 from MEND, and 1,953 from MHB. Cancer prevalence varied across cohorts: 52.2\% in MGI Anesthesiology, 24.7\% in MIPACT, 35.8\% in MEND, and 16.8\% in MHB. The distribution of biological sex was relatively consistent, except for MHB, where females comprised 62.7\% of the cohort. The heterogeneity in descriptive statistics across relevant variables for each cohort is summarized in details in Supplementary Table S3.\\

\noindent
\textbf{Analysis 2:} In this analysis, we use genotyped MGI data from 2018, comprising 38,360 participants. The analysis employs the LASSOSUM polygenic risk score (PRS), which demonstrated the highest predictive power in previous studies \citep{fritsche2020cancer}. To select the tuning parameters for LASSOSUM PRS, the data were divided into training and testing sets, with parameter selection performed exclusively on the training set. To prevent overfitting, the final analysis was conducted on the testing set, which consisted of 15,291 individuals. Since this dataset is a subset of the 2018 MGI cohort, the analysis focuses on the MGI Anesthesiology cohort. The prevalence of cancer in this sample is approximately 68\%.
\subsection{Specifics of Model Fitting}
\textbf{Analysis 1:} For each MGI cohort, distinct covariates (\(\boldsymbol{W}_K\)) were linked to different selection processes (see Supplementary Figure S1). We applied the four joint IPW methods under two scenarios: one excluding and one including cancer $(D)$ in the selection models. Initially, selection weights did not account for cancer due to its low prevalence in NHANES, but an adjustment factor was later introduced, as per \citet{kundu2024framework}. The JAIPW method was also implemented in four forms, considering both the inclusion and exclusion of cancer in the selection and auxiliary score models. Biological sex was used as \(\boldsymbol{Z}_{1\cap}\), assuming no influence on selection conditional the covariates $\boldsymbol X^{\text{mult}}$ to satisfy Condition C3. The rationale behind the choice of biological sex as \(\boldsymbol{Z}_{1\cap}\) is provided in Supplementary Section S9.2.
XGBoost was used to model the conditional relationship between biological sex and other covariates in constructing the auxiliary score.\\

\noindent
\textbf{Analysis 2:} For this analysis, we applied the unweighted and IPW methods designed for a single cohort. For the AIPW method, we used the single-cohort AIPW estimator described in Section \ref{sec:aipws}.
Supplementary Figure S2 shows the DAG relationships between disease and selection variables, where Cancer PRS is the new $Z_1$. The rationale behind the choice of Cancer PRS as \(\boldsymbol{Z}_{1}\) is provided in Supplementary Section S9.2. We constructed selection model similarly as Analysis 1 for both IPW and AIPW methods. For the auxiliary model, we tested both XGBoost and linear regression to model the relationship between Cancer PRS and the selection variables (\(\boldsymbol{X}\)). Since both methods produced similar results, we chose linear regression for its lower computational time. 
\subsection{Results}
\subsubsection{Analysis 1}\label{sec:cansex}
The benchmark national estimates from SEER show a lower lifetime cancer risk among women, with marginal log-odds ratios of -0.24 (2008-2010), -0.19 (2010-2012), -0.08 (2012-2014), and -0.07 (2014-2016). Here, we estimate the marginal log odds ratio (OR) between cancer and biological sex. Figure \ref{fig:realdata} presents the results for the Anesthesiology, MIPACT, MEND, and Combined cohorts using Unweighted logistic regression, JPL, JPS, and JAIPW. Results for the JCL and JSR estimators are included in Supplementary Figure S3. For the MHB cohort, results for all methods are shown in Supplementary Figure S3. Age-adjusted analyses are provided in Supplementary Section S9.3.\\

\noindent
\textbf{Anesthesiology:}
Unweighted logistic regression estimates the log OR to be 0.08 (95\% CI [-0.11, -0.05]),  where the upper limit is slightly outside the SEER range. Including cancer in the selection model improved estimates for PL, SR, PS, and CL. The estimates derived from the four IPW methods: PL, SR, PS, and CL are -0.15 (95\% CI [-0.19, -0.11]), -0.16 (95\% CI [-0.23, -0.10]), -0.13 (95\% CI [-0.17, -0.08]), and -0.13 (95\% CI [-0.17, -0.09]) respectively. Conversely, when cancer is not included in the selection model, the estimates for all methods exhibit bias in the opposite direction. For the AIPW method, estimates were -0.13 (95\% CI [-0.16, -0.09]) and -0.15 (95\% CI [-0.18, -0.11]) when cancer was included in both models and with inclusion/exclusion in the selection and auxiliary score model, respectively. When the selection model excluded cancer with the auxiliary score model including it, the result was -0.10 (95\% CI [-0.15, -0.06]). Hence the AIPW method, even when forced to use this misspecified ``no cancer" selection model was still able to produce a correct estimate that matched the SEER benchmark, when ``with cancer" auxiliary score model was used.\\

\noindent
\textbf{MIPACT:} Unweighted logistic regression estimated 0.33 (95\% CI [0.21, 0.45]), biased outside the SEER range. None of the IPW or AIPW estimates fell within the SEER range due to missing critical selection variables like income and education but helped to reduce the bias to some extent when cancer was included in either selection or the auxiliary score model.\\

\noindent
\textbf{MEND:} Unweighted logistic regression estimated 0.09 (95\% CI [-0.06, 0.23]). IPW methods showed bias opposite to national data, but PS yielded closer estimates (-0.08, 95\% CI [-0.27, 0.12]). AIPW performed better, with mean estimates within the SEER range when the auxiliary model excluded cancer (-0.13, 95\% CI [-0.24, -0.03]; -0.10, 95\% CI [-0.22, 0.02]). None of the 95\% CIs fully aligned with SEER data due to small sample sizes increasing variance.\\

\noindent
\textbf{Combined:} Unweighted logistic regression estimated -0.04 (95\% CI [-0.07, -0.01]), outside the SEER range. Including cancer in the selection model, estimates from JPL, JSR, JPS, and JCL were -0.08 (95\% CI [-0.13, -0.03]), -0.08 (95\% CI [-0.15, -0.01]), -0.10 (95\% CI [-0.13, -0.07]), and -0.08 (95\% CI [-0.11, -0.04]), respectively. However, the upper bounds for JPL, JCL, and JSR did not align within the SEER range due to variability. Excluding cancer from the selection model led to opposite biases in estimates for all IPW methods. JAIPW estimates, when cancer was included in the selection model, had 95\% CIs within the SEER range: -0.13 (95\% CI [-0.16, -0.10]) and -0.10 (95\% CI [-0.13, -0.07]). Excluding cancer in the selection model but including it in the auxiliary model resulted in estimates consistent with expectations, though the 95\% CI did not fully align with the SEER range due to bias induced by MIPACT.\\

\noindent
\textbf{MHB:} Unweighted logistic regression showed high bias with an estimate of 0.61 (95\% CI [0.34, 0.88]). IPW and AIPW methods reduced bias when accounting for cancer. CL with cancer included gave the closest estimate to SEER data (-0.05, 95\% CI [-0.77, 0.68]). High variance in estimates, likely due to the small MHB sample size, led to wider confidence intervals and less precise estimates.\\

\subsubsection{Analysis 2}\label{sec:prsresult}
Supplementary Figure S5 illustrates the age-adjusted association analysis between cancer and the cancer polygenic risk score (PRS). The unweighted method estimated a log odds ratio of 0.14 (95\% CI [0.09, 0.19]) with one unit change in Interquartile Range (IQR) of PRS. In comparison, the JIPW and JAIPW methods that incorporated cancer in their models consistently produced higher point estimates compared to those excluding cancer. Notably, even when cancer was included solely in the auxiliary score model and not in the selection model, the JAIPW method estimated a log odds ratio of 0.19 (95\% CI [0.09, 0.28]) with one unit change in IQR of PRS. This estimate closely aligned with those obtained from the JAIPW and JIPW methods that incorporated cancer in the selection models. These results highlight the robustness of the JAIPW method in delivering reliable estimates, comparable to those from correctly specified JIPW models, and its effectiveness in addressing potential model misspecifications.

\section{Discussion}
In this paper, we introduced the JAIPW method, designed to estimate disease model parameters while addressing outcome-dependent selection biases from overlapping cohorts. JAIPW shows double robustness by producing consistent estimates even when selection models are misspecified, provided a correctly specified auxiliary score model is used.\\

\noindent
\textbf{Limitations of Existing Methods:} Individual-level methods like JPL and JSR struggle with misspecified selection models, while JPS and JCL are always restricted by the limited availability of summary data from external probability sample/target population. JPS fails when exact joint distributions for numerous selection variables cannot be obtained. JCL finds it challenging to specify selection models accurately using only summary statistics on the marginal distribution of the selection variables.\\

\noindent
\textbf{Advantages of JAIPW:} JAIPW consistently estimates disease model parameters despite all selection models being misspecified if the auxiliary score model is correct. Applying JAIPW to the MGI data on cancer and biological sex demonstrates its superiority over traditional IPW methods in cases with potential selection model inaccuracies. In applying the JAIPW method to the cancer-PRS association, we found that the method remains effective even when PRS data is unavailable in the external probability dataset.\\

\noindent
\textbf{Limitations of the JAIPW method:} JAIPW's efficiency diminishes when crucial selection variables are missing, compromising the method’s ability to adjust for selection bias fully as observed in the case of MIPACT. If key variables influencing selection are absent, the Condition C3 of conditional independence between \(\boldsymbol{Z}_{1\cap}\) and \(S^{\text{mult}}\) given \(\boldsymbol{X}^{\text{mult}}\) fails. Methods to assess and quantify bias, such as those by \citet{west2021assessing} and \citet{neuhaus1999bias}, could be explored when facing missing selection variables.\\

\noindent
\textbf{Future Directions:} Meta-analysis is a standard statistical technique that provides an alternative approach for combining estimates from multiple studies aiming to produce a single pooled and more precise estimate. Fixed-effects and random-effects models are two common approaches for this purpose \citep{borenstein2010basic}. However, standard meta-analysis methods are more relevant for synthesizing evidence and may face challenges when applied to studies with small cohort sizes or rare disease prevalence. A potential avenue for future research is to integrate the proposed AIPW estimator with random-effects models \citep{han2011random} to better account for other sources of potential heterogeneity across different cohorts. \citet{dahabreh2020towards, dahabreh2023efficient} introduce a class of methods known as causally interpretable meta-analysis, which integrate information from multiple randomized trials to facilitate causal inference for a target population of substantive interest. Extending the idea of JAIPW in the context of causally interpretable meta-analysis might be an interesting direction.
\citet{wu2022statistical} proposes a non-parametric extension of the PL method using kernel smoothing estimators. \citet{chen2023dealing} proposes calibrated IPW and split population approaches to address stochastic or deterministic under-coverage. However, implementing such estimators poses significant challenges in our case of multi-dimensional selection variables. Adapting the density ratio function for estimating selection weights in data missing at random (MAR) contexts, as suggested by \citet{wang2021information}, is a potential direction. For scenarios where individual-level data cannot be shared across the cohorts, federated learning offers a privacy-preserving solution. This approach, which develops models using decentralized data without exchanging individual data points, can address selection bias while respecting privacy constraints. Research by \citet{luo2020renewable}, \citet{dang2022federated}, \citet{brisimi2018federated}, and \citet{jordan2018communication} provides a foundation for federated learning methodologies. Future work will focus on adapting these principles to mitigate selection bias in multi-center studies where data sharing is restricted. Finally, data from EHRs is subject to numerous biases beyond selection bias, such as misclassification, missing data, clinically informative patient encounter processes, confounding, lack of consistent data harmonization across cohorts, true heterogeneity of the studied populations, and alike. Developing robust strategies to simultaneously address a confluence of biases remains a major challenge.

\section*{Software}
Michigan Genomics Initiative Data are available after institutional review board approval to select researchers ( \url{https://precisionhealth.umich.edu/our-research/michigangenomics/}).
The link to the CRAN R-package EHRmuse can be found here, \url{https://CRAN.R-project.org/package=EHRmuse}.

\section*{Acknowledgments}
This work was supported through grant DMS1712933 from the National Science Foundation and MI-CARES grant 1UG3CA267907 from the National Cancer Institute. 
Data collection adhered to the Declaration of Helsinki principles. The University of Michigan Medical School Institutional Review Board reviewed and approved the consent forms and protocols of MGI study participants (IRB ID HUM00099605 and HUM00155849). Opt-in written informed consent was obtained. The authors acknowledge the Michigan Genomics Initiative participants, Precision Health at the University of Michigan, the University of Michigan Medical School Central Biorepository, the University of Michigan Medical School Data Office for Clinical and Translational Research, and the University of Michigan Advanced Genomics Core for providing data and specimen storage, management, processing, and distribution services, and the Center for Statistical Genetics in the Department of Biostatistics at the School of Public Health for genotype data curation, imputation, and management in support of the research reported in this work. 

\bibliographystyle{abbrvnat} 
\bibliography{references.bib}

@article{chen2020doubly,
  title={Doubly robust inference with nonprobability survey samples},
  author={Chen, Yilin and Li, Pengfei and Wu, Changbao},
  journal={Journal of the American Statistical Association},
  volume={115},
  number={532},
  pages={2011--2021},
  year={2020},
  publisher={Taylor \& Francis}
}

@article{yang2020doubly,
  title={Doubly robust inference when combining probability and non-probability samples with high dimensional data},
  author={Yang, Shu and Kim, Jae Kwang and Song, Rui},
  journal={Journal of the Royal Statistical Society: Series B (Statistical Methodology)},
  volume={82},
  number={2},
  pages={445--465},
  year={2020},
  publisher={Wiley Online Library}
}

@article{beesley2022statistical,
  title={Statistical inference for association studies using electronic health records: handling both selection bias and outcome misclassification},
  author={Beesley, Lauren J and Mukherjee, Bhramar},
  journal={Biometrics},
  volume={78},
  number={1},
  pages={214--226},
  year={2022},
  publisher={Wiley Online Library}
}

@article{tsiatis2006semiparametric,
  title={Semiparametric theory and missing data},
  author={Tsiatis, Anastasios A},
  year={2006},
  publisher={Springer}
}

@article{barndorff1991some,
  title={Some parametric models on the simplex},
  author={Barndorff-Nielsen, Ole E and J{\o}rgensen, Bent},
  journal={Journal of multivariate analysis},
  volume={39},
  number={1},
  pages={106--116},
  year={1991},
  publisher={Elsevier}
}

@article{wu2003optimal,
  title={Optimal calibration estimators in survey sampling},
  author={Wu, Changbao},
  journal={Biometrika},
  volume={90},
  number={4},
  pages={937--951},
  year={2003},
  publisher={Oxford University Press}
}

@article{zawistowski2021michigan,
  title={The Michigan Genomics Initiative: a biobank linking genotypes and electronic clinical records in {M}ichigan {M}edicine patients},
  author={Zawistowski, Matthew and Fritsche, Lars G and Pandit, Anita and Vanderwerff, Brett and Patil, Snehal and Scmidt, Ellen M and VanderHaar, Peter and Brummett, Chad M and Keterpal, Sachin and Zhou, Xiang and others},
  journal={medRxiv},
  year={2021},
  publisher={Cold Spring Harbor Laboratory Press}
}

@article{haneuse2016general,
  title={A general framework for considering selection bias in {EHR}-based studies: what data are observed and why?},
  author={Haneuse, Sebastien and Daniels, Michael},
  journal={eGEMs},
  volume={4},
  number={1},
  year={2016},
  publisher={Ubiquity Press}
}

@article{west2021assessing,
  title={Assessing selection bias in regression coefficients estimated from nonprobability samples with applications to genetics and demographic surveys},
  author={West, Brady T and Little, Roderick J and Andridge, Rebecca R and Boonstra, Philip S and Ware, Erin B and Pandit, Anita and Alvarado-Leiton, Fernanda},
  journal={The annals of applied statistics},
  volume={15},
  number={3},
  pages={1556--1581},
  year={2021},
  publisher={Institute of Mathematical Statistics}
}

@article{wang2021information,
  title={Information projection approach to propensity score estimation for handling selection bias under missing at random},
  author={Wang, Hengfang and Kim, Jae Kwang},
  journal={arXiv e-prints},
  pages={arXiv--2104},
  year={2021}
}

@article{neuhaus1999bias,
  title={Bias and efficiency loss due to misclassified responses in binary regression},
  author={Neuhaus, John M},
  journal={Biometrika},
  volume={86},
  number={4},
  pages={843--855},
  year={1999},
  publisher={Oxford University Press}
}

@article{fu2020assessment,
  title={Assessment of the impact of {EHR} heterogeneity for clinical research through a case study of silent brain infarction},
  author={Fu, Sunyang and Leung, Lester Y and Raulli, Anne-Olivia and Kallmes, David F and Kinsman, Kristin A and Nelson, Kristoff B and Clark, Michael S and Luetmer, Patrick H and Kingsbury, Paul R and Kent, David M and others},
  journal={BMC medical informatics and decision making},
  volume={20},
  number={1},
  pages={1--12},
  year={2020},
  publisher={BioMed Central}
}

@article{beesley2020analytic,
  title={An analytic framework for exploring sampling and observation process biases in genome and phenome-wide association studies using electronic health records},
  author={Beesley, Lauren J and Fritsche, Lars G and Mukherjee, Bhramar},
  journal={Statistics in Medicine},
  volume={39},
  number={14},
  pages={1965--1979},
  year={2020},
  publisher={Wiley Online Library}
}

@article{holt1979post,
  title={Post stratification},
  author={Holt, David and Smith, TM Fred},
  journal={Journal of the Royal Statistical Society: Series A (General)},
  volume={142},
  number={1},
  pages={33--46},
  year={1979},
  publisher={Wiley Online Library}
}

@article{bradley2021unrepresentative,
  title={Unrepresentative big surveys significantly overestimated {US} vaccine uptake},
  author={Bradley, Valerie C and Kuriwaki, Shiro and Isakov, Michael and Sejdinovic, Dino and Meng, Xiao-Li and Flaxman, Seth},
  journal={Nature},
  volume={600},
  number={7890},
  pages={695--700},
  year={2021},
  publisher={Nature Publishing Group UK London}
}

@article{kaplan2014big,
  title={Big data and large sample size: a cautionary note on the potential for bias},
  author={Kaplan, Robert M and Chambers, David A and Glasgow, Russell E},
  journal={Clinical and translational science},
  volume={7},
  number={4},
  pages={342--346},
  year={2014},
  publisher={Wiley Online Library}
}

@article{dahabreh2020towards,
  title={Towards causally interpretable meta-analysis: transporting inferences from multiple randomized trials to a new target population},
  author={Dahabreh, Issa J and Petito, Lucia C and Robertson, Sarah E and Hern{\'a}n, Miguel A and Steingrimsson, Jon A},
  journal={Epidemiology (Cambridge, Mass.)},
  volume={31},
  number={3},
  pages={334},
  year={2020},
  publisher={NIH Public Access}
}

@article{borenstein2010basic,
  title={A basic introduction to fixed-effect and random-effects models for meta-analysis},
  author={Borenstein, Michael and Hedges, Larry V and Higgins, Julian PT and Rothstein, Hannah R},
  journal={Research synthesis methods},
  volume={1},
  number={2},
  pages={97--111},
  year={2010},
  publisher={Wiley Online Library}
}

@article{dang2022federated,
  title={Federated learning for electronic health records},
  author={Dang, Trung Kien and Lan, Xiang and Weng, Jianshu and Feng, Mengling},
  journal={ACM Transactions on Intelligent Systems and Technology (TIST)},
  volume={13},
  number={5},
  pages={1--17},
  year={2022},
  publisher={ACM New York, NY}
}

@article{brisimi2018federated,
  title={Federated learning of predictive models from federated electronic health records},
  author={Brisimi, Theodora S and Chen, Ruidi and Mela, Theofanie and Olshevsky, Alex and Paschalidis, Ioannis Ch and Shi, Wei},
  journal={International journal of medical informatics},
  volume={112},
  pages={59--67},
  year={2018},
  publisher={Elsevier}
}

@article{robins1994estimation,
  title={Estimation of regression coefficients when some regressors are not always observed},
  author={Robins, James M and Rotnitzky, Andrea and Zhao, Lue Ping},
  journal={Journal of the American statistical Association},
  volume={89},
  number={427},
  pages={846--866},
  year={1994},
  publisher={Taylor \& Francis}
}

@article{scharfstein1999adjusting,
  title={Adjusting for nonignorable drop-out using semiparametric nonresponse models},
  author={Scharfstein, Daniel O and Rotnitzky, Andrea and Robins, James M},
  journal={Journal of the American Statistical Association},
  volume={94},
  number={448},
  pages={1096--1120},
  year={1999},
  publisher={Taylor \& Francis}
}

@article{williamson2012doubly,
  title={Doubly robust estimators of causal exposure effects with missing data in the outcome, exposure or a confounder},
  author={Williamson, Elizabeth Jane and Forbes, A and Wolfe, R},
  journal={Statistics in medicine},
  volume={31},
  number={30},
  pages={4382--4400},
  year={2012},
  publisher={Wiley Online Library}
}

@article{jordan2018communication,
  title={Communication-efficient distributed statistical inference},
  author={Jordan, Michael I and Lee, Jason D and Yang, Yun},
  journal={Journal of the American Statistical Association},
  year={2018},
  publisher={Taylor \& Francis}
}

@article{luo2020renewable,
  title={Renewable estimation and incremental inference in generalized linear models with streaming data sets},
  author={Luo, Lan and Song, Peter X-K},
  journal={Journal of the Royal Statistical Society Series B: Statistical Methodology},
  volume={82},
  number={1},
  pages={69--97},
  year={2020},
  publisher={Oxford University Press}
}

@article{du2024doubly,
  title={Doubly robust causal inference through penalized bias-reduced estimation: combining non-probability samples with designed surveys},
  author={Du, Jiacong and Shi, Xu and Zeng, Donglin and Mukherjee, Bhramar},
  journal={arXiv preprint arXiv:2403.18039},
  year={2024}
}

@article{wu2022statistical,
  title={Statistical inference with non-probability survey samples},
  author={Wu, Changbao},
  journal={Survey Methodology},
  volume={48},
  number={2},
  pages={283--311},
  year={2022},
  publisher={Statistics Canada}
}

@article{chen2023dealing,
  title={Dealing with undercoverage for non-probability survey samples},
  author={Chen, Yilin and Li, Pengfei and Wu, Changbao},
  journal={Survey Methodology},
  volume={49},
  number={2},
  year={2023},
  publisher={STATISTICS CANADA 100 TUNNEYS PASTURE DRIVEWAY, OTTAWA, ONTARIO K1A 0T6, CANADA}
}

@article{salvatore2024weight,
  title={To weight or not to weight? The effect of selection bias in 3 large electronic health record-linked biobanks and recommendations for practice},
  author={Salvatore, Maxwell and Kundu, Ritoban and Shi, Xu and Friese, Christopher R and Lee, Seunggeun and Fritsche, Lars G and Mondul, Alison M and Hanauer, David and Pearce, Celeste Leigh and Mukherjee, Bhramar},
  journal={Journal of the American Medical Informatics Association},
  pages={ocae098},
  year={2024},
  publisher={Oxford University Press}
}

@article{leese2023clinical,
  title={Clinical encounter heterogeneity and methods for resolving in networked EHR data: a study from N3C and RECOVER programs},
  author={Leese, Peter and Anand, Adit and Girvin, Andrew and Manna, Amin and Patel, Saaya and Yoo, Yun Jae and Wong, Rachel and Haendel, Melissa and Chute, Christopher G and Bennett, Tellen and others},
  journal={Journal of the American Medical Informatics Association},
  volume={30},
  number={6},
  pages={1125--1136},
  year={2023},
  publisher={Oxford University Press}
}

@article{kundu2024framework,
  title={A framework for understanding selection bias in real-world healthcare data},
  author={Kundu, Ritoban and Shi, Xu and Morrison, Jean and Barrett, Jessica and Mukherjee, Bhramar},
  journal={Journal of the Royal Statistical Society Series A: Statistics in Society},
  pages={qnae039},
  year={2024},
  publisher={Oxford University Press UK}
}

@article{fritsche2020cancer,
  title={Cancer PRSweb: an online repository with polygenic risk scores for major cancer traits and their evaluation in two independent biobanks},
  author={Fritsche, Lars G and Patil, Snehal and Beesley, Lauren J and VandeHaar, Peter and Salvatore, Maxwell and Ma, Ying and Peng, Robert B and Taliun, Daniel and Zhou, Xiang and Mukherjee, Bhramar},
  journal={The American Journal of Human Genetics},
  volume={107},
  number={5},
  pages={815--836},
  year={2020},
  publisher={Elsevier}
}

@article{schoeler2023participation,
  title={Participation bias in the UK Biobank distorts genetic associations and downstream analyses},
  author={Schoeler, Tabea and Speed, Doug and Porcu, Eleonora and Pirastu, Nicola and Pingault, Jean-Baptiste and Kutalik, Zolt{\'a}n},
  journal={Nature Human Behaviour},
  volume={7},
  number={7},
  pages={1216--1227},
  year={2023},
  publisher={Nature Publishing Group UK London}
}

@article{liu2025superpopulation,
  title={Superpopulation model inference for non probability samples under informative sampling with high-dimensional data},
  author={Liu, Zhan and Wang, Dianni and Pan, Yingli},
  journal={Communications in Statistics-Theory and Methods},
  volume={54},
  number={5},
  pages={1370--1390},
  year={2025},
  publisher={Taylor \& Francis}
}

@article{dahabreh2023efficient,
  title={Efficient and robust methods for causally interpretable meta-analysis: Transporting inferences from multiple randomized trials to a target population},
  author={Dahabreh, Issa J and Robertson, Sarah E and Petito, Lucia C and Hern{\'a}n, Miguel A and Steingrimsson, Jon A},
  journal={Biometrics},
  volume={79},
  number={2},
  pages={1057--1072},
  year={2023},
  publisher={Wiley Online Library}
}

@article{han2011random,
  title={Random-effects model aimed at discovering associations in meta-analysis of genome-wide association studies},
  author={Han, Buhm and Eskin, Eleazar},
  journal={The American Journal of Human Genetics},
  volume={88},
  number={5},
  pages={586--598},
  year={2011},
  publisher={Elsevier}
}

@article{rubin1976inference,
  title={Inference and missing data},
  author={Rubin, Donald B},
  journal={Biometrika},
  volume={63},
  number={3},
  pages={581--592},
  year={1976},
  publisher={Oxford University Press}
}

@book{little2019statistical,
  title={Statistical analysis with missing data},
  author={Little, Roderick JA and Rubin, Donald B},
  year={2019},
  publisher={John Wiley \& Sons}
}

@misc{chernozhukov2018double,
  title={Double/debiased machine learning for treatment and structural parameters},
  author={Chernozhukov, Victor and Chetverikov, Denis and Demirer, Mert and Duflo, Esther and Hansen, Christian and Newey, Whitney and Robins, James},
  year={2018},
  publisher={Oxford University Press Oxford, UK}
}

\newpage
\begin{table}[H]
\renewcommand{\arraystretch}{2.2}
\centering
\begin{adjustbox}{max width=0.9\textwidth}
\begin{tabular}{|cccccccccccccc|}
\hline
\textbf{Method}                 & \textbf{\begin{tabular}[c]{@{}c@{}}Selection \\ Model 1\end{tabular}} & \textbf{\begin{tabular}[c]{@{}c@{}}Selection \\ Model 2\end{tabular}} & \textbf{\begin{tabular}[c]{@{}c@{}}Selection \\ Model 3\end{tabular}} & \textbf{\begin{tabular}[c]{@{}c@{}}Auxiliary\\ Model\end{tabular}} & \textbf{RBP$(\theta_1)$} & \textbf{RBP $(\theta_2)$} & \textbf{RBP $(\theta_3)$} & \textbf{RMSE $(\theta_1)$} & \textbf{RMSE $(\theta_2)$} & \textbf{RMSE $(\theta_3)$} & \textbf{CP $(\theta_1)$} & \textbf{CP $(\theta_2)$} & \textbf{CP $(\theta_3)$} \\ \hline
\textbf{Unweighted}             & -                                                                     & -                                                                     & -                                                                     & -                                                                  & -25.08\% (0.22\%)                & \textbf{-41.00\%} (0.18\%)               & \textbf{-53.84\%} (0.30\%)               & 1.00 (0.02)                      & 1.00 (0.01)                      & 1.00 (0.02)                      & 0.00 (0.00)                    & 0.00 (0.00)                    & 0.00 (0.00)                    \\ \hline
\textbf{Unweighted Diff}        & -                                                                     & -                                                                     & -                                                                     & -                                                                  & -28.94\% (0.22\%)                & -74.22\% (0.18\%)                  & -14.79\% (0.30\%)                  & 1.70 (0.03)                        & 4.13 (0.03)                        & 0.15 (0.02)                      & 0.00 (0.00)                    & 0.00 (0.00)                    & 0.48 (0.02)                    \\ \hline
\multirow{4}{*}{\textbf{JPL}}   & Correct                                                               & Correct                                                               & Correct                                                               & -                                                                  & \textbf{0.20\%}  (0.23\%)                & 0.08\% (0.19\%)                  & -0.21\% (0.33\%)                   & \textbf{0.04} (0.01)                        & 0.01 (0.01)                        & 0.02 (0.01)                      & 0.93 (0.01)                    & 0.90 (0.01)                    & 0.91 (0.01)                    \\ \cline{2-14} 
                                & Incorrect                                                             & Correct                                                               & Correct                                                               & -                                                                  & -9.01\% (0.26\%)                 & -12.01\% (0.22\%)                  & -38.42\% (0.37\%)                  & 0.17 (0.01)                        & 0.10 (0.01)                        & 0.52(0.01)                       & 0.65 (0.02)                    & 0.28 (0.02)                    & 0.01  (0.01)                   \\ \cline{2-14} 
                                & Incorrect                                                             & Incorrect                                                             & Correct                                                               & -                                                                  & -21.38\% (0.28\%)                & -36.81\% (0.24\%)                  & -49.30\% (0.38\%)                  & 0.76 (0.02)                        & 0.82 (0.01)                        & 0.85 (0.02)                      & 0.05 (0.01)                    & 0.00 (0.00)                    & 0.00 (0.00)                    \\ \cline{2-14} 
                                & Incorrect                                                             & Incorrect                                                             & Incorrect                                                             & -                                                                  & -24.66\% (0.27\%)                & -34.80\% (0.23\%)                  & -54.71\% (0.38\%)                  & 0.99 (0.03)                        & 0.73 (0.01)                        & 1.04 (0.02)                      & 0.02 (0.02)                    & 0.00 (0.00)                    & 0.00 (0.00)                    \\ \hline
\multirow{4}{*}{\textbf{JSR}}   & Correct                                                               & Correct                                                               & Correct                                                               & -                                                                  & -1.03\% (0.23\%)                 & -1.05\% (0.19\%)                   & -4.49\% (0.32\%)                   & 0.04 (0.01)                        & 0.01 (0.01)                        & 0.02 (0.01)                      & 0.95 (0.01)                    & 0.93 (0.01)                    & 0.89 (0.01)                    \\ \cline{2-14} 
                                & Incorrect                                                             & Correct                                                               & Correct                                                               & -                                                                  & -9.98\% (0.25\%)                 & -12.63\% (0.20\%)                  & -42.52\% (0.36\%)                  & 0.20 (0.01)                        & 0.11 (0.01)                        & 0.63 (0.01)                      & 0.61 (0.02)                    & 0.20 (0.02)                    & 0.00 (0.00)                    \\ \cline{2-14} 
                                & Incorrect                                                             & Incorrect                                                             & Correct                                                               & -                                                                  & -20.96\% (0.27\%)                & -35.64\% (0.24\%)                  & -50.16\% (0.36\%)                  & 0.73 (0.02)                        & 0.77 (0.01)                        & 0.88 (0.02)                      & 0.05 (0.01)                    & 0.00 (0.00)                    & 0.00 (0.00)                    \\ \cline{2-14} 
                                & Incorrect                                                             & Incorrect                                                             & Incorrect                                                             & -                                                                  & -24.30\% (0.27\%)                & -33.97\% (0.23\%)                  & -55.05\% (0.36\%)                  & 0.96 (0.03)                        & 0.88 (0.01)                        & 1.06 (0.02)                      & 0.01 (0.01)                    & 0.00 (0.00)                    & 0.00 (0.00)                    \\ \hline
\multirow{4}{*}{\textbf{JCL}}   & Correct                                                               & Correct                                                               & Correct                                                               & -                                                                  & 0.24\% (0.23\%)                  & 0.24\% (0.19\%)                    & -0.02\% (0.30\%)                   & 0.04 (0.01)                        & 0.01 (0.01)                        & 0.03 (0.01)                      & 0.95 (0.01)                    & 0.95 (0.01)                    & 0.97 (0.01)                    \\ \cline{2-14} 
                                & Incorrect                                                             & Correct                                                               & Correct                                                               & -                                                                  & -9.67\% (0.25\%)                 & -12.08\% (0.20\%)                  & -38.67\% (0.35\%)                  & 0.19 (0.01)                        & 0.10 (0.01)                        & 0.53 (0.01)                      & 0.62 (0.02)                    & 0.24 (0.02)                    & 0.01 (0.01)                    \\ \cline{2-14} 
                                & Incorrect                                                             & Incorrect                                                             & Correct                                                               & -                                                                  & -21.69\% (0.27\%)                & -37.28\% (0.23\%)                  & -49.37\% (0.37\%)                  & 0.78 (0.03)                        & 0.83 (0.01)                        & 0.85 (0.02)                      & 0.05 (0.01)                    & 0.00 (0.00)                    & 0.00 (0.00)                    \\ \cline{2-14} 
                                & Incorrect                                                             & Incorrect                                                             & Incorrect                                                             & -                                                                  & -24.77\% (0.27\%)                & -34.98\% (0.23\%)                  & -54.81\% (0.36\%)                  & 0.99 (0.03)                        & 0.74 (0.01)                        & 1.04 (0.02)                      & 0.01 (0.01)                    & 0.00 (0.00)                    & 0.00 (0.00)                    \\ \hline
\textbf{JPS Exact}              & -                                                                     & -                                                                     & -                                                                     & -                                                                  & \textbf{-7.21\%} (0.23\%)                 & -4.36\% (0.18\%)                   & 6.37\% (0.31\%)                    & \textbf{0.12} (0.01)                        & 0.02 (0.01)                        & 0.03 (0.01)                      & 0.72 (0.02)                    & 0.82 (0.02)                    & 0.86 (0.02)                    \\ \hline
\textbf{JPS Approximate}        & -                                                                     & -                                                                     & -                                                                     & -                                                                  & -11.74\% (0.23\%)                & \textbf{-20.26\%} (0.18\%)                 & 20.54\% (0.31\%)                   & 0.25 (0.01)                        & 0.25 (0.01)                        & 0.16 (0.01)                      & 0.37 (0.02)                    & 0.01 (0.01)                    & 0.14 (0.02)                    \\ \hline
\multirow{4}{*}{\textbf{JAIPW}} & Correct                                                               & Correct                                                               & Correct                                                               & Correct                                                            & 0.48\% (0.28\%)                  & -0.05\% (0.24\%)                   & 0.20\% (0.46\%)                    & 0.06 (0.01)                        & 0.02 (0.01)                        & 0.04 (0.01)                      & 0.96 (0.01)                    & 0.95 (0.01)                    & 0.96 (0.01)                    \\ \cline{2-14} 
                                & Correct                                                               & Correct                                                               & Correct                                                               & Incorrect                                                          & 0.58\% (0.28\%)                  & -0.37\% (0.26\%)                   & -0.93\% (0.43\%)                   & 0.06 (0.01)                        & 0.02 (0.01)                        & 0.03 (0.01)                      & 0.96 (0.01)                    & 0.96 (0.01)                    & 0.95 (0.01)                    \\ \cline{2-14} 
                                & Incorrect                                                             & Incorrect                                                             & Incorrect                                                             & Correct                                                            & \textbf{-2.34\%} (0.30\%)                 & 0.33\% (0.24\%)                    & 1.31\% (0.46\%)                    & \textbf{0.08} (0.01)                        & 0.02 (0.01)                        & 0.04 (0.01)                      & \textbf{0.96} (0.01)                    & \textbf{0.96} (0.01)                    & \textbf{0.95} (0.01)                    \\ \cline{2-14} 
                                & Incorrect                                                             & Incorrect                                                             & Incorrect                                                             & Incorrect                                                          & -23.79\% (0.30\%)                & -36.74\% (0.27\%)                  & -49.88\% (0.42\%)                  & 0.94 (0.03)                        & 0.82 (0.01)                        & 0.88 (0.02)                      & 0.05 (0.01)                    & 0.01 (0.01)                    & 0.01 (0.01)                    \\ \hline
\end{tabular}
\end{adjustbox}
\vspace{0.5cm}
\caption{Comparison of Relative Bias Percentage (RBP), Relative Mean Squared Error (RMSE), and Coverage Probabilities (CP) across the unweighted, four joint IPW methods (JPL, JSR, JCL, JPS), and JAIPW under simulation setup 1. Bolded values indicate key findings discussed in the text. For each metric, the mean estimate is reported with its monte carlo standard error in parentheses. The results are obtained using number of simulation replications as $R=500$.\\
\textbf{Abbreviations:} Unweighted = Unweighted Logistic Regression; JPL = Joint Pseudolikelihood; JSR = Joint Simplex Regression; JPS = Joint Post Stratification; JCL = Joint Calibration; JAIPW = Joint Augmented Inverse Probability Weighted.}
\label{tab:table_r1}
\end{table}

\begin{table}[H]
\renewcommand{\arraystretch}{2.2}
\centering
\begin{adjustbox}{max width=0.9\textwidth}
\begin{tabular}{|cccccccccccccc|}
\hline
\textbf{Method}                 & \textbf{\begin{tabular}[c]{@{}c@{}}Selection \\ Model 1\end{tabular}} & \textbf{\begin{tabular}[c]{@{}c@{}}Selection \\ Model 2\end{tabular}} & \textbf{\begin{tabular}[c]{@{}c@{}}Selection \\ Model 3\end{tabular}} & \textbf{\begin{tabular}[c]{@{}c@{}}Auxiliary\\ Model\end{tabular}} & \textbf{RBP$(\theta_1)$} & \textbf{RBP $(\theta_2)$} & \textbf{RBP $(\theta_3)$} & \textbf{RMSE $(\theta_1)$} & \textbf{RMSE $(\theta_2)$} & \textbf{RMSE $(\theta_3)$} & \textbf{CP $(\theta_1)$} & \textbf{CP $(\theta_2)$} & \textbf{CP $(\theta_3)$} \\ \hline
\textbf{Unweighted}             & -                                                                     & -                                                                     & -                                                                     & -                                                                  & -15.21\% (0.23\%)                & \textbf{-32.46\%} (0.18\%)               & \textbf{-31.02\%} (0.30\%)
& 1.00 (0.00)                      & 1.00 (0.00)                      & 1.00 (0.00)                      & 0.00 (0.00)                     & 0.00 (0.00)                    & 0.00 (0.00)                    \\ \hline
\textbf{Unweighted Diff}        & -                                                                     & -                                                                     & -                                                                     & -                                                                  & -22.79\% (0.23\%)                & -82.29\% (0.19\%)                  & 21.78\% (0.32\%)                   & 2.13 (0.07)                      & 6.35 (0.08)                      & 0.52 (0.05)                      & 0.00 (0.00)                    & 0.00 (0.00)                    & 0.20 (0.02)                    \\ \hline
\multirow{4}{*}{\textbf{JPL}}   & Correct                                                               & Correct                                                               & Correct                                                               & -                                                                  & -1.83\% (0.68\%)
& -4.83\% (0.88\%)                   & -3.04\% (1.02\%)                   & \textbf{0.92} (0.15)                      & 0.38 (0.05)                      & 0.53 (0.05)                      & 0.97 (0.01)                    & 0.96 (0.01)                    & 0.98 (0.01)                    \\ \cline{2-14} 
                                & Incorrect                                                             & Correct                                                               & Correct                                                               & -                                                                  & 16.55\% (0.24\%)                 & 15.25\% (0.19\%)                   & 37.98\% (0.34\%)                   & 1.17 (0.04)                      & 0.23 (0.01)                      & 1.49 (0.04)                      & 0.15 (0.02)                    & 0.09 (0.01)                    & 0.02 (0.01)                    \\ \cline{2-14} 
                                & Incorrect                                                             & Incorrect                                                             & Correct                                                               & -                                                                  & 20.36\%  (0.24)                & 19.60\% (0.20)                   & 42.78\% (0.34)                   & 1.71 (0.06)                      & 0.38 (0.01)                      & 1.87 (0.05)                      & 0.04 (0.01)                    & 0.01 (0.01)                    & 0.00 (0.01)                    \\ \cline{2-14} 
                                & Incorrect                                                             & Incorrect                                                             & Incorrect                                                             & -                                                                  & 19.75\% (0.23\%)                 & 15.77\% (0.19\%)                   & 41.81\% (0.33\%)                   & \textbf{1.62} (0.06)                      & 0.25 (0.01)                      & \textbf{1.79} (0.04)                      & 0.05 (0.01)                    & 0.05 (0.01)                    & 0.00 (0.01)                    \\ \hline
\multirow{4}{*}{\textbf{JSR}}   & Correct                                                               & Correct                                                               & Correct                                                               & -                                                                  & 0.71\% (0.37\%)                    & -0.19\% (0.34\%)                   & 1.32\% (0.53\%)                    & \textbf{0.28} (0.03)                      & 0.06 (0.01)                      & 0.14 (0.02)                      & 0.94 (0.01)                    & 0.93 (0.01)                    & 0.92 (0.01)                    \\ \cline{2-14} 
                                & Incorrect                                                             & Correct                                                               & Correct                                                               & -                                                                  & 15.39\% (0.23\%)                 & 14.97\% (0.19\%)                   & 33.74\% (0.33\%)
                                & 1.03 (0.04)                      & 0.23 (0.01)                      & 1.18 (0.03)                      & 0.17 (0.02)                    & 0.05 (0.01)                    & 0.02 (0.01)                    \\ \cline{2-14} 
                                & Incorrect                                                             & Incorrect                                                             & Correct                                                               & -                                                                  & 18.63\% (0.23\%)                 & 18.33\% (0.18\%)                   & 37.40\% (0.33\%)                   & 1.45 (0.05)                      & 0.33 (0.01)                      & 1.44 (0.04)                      & 0.06 (0.01)                    & 0.01 (0.01)                    & 0.00 (0.00)                    \\ \cline{2-14} 
                                & Incorrect                                                             & Incorrect                                                             & Incorrect                                                             & -                                                                  & 18.20\% (0.23\%)                 & 15.28\% (0.18\%)                   & 36.68\% (0.33\%)                   & \textbf{1.39} (0.05)                      & 0.23 (0.01)                      & \textbf{1.39} (0.04)                      & 0.06 (0.01)                    & 0.04 (0.01)                    & 0.00 (0.00)                    \\ \hline
\multirow{4}{*}{\textbf{JCL}}   & Correct                                                               & Correct                                                               & Correct                                                               & -                                                                  & 0.24\% (0.30\%)                  & -0.14\% (0.23\%)                   & -0.64\% (0.35\%)                    & 0.18 (0.01)                      & 0.03 (0.02)                      & 0.06 (0.01)                      & 0.93 (0.01)                    & 0.94 (0.01)                    & 0.96 (0.01)                    \\ \cline{2-14} 
                                & Incorrect                                                             & Correct                                                               & Correct                                                               & -                                                                  & 14.44\% (0.23\%)                 & 13.65\% (0.19\%)                   & 32.22\% (0.32\%)                   & 0.92 (0.04)                      & 0.19 (0.01)                      & 1.08 (0.03)                      & 0.21 (0.02)                    & 0.11 (0.01)                    & 0.01 (0.01)                    \\ \cline{2-14} 
                                & Incorrect                                                             & Incorrect                                                             & Correct                                                               & -                                                                  & 16.10\% (0.23\%)                 & 13.88\% (0.19\%)                   & 35.10\% (0.32\%)                   & 1.11 (0.04)                      & 0.20 (0.01)                      & 1.27 (0.03)                      & 0.13 (0.02)                    & 0.09 (0.01)                    & 0.00 (0.00)                    \\ \cline{2-14} 
                                & Incorrect                                                             & Incorrect                                                             & Incorrect                                                             & -                                                                  & 15.97\% (0.23\%)                 & 16.56\% (0.19\%)                   & 34.98\% (0.32\%)                   & 1.09 (0.04)                      & 0.27 (0.01)                      & 1.26 (0.03)                      & 0.15 (0.02)                    & 0.02 (0.01)                    & 0.00 (0.00)                    \\ \hline
\textbf{JPS Exact}              & -                                                                     & -                                                                     & -                                                                     & -                                                                  & 4.06\% (0.24\%)                  & 2.10\% (0.19\%)                    & 19.52\% (0.31\%)                   & 0.18 (0.01)                      & 0.02 (0.01)                      & 0.43 (0.02)                      & 0.87 (0.02)                    & 0.92 (0.01)                    & 0.23 (0.02)                    \\ \hline
\textbf{JPS Approximate}        & -                                                                     & -                                                                     & -                                                                     & -                                                                  & 1.40\% (0.24\%)                  & -10.73\% (0.19\%)                  & 49.19\% (0.30\%)                   & 0.12 (0.01)                      & 0.12 (0.01)                      & 2.45 (0.06)                      & 0.94 (0.01)                    & 0.27 (0.02)                    & 0.00 (0.00)                    \\ \hline
\multirow{4}{*}{\textbf{JAIPW}} & Correct                                                               & Correct                                                               & Correct                                                               & Correct                                                            & 0.10\% (0.45\%)                  & 0.10\% (0.26\%)                    & 0.26\% (0.50\%)                    & \textbf{0.39} (0.04)                      & 0.03 (0.01)                      & 0.12 (0.01)                      & 0.93 (0.01)                    & 0.96 (0.01)                    & 0.96 (0.01)                    \\ \cline{2-14} 
                                & Correct                                                               & Correct                                                               & Correct                                                               & Incorrect                                                          & 1.15\% (0.45\%)                    & 0.62\% (0.35\%)                    & 1.19\% (0.63\%)                    & 0.39 (0.03)                      & 0.06 (0.01)                      & 0.19 (0.01)                      & \textbf{0.87} (0.02)                    & \textbf{0.89} (0.01)                    & \textbf{0.86} (0.02)                    \\ \cline{2-14} 
                                & Incorrect                                                             & Incorrect                                                             & Incorrect                                                             & Correct                                                            & \textbf{3.24\%} (0.26\%)                  & \textbf{-0.80\%} (0.24\%)                   & \textbf{-1.46\%} (0.45\%)                   & \textbf{0.17} (0.01)                      & 0.03 (0.01)                      & 0.10 (0.01)                      & \textbf{0.92} (0.01)                    & \textbf{0.95} (0.01)                    & \textbf{0.95} (0.01)                    \\ \cline{2-14} 
                                & Incorrect                                                             & Incorrect                                                             & Incorrect                                                             & Incorrect                                                          & 19.30\% (0.27\%)                 & 13.01\% (0.23\%)                   & 38.00\% (0.41\%)                   & 1.58 (0.06)                      & 0.18 (0.01)                      & 1.52 (0.04)                      & 0.11 (0.01)                    & 0.24 (0.02)                    & 0.01 (0.01)                    \\ \hline
\end{tabular}
\end{adjustbox}
\vspace{0.5cm}
\caption{Comparison of Relative Bias Percentage (RBP), Relative Mean Squared Error (RMSE), and Coverage Probabilities (CP) across the unweighted, four joint IPW methods (JPL, JSR, JCL, JPS), and JAIPW under simulation setup 2. Bolded values indicate key findings discussed in the text. For each metric, the mean estimate is reported with its monte carlo standard error in parentheses. The results are obtained using number of simulation replications as $R=500$.\\
\textbf{Abbreviations:} Unweighted = Unweighted Logistic Regression; JPL = Joint Pseudolikelihood; JSR = Joint Simplex Regression; JPS = Joint Post Stratification; JCL = Joint Calibration; JAIPW = Joint Augmented Inverse Probability Weighted.}
\label{tab:table_r2}
\end{table}
\begin{figure}[H]
    \centering
    \includegraphics[width=0.9\linewidth]{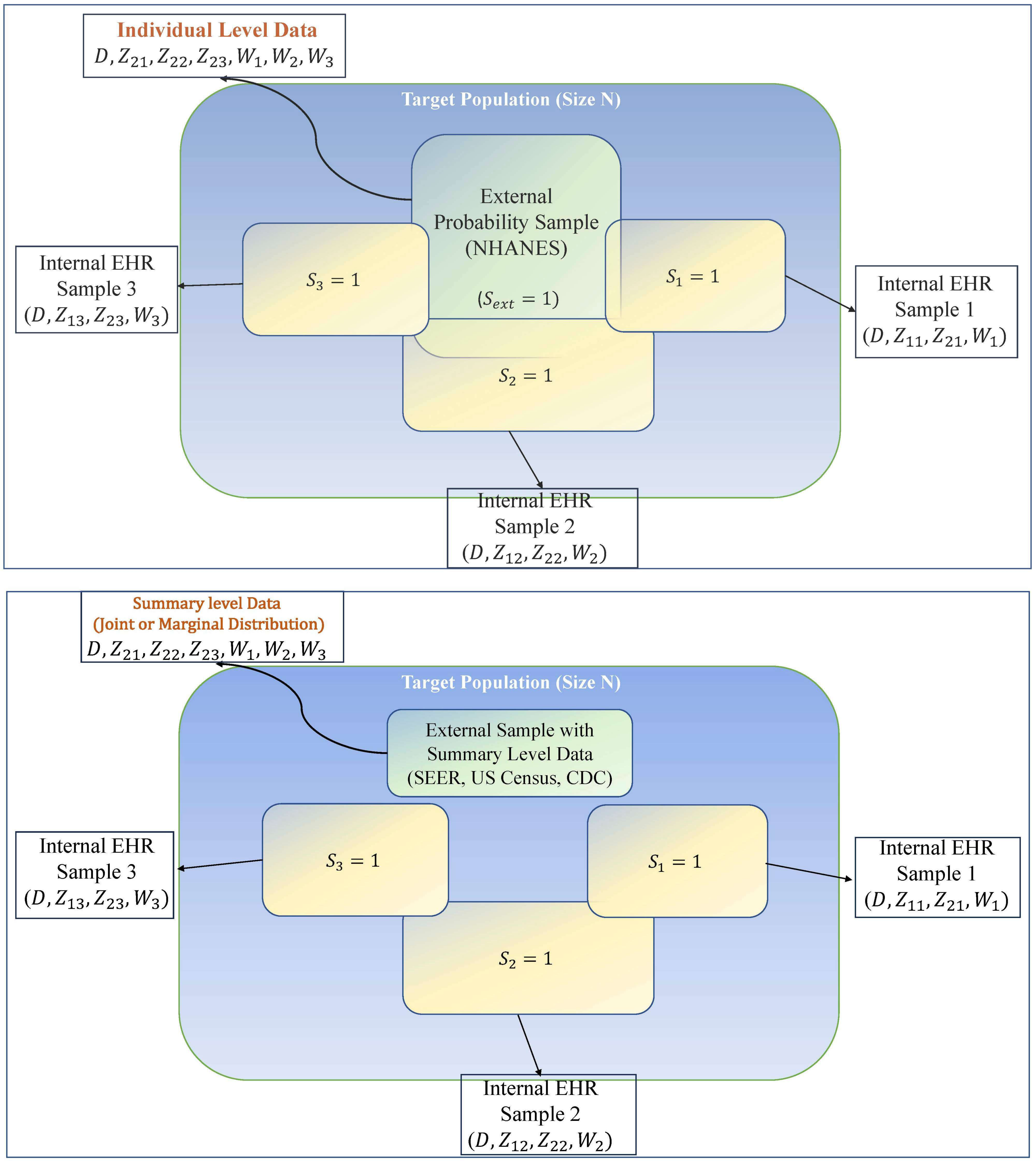}
    \caption{Figure depicts the relationship between the target population, three internal non-probability samples ($K=3$) and external data sources. $S_1$, $S_2$ and $S_3$ are the selection indicator variables of the three internal samples respectively. In the Sub-figure 1, $S_{\text{ext}}=1$ denote the individual-level external data which is NHANES in our work. In Sub-figure 2, external data sources include SEER, US Census and CDC. $D$ is the outcome of interest. For $k=1,2,3$, $\boldsymbol Z_{1k}\rightarrow D,\boldsymbol Z_{1k}\not \rightarrow S_k$, $\boldsymbol Z_{2k}\rightarrow D, \boldsymbol Z_{2k}\rightarrow S_k$ and $\boldsymbol W_{k}\not\rightarrow D,\boldsymbol W_{k}\rightarrow S_k$.}\label{fig:popumulti}
\end{figure}

\begin{figure}[H]
    \centering
    \includegraphics[width=\linewidth]{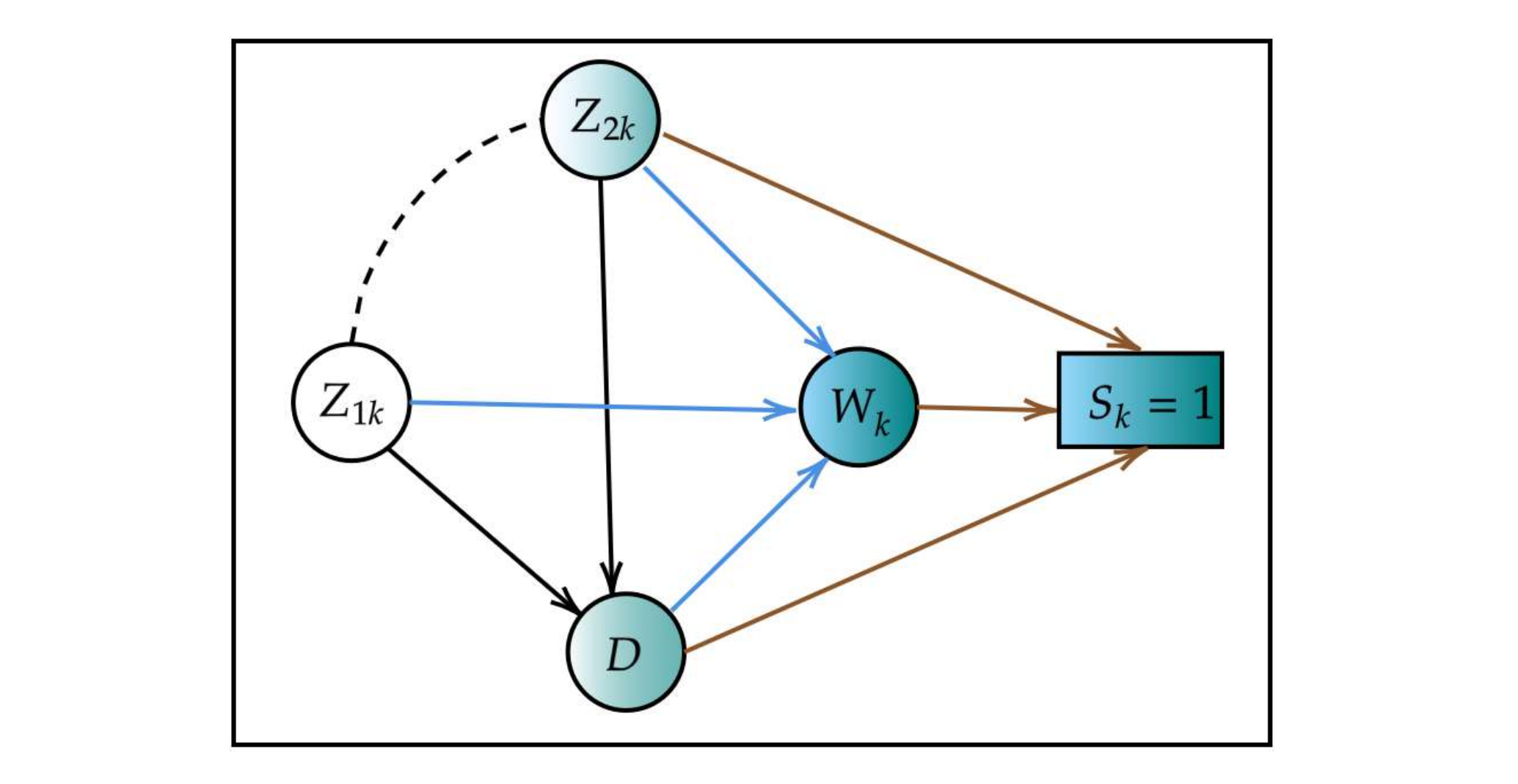}
    \caption{Selection Directed Acyclic Graphs (DAGs)  representing the relationships between different variables of interest in cohort $k$, where $k \in \{1,2,...,K\}$: $D$ (Disease Indicator), $S_k$ (Selection Indicator into the cohort $k$), $\boldsymbol Z_{1k}\rightarrow D,\boldsymbol Z_{1k}\not \rightarrow S_k$, $\boldsymbol Z_{2k}\rightarrow D, \boldsymbol Z_{2k}\rightarrow S_k$ and $\boldsymbol W_{k}\not\rightarrow D,\boldsymbol W_{k}\rightarrow S_k$. The dotted line between $\boldsymbol Z_{1k}$ and  $\boldsymbol Z_{2k}$ denotes association between those two variables.}\label{fig:dagmulti}
\end{figure}
\begin{figure}[H]
    \centering
    \includegraphics[width =\linewidth]{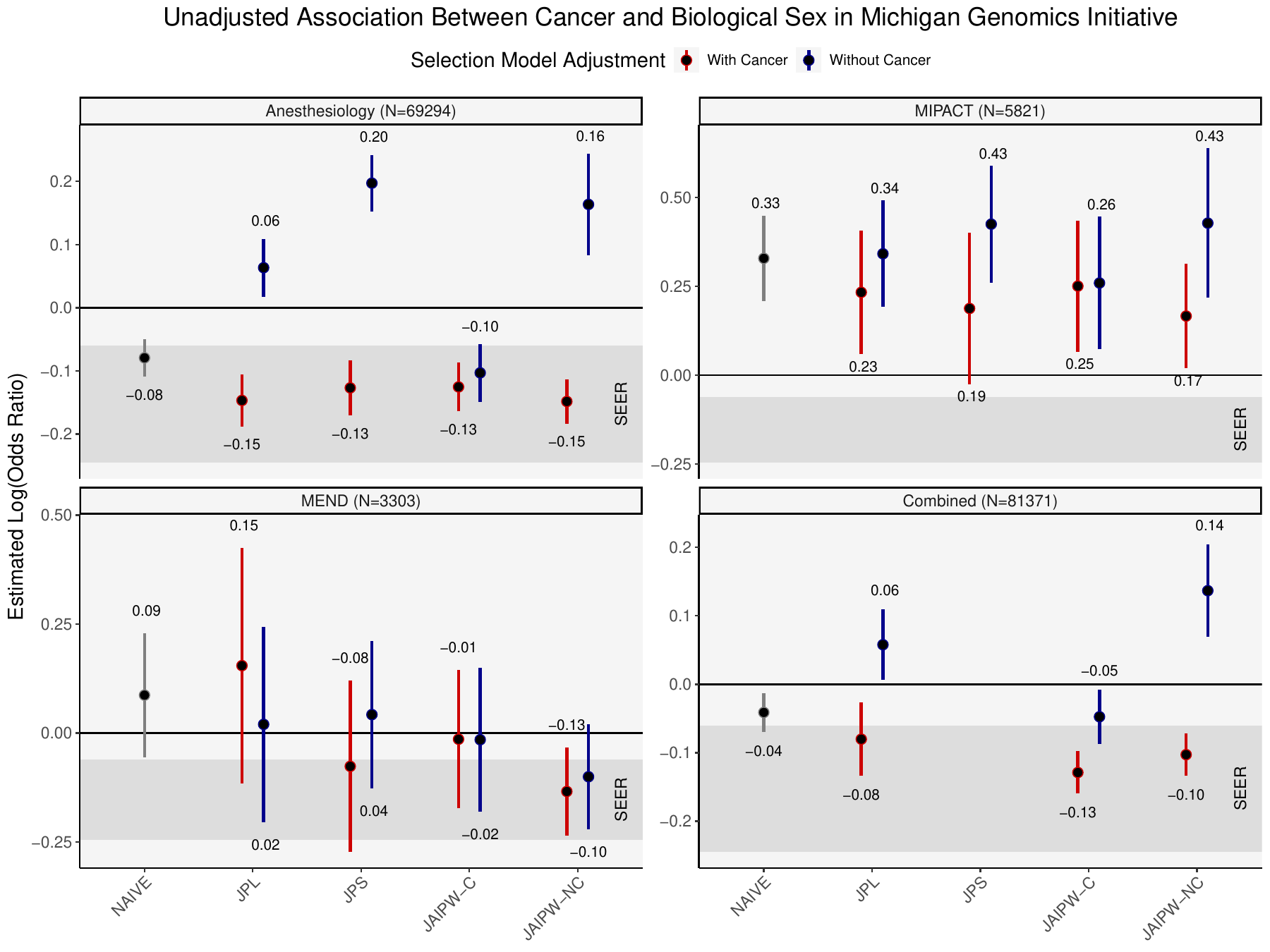}
    \caption{Estimates of the marginal log (Odds Ratio) for the association between cancer and biological sex, with 95\% confidence intervals, across Anesthesiology, MIPACT, MEND and the combined cohort. Comparisons are shown for the unweighted logistic regression (NAIVE, gray) and methods adjusting for selection bias (JPL, JPS, JAIPW-C, JAIPW-NC). For the IPW and JAIPW methods, estimates are shown either in red or blue depending on the selection model including cancer status (red, Cancer) or not (blue, No Cancer). JAIPW-C includes cancer in the auxiliary score model, while JAIPW-NC do not. The gray horizontal band represents the estimates of the log (Odds Ratio) from the SEER 2008-2016 registry.}\label{fig:realdata}
\end{figure}

\newpage
\section*{Supplementary Materials: A Doubly Robust Framework for Addressing Outcome-Dependent Selection Bias in Multi-Cohort EHR Studies}

\pagenumbering{arabic}

\setcounter{figure}{0}
\renewcommand{\thefigure}{S\arabic{figure}}
\setcounter{table}{0}
\renewcommand{\thetable}{S\arabic{table}}
\setcounter{equation}{0}
\renewcommand{\theequation}{S\arabic{equation}}
\setcounter{section}{0}
\renewcommand{\thesection}{S\arabic{section}}
\setcounter{theorem}{0}
\renewcommand{\thetheorem}{S\arabic{theorem}}
\section{Supplementary Assumptions} \label{sec:assumptions}
\textit{A}1. Let \(\boldsymbol{\theta}^*_{\boldsymbol{Z}}\) denote the true value of \(\boldsymbol{\theta}_{\boldsymbol{Z}}\), as defined in equation (2.1) of the main text. We assume that \(\boldsymbol{\theta}^*_{\boldsymbol{Z}}\) belongs to the compact parameter space \(\Theta_{\boldsymbol{Z}}\).\\

\noindent
\textit{A}2. Let the sample sizes of the $K$ different cohorts and the external probability sample be $n_1,n_2,...n_K$ and $n_b$. We assume that $n_1,n_2,...n_K,n_b=\mathcal{O}(N)$ where $N$ is the target population size.\\

\noindent
\textit{A}3. \textbf{IPW Variance Estimation Without Nuisance Parameter Contribution:}\\ $\mathbb{E}[\sup_{\boldsymbol \theta_{\boldsymbol Z} \in \Theta_{\boldsymbol Z}}(G_{\boldsymbol \theta_{\boldsymbol Z}})]<\infty$ and $\mathbb{E}[\sup_{\boldsymbol \theta_{\boldsymbol Z} \in \Theta_{\boldsymbol Z}}\{\boldsymbol g(\boldsymbol \theta_{\boldsymbol Z}) \boldsymbol g(\boldsymbol \theta_{\boldsymbol Z})'\}]<\infty$.\\

\noindent
\textit{A}4. The true selection model parameters, \(\boldsymbol{\alpha}^{\text{mult}*}\), are assumed to lie within the compact parameter space \(\Theta_{\boldsymbol{\alpha}^{\text{mult}}}\), where \(\boldsymbol{\alpha}^{\text{mult}*}\) represents the true value of \(\boldsymbol{\alpha}^{\text{mult}}\), as defined in equation (3.14) of the main text.\\

\noindent
\textit{A}5. \textbf{PL Variance Estimation:} Let $\eta=(\boldsymbol \theta_{\boldsymbol Z},\boldsymbol \alpha^{\text{mult}})$ and $\mathcal{N}=\Theta_{\boldsymbol Z} \times \Theta_{\boldsymbol \alpha^{\text{mult}}}$. Each of the following expectations are finite.
\begin{align*}
    & \mathbb{E}[\sup_{\boldsymbol \eta \in \mathcal{N}}(G_{\boldsymbol \theta_{\boldsymbol Z}})]<\infty \hspace{0.8cm} 
    \mathbb{E}[\sup_{\boldsymbol \eta \in \mathcal{N}}(G_{\boldsymbol \alpha^{\text{mult}}})]<\infty \hspace{0.8cm} \mathbb{E}[\sup_{\boldsymbol \alpha^{\text{mult}} \in \Theta_{\boldsymbol \alpha^{\text{mult}}}}(H)]<\infty\\
    & \mathbb{E}\left(\sup_{\boldsymbol \eta \in \mathcal{N}}\left[S \cdot \left\{\frac{1}{\pi(\boldsymbol X^{\text{mult}},\boldsymbol \alpha^{\text{mult}})}\right\}^2\left\{D-\frac{e^{\boldsymbol \theta_{\boldsymbol Z}' \boldsymbol Z}}{(1+e^{\boldsymbol \theta_{\boldsymbol Z}'\boldsymbol Z})}\right\}^2\cdot \boldsymbol Z \boldsymbol Z'\right]\right)<\infty\\
    & \mathbb{E}\left(\sup_{\boldsymbol \eta \in \mathcal{N}}\left[\frac{1}{\pi(\boldsymbol X^{\text{mult}},\boldsymbol \alpha^{\text{mult}})}\begin{pmatrix}
     S_{1}\boldsymbol X_{1}\\
     \vdots\\
      S_{K}\boldsymbol X_{K}\\
         \end{pmatrix}\cdot \left\{D\boldsymbol Z'-\frac{e^{\boldsymbol \theta_{\boldsymbol Z}'\boldsymbol Z}}{(1+e^{\boldsymbol \theta_{\boldsymbol Z}' \boldsymbol Z})} \cdot \boldsymbol Z'\right\}\right.\right.\\
         & \left.\left.-\frac{S\cdot S_{\text{ext}}}{\pi_{\text{ext}}(\boldsymbol M)} \begin{pmatrix}
       \pi_1(\boldsymbol X_{1},\boldsymbol \alpha_1).\boldsymbol X_{1}\\
     \vdots\\
     \pi_K(\boldsymbol X_{K},\boldsymbol \alpha_K).\boldsymbol X_{K}\\
         \end{pmatrix}\cdot \left\{D \boldsymbol Z'-\frac{e^{\boldsymbol \theta_{\boldsymbol Z}' \boldsymbol Z}}{(1+e^{\boldsymbol \theta_{\boldsymbol Z}' \boldsymbol Z})} \cdot \boldsymbol Z'\right\}\right]\right)<\infty\\
    & \mathbb{E}\left(\sup_{\boldsymbol \alpha^{\text{mult}}\in \Theta_{\boldsymbol \alpha^{\text{mult}}}}\left[\begin{pmatrix}
    S_1 \boldsymbol X_{1}\\
     \vdots\\
      S_K \boldsymbol X_{K}\\
         \end{pmatrix}\cdot(S_1 \boldsymbol X_{1}',\cdot \cdot \cdot,S_K \boldsymbol X_{K}')-2S_{\text{ext}}\begin{pmatrix}
    \frac{\pi_1(\boldsymbol X_{1},\boldsymbol \alpha_1)}{\pi_{\text{ext}}(\boldsymbol M)}\cdot \boldsymbol X_{1}\\
     \vdots\\
    \frac{\pi_K(\boldsymbol X_{K},\boldsymbol \alpha_K)}{\pi_{\text{ext}}(\boldsymbol M)}\cdot \boldsymbol X_{K}\\
         \end{pmatrix}\cdot (S_1\boldsymbol X_{1}',\cdot \cdot \cdot,S_K\boldsymbol X_{K}')\right.\right.\\
    & \left.\left.+  S_{\text{ext}} \cdot \left\{\frac{1}{\pi_{\text{ext}}(\boldsymbol M)}\right\}^2 \cdot \begin{pmatrix}
    \frac{\pi_1(\boldsymbol X_{1},\boldsymbol \alpha_1)}{\pi_{\text{ext}}(\boldsymbol M)}\cdot \boldsymbol X_{1}\\
     \vdots\\
    \frac{\pi_K(\boldsymbol X_{K},\boldsymbol \alpha_K)}{\pi_{\text{ext}}(\boldsymbol M)}\cdot \boldsymbol X_{K}\\
         \end{pmatrix}\left\{ \frac{\pi_1(\boldsymbol X_{1},\boldsymbol \alpha_1)}{\pi_{\text{ext}}(\boldsymbol M)}\cdot \boldsymbol X_{1}',\cdot\cdot\cdot, \frac{\pi_K(\boldsymbol X_{K},\boldsymbol \alpha_K)}{\pi_{\text{ext}}(\boldsymbol M)}\cdot \boldsymbol X_{K}' \right\}\right]\right) <\infty \cdot
\end{align*}
\noindent
\textit{A}6. \textbf{CL Variance Estimation:}  Each of the following expectations are finite.
\begin{align*}
    & \mathbb{E}[\sup_{\boldsymbol \eta \in \mathcal{N}}(G_{\boldsymbol \theta_{\boldsymbol Z}})]<\infty \hspace{0.8cm} 
    \mathbb{E}[\sup_{\boldsymbol \eta \in \mathcal{N}}(G_{\boldsymbol \alpha^{\text{mult}}})]<\infty \hspace{0.8cm} \mathbb{E}[\sup_{\boldsymbol \alpha^{\text{mult}}\in \Theta_{\boldsymbol \alpha^{\text{mult}}}}(H)]<\infty\\
    & \mathbb{E}\left(\sup_{\boldsymbol \eta \in \mathcal{N}}\left[S \cdot \left\{\frac{1}{\pi(\boldsymbol X^{\text{mult}},\boldsymbol \alpha^{\text{mult}})}\right\}^2\left\{D-\frac{e^{\boldsymbol \theta_{\boldsymbol Z}' \boldsymbol Z}}{(1+e^{\boldsymbol \theta_{\boldsymbol Z}'\boldsymbol Z})}\right\}^2\cdot \boldsymbol Z \boldsymbol Z'\right]\right)<\infty\\
    & \mathbb{E}\left(\sup_{\boldsymbol \eta \in \mathcal{N}}\left[\frac{1}{\pi(\boldsymbol X^{\text{mult}},\boldsymbol \alpha^{\text{mult}})}\begin{pmatrix}
     \frac{S_{1}\boldsymbol X_{1}}{\pi_1(\boldsymbol X_1,\boldsymbol \alpha_1)}\\
     \vdots\\
       \frac{S_{K}\boldsymbol X_{K}}{\pi_K(\boldsymbol X_K,\boldsymbol \alpha_K)}\\
         \end{pmatrix}\cdot \left\{D\boldsymbol Z'-\frac{e^{\boldsymbol \theta_{\boldsymbol Z}'\boldsymbol Z}}{(1+e^{\boldsymbol \theta_{\boldsymbol Z}' \boldsymbol Z})} \cdot \boldsymbol Z'\right\}\right.\right.\\
         & \left.\left.-S
         \begin{pmatrix}
      \boldsymbol X_{1}\\
     \vdots\\
    \boldsymbol X_{K}
         \end{pmatrix}
         \cdot \left\{D \boldsymbol Z'-\frac{e^{\boldsymbol \theta_{\boldsymbol Z}' \boldsymbol Z}}{(1+e^{\boldsymbol \theta_{\boldsymbol Z}' \boldsymbol Z})} \cdot \boldsymbol Z'\right\}\right]\right)<\infty
\end{align*}
\begin{align*}
     & \mathbb{E}\left(\sup_{\boldsymbol \alpha^{\text{mult}}\in \Theta_{\boldsymbol \alpha^{\text{mult}}}}\left[\begin{pmatrix}
     \frac{S_{1}\boldsymbol X_{1}}{\pi_1(\boldsymbol X_1,\boldsymbol \alpha_1)}\\
     \vdots\\
       \frac{S_{K}\boldsymbol X_{K}}{\pi_K(\boldsymbol X_K,\boldsymbol \alpha_K)}\\
         \end{pmatrix}\cdot\left(\frac{S_{1}\boldsymbol X_{1}'}{\pi_1(\boldsymbol X_1,\boldsymbol \alpha_1)},\cdot \cdot \cdot, \frac{S_{K}\boldsymbol X_{K}'}{\pi_K(\boldsymbol X_K,\boldsymbol \alpha_K)}\right)\right.\right.\\
         & \left.\left. -2\begin{pmatrix}
    \boldsymbol X_{1}\\
     \vdots\\
    \boldsymbol X_{K}\\
         \end{pmatrix}\cdot \left(\frac{S_{1}\boldsymbol X_{1}'}{\pi_1(\boldsymbol X_1,\boldsymbol \alpha_1)},\cdot \cdot \cdot, \frac{S_{K}\boldsymbol X_{K}'}{\pi_K(\boldsymbol X_K,\boldsymbol \alpha_K)}\right).+ \begin{pmatrix}
    \boldsymbol X_{1}\\
     \vdots\\
    \boldsymbol X_{K}\\
         \end{pmatrix}(\boldsymbol X_1',\cdot \cdot \cdot ,\boldsymbol X_{K}')\right]\right) <\infty \cdot
\end{align*}
\noindent
\textit{A}7. Let the combined nuisance parameter be $\boldsymbol\eta^{\text{mult}}=(\boldsymbol \alpha^{\text{mult}},\boldsymbol \gamma^{\text{mult}})$. Then, the true nuisance parameters, \(\boldsymbol{\eta}^{\text{mult}*}\), belong to the compact parameter space \(\Theta_{\boldsymbol{\eta}^{\text{mult}*}}\), where \(\boldsymbol{\eta}^{\text{mult}*}\) represents the true value of \(\boldsymbol{\eta}^{\text{mult}}\). \\

\noindent
\textit{A}8. \textbf{Parametric JAIPW Variance Estimation:} Let the combined parameter space be $\mathcal{N}=\Theta_{\boldsymbol Z} \times \Theta_{\boldsymbol \alpha^{\text{mult}}} \times \Theta_{\boldsymbol \gamma^{\text{mult}}}$. Each of the following expectations are finite.
\begin{align*}
    & \mathbb{E}[\sup_{\boldsymbol \eta \in \mathcal{N}}(G_{\boldsymbol \theta_{\boldsymbol Z}})]<\infty \hspace{0.8cm} 
    \mathbb{E}[\sup_{\boldsymbol \eta \in \mathcal{N}}(G_{\boldsymbol \alpha^{\text{mult}},\boldsymbol \gamma^{\text{mult}}})]<\infty \hspace{0.8cm} \mathbb{E}\left[\sup_{\boldsymbol \eta \in \mathcal{N}}(H)\right]<\infty\\
    & \mathbb{E}\left[\sup_{\boldsymbol \eta \in \mathcal{N}}\left\{\boldsymbol g(\boldsymbol \theta_{\boldsymbol Z},\boldsymbol \alpha^{\text{mult}},\boldsymbol \gamma^{\text{mult}})\cdot  g(\boldsymbol \theta_{\boldsymbol Z},\boldsymbol \alpha^{\text{mult}},\boldsymbol \gamma^{\text{mult}})'\right\}\right]<\infty\\
    & \mathbb{E}\left[\sup_{\boldsymbol \eta \in \mathcal{N}}\left\{\boldsymbol g(\boldsymbol \theta_{\boldsymbol Z},\boldsymbol \alpha^{\text{mult}},\boldsymbol \gamma^{\text{mult}})\cdot  \boldsymbol h(\boldsymbol X^{\text{mult}},\boldsymbol Z_{1\cap},\boldsymbol M,\boldsymbol \alpha^{\text{mult}},\boldsymbol \gamma^{\text{mult}})'\right\}\right]<\infty\\
     & \mathbb{E}\left[\sup_{\boldsymbol \eta \in \mathcal{N}}\left\{\boldsymbol h(\boldsymbol X^{\text{mult}},\boldsymbol Z_{1\cap},\boldsymbol M,\boldsymbol \alpha^{\text{mult}},\boldsymbol \gamma^{\text{mult}})\cdot  \boldsymbol h(\boldsymbol X^{\text{mult}},\boldsymbol Z_{1\cap},\boldsymbol M,\boldsymbol \alpha^{\text{mult}},\boldsymbol \gamma^{\text{mult}})'\right\}\right]<\infty \cdot
\end{align*}
\textit{A}9. Let the combined nuisance parameter function be $\boldsymbol\eta^{\text{mult}}=(\boldsymbol \alpha^{\text{mult}},\boldsymbol f)$. Let the combined sample size be $N_K$. For all \( N_K \geq 3 \), the following conditions hold:
\begin{enumerate}
      \item The map \( (\boldsymbol \theta_{\boldsymbol Z}, \boldsymbol\eta^{\text{mult}})) \mapsto \mathbb{E}[\boldsymbol g^{\text{mult}}( \boldsymbol \theta_{\boldsymbol Z}, \boldsymbol\eta^{\text{mult}})] \) is twice continuously Gateaux-differentiable on \( \Theta_{\boldsymbol Z} \times  \Theta_{\boldsymbol\eta^{\text{mult}}} \).
       \item For all \( \boldsymbol \theta_{\boldsymbol Z} \in \Theta_{\boldsymbol Z} \), the identification condition holds:
       \begin{equation*}
        2 \left\| \mathbb{E}[\boldsymbol g^{\text{mult}}( \boldsymbol \theta_{\boldsymbol Z}, \boldsymbol\eta^{\text{mult}})] \right\| \geq \left\| J_0 ( \boldsymbol \theta_{\boldsymbol Z} -  \boldsymbol \theta_{\boldsymbol Z}^*) \right\| \wedge c_0,
    \end{equation*}
    where the Jacobian matrix $J_0 := \left. \frac{\partial}{\partial \boldsymbol \theta_{\boldsymbol Z}'} \, 
\mathbb{E}\left[\boldsymbol g^{\text{mult}}( \boldsymbol \theta_{\boldsymbol Z}, \boldsymbol\eta^{\text{mult}*})\right] 
\right|_{\boldsymbol \theta_{\boldsymbol Z} = \boldsymbol \theta_{\boldsymbol Z}^*}$ has singular values bounded between \( c_0 \) and \( c_1 \).
\end{enumerate}

\noindent
\textit{A}10.
Let \( L \) be a fixed integer. For all \( N _K\geq 3 \) the following conditions hold:
\begin{enumerate}
    \item Given a random subset \( I \subset \{1, \dots, N_K\} \) of size \( n = N_K/L \), the nuisance parameter estimator \( \hat{\boldsymbol\eta}^{\text{mult}}_I = \hat{\boldsymbol\eta}^{\text{mult}}(i \in I^c) \) belongs to the realization set \( \Theta^N_{\boldsymbol\eta^{\text{mult}}}\) with probability at least \( 1 - \Delta_N \), where \( \Theta_{\boldsymbol\eta^{\text{mult}}}^N\subset\Theta_{\boldsymbol\eta^{\text{mult}}}\) contains \( \boldsymbol\eta^{\text{mult}} \) and satisfies the regularity conditions below.

    \item The parameter space \( \Theta_{\boldsymbol Z} \) is bounded, and for each \( \boldsymbol\eta^{\text{mult}} \in \Theta_{\boldsymbol\eta^{\text{mult}}}^N\), the function class
    \[
    \mathcal{F}_{1,\boldsymbol\eta^{\text{mult}}} = \left\{ g^{\text{mult}}_j( \boldsymbol \theta_{\boldsymbol Z}, \boldsymbol\eta^{\text{mult}}) : j = 1, \dots, p, \; \boldsymbol \theta_{\boldsymbol Z} \in \Theta_{\boldsymbol Z} \right\}
    \]
    is suitably measurable and its uniform covering entropy obeys
    \[
    \sup_Q \log N\left( \epsilon \| \mathcal{F}_{1,\boldsymbol\eta^{\text{mult}}} \|_{Q,2}, \mathcal{F}_{1,\boldsymbol\eta^{\text{mult}}}, \| \cdot \|_{Q,2} \right) \leq v \log\left( \frac{a}{\epsilon} \right), \quad \text{for all } 0 < \epsilon \leq 1,
    \]
    where \( \mathcal{F}_{1,\boldsymbol\eta^{\text{mult}}} \) has a measurable envelope function with \( \| \mathcal{F}_{1,\boldsymbol\eta^{\text{mult}}} \|_{P,q} \leq c_1 \).

    \item The following conditions on the statistical rates \( r_N, r_N', \lambda_N' \) hold:
    \begin{align*}
        &r_N := \sup_{\boldsymbol\eta^{\text{mult}} \in \Theta_{\boldsymbol\eta^{\text{mult}}}^N,\, \boldsymbol \theta_{\boldsymbol Z} \in \Theta_{\boldsymbol Z}} 
        \left\| \mathbb{E} \left[ \boldsymbol g^{\text{mult}}(\boldsymbol \theta_{\boldsymbol Z}, \boldsymbol\eta^{\text{mult}}) - \mathbb{E}[ \boldsymbol g^{\text{mult}}(\boldsymbol \theta_{\boldsymbol Z}, \boldsymbol\eta^{\text{mult}*}) ] \right] \right\| \leq \delta_N \tau_N, \\
        &r_N' = \lambda_N' := \sup_{\boldsymbol\eta^{\text{mult}} \in \Theta_{\boldsymbol\eta^{\text{mult}}}^N,\, \|\boldsymbol \theta_{\boldsymbol Z} - \boldsymbol \theta_{\boldsymbol Z}^*\| \leq \tau_N}
        \left( \mathbb{E}\left[ \| \boldsymbol g^{\text{mult}}(\boldsymbol \theta_{\boldsymbol Z}, \boldsymbol\eta^{\text{mult}}) - \boldsymbol g^{\text{mult}}(\boldsymbol \theta_{\boldsymbol Z}^*, \boldsymbol\eta^{\text{mult}*}) \|^2 \right] \right)^{1/2}, \\
        &r_N' \sqrt{\log(1/r_N')} \leq \delta_N, \quad \text{and} \\
        &\sup_{r \in (0,1),\, \boldsymbol\eta^{\text{mult}} \in \Theta_{\boldsymbol\eta^{\text{mult}}}^N,\, \| \boldsymbol \theta_{\boldsymbol Z} - \boldsymbol \theta_{\boldsymbol Z}^* \| \leq \tau_N}
      \left\| \frac{\partial^2}{\partial r^2} \mathbb{E} \left[
        \boldsymbol g^{\text{mult}}\left(
\boldsymbol \theta_{\boldsymbol Z}^* + r(\boldsymbol \theta_{\boldsymbol Z} - \boldsymbol \theta_{\boldsymbol Z}^*)\right.\right.\right.\\
&\left.\left.\left.\hspace{4cm},\boldsymbol\eta^{\text{mult}*} + r(\boldsymbol\eta^{\text{mult}} - \boldsymbol\eta^{\text{mult}*})
\right) \right] \right\| \leq \delta_N N^{-1/2}.
    \end{align*}
    \item The variance of the score is non-degenerate: all eigenvalues of the matrix
    \[
    \mathbb{E} \left[ \boldsymbol g^{\text{mult}}(\boldsymbol \theta_{\boldsymbol Z}^*, \boldsymbol\eta^{\text{mult}*}) \, \boldsymbol g^{\text{mult}}(\boldsymbol \theta_{\boldsymbol Z}^*, \boldsymbol\eta^{\text{mult}*})' \right]
    \]
    are bounded below by \( c_0 > 0 \).
\end{enumerate}
\section{Known Selection Probabilities}
\subsection{Theorem S1} \label{sec:pr1}
\begin{theorem}
    Under condition C1 of the main text and assumptions \textit{A}1 and \textit{A}2 in Supplementary Section \ref{sec:assumptions} if $\pi_1(.)$, $\pi_2(.)$, ..., $\pi_K(.)$ are known for each individual in the target population, $\widehat{\boldsymbol \theta}_{\boldsymbol Z}$ estimated from the following weighted logistic regression estimating equation is consistent for $\theta^*_{\boldsymbol Z}$, where $\theta^*_{\boldsymbol Z}$ is the true value of $\theta_{\boldsymbol Z}$ satisfying equation (2.1) of the main text.
    \begin{equation}
     \frac{1}{N}\sum_{i=1}^{N}\frac{S^{\text{mult}}_i}{\pi(\boldsymbol X^{\text{mult}}_i)}\left\{D_i\boldsymbol Z_i-\frac{e^{\boldsymbol\theta_{\boldsymbol Z}'\boldsymbol Z_i}}{(1+e^{\boldsymbol\theta_{\boldsymbol Z}'\boldsymbol Z_i})}\cdot \boldsymbol Z_i\right\}=\mathbf{0}.
\end{equation}\label{thm1}
\end{theorem}
\begin{proof}
 Let  $$\phi_N(\boldsymbol\theta_{\boldsymbol Z})=\frac{1}{N}\sum_{i=1}^{N}\frac{S^{\text{mult}}_i}{\pi(\boldsymbol X^{\text{mult}}_i)}\left\{D_i \boldsymbol Z_i-\frac{e^{\boldsymbol\theta_{\boldsymbol Z}'\boldsymbol Z_i}}{(1+e^{\boldsymbol\theta_{\boldsymbol Z}'\boldsymbol Z_i})}\cdot \boldsymbol Z_i\right\}\cdot$$
From \citet{tsiatis2006semiparametric}, under condition C1 of the main text and assumptions \textit{A}1 and \textit{A}2 in Supplementary Section \ref{sec:assumptions} it is enough to show that $\mathbb{E}(\phi_N(\boldsymbol\theta^*_{\boldsymbol Z}))=0$, in order to prove that $\widehat{\boldsymbol\theta_{\boldsymbol Z}}\xrightarrow{p}\boldsymbol\theta^*_{\boldsymbol Z}$, where $\widehat{\boldsymbol\theta_{\boldsymbol Z}}$ is obtained from solving $\phi_N(\boldsymbol\theta_{\boldsymbol Z})=\mathbf{0}$.
\begin{align*}
     \mathbb{E}[\phi_N(\boldsymbol\theta^*_{\boldsymbol Z})] & = \mathbb{E}\left[\frac{1}{N}\sum_{i=1}^{N}\frac{S^{\text{mult}}_i}{\pi(\boldsymbol X^{\text{mult}}_i)}\left\{D_i \boldsymbol Z_i-\frac{e^{\boldsymbol\theta_{\boldsymbol Z}^{*'}\boldsymbol Z_i}}{(1+e^{\boldsymbol\theta_{\boldsymbol Z}^{*'}\boldsymbol Z_i})}\cdot \boldsymbol Z_i\right\}\right]\\
     & =\frac{1}{N}
   \sum_{i=1}^{N} \mathbb{E}\left[\frac{S^{\text{mult}}_i}{\pi(\boldsymbol X^{\text{mult}}_i)}\left\{D_i \boldsymbol Z_i-\frac{e^{\boldsymbol\theta_{\boldsymbol Z}^{*'}\boldsymbol Z_i}}{(1+e^{\boldsymbol\theta_{\boldsymbol Z}^{*'}\boldsymbol Z_i})}\cdot \boldsymbol Z_i\right\}\right]\\
   & = \frac{1}{N}
   \sum_{i=1}^{N} \mathbb{E}_{\boldsymbol X^{\text{mult}}_i,\boldsymbol Z_{1\cap,i}}\left[\mathbb{E}\left\{\frac{S^{\text{mult}}_i}{\pi(\boldsymbol X^{\text{mult}}_i)}\left.\left(D_i \boldsymbol Z_i-\frac{e^{\boldsymbol\theta_{\boldsymbol Z}^{*'}\boldsymbol Z_i}}{(1+e^{\boldsymbol\theta_{\boldsymbol Z}^{*'}\boldsymbol Z_i})}\cdot \boldsymbol Z_i\right)\right| \boldsymbol X^{\text{mult}}_i,\boldsymbol Z_{1\cap,i}\right\}\right]\\
   & = \frac{1}{N}
   \sum_{i=1}^{N} \mathbb{E}_{\boldsymbol X^{\text{mult}}_i,\boldsymbol Z_{1\cap,i}}\left[\left\{D_i \boldsymbol Z_i-\frac{e^{\boldsymbol\theta_{\boldsymbol Z} ^{*'}\boldsymbol Z_{i}}}{(1+e^{\boldsymbol\theta_{\boldsymbol Z} ^{*'}\boldsymbol Z_i})}\cdot \boldsymbol Z_i\right\}\cdot \frac{1}{\pi(\boldsymbol X^{\text{mult}}_i)}\mathbb{E}\left(S^{\text{mult}}_i|\boldsymbol X^{\text{mult}}_i,\boldsymbol Z_{1\cap,i}\right)\right]\cdot
\end{align*}
Using definition of $\boldsymbol  Z_{1\cap}$ we obtain that $\mathbb{E}(S|\boldsymbol X^{\text{mult}},\boldsymbol  Z_{1\cap})=
\mathbb{E}(S|\boldsymbol X^{\text{mult}})=\pi(\boldsymbol X^{\text{mult}})$. Hence we obtain
\begin{align*}
     &\mathbb{E}[\phi_N(\boldsymbol\theta^*_{\boldsymbol Z})]  =\frac{1}{N}
   \sum_{i=1}^{N} \mathbb{E}_{\boldsymbol X^{\text{mult}}_i,\boldsymbol  Z_{1\cap,i}}\left\{D_i \boldsymbol Z_i-\frac{e^{\boldsymbol\theta_{\boldsymbol Z}^{*'}\boldsymbol Z_i}}{(1+e^{\boldsymbol\theta_{\boldsymbol Z}^{*'}\boldsymbol Z_i})}\cdot \boldsymbol Z_i\right\}\\
   &=\frac{1}{N}
   \sum_{i=1}^{N} \mathbb{E}_{D_i,\boldsymbol Z_i}\left\{D_i \boldsymbol Z_i-\frac{e^{\boldsymbol\theta_{\boldsymbol Z}^{*'}\boldsymbol Z_i}}{(1+e^{\boldsymbol\theta_{\boldsymbol Z}^{*'}\boldsymbol Z_i})}\cdot \boldsymbol Z_i\right\}
\end{align*}
\begin{align*}
   & = \frac{1}{N}
   \sum_{i=1}^{N} \mathbb{E}_{\boldsymbol z_i}\left[\mathbb{E}\left \{\left. \left(D_i \boldsymbol Z_i-\frac{e^{\boldsymbol\theta_{\boldsymbol Z}^{*'}\boldsymbol Z_i}}{(1+e^{\boldsymbol\theta_{\boldsymbol Z}^{*'}\boldsymbol Z_i})}\cdot \boldsymbol Z_i\right) \right|\boldsymbol Z_i\right\}\right] \\
   &= \frac{1}{N}
   \sum_{i=1}^{N} \mathbb{E}\left\{\frac{e^{\boldsymbol\theta_{\boldsymbol Z}^{*'}\boldsymbol Z_i}}{(1+e^{\boldsymbol\theta_{\boldsymbol Z}^{*'}\boldsymbol Z_i})}\cdot \boldsymbol Z_i-\frac{e^{\boldsymbol\theta_{\boldsymbol Z}^{*'}\boldsymbol Z_i}}{(1+e^{\boldsymbol\theta_{\boldsymbol Z}^{*'}\boldsymbol Z_i})} \cdot \boldsymbol Z_i\right\}=\mathbf{0}\cdot
\end{align*}
The last step is obtained from the relation between $D$ and $\boldsymbol Z$ given by equation (2.1) of the main text.
\end{proof}
\subsection{Theorem S2}\label{sec:pr2}
\begin{theorem}
    Under assumptions of Theorem S1 when the selection weights are known and we do not take into consideration estimation of selection model parameters, the asymptotic distribution of $\widehat{\boldsymbol \theta}_{\boldsymbol Z}$ is given by 
$$\sqrt{N}(\widehat{\boldsymbol \theta_{\boldsymbol Z}}-\boldsymbol\theta^*_{\boldsymbol Z})\xrightarrow{d}\mathcal{N}(0,V)\hspace{0.5cm} \text{as} \hspace{0.5cm} N\rightarrow \infty\cdot$$
where
\begin{align*}
    & V=(G_{\boldsymbol \theta_{\boldsymbol Z}^*})^{-1}\cdot \mathbb{E}[\boldsymbol g\cdot \boldsymbol g']\cdot(G_{\boldsymbol \theta_{\boldsymbol Z}^*}^{-1})^{'}\hspace{0.8cm} G_{\boldsymbol \theta_{\boldsymbol Z}^*} = \mathbb{E}\left\{-\frac{S^{\text{mult}}}{\pi(\boldsymbol X^{\text{mult}})}\cdot \frac{e^{\boldsymbol \theta_{\boldsymbol Z}^{*'}\boldsymbol Z}}{(1+e^{\boldsymbol \theta_{\boldsymbol Z}^{*'}\boldsymbol Z})^2}\cdot \boldsymbol Z\boldsymbol Z'\right\}\\
    & \boldsymbol g(\boldsymbol \theta_{\boldsymbol Z}^*) = \frac{S^{\text{mult}}}{\pi(\boldsymbol X^{\text{mult}})}\left\{D \boldsymbol Z-\frac{e^{\boldsymbol \theta_{\boldsymbol Z}^{*'}\boldsymbol Z}}{(1+e^{\boldsymbol \theta_{\boldsymbol Z}^{*'}\boldsymbol Z})}\cdot \boldsymbol Z\right\}\cdot
\end{align*}\label{thm2}
\end{theorem}
\begin{proof}
By \citet{tsiatis2006semiparametric}'s arguments for a Z-estimation problem under assumptions of Theorem S1 we obtain  
$$\sqrt{N}(\widehat{\boldsymbol \theta_{\boldsymbol Z}}-\boldsymbol\theta^*_{\boldsymbol Z})\xrightarrow{d}\mathcal{N}(0,V)\hspace{0.5cm} \text{as} \hspace{0.5cm} N\rightarrow \infty\cdot$$
where
\begin{align*}
    & V=(G_{\boldsymbol \theta_{\boldsymbol Z}^*})^{-1}\cdot \mathbb{E}[\boldsymbol g\cdot \boldsymbol g']\cdot(G_{\boldsymbol \theta_{\boldsymbol Z}^*}^{-1})^{'}\hspace{0.8cm} G_{\boldsymbol \theta_{\boldsymbol Z}^*} = \left.\mathbb{E}\left\{\frac{\partial g(\boldsymbol \theta_{\boldsymbol Z}^*)}{\partial \boldsymbol \theta_{\boldsymbol Z}}\right\}\right|_{\boldsymbol\theta=\boldsymbol \theta_{\boldsymbol Z}^*}\\
    & \boldsymbol g(\boldsymbol \theta_{\boldsymbol Z}^*) = \frac{S^{\text{mult}}}{\pi(\boldsymbol X^{\text{mult}})}\left\{D \boldsymbol Z-\frac{e^{\boldsymbol \theta_{\boldsymbol Z}^{*'}\boldsymbol Z}}{(1+e^{\boldsymbol \theta_{\boldsymbol Z}^{*'}\boldsymbol Z})}\cdot \boldsymbol Z\right\}\cdot
\end{align*}
    This proof just requires the calculation of $G_{\boldsymbol \theta_{\boldsymbol Z}^*}$.
\subsubsection*{Calculation for $G_{\boldsymbol \theta_{\boldsymbol Z}^*}$}
\begin{align*}
    & \left.\frac{\partial \boldsymbol g(\boldsymbol \theta_{\boldsymbol Z})}{\partial \boldsymbol \theta_{\boldsymbol Z}}\right|_{\boldsymbol \theta_{\boldsymbol Z}= \boldsymbol \theta_{\boldsymbol Z}^*}=\left.\frac{\partial}{\partial \boldsymbol \theta_{\boldsymbol Z}}\left[\frac{S^{\text{mult}}}{\pi(\boldsymbol X^{\text{mult}})}\left\{D\boldsymbol Z-\frac{e^{\boldsymbol \theta_{\boldsymbol Z}^{*'}\boldsymbol Z}}{(1+e^{\boldsymbol \theta_{\boldsymbol Z}^{*'}\boldsymbol Z})}\cdot \boldsymbol Z\right\}\right]\right|_{\boldsymbol \theta_{\boldsymbol Z}= \boldsymbol \theta_{\boldsymbol Z}^*}\\
    & = -\frac{S^{\text{mult}}}{\pi(\boldsymbol X^{\text{mult}})} \cdot \frac{e^{\boldsymbol \theta_{\boldsymbol Z}^{*'}\boldsymbol Z}}{(1+e^{\boldsymbol \theta_{\boldsymbol Z}^{*'}\boldsymbol Z})^2}\cdot \boldsymbol Z\boldsymbol Z'
\end{align*} 
Therefore we obtain 
\begin{align*}
     G_{\boldsymbol \theta_{\boldsymbol Z}^*} = \mathbb{E}\left\{-\frac{S^{\text{mult}}}{\pi(\boldsymbol X^{\text{mult}})} \cdot \frac{e^{\boldsymbol \theta_{\boldsymbol Z}^{*'} \boldsymbol Z}}{(1+e^{\boldsymbol \theta_{\boldsymbol Z}^{*'}\boldsymbol Z})^2}\cdot \boldsymbol Z \boldsymbol Z'\right\}\cdot
\end{align*}
\end{proof}
\subsection{Theorem S3}\label{sec:pr3}
\begin{theorem}
     Under Condition C1 in the main text and assumptions \textit{A}1, \textit{A}2 and \textit{A}3 in Supplementary Section \ref{sec:assumptions}, $\frac{1}{N} \cdot \widehat{G_{\boldsymbol \theta_{\boldsymbol Z}}}^{-1}\cdot \widehat{E}\cdot (\widehat{G_{\boldsymbol \theta_{\boldsymbol Z}}}^{-1})^{'}$ is a consistent estimator of the asymptotic variance of $\widehat{\boldsymbol \theta}_{\boldsymbol Z}$ where
\begin{align*}
    & \widehat{G_{\boldsymbol \theta_{\boldsymbol Z}}}=\frac{1}{N}\sum_{i=1}^N \left\{\frac{S^{\text{mult}}_i}{\pi(\boldsymbol X^{\text{mult}}_i)}\cdot \frac{e^{\widehat{\boldsymbol \theta_{\boldsymbol Z}}'\boldsymbol Z_i} }{(1+e^{\widehat{\boldsymbol \theta_{\boldsymbol Z}}'\boldsymbol Z_i})^2}\cdot \boldsymbol Z_i \boldsymbol Z_i'\right\}\cdot \\
    & \widehat{E}=\frac{1}{N}\sum_{i =1}^N S^{\text{mult}}_i \cdot \left\{\frac{1}{\pi(\boldsymbol X^{\text{mult}}_i)}\right\}^2\left\{D_i-\frac{e^{\widehat{\boldsymbol \theta_{\boldsymbol Z}}'\boldsymbol Z_i}}{(1+e^{\widehat{\boldsymbol\theta}' \boldsymbol Z_i})}\right\}^2\cdot \boldsymbol Z_i \boldsymbol Z_i' \cdot
\end{align*}\label{thm3}
\end{theorem}
\begin{proof}
First we prove that as $N \rightarrow \infty$
$$\frac{1}{N}\sum_{i=1}^N \left\{\frac{S^{\text{mult}}_i}{\pi(\boldsymbol X^{\text{mult}}_i)}\cdot \frac{e^{\widehat{\boldsymbol \theta_{\boldsymbol Z}}'\boldsymbol Z_i} }{(1+e^{\widehat{\boldsymbol \theta_{\boldsymbol Z}}'\boldsymbol Z_i})^2}\cdot \boldsymbol Z_i \boldsymbol Z_i'\right\}
\xrightarrow{p}\mathbb{E}\left\{\frac{S^{\text{mult}}}{\pi(\boldsymbol X^{\text{mult}})}\cdot \frac{e^{\boldsymbol \theta_{\boldsymbol Z}^{*'}\boldsymbol Z}}{(1+e^{\boldsymbol \theta_{\boldsymbol Z}^{*'}\boldsymbol Z})^2}\cdot \boldsymbol Z\boldsymbol Z'\right\} \cdot$$
Using Condition C1 in the main text and assumptions \textit{A}1, \textit{A}2 and \textit{A}3 in Supplementary Section \ref{sec:assumptions} and Uniform Law of Large Numbers (ULLN), we obtain 
\begin{align}
    \sup_{\boldsymbol \theta_{\boldsymbol Z} \in \Theta_{\boldsymbol Z}}\left|\left|
    \frac{1}{N}\sum_{i=1}^N \left\{\frac{S^{\text{mult}}_i}{\pi(\boldsymbol X^{\text{mult}}_i)}\cdot \frac{e^{\boldsymbol \theta_{\boldsymbol Z}'\boldsymbol Z_i} }{(1+e^{\boldsymbol \theta_{\boldsymbol Z}'\boldsymbol Z_i})^2}\cdot \boldsymbol Z_i \boldsymbol Z_i'\right\}-\mathbb{E}\left\{\frac{S^{\text{mult}}}{\pi(\boldsymbol X^{\text{mult}})}\cdot \frac{e^{\boldsymbol \theta_{\boldsymbol Z}^{'}\boldsymbol Z}}{(1+e^{\boldsymbol \theta_{\boldsymbol Z}^{'}\boldsymbol Z})^2}\cdot \boldsymbol Z\boldsymbol Z'\right\}
    \right|\right|\xrightarrow{p}\mathbf{0}\cdot \label{eq:st11}
\end{align}
Since this above expression holds for any $\boldsymbol \theta_{\boldsymbol Z} \in \Theta_{\boldsymbol Z}$, therefore it is true for $\widehat{\boldsymbol \theta}_{\boldsymbol Z}$. This implies
\begin{align}
    \left|\left|
    \frac{1}{N}\sum_{i=1}^N \left. \left\{\frac{S^{\text{mult}}_i}{\pi(\boldsymbol X^{\text{mult}}_i)}\cdot \frac{e^{\widehat{\boldsymbol \theta_{\boldsymbol Z}}'\boldsymbol Z_i} }{(1+e^{\widehat{\boldsymbol \theta_{\boldsymbol Z}}'\boldsymbol Z_i})^2}\cdot \boldsymbol Z_i \boldsymbol Z_i'\right\}-\mathbb{E}\left\{\frac{S^{\text{mult}}}{\pi(\boldsymbol X^{\text{mult}})}\cdot \frac{e^{\boldsymbol \theta_{\boldsymbol Z}^{'}\boldsymbol Z}}{(1+e^{\boldsymbol \theta_{\boldsymbol Z}^{'}\boldsymbol Z})^2}\cdot \boldsymbol Z\boldsymbol Z'\right\}\right|_{\boldsymbol \theta_{\boldsymbol Z}=\widehat{\boldsymbol \theta_{\boldsymbol Z}}}
    \right|\right|\xrightarrow{p}\mathbf{0}\cdot \label{eq:st12}
\end{align}
Using Triangle Inequality we obtain
\begin{align}
    &  \left|\left|\frac{1}{N}\sum_{i=1}^N \left\{\frac{S^{\text{mult}}_i}{\pi(\boldsymbol X^{\text{mult}}_i)}\cdot \frac{e^{\widehat{\boldsymbol \theta_{\boldsymbol Z}}'\boldsymbol Z_i} }{(1+e^{\boldsymbol \theta_{\boldsymbol Z}'\boldsymbol Z_i})^2}\cdot \boldsymbol Z_i \boldsymbol Z_i'\right\}-\mathbb{E}\left\{\frac{S^{\text{mult}}}{\pi(\boldsymbol X^{\text{mult}})}\cdot \frac{e^{\boldsymbol \theta_{\boldsymbol Z}^{*'}\boldsymbol Z}}{(1+e^{\boldsymbol \theta_{\boldsymbol Z}^{*'}\boldsymbol Z})^2}\cdot \boldsymbol Z\boldsymbol Z'\right\}
    \right|\right|\leq \label{eq:st132}\\
    &  \left|\left|
    \frac{1}{N}\sum_{i=1}^N \left. \left\{\frac{S^{\text{mult}}_i}{\pi(\boldsymbol X^{\text{mult}}_i)}\cdot \frac{e^{\widehat{\boldsymbol \theta_{\boldsymbol Z}}'\boldsymbol Z_i} }{(1+e^{\widehat{\boldsymbol \theta_{\boldsymbol Z}}'\boldsymbol Z_i})^2}\cdot \boldsymbol Z_i \boldsymbol Z_i'\right\}-\mathbb{E}\left\{\frac{S^{\text{mult}}}{\pi(\boldsymbol X^{\text{mult}})}\cdot \frac{e^{\boldsymbol \theta_{\boldsymbol Z}^{'}\boldsymbol Z}}{(1+e^{\boldsymbol \theta_{\boldsymbol Z}^{'}\boldsymbol Z})^2}\cdot \boldsymbol Z\boldsymbol Z'\right\}\right|_{\boldsymbol \theta_{\boldsymbol Z}=\widehat{\boldsymbol \theta_{\boldsymbol Z}}}
    \right|\right|+\label{eq:st131}\\
    & \left. \left|\left|
   \mathbb{E}\left\{\frac{S^{\text{mult}}}{\pi(\boldsymbol X^{\text{mult}})}\cdot \frac{e^{\boldsymbol \theta_{\boldsymbol Z}^{'}\boldsymbol Z}}{(1+e^{\boldsymbol \theta_{\boldsymbol Z}^{'}\boldsymbol Z})^2}\cdot \boldsymbol Z\boldsymbol Z'\right\}\right|_{\boldsymbol \theta_{\boldsymbol Z}=\widehat{\boldsymbol \theta_{\boldsymbol Z}}}- \mathbb{E}\left\{\frac{S^{\text{mult}}}{\pi(\boldsymbol X^{\text{mult}})}\cdot \frac{e^{\boldsymbol \theta_{\boldsymbol Z}^{*'}\boldsymbol Z}}{(1+e^{\boldsymbol \theta_{\boldsymbol Z}^{*'}\boldsymbol Z})^2}\cdot \boldsymbol Z\boldsymbol Z'\right\}
    \right|\right|\cdot \label{eq:st13}
\end{align}
Since we proved that $\widehat{\boldsymbol \theta_{\boldsymbol Z}}\xrightarrow{p} \boldsymbol \theta_{\boldsymbol Z}^{*}$ in Theorem S1, therefore by Continuous Mapping Theorem
\begin{align}
    \left. \left|\left|
   \mathbb{E}\left\{\frac{S^{\text{mult}}}{\pi(\boldsymbol X^{\text{mult}})}\cdot \frac{e^{\boldsymbol \theta_{\boldsymbol Z}^{'}\boldsymbol Z}}{(1+e^{\boldsymbol \theta_{\boldsymbol Z}^{'}\boldsymbol Z})^2}\cdot \boldsymbol Z\boldsymbol Z'\right\}\right|_{\boldsymbol \theta_{\boldsymbol Z}=\widehat{\boldsymbol \theta_{\boldsymbol Z}}}- \mathbb{E}\left\{\frac{S^{\text{mult}}}{\pi(\boldsymbol X^{\text{mult}})}\cdot \frac{e^{\boldsymbol \theta_{\boldsymbol Z}^{*'}\boldsymbol Z}}{(1+e^{\boldsymbol \theta_{\boldsymbol Z}^{*'}\boldsymbol Z})^2}\cdot \boldsymbol Z\boldsymbol Z'\right\}
    \right|\right|\xrightarrow{p} \mathbf{0}\cdot \label{eq:st14}
\end{align}
Using equations \eqref{eq:st12}, \eqref{eq:st132}, \eqref{eq:st131}, \eqref{eq:st13} and \eqref{eq:st14} we obtain 
$$\frac{1}{N}\sum_{i=1}^N \left\{\frac{S^{\text{mult}}_i}{\pi(\boldsymbol X^{\text{mult}}_i)}\cdot \frac{e^{\widehat{\boldsymbol \theta_{\boldsymbol Z}}'\boldsymbol Z_i} }{(1+e^{\widehat{\boldsymbol \theta_{\boldsymbol Z}}'\boldsymbol Z_i})^2}\cdot \boldsymbol Z_i \boldsymbol Z_i'\right\}
\xrightarrow{p}\mathbb{E}\left\{\frac{S^{\text{mult}}}{\pi(\boldsymbol X^{\text{mult}})}\cdot \frac{e^{\boldsymbol \theta_{\boldsymbol Z}^{*'}\boldsymbol Z}}{(1+e^{\boldsymbol \theta_{\boldsymbol Z}^{*'}\boldsymbol Z})^2}\cdot \boldsymbol Z\boldsymbol Z'\right\}\cdot$$
Therefore we obtain  
$$-\widehat{G_{\boldsymbol \theta_{\boldsymbol Z}}}\xrightarrow{p} -G_{\boldsymbol \theta_{\boldsymbol Z}^*}\hspace{0.5cm}
    \text{which implies} \hspace{0.5cm} \widehat{G_{\boldsymbol \theta_{\boldsymbol Z}}}\xrightarrow{p} G_{\boldsymbol \theta_{\boldsymbol Z}^*}\cdot$$
Using the exact same set of arguments we obtain 
\begin{align*}
    & \widehat{E} =\frac{1}{N}\sum_{i =1}^N S^{\text{mult}}_i \cdot \left\{\frac{1}{\pi(\boldsymbol X^{\text{mult}}_i)}\right\}^2\left\{D_i-\frac{e^{\widehat{\boldsymbol \theta_{\boldsymbol Z}}'\cdot \boldsymbol Z_i}}{(1+e^{\hat{\theta}' \cdot \boldsymbol Z_i})}\right\}^2\cdot \boldsymbol Z_i \boldsymbol Z_i'\\
     & \xrightarrow{p}\mathbb{E}[\boldsymbol g(\boldsymbol \theta_{\boldsymbol Z}^*)\cdot \boldsymbol g(\boldsymbol \theta_{\boldsymbol Z}^*)']=  \mathbb{E}\left[S \cdot \left\{\frac{1}{\pi(\boldsymbol X^{\text{mult}})}\right\}^2\cdot\left\{D-\frac{e^{\boldsymbol \theta_{\boldsymbol Z}^{*'}\boldsymbol Z}}{(1+e^{\boldsymbol \theta_{\boldsymbol Z}^{*'}\boldsymbol Z})}\right\}^2\cdot \boldsymbol Z\boldsymbol Z'\right]\cdot
\end{align*}
Combining together the consistency of $\widehat{G_{\boldsymbol \theta_{\boldsymbol Z}}}$ and $\widehat{E}$ we obtain
$$\widehat{V}=\widehat{G_{\boldsymbol \theta_{\boldsymbol Z}}}^{-1}\cdot \widehat{E}\cdot(\widehat{G_{\boldsymbol \theta_{\boldsymbol Z}}}^{-1})^{'}\xrightarrow{p}V=(G_{\boldsymbol \theta_{\boldsymbol Z}^*})^{-1}\cdot \mathbb{E}[g(\boldsymbol \theta_{\boldsymbol Z}^*)\cdot \boldsymbol g(\boldsymbol \theta_{\boldsymbol Z}^*)'].(G_{\boldsymbol \theta_{\boldsymbol Z}^*}^{-1})^{'}\cdot$$
From Theorem \ref{thm2}, $\sqrt{N}(\widehat{\boldsymbol \theta_{\boldsymbol Z}}-\boldsymbol \theta_{\boldsymbol Z}^*)\xrightarrow{d}\mathcal{N}\{0,(G_{\boldsymbol \theta_{\boldsymbol Z}^*})^{-1}\cdot\mathbb{E}[g(\boldsymbol \theta_{\boldsymbol Z}^*)\cdot \boldsymbol g(\boldsymbol \theta_{\boldsymbol Z}^*)']\cdot (G_{\boldsymbol \theta_{\boldsymbol Z}^*}^{-1})^{'}\}$. Therefore we obtain
\begin{align*}
    & \text{Var}\{ \sqrt{N}(\widehat{\boldsymbol \theta_{\boldsymbol Z}}-\boldsymbol \theta_{\boldsymbol Z}^*)\}=(G_{\boldsymbol \theta_{\boldsymbol Z}^*})^{-1}\cdot\mathbb{E}[g(\boldsymbol \theta_{\boldsymbol Z}^*)\cdot \boldsymbol g(\boldsymbol \theta_{\boldsymbol Z}^*)']\cdot(G_{\boldsymbol \theta_{\boldsymbol Z}^*}^{-1})^{'} + O_p(1)\\
    & N \cdot\text{Var}(\widehat{\boldsymbol \theta_{\boldsymbol Z}})=(G_{\boldsymbol \theta_{\boldsymbol Z}^*})^{-1}\cdot\mathbb{E}[g(\boldsymbol \theta_{\boldsymbol Z}^*)\cdot \boldsymbol g(\boldsymbol \theta_{\boldsymbol Z}^*)']\cdot(G_{\boldsymbol \theta_{\boldsymbol Z}^*}^{-1})^{'} + O_p(1)\\
     &\text{Var} (\widehat{\boldsymbol \theta_{\boldsymbol Z}})=\frac{1}{N}\cdot(G_{\boldsymbol \theta_{\boldsymbol Z}^*})^{-1}\cdot\mathbb{E}[g(\boldsymbol \theta_{\boldsymbol Z}^*)\cdot \boldsymbol g(\boldsymbol \theta_{\boldsymbol Z}^*)']\cdot(G_{\boldsymbol \theta_{\boldsymbol Z}^*}^{-1})^{'} + O_p\left(\frac{1}{N}\right)\\
     & \text{Var}(\widehat{\boldsymbol \theta_{\boldsymbol Z}})= \frac{1}{N}\cdot\widehat{G_{\boldsymbol \theta_{\boldsymbol Z}}}^{-1}\cdot\widehat{E}\cdot(\widehat{G_{\boldsymbol \theta_{\boldsymbol Z}}}^{-1})^{'}+ o_p\left(1\right)=\frac{\widehat{V}}{N}+ o_p\left(1\right)\cdot
\end{align*} 
Therefore $\frac{\widehat{V}}{N}$ is a consistent estimator of asymptotic variance of $\widehat{\boldsymbol \theta}_{\boldsymbol Z}$.
\end{proof}
\section{Joint Pseudolikelihood-based Estimating equation (JPL)}\label{sec:jpl}
\subsection{Theorem S4}\label{sec:pr4}
\begin{theorem}
     Under condition C1 in the main text and assumptions \textit{A}1 and \textit{A}2 in Supplementary Section \ref{sec:assumptions} and assuming the selection model is correctly specified, that is, $\pi(\boldsymbol X^{\text{mult}},\boldsymbol \alpha^{\text{mult}*})=P(S=1|\boldsymbol X^{\text{mult}})=\pi(\boldsymbol X^{\text{mult}})$ where $\widehat{\boldsymbol \alpha^{\text{mult}}}\xrightarrow{p}\boldsymbol \alpha^{\text{mult}*}$ and $\boldsymbol \alpha^{\text{mult}*}$ is the true value of $\boldsymbol \alpha$, then $\widehat{\boldsymbol \theta}_{\boldsymbol Z}$ estimated using JPL is consistent for $\boldsymbol \theta_{\boldsymbol Z}^*$ as $N\rightarrow \infty$. \label{thm4}
\end{theorem}
\begin{proof}
    In this case, the estimating equation consists of both selection model and disease model parameter estimation. The two step estimating equation is given by
\begin{align*}
     & \delta_n(\boldsymbol \alpha^{\text{mult}})=\begin{pmatrix}
     \delta_{1_n}(\boldsymbol \alpha_1)= \frac{1}{N}\sum_{i=1}^N S_{1i}\boldsymbol X_{1i}- \frac{1}{N}\cdot \sum_{i=1}^ N \left(\frac{S_{\text{ext},i}}{\pi_{\text{ext},i}}\right)\cdot \pi_k(\boldsymbol X_{1i},\boldsymbol \alpha_1)\cdot \boldsymbol X_{1i}\\
     \vdots\\
       \delta_{K_n}(\boldsymbol \alpha_K)=\frac{1}{N}\sum_{i=1}^N S_{Ki}\boldsymbol X_{Ki}- \frac{1}{N}\cdot \sum_{i=1}^ N \left(\frac{S_{\text{ext},i}}{\pi_{\text{ext},i}}\right)\cdot \pi_k(\boldsymbol X_{Ki},\boldsymbol \alpha_K)\cdot \boldsymbol X_{Ki}\\
         \end{pmatrix}\\
     & \phi_n(\boldsymbol \theta_{\boldsymbol Z},\widehat{\boldsymbol \alpha^{\text{mult}}})= \frac{1}{N}\sum_{i=1}^{N}\frac{S^{\text{mult}}_i}{\pi(\boldsymbol X^{\text{mult}}_i,\widehat{\boldsymbol \alpha^{\text{mult}}})}\left\{D_i \boldsymbol Z_i-\frac{e^{\boldsymbol\theta'\boldsymbol Z_i}}{(1+e^{\boldsymbol\theta'\boldsymbol Z_i})}\cdot \boldsymbol Z_i\right\}\cdot
\end{align*}
From \citet{tsiatis2006semiparametric} to show  $\widehat{\boldsymbol \theta_{\boldsymbol Z}}\xrightarrow{p}\boldsymbol \theta_{\boldsymbol Z}^*$ in a two step estimation procedure, we need to prove $\mathbb{E}(\phi_N(\boldsymbol \theta_{\boldsymbol Z}^*,\boldsymbol \alpha^{\text{mult}*}))=\mathbf{0}$. Under correct specification of selection model, $\pi(\boldsymbol X^{\text{mult}},\boldsymbol \alpha^{\text{mult}*})=P(S=1|\boldsymbol X^{\text{mult}})=\pi(\boldsymbol X^{\text{mult}})$ where $\widehat{\boldsymbol \alpha^{\text{mult}}}\xrightarrow{p}\boldsymbol \alpha^{\text{mult}*}$ and $\boldsymbol \alpha^{\text{mult}*}$ is the true value of $\boldsymbol \alpha^{\text{mult}}$. Using this equality, the proof to show that $\mathbb{E}(\phi_N(\boldsymbol \theta_{\boldsymbol Z}^*,\boldsymbol \alpha^{\text{mult}*}))=\mathbf{0}$ is exactly as Theorem S1. Therefore we obtain $\widehat{\boldsymbol \theta}_{\boldsymbol Z}$ is a consistent estimator of $\boldsymbol \theta_{\boldsymbol Z}^*$.
\end{proof}
\subsection{Theorem S5}\label{sec:pr5}
\begin{theorem}
    Under assumptions of Theorem S4 and assumption \textit{A}4 in Supplementary Section \ref{sec:assumptions} and suppose $\widehat{\boldsymbol \theta_{\boldsymbol Z}}\xrightarrow{p}\boldsymbol \theta_{\boldsymbol Z}^*$, the asymptotic distribution of $\widehat{\boldsymbol \theta}_{\boldsymbol Z}$ using JPL is given by
\begin{equation}
\sqrt{N}(\widehat{\boldsymbol \theta_{\boldsymbol Z}}-\boldsymbol\theta^*_{\boldsymbol Z})\xrightarrow{d}\mathcal{N}(\mathbf{0},V).
\end{equation}
where
\begin{align*}
    & V=(G_{\boldsymbol \theta_{\boldsymbol Z}^*})^{-1}.\mathbb{E}[\{\boldsymbol g(\boldsymbol \theta_{\boldsymbol Z}^*,\boldsymbol \alpha^{\text{mult}*})+G_{\boldsymbol \alpha^{\text{mult}*}}.\boldsymbol \Psi(\boldsymbol \alpha^{\text{mult}*})\}\\
    & \hspace{5cm}\{\boldsymbol g(\boldsymbol \theta_{\boldsymbol Z}^*,\boldsymbol \alpha^{\text{mult}*})+G_{\boldsymbol \alpha^{\text{mult}*}}.\boldsymbol \Psi(\boldsymbol \alpha^{\text{mult}*})\}'].(G_{\boldsymbol \theta_{\boldsymbol Z}^*}^{-1})^{'}\\
    & G_{\boldsymbol \theta_{\boldsymbol Z}^*} = \mathbb{E}\left\{-\frac{S^{\text{mult}}}{\pi(\boldsymbol X^{\text{mult}},\boldsymbol \alpha^{\text{mult}*})}\cdot \frac{e^{\boldsymbol \theta_{\boldsymbol Z}^{*'}\boldsymbol Z}}{(1+e^{\boldsymbol \theta_{\boldsymbol Z}^{*'}\boldsymbol Z})^2}\cdot \boldsymbol Z\boldsymbol Z'\right\}\\
    & \boldsymbol g(\boldsymbol \theta_{\boldsymbol Z}^*,\boldsymbol \alpha^{\text{mult}*}) = \frac{S^{\text{mult}}}{\pi(\boldsymbol X^{\text{mult}},\boldsymbol \alpha^{\text{mult}*})}\left\{D \boldsymbol Z-\frac{e^{\boldsymbol \theta_{\boldsymbol Z}^{*'}\boldsymbol Z}}{(1+e^{\boldsymbol \theta_{\boldsymbol Z}^{*'}\boldsymbol Z})}\cdot \boldsymbol Z\right\}
\end{align*}
\setlength\arraycolsep{-20pt}
\begin{align*}
     & G_{\boldsymbol \alpha^{\text{mult}*}}=\mathbb{E}\left[-\frac{S^{\text{mult}}}{\pi(\boldsymbol X^{\text{mult}},\boldsymbol\alpha^{\text{mult}*})^2}.\{1-\pi_1(\boldsymbol X_1,\boldsymbol \alpha^*_1)\}\cdot \cdot \cdot\{1-\pi_K(\boldsymbol X_K,\boldsymbol \alpha^*_K)\}\cdot\right.\\
    & \hspace{2cm} \left.\boldsymbol Z\cdot \{\pi_1(\boldsymbol X_1,\boldsymbol \alpha_1^*)\cdot \boldsymbol X_1',\cdot \cdot \cdot,\pi_K(\boldsymbol X_K,\boldsymbol \alpha^*_K)\cdot \boldsymbol X_K'\}\cdot
   \left\{D-\frac{e^{\boldsymbol \theta_{\boldsymbol Z}^{*'}\boldsymbol Z}}{(1+e^{\boldsymbol \theta_{\boldsymbol Z}^{*'}\boldsymbol Z})}\right\}\right]\\
    & \boldsymbol h(\boldsymbol X^{\text{mult}},\boldsymbol M,\boldsymbol \alpha^{\text{mult}*}) = \begin{pmatrix}
   S_1\boldsymbol X_1-\frac{\pi_1(\boldsymbol X_1,\boldsymbol \alpha_1)}{\pi_{\text{ext}}(\boldsymbol M)} \cdot S_{\text{ext}}\cdot \boldsymbol X_1\\
     \vdots\\
         S_K\boldsymbol X_K-\frac{\pi_1(\boldsymbol X_K,\boldsymbol \alpha_K)}{\pi_{\text{ext}}(\boldsymbol M)} \cdot S_{\text{ext}}\cdot \boldsymbol X_K\\
         \end{pmatrix}\\
         &  H = - \mathbb{E} \begin{pmatrix}
     \frac{\pi_1(\boldsymbol X_1,\boldsymbol \alpha_1)}{\pi_{\text{ext}}(\boldsymbol M)}\{1-\pi_1(\boldsymbol X_1,\boldsymbol \alpha_1)\}S_{\text{ext}}\boldsymbol X_1\boldsymbol X_1' & \hdots & 0 \\
     \vdots & \ddots & \vdots \\
     0 & \hdots &   \frac{\pi_K(\boldsymbol X_K,\boldsymbol \alpha_K)}{\pi_{\text{ext}}(\boldsymbol M)}\{1-\pi_K(\boldsymbol X_K,\boldsymbol \alpha_K)\} S_{\text{ext}} \boldsymbol X_K\boldsymbol X_K'
         \end{pmatrix}\\
    & \boldsymbol \Psi(\boldsymbol \alpha^{\text{mult}*})=-H^{-1}\cdot  \boldsymbol h(\boldsymbol X^{\text{mult}},\boldsymbol M,\boldsymbol \alpha^{\text{mult}*})
\end{align*} \label{thm5}
\end{theorem}
\begin{proof}
By \citet{tsiatis2006semiparametric}'s arguments on a two step Z-estimation problem, under assumptions of Theorem S4  and assumption \textit{A}4 in Supplementary Section \ref{sec:assumptions} and since $\widehat{\boldsymbol \theta_{\boldsymbol Z}}\xrightarrow{p}\boldsymbol \theta_{\boldsymbol Z}^*$, we obtain that 
\begin{align*}
    &\sqrt{N}(\widehat{\boldsymbol \theta_{\boldsymbol Z}}-\boldsymbol\theta_{\boldsymbol Z}^*)\xrightarrow{d}\mathcal{N}(0,(G_{\boldsymbol \theta_{\boldsymbol Z}^*})^{-1}.\mathbb{E}[\{\boldsymbol g(\boldsymbol \theta_{\boldsymbol Z}^*,\boldsymbol \alpha^{\text{mult}*})+G_{\boldsymbol \alpha^{\text{mult}*}}.\boldsymbol \Psi(\boldsymbol \alpha^{\text{mult}*})\}\\
    & \hspace{4cm}\{\boldsymbol g(\boldsymbol \theta_{\boldsymbol Z}^*,\boldsymbol \alpha^{\text{mult}*})+G_{\boldsymbol \alpha^{\text{mult}*}}.\boldsymbol \Psi(\boldsymbol \alpha^{\text{mult}*})\}'].(G_{\boldsymbol \theta_{\boldsymbol Z}^*}^{-1})^{'}\cdot
\end{align*}
We derive the expression of each of the terms in the above expression.
\begin{align*}
    & G_{\boldsymbol \theta_{\boldsymbol Z}^*}=\left.\mathbb{E}\left\{\frac{\partial g(\boldsymbol \theta_{\boldsymbol Z},\boldsymbol \alpha^{\text{mult}*})}{\partial \boldsymbol \theta_{\boldsymbol Z}}\right\}\right|_{\boldsymbol \theta_{\boldsymbol Z}=\boldsymbol \theta_{\boldsymbol Z}^*}\hspace{1cm} \left.G_{\boldsymbol \alpha^{\text{mult}*}}=\mathbb{E}\left\{\frac{\partial g(\boldsymbol \theta_{\boldsymbol Z}^*,\boldsymbol \alpha^{\text{mult}})}{\partial \boldsymbol \alpha^{\text{mult}}}\right\}\right|_{\boldsymbol\alpha^{\text{mult}}=\boldsymbol \alpha^{\text{mult}*}}\\
    & H = \left.\mathbb{E}\left\{\frac{\partial \boldsymbol h(\boldsymbol \alpha^{\text{mult}})}{\partial \boldsymbol \alpha^{\text{mult}}}\right\}\right|_{\boldsymbol\alpha^{\text{mult}}=\boldsymbol \alpha^{\text{mult}*}}\hspace{1.9cm} \boldsymbol \Psi(\boldsymbol \alpha^{\text{mult}*})=-H^{-1}\boldsymbol h(\boldsymbol \alpha^{\text{mult}*}) \cdot
\end{align*}
First we calculate $G_{\boldsymbol \theta_{\boldsymbol Z}^*}$.
\subsubsection*{Calculation for $G_{\boldsymbol \theta_{\boldsymbol Z}^*}$}
\begin{align*}
    & \left.\frac{\partial \boldsymbol g(\boldsymbol \theta_{\boldsymbol Z},\boldsymbol \alpha^{\text{mult}*})}{\partial \boldsymbol \theta_{\boldsymbol Z}}\right|_{\boldsymbol \theta_{\boldsymbol Z}=\boldsymbol \theta_{\boldsymbol Z}^*}= \left.\frac{\partial}{\partial \boldsymbol \theta_{\boldsymbol Z}}\left[\frac{S^{\text{mult}}}{\pi(\boldsymbol X^{\text{mult}},\boldsymbol \alpha^{\text{mult}*})}\left\{D\boldsymbol Z-\frac{e^{\boldsymbol \theta_{\boldsymbol Z}^{*'}\boldsymbol Z}}{(1+e^{\boldsymbol \theta_{\boldsymbol Z}^{*'}\boldsymbol Z})}\cdot \boldsymbol Z\right\}\right]\right|_{\boldsymbol \theta_{\boldsymbol Z}=\boldsymbol \theta_{\boldsymbol Z}^*}\\
    & = -\frac{S^{\text{mult}}}{\pi(\boldsymbol X^{\text{mult}},\boldsymbol \alpha^{\text{mult}*})} \cdot \frac{e^{\boldsymbol \theta_{\boldsymbol Z}^{*'}\boldsymbol Z}}{(1+e^{\boldsymbol \theta_{\boldsymbol Z}^{*'}\boldsymbol Z})^2}\cdot \boldsymbol Z\boldsymbol Z'\cdot
\end{align*}
Therefore we obtain 
\begin{align*}
    G_{\boldsymbol \theta_{\boldsymbol Z}^*} = \mathbb{E}\left[-\frac{S^{\text{mult}}}{\pi(\boldsymbol X^{\text{mult}},\boldsymbol \alpha^{\text{mult}*})} \cdot \frac{e^{\boldsymbol \theta_{\boldsymbol Z}^{*'}\boldsymbol Z}}{(1+e^{\boldsymbol \theta_{\boldsymbol Z}^{*'}\boldsymbol Z})^2}\cdot \boldsymbol Z \boldsymbol Z'\right]\cdot
\end{align*}
Next we calculate $G_{\boldsymbol \alpha^{\text{mult}*}}$.
\subsubsection*{Calculation for $G_{\boldsymbol \alpha^{\text{mult}*}}$}
\begin{align*}
    & \left.\frac{\partial \boldsymbol g(\boldsymbol \theta_{\boldsymbol Z}^*,\boldsymbol \alpha^{\text{mult}})}{\partial \boldsymbol \alpha^{\text{mult}}}\right|_{\boldsymbol\alpha^{\text{mult}}=\boldsymbol \alpha^{\text{mult}*}}=\left[\left.\frac{\partial \boldsymbol g(\boldsymbol \theta_{\boldsymbol Z}^*,\boldsymbol \alpha^{\text{mult}})}{\partial \boldsymbol \alpha_1}\right|_{\boldsymbol\alpha^{\text{mult}}=\boldsymbol \alpha^{\text{mult}*}},...,\left.\frac{\partial \boldsymbol g(\boldsymbol \theta_{\boldsymbol Z}^*,\boldsymbol \alpha^{\text{mult}})}{\partial \boldsymbol \alpha_K}\right|_{\boldsymbol\alpha^{\text{mult}}=\boldsymbol \alpha^{\text{mult}*}}\right]
\end{align*}
Moreover we obtain $\forall$  $k \{1,2,...,K\}$
\begin{align*}
    & \left.\frac{\partial \boldsymbol g(\boldsymbol \theta_{\boldsymbol Z}^*,\boldsymbol \alpha^{\text{mult}})}{\partial \boldsymbol \alpha_k}\right|_{\boldsymbol\alpha^{\text{mult}}=\boldsymbol \alpha^{\text{mult}*}}= \left.\frac{\partial}{\partial \boldsymbol\alpha_k}\left[\frac{S^{\text{mult}}}{\pi(\boldsymbol X^{\text{mult}},\boldsymbol \alpha^{\text{mult}})}\left\{D\boldsymbol Z-\frac{e^{\boldsymbol \theta_{\boldsymbol Z}^{*'}\boldsymbol Z}}{(1+e^{\boldsymbol \theta_{\boldsymbol Z}^{*'}\boldsymbol Z})}\cdot \boldsymbol Z\right\}\right]\right|_{\boldsymbol\alpha^{\text{mult}}=\boldsymbol \alpha^{\text{mult}*}}\\
    & = -\frac{S^{\text{mult}}}{\pi(\boldsymbol X^{\text{mult}},\boldsymbol \alpha^{\text{mult}*})^2}\cdot\left\{D\boldsymbol Z-\frac{e^{\boldsymbol \theta_{\boldsymbol Z}^{*'}\boldsymbol Z}}{(1+e^{\boldsymbol \theta_{\boldsymbol Z}^{*'}\boldsymbol Z})}\cdot \boldsymbol Z\right\}\cdot \left.\frac{\partial \pi(\boldsymbol X^{\text{mult}},\boldsymbol \alpha^{\text{mult}})}{\partial \boldsymbol \alpha_k}\right|_{\boldsymbol\alpha^{\text{mult}}=\boldsymbol \alpha^{\text{mult}*}}\\
    & =-\frac{S^{\text{mult}}}{\pi(\boldsymbol X^{\text{mult}},\boldsymbol\alpha^{\text{mult}*})^2}.\{1-\pi_1(\boldsymbol X_1,\boldsymbol \alpha^*_1)\}\cdot \cdot \cdot\{1-\pi_K(\boldsymbol X_K,\boldsymbol \alpha^*_K)\}\boldsymbol Z\cdot \{\pi_1(\boldsymbol X_1,\boldsymbol \alpha_1^*)\cdot \boldsymbol X_1'\}\cdot\\
    & \hspace{7cm}
   \left\{D-\frac{e^{\boldsymbol \theta_{\boldsymbol Z}^{*'}\boldsymbol Z}}{(1+e^{\boldsymbol \theta_{\boldsymbol Z}^{*'}\boldsymbol Z})}\right\}
\end{align*}
Therefore we obtain 
\begin{align*}
    & G_{\boldsymbol \alpha^{\text{mult}*}}=\mathbb{E}\left[-\frac{S^{\text{mult}}}{\pi(\boldsymbol X^{\text{mult}},\boldsymbol\alpha^{\text{mult}*})^2}.\{1-\pi_1(\boldsymbol X_1,\boldsymbol \alpha^*_1)\}\cdot \cdot \cdot\{1-\pi_K(\boldsymbol X_K,\boldsymbol \alpha^*_K)\}\right.\\
    & \hspace{2cm} \left.\boldsymbol Z\cdot \{\pi_1(\boldsymbol X_1,\boldsymbol \alpha_1^*)\cdot \boldsymbol X_1',\cdot \cdot \cdot,\pi_K(\boldsymbol X_K,\boldsymbol \alpha^*_K)\cdot \boldsymbol X_K'\}\cdot
   \left\{D-\frac{e^{\boldsymbol \theta_{\boldsymbol Z}^{*'}\boldsymbol Z}}{(1+e^{\boldsymbol \theta_{\boldsymbol Z}^{*'}\boldsymbol Z})}\right\}\right]
\end{align*}
\setlength\arraycolsep{-8pt}
\subsubsection*{Calculation for $\boldsymbol\Psi(\boldsymbol \alpha^{\text{mult}*})$}
\begin{align*}
   & \left.\frac{ \partial \boldsymbol h(\boldsymbol \alpha^{\text{mult}})}{\partial \boldsymbol \alpha^{\text{mult}}}\right|_{\boldsymbol \alpha=\boldsymbol \alpha^{\text{mult}*}}=\begin{pmatrix}
     \left.\frac{ \partial \boldsymbol h_1(\boldsymbol \alpha_1)}{\partial \boldsymbol \alpha_1}\right|_{\boldsymbol \alpha_1=\boldsymbol \alpha^*_1} & \hdots & 0 \\
     \vdots & \ddots & \vdots \\
     0 & \hdots &  \left.\frac{ \partial \boldsymbol h_K(\boldsymbol \alpha_K)}{\partial \boldsymbol \alpha_K}\right|_{\boldsymbol \alpha_K=\boldsymbol \alpha^*_K}
         \end{pmatrix}
\end{align*}
Moreover we obtain $\forall$  $k \{1,2,...,K\}$
\begin{align*}
     & \left.\frac{ \partial \boldsymbol h_k(\boldsymbol \alpha_k)}{\partial \boldsymbol \alpha_k}\right|_{\boldsymbol \alpha_k=\boldsymbol \alpha^*_k}= \left.\frac{ \partial}{\partial \boldsymbol \alpha_k}\left\{S_k\boldsymbol X_k-S_{\text{ext}}\cdot \frac{\pi(\boldsymbol X_k,\boldsymbol \alpha_k)}{\pi_{\text{ext}}(\boldsymbol M)}\cdot \boldsymbol X_k \right\}\right|_{\boldsymbol \alpha_k=\boldsymbol \alpha^*_k}\\
     & =-\frac{\left.\frac{ \partial}{\partial \alpha_k}\pi(\boldsymbol X_k,\boldsymbol \alpha_k)\right|_{\boldsymbol \alpha_k=\boldsymbol \alpha^*_k}}{\pi_{\text{ext}}(\boldsymbol M)}\cdot S_{\text{ext}}\cdot \boldsymbol X_k  = -\frac{\pi(\boldsymbol X_k,\boldsymbol \alpha^*_k)}{\pi_{\text{ext}}(\boldsymbol M)}\cdot \{1-\pi(\boldsymbol X_k,\boldsymbol \alpha^*_k)\}\cdot S_{\text{ext}}\cdot \boldsymbol X_k \cdot \boldsymbol X_k'
\end{align*}
\setlength\arraycolsep{-20pt}
This implies
\begin{align*}
   H = - \mathbb{E} \begin{pmatrix}
     \frac{\pi_1(\boldsymbol X_1,\boldsymbol \alpha_1)}{\pi_{\text{ext}}(\boldsymbol M)}\{1-\pi_1(\boldsymbol X_1,\boldsymbol \alpha_1)\}S_{\text{ext}}\boldsymbol X_1\boldsymbol X_1' & \hdots & 0 \\
     \vdots & \ddots & \vdots \\
     0 & \hdots &   \frac{\pi_K(\boldsymbol X_K,\boldsymbol \alpha_K)}{\pi_{\text{ext}}(\boldsymbol M)}\{1-\pi_K(\boldsymbol X_K,\boldsymbol \alpha_K)\} S_{\text{ext}} \boldsymbol X_K\boldsymbol X_K'
         \end{pmatrix}
\end{align*}
This gives the asymptotic distribution of $\widehat{\boldsymbol \theta}_{\boldsymbol Z}$ for JPL.
\end{proof}
\subsection{Theorem S6}\label{sec:pr6}
\setlength\arraycolsep{-20pt}
\begin{theorem}
    Under all the assumptions of Theorems S4 and S5 and \textit{A}5 in Supplementary Section \ref{sec:assumptions}, $\frac{1}{N} \cdot \widehat{G_{\boldsymbol \theta_{\boldsymbol Z}}}^{-1}\cdot \widehat{E}\cdot (\widehat{G_{\boldsymbol \theta_{\boldsymbol Z}}}^{-1})^{'}$ is a consistent estimator of the variance of $\widehat{\boldsymbol \theta}_{\boldsymbol Z}$ for JPL where
\begin{align*}
   & \widehat{G_{\boldsymbol \theta_{\boldsymbol Z}}}=-\frac{1}{N}\sum_{i=1}^N S^{\text{mult}}_i \cdot \frac{1}{\pi(\boldsymbol X^{\text{mult}}_i,\widehat{\boldsymbol \alpha^{\text{mult}}})}\cdot \frac{e^{\widehat{\boldsymbol \theta_{\boldsymbol Z}}'\boldsymbol Z_i}}{(1+e^{\widehat{\boldsymbol \theta_{\boldsymbol Z}}'\boldsymbol Z_i})^2}\cdot\boldsymbol Z_i \boldsymbol Z_i'\\
   & \widehat{H} = -\frac{1}{N}\sum_{i=1}^N S_{\text{ext},i} \begin{pmatrix}
     \frac{\pi_1(\boldsymbol X_{1i},\widehat{\boldsymbol \alpha_1)}}{\pi_{\text{ext,i}}(\boldsymbol M_i)}\{1-\pi_1(\boldsymbol X_{1i},\widehat{\boldsymbol \alpha_1)}\}\boldsymbol X_{1i}\boldsymbol X_{1i}' & \hdots & 0 \\
     \vdots & \ddots & \vdots \\
     0 & \hdots &   \frac{\pi_K(\boldsymbol X_{Ki},\widehat{\boldsymbol \alpha_K)}}{\pi_{\text{ext,i}}(\boldsymbol M_i)}\{1-\pi_K(\boldsymbol X_{Ki},\widehat{\boldsymbol \alpha_K)}\}\boldsymbol X_{Ki}\boldsymbol X_{Ki}'
         \end{pmatrix}\\
   & \widehat{G_{\boldsymbol \alpha^{\text{mult}}}}=-\frac{1}{N}\sum_{i=1}^N \frac{S^{\text{mult}}_i}{\pi(\boldsymbol X^{\text{mult}}_i,\widehat{\boldsymbol \alpha^{\text{mult}}})^2}.\{1-\pi_1(\boldsymbol X_{1i},\widehat{\boldsymbol\alpha_1})\}\cdot \cdot \cdot\{1-\pi_K(\boldsymbol X_{Ki},\widehat{\boldsymbol\alpha_K})\}\cdot\\
    & \hspace{2cm} \boldsymbol Z_i\cdot \{\pi_1(\boldsymbol X_{1i},\widehat{\boldsymbol\alpha_1})\cdot \boldsymbol X_{1i}',\cdot \cdot \cdot,\pi_K(\boldsymbol X_{Ki},\widehat{\boldsymbol\alpha_K})\cdot \boldsymbol X_{Ki}'\}\cdot
   \left\{D_i-\frac{e^{\widehat{\boldsymbol \theta_{\boldsymbol Z}}\boldsymbol Z_i}}{(1+e^{\widehat{\boldsymbol \theta_{\boldsymbol Z}}\boldsymbol Z_i})}\right\}\\
   & \widehat{E_1}=\frac{1}{N}\sum_{i =1}^N S^{\text{mult}}_i \cdot \left\{\frac{1}{\pi(\boldsymbol X^{\text{mult}}_i,\widehat{\boldsymbol \alpha^{\text{mult}}})}\right\}^2\left\{D_i-\frac{e^{\widehat{\boldsymbol \theta_{\boldsymbol Z}}'\boldsymbol Z_i}}{(1+e^{\widehat{\boldsymbol \theta_{\boldsymbol Z}}' \boldsymbol Z_i})}\right\}^2\cdot \boldsymbol Z_i \boldsymbol Z_i'\\
    & \widehat{E}_2 =  \widehat{E}_3'=\frac{1}{N}\sum_{i=1}^N  \widehat{G_{\boldsymbol \alpha^{\text{mult}}}}\cdot \widehat{H}^{-1} \cdot\frac{1}{\pi(\boldsymbol X^{\text{mult}}_i,\widehat{\boldsymbol \alpha^{\text{mult}})}}\begin{pmatrix}
     S_{1i}\boldsymbol X_{1i}\\
     \vdots\\
      S_{Ki}\boldsymbol X_{Ki}\\
         \end{pmatrix}
         \cdot\left\{D_i\boldsymbol Z_i'-\frac{e^{\widehat{\boldsymbol \theta_{\boldsymbol Z}}'\boldsymbol Z_i}}{(1+e^{\widehat{\boldsymbol \theta_{\boldsymbol Z}}' \boldsymbol Z_i})} \cdot \boldsymbol Z_i'\right\}\\
         & -\frac{1}{N}\sum_{i=1}^N  \widehat{G_{\boldsymbol \alpha^{\text{mult}}}}\cdot \widehat{H}^{-1} \cdot \frac{S^{\text{mult}}_i}{\pi(\boldsymbol X^{\text{mult}}_i,\widehat{\boldsymbol \alpha^{\text{mult}}})} \cdot\frac{S_{\text{ext},i}}{\pi_{\text{ext}}(\boldsymbol M_i)} \begin{pmatrix}
       \pi_1(\boldsymbol X_{1i},\widehat{\boldsymbol \alpha_1})\cdot \boldsymbol X_{1i}\\
     \vdots\\
     \pi_K(\boldsymbol X_{Ki},\widehat{\boldsymbol \alpha_K})\cdot\boldsymbol X_{Ki}\\
         \end{pmatrix}\cdot\left\{D_i\boldsymbol Z_i'-\frac{e^{\widehat{\boldsymbol \theta_{\boldsymbol Z}}'\boldsymbol Z_i}}{(1+e^{\widehat{\boldsymbol \theta_{\boldsymbol Z}}' \boldsymbol Z_i})} \cdot \boldsymbol Z_i'\right\}
         \end{align*}
\begin{align*}
    & \widehat{E_4} = \frac{1}{N}\cdot\widehat{G_{\boldsymbol \alpha^{\text{mult}}}}\cdot \widehat{H}^{-1} \sum_{i=1}^N  \left[\begin{pmatrix}
    S_{1i} \boldsymbol X_{1i}\\
     \vdots\\
      S_{Ki} \boldsymbol X_{Ki}\\
         \end{pmatrix}\cdot(S_{1i} \boldsymbol X_{1i}',\cdot \cdot \cdot,S_{Ki} \boldsymbol X_{Ki}')\right.\\
        & \left.-2\cdot S_{\text{ext},i}\begin{pmatrix}
    \frac{\pi_1(\boldsymbol X_{1i},\widehat{\boldsymbol \alpha_1)}}{\pi_{\text{ext}}(\boldsymbol M_i)}\cdot \boldsymbol X_{1i}\\
     \vdots\\
     \frac{\pi_K(\boldsymbol X_{Ki},\widehat{\boldsymbol \alpha_K)}}{\pi_{\text{ext}}(\boldsymbol M_i)}\cdot \boldsymbol X_{Ki}\\
         \end{pmatrix}\cdot (S_{1i}\boldsymbol X_{1i}',\cdot \cdot \cdot,S_{Ki}\boldsymbol X_{Ki}')\right.\\
    & \left.+ S_{\text{ext},i} \cdot \begin{pmatrix}
    \frac{\pi_1(\boldsymbol X_{1i},\widehat{\boldsymbol \alpha_1)}}{\pi_{\text{ext}}(\boldsymbol M_i)}\cdot \boldsymbol X_{1i}\\
     \vdots\\
     \frac{\pi_K(\boldsymbol X_{Ki},\widehat{\boldsymbol \alpha_K)}}{\pi_{\text{ext}}(\boldsymbol M_i)}\cdot \boldsymbol X_{Ki}\\
         \end{pmatrix}\left\{  \frac{\pi_1(\boldsymbol X_{1i},\widehat{\boldsymbol \alpha_1)}}{\pi_{\text{ext}}(\boldsymbol M_i)}\cdot \boldsymbol X_{1i},\cdot\cdot\cdot, \frac{\pi_K(\boldsymbol X_{Ki},\widehat{\boldsymbol \alpha_K)}}{\pi_{\text{ext}}(\boldsymbol M_i)}\cdot \boldsymbol X_{Ki}\right\}\right]\cdot (\widehat{H}^{-1})'\cdot(\widehat{G_{\boldsymbol \alpha^{\text{mult}}}})'
\end{align*}
$\widehat{E}=\widehat{E}_1-\widehat{E}_2-\widehat{E}_3 + \widehat{E}_4$.
\end{theorem}
\begin{proof}
    Under all the assumptions of Theorems S4, S5 and \textit{A}5 in Supplementary Section \ref{sec:assumptions} and using ULLN and Continuous Mapping Theorem, the proof of consistency for each of the following sample quantities are exactly same as the approach in Theorem S3. Using the exact same steps on the joint parameters $\boldsymbol \eta$ instead of $\boldsymbol \theta_{\boldsymbol Z}$ (as in Theorem S3) we obtain
\begin{align*}
     & \widehat{G_{\boldsymbol \theta_{\boldsymbol Z}}}=-\frac{1}{N}\sum_{i=1}^N S^{\text{mult}}_i \cdot \frac{1}{\pi(\boldsymbol X^{\text{mult}}_i,\widehat{\boldsymbol \alpha^{\text{mult}}})}\cdot \frac{e^{\widehat{\boldsymbol \theta_{\boldsymbol Z}}'\boldsymbol Z_i}}{(1+e^{\widehat{\boldsymbol \theta_{\boldsymbol Z}}'\boldsymbol Z_i})^2}\cdot\boldsymbol Z_i \boldsymbol Z_i' \\
     & \xrightarrow{p}G_{\boldsymbol \theta_{\boldsymbol Z}^*} = \mathbb{E}\left\{-\frac{S^{\text{mult}}}{\pi(\boldsymbol X^{\text{mult}},\boldsymbol \alpha^{\text{mult}*})}\cdot \frac{e^{\boldsymbol \theta_{\boldsymbol Z}^{*'}\boldsymbol Z}}{(1+e^{\boldsymbol \theta_{\boldsymbol Z}^{*'}\boldsymbol Z})^2}\cdot \boldsymbol Z\boldsymbol Z'\right\}\cdot
\end{align*}
Similarly we obtain 
\begin{align*}
     &  \widehat{G_{\boldsymbol \alpha^{\text{mult}}}}=-\frac{1}{N}\sum_{i=1}^N \frac{S^{\text{mult}}_i}{\pi(\boldsymbol X^{\text{mult}}_i,\widehat{\boldsymbol \alpha^{\text{mult}}})^2}.\{1-\pi_1(\boldsymbol X_{1i},\widehat{\boldsymbol\alpha_1})\}\cdot \cdot \cdot\{1-\pi_K(\boldsymbol X_{Ki},\widehat{\boldsymbol\alpha_K})\}\cdot\\
    & \hspace{2cm} \boldsymbol Z_i\cdot \{\pi_1(\boldsymbol X_{1i},\widehat{\boldsymbol\alpha_1})\cdot \boldsymbol X_{1i}',\cdot \cdot \cdot,\pi_K(\boldsymbol X_{Ki},\widehat{\boldsymbol\alpha_K})\cdot \boldsymbol X_{Ki}'\}\cdot
   \left\{D_i-\frac{e^{\widehat{\boldsymbol \theta_{\boldsymbol Z}}\boldsymbol Z_i}}{(1+e^{\widehat{\boldsymbol \theta_{\boldsymbol Z}}\boldsymbol Z_i})}\right\}\\
     & \xrightarrow{p} G_{\boldsymbol \alpha^{\text{mult}*}}=\mathbb{E}\left[-\frac{S^{\text{mult}}}{\pi(\boldsymbol X^{\text{mult}},\boldsymbol\alpha^{\text{mult}*})^2}.\{1-\pi_1(\boldsymbol X_1,\boldsymbol \alpha^*_1)\}\cdot \cdot \cdot\{1-\pi_K(\boldsymbol X_K,\boldsymbol \alpha^*_K)\}\cdot\right.\\
    & \hspace{2cm} \left.\boldsymbol Z\cdot \{\pi_1(\boldsymbol X_1,\boldsymbol \alpha_1^*)\cdot \boldsymbol X_1',\cdot \cdot \cdot,\pi_K(\boldsymbol X_K,\boldsymbol \alpha^*_K)\cdot \boldsymbol X_K'\}\cdot
   \left\{D-\frac{e^{\boldsymbol \theta_{\boldsymbol Z}^{*'}\boldsymbol Z}}{(1+e^{\boldsymbol \theta_{\boldsymbol Z}^{*'}\boldsymbol Z})}\right\}\right]
\end{align*}
\setlength\arraycolsep{-20pt}
\begin{align*}
     & \widehat{H} = -\frac{1}{N}\sum_{i=1}^N S_{\text{ext},i} \begin{pmatrix}
     \frac{\pi_1(\boldsymbol X_{1i},\widehat{\boldsymbol \alpha_1)}}{\pi_{\text{ext,i}}(\boldsymbol M_i)}\{1-\pi_1(\boldsymbol X_{1i},\widehat{\boldsymbol \alpha_1)}\} \boldsymbol X_{1i}\boldsymbol X_{1i}' & \hdots & 0 \\
     \vdots & \ddots & \vdots \\
     0 & \hdots &   \frac{\pi_K(\boldsymbol X_{Ki},\widehat{\boldsymbol \alpha_K)}}{\pi_{\text{ext,i}}(\boldsymbol M_i)}\{1-\pi_K(\boldsymbol X_{Ki},\widehat{\boldsymbol \alpha_K)}\} \boldsymbol X_{Ki}\boldsymbol X_{Ki}'
         \end{pmatrix}\\
     &\xrightarrow{p}   H = - \mathbb{E} \begin{pmatrix}
     \frac{\pi_1(\boldsymbol X_1,\boldsymbol \alpha^*_1)}{\pi_{\text{ext}}(\boldsymbol M)}\{1-\pi_1(\boldsymbol X_1,\boldsymbol \alpha^*_1)\} S_{\text{ext}} \boldsymbol X_1\boldsymbol X_1' & \hdots & 0 \\
     \vdots & \ddots & \vdots \\
     0 & \hdots &   \frac{\pi_K(\boldsymbol X_K,\boldsymbol \alpha^*_K)}{\pi_{\text{ext}}(\boldsymbol M)}\{1-\pi_K(\boldsymbol X_K,\boldsymbol \alpha^*_K)\} S_{\text{ext}} \boldsymbol X_K\boldsymbol X_K'
         \end{pmatrix}
\end{align*}
Similarly we obtain
\begin{align*}
     & \widehat{E_1}=\frac{1}{N}\sum_{i =1}^N S^{\text{mult}}_i \cdot \left\{\frac{1}{\pi(\boldsymbol X^{\text{mult}}_i,\widehat{\boldsymbol \alpha^{\text{mult}}})}\right\}^2\left\{D_i-\frac{e^{\widehat{\boldsymbol \theta_{\boldsymbol Z}}'\boldsymbol Z_i}}{(1+e^{\widehat{\boldsymbol \theta_{\boldsymbol Z}}' \boldsymbol Z_i})}\right\}^2\cdot \boldsymbol Z_i \boldsymbol Z_i'\\
     & \xrightarrow{p}  \mathbb{E}\left[S \cdot \left\{\frac{1}{\pi(\boldsymbol X^{\text{mult}},\boldsymbol \alpha^{\text{mult}*})}\right\}^2\cdot \left\{D-\frac{e^{\boldsymbol \theta_{\boldsymbol Z}^{*'}\boldsymbol Z_i}}{(1+e^{\boldsymbol \theta_{\boldsymbol Z}^{*'} \boldsymbol Z})}\right\}^2\cdot \boldsymbol Z\boldsymbol Z'\right]\cdot
\end{align*}
\begin{align*}
     &  \widehat{E}_2 =  \widehat{E}_3'=\frac{1}{N}\sum_{i=1}^N  \widehat{G_{\boldsymbol \alpha^{\text{mult}}}}\cdot \widehat{H}^{-1} \cdot\frac{1}{\pi(\boldsymbol X^{\text{mult}}_i,\widehat{\boldsymbol \alpha^{\text{mult}})}}\begin{pmatrix}
     S_{1i}\boldsymbol X_{1i}\\
     \vdots\\
      S_{Ki}\boldsymbol X_{Ki}\\
         \end{pmatrix}
         \cdot\left\{D_i\boldsymbol Z_i'-\frac{e^{\widehat{\boldsymbol \theta_{\boldsymbol Z}}'\boldsymbol Z_i}}{(1+e^{\widehat{\boldsymbol \theta_{\boldsymbol Z}}' \boldsymbol Z_i})} \cdot \boldsymbol Z_i'\right\}\\
         & -\frac{1}{N}\sum_{i=1}^N  \widehat{G_{\boldsymbol \alpha^{\text{mult}}}}\cdot \widehat{H}^{-1} \cdot \frac{S^{\text{mult}}_i}{\pi(\boldsymbol X^{\text{mult}}_i,\widehat{\boldsymbol \alpha^{\text{mult}}})} \cdot\frac{S_{\text{ext},i}}{\pi_{\text{ext}}(\boldsymbol M_i)} \begin{pmatrix}
       \pi_1(\boldsymbol X_{1i},\widehat{\boldsymbol \alpha_1})\cdot \boldsymbol X_{1i}\\
     \vdots\\
     \pi_K(\boldsymbol X_{Ki},\widehat{\boldsymbol \alpha_K})\cdot\boldsymbol X_{Ki}\\
         \end{pmatrix}\cdot\left\{D_i\boldsymbol Z_i'-\frac{e^{\widehat{\boldsymbol \theta_{\boldsymbol Z}}'\boldsymbol Z_i}}{(1+e^{\widehat{\boldsymbol \theta_{\boldsymbol Z}}' \boldsymbol Z_i})} \cdot \boldsymbol Z_i'\right\}\cdot
\end{align*}
\begin{align*}
    & \widehat{E_4} = \frac{1}{N}\cdot\widehat{G_{\boldsymbol \alpha^{\text{mult}}}}\cdot \widehat{H}^{-1} \sum_{i=1}^N  \left[\begin{pmatrix}
    S_{1i} \boldsymbol X_{1i}\\
     \vdots\\
      S_{Ki} \boldsymbol X_{Ki}\\
         \end{pmatrix}\cdot(S_{1i} \boldsymbol X_{1i}',\cdot \cdot \cdot,S_{Ki} \boldsymbol X_{Ki}')\right.\\
        & \left.-2\cdot S_{\text{ext},i}\begin{pmatrix}
    \frac{\pi_1(\boldsymbol X_{1i},\widehat{\boldsymbol \alpha_1)}}{\pi_{\text{ext}}(\boldsymbol M_i)}\cdot \boldsymbol X_{1i}\\
     \vdots\\
     \frac{\pi_K(\boldsymbol X_{Ki},\widehat{\boldsymbol \alpha_K)}}{\pi_{\text{ext}}(\boldsymbol M_i)}\cdot \boldsymbol X_{Ki}\\
         \end{pmatrix}\cdot (S_{1i}\boldsymbol X_{1i}',\cdot \cdot \cdot,S_{Ki}\boldsymbol X_{Ki}')\right.\\
    & \left.+ S_{\text{ext},i} \cdot \begin{pmatrix}
    \frac{\pi_1(\boldsymbol X_{1i},\widehat{\boldsymbol \alpha_1)}}{\pi_{\text{ext}}(\boldsymbol M_i)}\cdot \boldsymbol X_{1i}\\
     \vdots\\
     \frac{\pi_K(\boldsymbol X_{Ki},\widehat{\boldsymbol \alpha_K)}}{\pi_{\text{ext}}(\boldsymbol M_i)}\cdot \boldsymbol X_{Ki}\\
         \end{pmatrix}\left\{  \frac{\pi_1(\boldsymbol X_{1i},\widehat{\boldsymbol \alpha_1)}}{\pi_{\text{ext}}(\boldsymbol M_i)}\cdot \boldsymbol X_{1i},\cdot\cdot\cdot, \frac{\pi_K(\boldsymbol X_{Ki},\widehat{\boldsymbol \alpha_K)}}{\pi_{\text{ext}}(\boldsymbol M_i)}\cdot \boldsymbol X_{Ki}\right\}\right]\\
    & \hspace{7cm}\cdot (\widehat{H}^{-1})'\cdot(\widehat{G_{\boldsymbol \alpha^{\text{mult}}}})'\\
    & \xrightarrow{p} \mathbb{E}\{G_{\boldsymbol \alpha^{\text{mult}*}}H^{-1}\cdot \boldsymbol h(\boldsymbol \alpha^{\text{mult}*})\cdot \boldsymbol h(\boldsymbol \alpha^{\text{mult}*})'\cdot (H^{-1})'\cdot (G_{\boldsymbol\alpha^{\text{mult}*}})'\} 
\end{align*}
Therefore we obtain
\begin{align*}
    &\widehat{E}=\widehat{E}_1-\widehat{E}_2-\widehat{E}_3 + \widehat{E}_4\xrightarrow{p}\mathbb{E}[\{\boldsymbol g(\boldsymbol \theta_{\boldsymbol Z}^*,\boldsymbol \alpha^{\text{mult}*})+G_{\boldsymbol \alpha^{\text{mult}*}}\cdot \boldsymbol \Psi(\boldsymbol \alpha^{\text{mult}*})\}
    \{\boldsymbol g(\boldsymbol \theta_{\boldsymbol Z}^*,\boldsymbol \alpha^{\text{mult}*})+G_{\boldsymbol \alpha^{\text{mult}*}}.\boldsymbol \Psi(\boldsymbol \alpha^{\text{mult}*})\}']\cdot
\end{align*}
Using all the above results we obtain $\widehat{G_{\boldsymbol \theta_{\boldsymbol Z}}}^{-1}\cdot\widehat{E}\cdot(\widehat{G_{\boldsymbol \theta_{\boldsymbol Z}}}^{-1})^{'}$
\begin{align*}
    &\xrightarrow{p}(G_{\boldsymbol \theta_{\boldsymbol Z}^*})^{-1}\cdot \mathbb{E}[\{\boldsymbol g(\boldsymbol \theta_{\boldsymbol Z}^*,\boldsymbol \alpha^{\text{mult}*})+G_{\boldsymbol \alpha^{\text{mult}*}}\cdot \boldsymbol \Psi(\boldsymbol \alpha^{\text{mult}*})\}\{\boldsymbol g(\boldsymbol \theta_{\boldsymbol Z}^*,\boldsymbol \alpha^{\text{mult}*})+G_{\boldsymbol \alpha^{\text{mult}*}}.\boldsymbol \Psi(\boldsymbol \alpha^{\text{mult}*})\}'].(G_{\boldsymbol \theta_{\boldsymbol Z}^*}^{-1})^{'}\cdot
\end{align*}
From Theorem S5 using the same approach used in the last step of Theorem S3, we obtain that $\frac{1}{N}\cdot\widehat{G_{\boldsymbol \theta_{\boldsymbol Z}}}^{-1}\cdot\hat{E}\cdot(\widehat{G_{\boldsymbol \theta_{\boldsymbol Z}}}^{-1})^{'}$ is a consistent estimator of the asymptotic variance of $\widehat{\boldsymbol \theta}_{\boldsymbol Z}$.
\end{proof}
\section{Joint Simplex Regression (JSR)}\label{sec:jsr}
We extend the method of Simplex Regression (SR) \citep{kundu2024framework,beesley2022statistical,barndorff1991some} to multiple cohorts using the joint approach. The main idea underlying this method is based on the identity
\begin{equation}
       \pi_k(\boldsymbol X_k)=P(S_k=1|\boldsymbol X_k)=P(S_{\text{ext}}=1|\boldsymbol X_k)\cdot \left(\frac{p_{11k}(\boldsymbol X_k)+p_{10k}(\boldsymbol X_k)}{p_{11k}(\boldsymbol X_k)+p_{01k}(\boldsymbol X_k)}\right).\label{eq:eqs6}
 \end{equation}
where, $p_{ijk}(\boldsymbol X_k)=P(S_k=i,S_{\text{ext}}=j|\boldsymbol X_k,S_k=1\text{  or   }S_{\text{ext}}=1)$. In equation \eqref{eq:eqs6},
we estimate $P(S_{\text{ext}}=1|\boldsymbol X_k)$ and $p_{11k}(\boldsymbol X_k), p_{10k}(\boldsymbol X_k), p_{01k}(\boldsymbol X_k)$ using Simplex Regression and Multinomial Regression respectively. The Simple Regression step models the dependency of $P(S_{\text{ext}}=1|\boldsymbol X_k)$ on $\boldsymbol X_k$. It is fitted on the external probability sample where $\pi_{\text{ext}}$ serve as the response variable. We predict $P(S_{\text{ext}}=1|\boldsymbol X_k)$ for all the individuals in the $k^{\text{th}}$ cohort using the estimated parameters in the previous step. Therefore this step is based on a stringent assumption that for all individuals in the external probability, $\pi_{\text{ext}}=P(S_{\text{ext}}=1|\boldsymbol X_k)$ $\forall k \in \{1,2,\cdot\cdot,K\}$. On the other hand, $p_{ijk}(\boldsymbol X_k)$ is estimated based on the sample combining external probability and $k^{\text{th}}$ internal non-probability sample. We define a nominal categorical variable with three levels corresponding to different values of $(i,j)$ pairs ($(i,j)=(1,1),(1,0),(0,1)$). An individual with level (1,1) is a member of both samples; (0,1) indicates a member of the exterior sample only, whereas (1,0) corresponds to the internal sample only. The multicategory response is again regressed on the internal selection model variables, $\boldsymbol X_k$ using a multinomial regression model and we obtain estimates of $p_{ijk}(\boldsymbol X_k)$.
Using these estimates, the selection probabilities for the internal sample $k$,  $P(S_k=1|\boldsymbol X_k)$, were estimated from equation \eqref{eq:eqs6}. Consequently, using equation (3.3) of the main text, we calculate the joint selection probability $\widehat{\pi}(\boldsymbol X^{\text{mult}})$ which serves as $\pi(\boldsymbol X^{\text{mult}})$ in the IPW equation (3.4) of the main text.\\

\noindent
For JSR, due to composite nature of the selection model, we use an approximation of the variance ignoring uncertainty in the estimates of the selection model parameters. Hence we used the asymptotic distribution and the corresponding variance estimator provided in Theorems S1 and S2. 
\section{Joint Post Stratification (JPS)}\label{sec:jps}
We extend the method of Post Stratification (PS) \citep{kundu2024framework,beesley2022statistical,holt1979post} to multiple cohorts using the joint approach. 
We assume the joint distribution of the selection variables in the target population, namely $P(\boldsymbol X_k)$ are available to us. In case of continuous selection variables, we can at best expect to have access to joint probabilities of discretized versions of those variables.  Beyond this coarsening, obtaining joint probabilities of a large multivariate set of predictors become extremely challenging. In such cases, several conditional independence assumptions will be needed to specify a joint distribution from sub-conditionals.\\

\noindent
We consider the scenario where both $\boldsymbol Z_{2k}$ and $\boldsymbol W_k$ are continuous variables. Let $\boldsymbol Z_{2k}'$ and $\boldsymbol W_k'$ be the discretized versions of $\boldsymbol Z_{2k}$ and $\boldsymbol W_k$ respectively. The post stratification method estimates the selection probabilities into the internal sample $k$ for each individual by,
$$\pi_k(\boldsymbol X_k)= \frac{P(\boldsymbol X_{k}'|S_k=1)\cdot P(S_k=1)}{P(\boldsymbol X_{k}')}\cdot$$
$P(\boldsymbol X_{k}'|S_k=1)$ is empirically estimated from the internal sample $k$. $P(\boldsymbol X_{k}')$ is the known population level joint distribution for the discretized selection variables $\boldsymbol X_i^{c'}$  obtained from external sources. Finally, $P(S_k=1)$ is empirically estimated by $\frac{n_k}{N}$, where $n_k$ is the size of internal sample $k$. 
Consequently, using equation (3.3) of the main text, we calculate the joint selection probability $\widehat{\pi}(\boldsymbol X^{\text{mult}})$ which serves as $\pi(\boldsymbol X^{\text{mult}})$ in the IPW equation (3.4) of the main text. For JPS, the selection weights are derived from the summary statistics of the target population. Consequently, there is no need to estimate selection model parameters within this framework. Therefore, we utilize the asymptotic distribution and the corresponding consistent variance estimator as laid out in Theorems S1 and S2, which operate under the premise of known selection weights.
\section{Joint Calibration (JCL)}\label{sec:jcl}
\subsection{Method}
We extend the method of Calibration (CL) \citep{kundu2024framework,wu2003optimal} to multiple cohorts using the joint approach. Similar to JPL, we estimate internal selection probabilities $\pi_k(\boldsymbol X_k)$ by a model, $\pi_k(\boldsymbol X_k,\boldsymbol \alpha_k)$ indexed by parameters $\boldsymbol\alpha_k$, when marginal population means of the selection variables $\boldsymbol X_k$ are available from external sources.  We obtain the estimate of $\boldsymbol \alpha_k$ by solving the following calibration equation,
\begin{equation}
    \sum_{i=1}^N\frac{S_{ki}\boldsymbol X_{ki}}{\pi_k(\boldsymbol X_{ki},\boldsymbol \alpha_k)}=\sum_{i=1}^ N\boldsymbol X_{ki}.  \label{eq:eqs7}
\end{equation}
\noindent
Newton-Raphson method is used to solve equation \eqref{eq:eqs7} to estimate $\boldsymbol \alpha_k$ and henceforth obtain $\widehat{\pi_k}(\boldsymbol X_{ki},\boldsymbol \alpha_k)$. We used a logistic specification of $\pi(\boldsymbol X_k,\boldsymbol\alpha_k)$ in our work. Consequently, using equation (3.3) of the main text, we calculate the joint selection probability $\widehat{\pi}(\boldsymbol X^{\text{mult}})$ which serves as $\pi(\boldsymbol X^{\text{mult}})$ in the IPW equation (3.4) of the main text. 
\subsection{Theorem S7}\label{sec:pr7}
\begin{theorem}
     Under condition C1 in the main text and assumptions \textit{A}1 and \textit{A}2 in Supplementary Section \ref{sec:assumptions} and assuming the selection model is correctly specified, that is, $\pi(\boldsymbol X^{\text{mult}},\boldsymbol \alpha^{\text{mult}*})=P(S=1|\boldsymbol X^{\text{mult}})=\pi(\boldsymbol X^{\text{mult}})$ where $\widehat{\boldsymbol \alpha^{\text{mult}}}\xrightarrow{p}\boldsymbol \alpha^{\text{mult}*}$ and $\boldsymbol \alpha^{\text{mult}*}$ is the true value of $\boldsymbol \alpha$, then $\widehat{\boldsymbol \theta}_{\boldsymbol Z}$ estimated using JCL is consistent for $\boldsymbol \theta^*_{\boldsymbol Z}$ as $N\rightarrow \infty$. \label{thm7}
\end{theorem}
\begin{proof}
    In this case, the estimating equation consists of both selection model and disease model parameter estimation. The two step estimating equation is given by
\begin{align*}
     & \delta_n(\boldsymbol \alpha^{\text{mult}})=\begin{pmatrix}
     \delta_{1_n}(\boldsymbol \alpha_1)= \sum_{i=1}^N\frac{S_{1i}\boldsymbol X_{1i}}{\pi_1(\boldsymbol X_{1i},\boldsymbol \alpha_1)}-\sum_{i=1}^ N\boldsymbol X_{1i}\\
     \vdots\\
       \delta_{K_n}(\boldsymbol \alpha_K)=\sum_{i=1}^N\frac{S_{Ki}\boldsymbol X_{Ki}}{\pi_K(\boldsymbol X_{Ki},\boldsymbol \alpha_K)}-\sum_{i=1}^ N\boldsymbol X_{Ki}\\
         \end{pmatrix}\\
     & \phi_n(\boldsymbol \theta_{\boldsymbol Z},\widehat{\boldsymbol \alpha^{\text{mult}}})= \frac{1}{N}\sum_{i=1}^{N}\frac{S^{\text{mult}}_i}{\pi(\boldsymbol X^{\text{mult}}_i,\widehat{\boldsymbol \alpha^{\text{mult}}})}\left\{D_i \boldsymbol Z_i-\frac{e^{\boldsymbol\theta'\boldsymbol Z_i}}{(1+e^{\boldsymbol\theta'\boldsymbol Z_i})}\cdot \boldsymbol Z_i\right\}\cdot
\end{align*}
From \citet{tsiatis2006semiparametric} to show  $\widehat{\boldsymbol \theta_{\boldsymbol Z}}\xrightarrow{p}\boldsymbol \theta_{\boldsymbol Z}^*$ in a two step estimation procedure, we need to prove $\mathbb{E}(\phi_N(\boldsymbol \theta_{\boldsymbol Z}^*,\boldsymbol \alpha^{\text{mult}*}))=\mathbf{0}$. Under correct specification of selection model, $\pi(\boldsymbol X^{\text{mult}},\boldsymbol \alpha^{\text{mult}*})=P(S=1|\boldsymbol X^{\text{mult}})=\pi(\boldsymbol X^{\text{mult}})$ where $\widehat{\boldsymbol \alpha^{\text{mult}}}\xrightarrow{p}\boldsymbol \alpha^{\text{mult}*}$ and $\boldsymbol \alpha^{\text{mult}*}$ is the true value of $\boldsymbol \alpha^{\text{mult}}$. Using this equality, the proof to show that $\mathbb{E}(\phi_N(\boldsymbol \theta_{\boldsymbol Z}^*,\boldsymbol \alpha^{\text{mult}*}))=\mathbf{0}$ is exactly as Theorem S1. Therefore we obtain $\widehat{\boldsymbol \theta}_{\boldsymbol Z}$ is a consistent estimator of $\boldsymbol \theta_{\boldsymbol Z}^*$.
\end{proof}
\subsection{Theorem S8}\label{sec:pr8}
\begin{theorem}
    Under assumptions of Theorem S7 and assumption \textit{A}4 in Supplementary Section \ref{sec:assumptions}, suppose $\widehat{\boldsymbol \theta_{\boldsymbol Z}}\xrightarrow{p}\boldsymbol \theta_{\boldsymbol Z}^*$, the asymptotic distribution of $\widehat{\boldsymbol \theta}_{\boldsymbol Z}$ using JCL is given by
\begin{equation}
\sqrt{N}(\widehat{\boldsymbol \theta_{\boldsymbol Z}}-\boldsymbol\theta^*_{\boldsymbol Z})\xrightarrow{d}\mathcal{N}(\mathbf{0},V).
\end{equation}
where
\setlength\arraycolsep{-8pt}
\begin{align*}
    & V=(G_{\boldsymbol \theta_{\boldsymbol Z}^*})^{-1}.\mathbb{E}[\{\boldsymbol g(\boldsymbol \theta_{\boldsymbol Z}^*,\boldsymbol \alpha^{\text{mult}*})+G_{\boldsymbol \alpha^{\text{mult}*}}.\boldsymbol \Psi(\boldsymbol \alpha^{\text{mult}*})\}\{\boldsymbol g(\boldsymbol \theta_{\boldsymbol Z}^*,\boldsymbol \alpha^{\text{mult}*})+G_{\boldsymbol \alpha^{\text{mult}*}}\boldsymbol \Psi(\boldsymbol \alpha^{\text{mult}*})\}'](G_{\boldsymbol \theta_{\boldsymbol Z}^*}^{-1})^{'}\\
    & G_{\boldsymbol \theta_{\boldsymbol Z}^*} = \mathbb{E}\left\{-\frac{S^{\text{mult}}}{\pi(\boldsymbol X^{\text{mult}},\boldsymbol \alpha^{\text{mult}*})}\cdot \frac{e^{\boldsymbol \theta_{\boldsymbol Z}^{*'}\boldsymbol Z}}{(1+e^{\boldsymbol \theta_{\boldsymbol Z}^{*'}\boldsymbol Z})^2}\cdot \boldsymbol Z\boldsymbol Z'\right\}, \quad \boldsymbol g(\boldsymbol \theta_{\boldsymbol Z}^*,\boldsymbol \alpha^{\text{mult}*}) = \frac{S^{\text{mult}`}}{\pi(\boldsymbol X^{\text{mult}},\boldsymbol \alpha^{\text{mult}*})}\left\{D \boldsymbol Z-\frac{e^{\boldsymbol \theta_{\boldsymbol Z}^{*'}\boldsymbol Z}}{(1+e^{\boldsymbol \theta_{\boldsymbol Z}^{*'}\boldsymbol Z})}\cdot \boldsymbol Z\right\}
\end{align*}
\begin{align*}
    & G_{\boldsymbol \alpha^{\text{mult}*}}=\mathbb{E}\left[-\frac{S^{\text{mult}}}{\pi(\boldsymbol X^{\text{mult}},\boldsymbol\alpha^{\text{mult}*})^2}.\{1-\pi_1(\boldsymbol X_1,\boldsymbol \alpha^*_1)\}\cdot \cdot \cdot\{1-\pi_K(\boldsymbol X_K,\boldsymbol \alpha^*_K)\}\cdot\right.\\
    & \hspace{2cm} \left.\boldsymbol Z\cdot \{\pi_1(\boldsymbol X_1,\boldsymbol \alpha_1^*)\cdot \boldsymbol X_1',\cdot \cdot \cdot,\pi_K(\boldsymbol X_K,\boldsymbol \alpha^*_K)\cdot \boldsymbol X_K'\}\cdot
   \left\{D-\frac{e^{\boldsymbol \theta_{\boldsymbol Z}^{*'}\boldsymbol Z}}{(1+e^{\boldsymbol \theta_{\boldsymbol Z}^{*'}\boldsymbol Z})}\right\}\right]\\
    & \boldsymbol h(\boldsymbol X^{\text{mult}},\boldsymbol \alpha^{\text{mult}*}) = \begin{pmatrix}
   \frac{S_1\boldsymbol X_1}{\pi_1(\boldsymbol X_1,\boldsymbol \alpha_1)}-\boldsymbol X_1\\
     \vdots\\
         \frac{S_K\boldsymbol X_K}{\pi_K(\boldsymbol X_K,\boldsymbol \alpha_K)}-\boldsymbol X_K\\
         \end{pmatrix}\\
         &  H = - \mathbb{E} \begin{pmatrix}
     \frac{S_1\boldsymbol X_1\boldsymbol X_1'}{\pi_1(\boldsymbol X_1,\boldsymbol \alpha_1^*)}\cdot \{1- \pi_1(\boldsymbol X_1,\boldsymbol \alpha_1)\} & \hdots & 0 \\
     \vdots & \ddots & \vdots \\
     0 & \hdots &   \frac{S_K\boldsymbol X_K\boldsymbol X_K'}{\pi_K(\boldsymbol X_K,\boldsymbol \alpha_K^*)}\cdot \{1- \pi_1(\boldsymbol X_1,\boldsymbol \alpha_K)\}
         \end{pmatrix}\\
    & \boldsymbol \Psi(\boldsymbol \alpha^{\text{mult}*})=-H^{-1}\cdot  \boldsymbol h(\boldsymbol X^{\text{mult}},\boldsymbol \alpha^{\text{mult}*})
\end{align*} \label{thm8}
\end{theorem}
\begin{proof}
By \citet{tsiatis2006semiparametric}'s arguments on a two step Z-estimation problem, under assumptions of Theorem \ref{thm7} and since $\widehat{\boldsymbol \theta_{\boldsymbol Z}}\xrightarrow{p}\boldsymbol \theta_{\boldsymbol Z}^*$, we obtain that 
\begin{align*}
    &\sqrt{N}(\widehat{\boldsymbol \theta_{\boldsymbol Z}}-\boldsymbol\theta_{\boldsymbol Z}^*)\xrightarrow{d}\mathcal{N}(0,(G_{\boldsymbol \theta_{\boldsymbol Z}^*})^{-1}.\mathbb{E}[\{\boldsymbol g(\boldsymbol \theta_{\boldsymbol Z}^*,\boldsymbol \alpha^{\text{mult}*})+G_{\boldsymbol \alpha^{\text{mult}*}}.\boldsymbol \Psi(\boldsymbol \alpha^{\text{mult}*})\}\{\boldsymbol g(\boldsymbol \theta_{\boldsymbol Z}^*,\boldsymbol \alpha^{\text{mult}*})+G_{\boldsymbol \alpha^{\text{mult}*}}.\boldsymbol \Psi(\boldsymbol \alpha^{\text{mult}*})\}'].(G_{\boldsymbol \theta_{\boldsymbol Z}^*}^{-1})^{'}\cdot
\end{align*}

We derive the expression of each of the terms in the above expression.
\begin{align*}
    & G_{\boldsymbol \theta_{\boldsymbol Z}^*}=\left.\mathbb{E}\left\{\frac{\partial g(\boldsymbol \theta_{\boldsymbol Z},\boldsymbol \alpha^{\text{mult}*})}{\partial \boldsymbol \theta_{\boldsymbol Z}}\right\}\right|_{\boldsymbol \theta_{\boldsymbol Z}=\boldsymbol \theta_{\boldsymbol Z}^*}\hspace{1cm} \left.G_{\boldsymbol \alpha^{\text{mult}*}}=\mathbb{E}\left\{\frac{\partial g(\boldsymbol \theta_{\boldsymbol Z}^*,\boldsymbol \alpha^{\text{mult}})}{\partial \boldsymbol \alpha^{\text{mult}}}\right\}\right|_{\boldsymbol\alpha^{\text{mult}}=\boldsymbol \alpha^{\text{mult}*}}\\
    & H = \left.\mathbb{E}\left\{\frac{\partial \boldsymbol h(\boldsymbol \alpha^{\text{mult}})}{\partial \boldsymbol \alpha^{\text{mult}}}\right\}\right|_{\boldsymbol\alpha^{\text{mult}}=\boldsymbol \alpha^{\text{mult}*}}\hspace{1.9cm} \boldsymbol \Psi(\boldsymbol \alpha^{\text{mult}*})=-H^{-1}\boldsymbol h(\boldsymbol \alpha^{\text{mult}*}) \cdot
\end{align*}
First we calculate $G_{\boldsymbol \theta_{\boldsymbol Z}^*}$.
\subsubsection*{Calculation for $G_{\boldsymbol \theta_{\boldsymbol Z}^*}$}
\begin{align*}
    & \left.\frac{\partial \boldsymbol g(\boldsymbol \theta_{\boldsymbol Z},\boldsymbol \alpha^{\text{mult}*})}{\partial \boldsymbol \theta_{\boldsymbol Z}}\right|_{\boldsymbol \theta_{\boldsymbol Z}=\boldsymbol \theta_{\boldsymbol Z}^*}= \left.\frac{\partial}{\partial \boldsymbol \theta_{\boldsymbol Z}}\left[\frac{S^{\text{mult}}}{\pi(\boldsymbol X^{\text{mult}},\boldsymbol \alpha^{\text{mult}*})}\left\{D\boldsymbol Z-\frac{e^{\boldsymbol \theta_{\boldsymbol Z}^{*'}\boldsymbol Z}}{(1+e^{\boldsymbol \theta_{\boldsymbol Z}^{*'}\boldsymbol Z})}\cdot \boldsymbol Z\right\}\right]\right|_{\boldsymbol \theta_{\boldsymbol Z}=\boldsymbol \theta_{\boldsymbol Z}^*}\\
    & = -\frac{S^{\text{mult}}}{\pi(\boldsymbol X^{\text{mult}},\boldsymbol \alpha^{\text{mult}*})} \cdot \frac{e^{\boldsymbol \theta_{\boldsymbol Z}^{*'}\boldsymbol Z}}{(1+e^{\boldsymbol \theta_{\boldsymbol Z}^{*'}\boldsymbol Z})^2}\cdot \boldsymbol Z\boldsymbol Z'\cdot
\end{align*}
Therefore we obtain 
\begin{align*}
    G_{\boldsymbol \theta_{\boldsymbol Z}^*} = \mathbb{E}\left[-\frac{S^{\text{mult}}}{\pi(\boldsymbol X^{\text{mult}},\boldsymbol \alpha^{\text{mult}*})} \cdot \frac{e^{\boldsymbol \theta_{\boldsymbol Z}^{*'}\boldsymbol Z}}{(1+e^{\boldsymbol \theta_{\boldsymbol Z}^{*'}\boldsymbol Z})^2}\cdot \boldsymbol Z \boldsymbol Z'\right]\cdot
\end{align*}
Next we calculate $G_{\boldsymbol \alpha^{\text{mult}*}}$.
\subsubsection*{Calculation for $G_{\boldsymbol \alpha^{\text{mult}*}}$}
\begin{align*}
    & \left.\frac{\partial \boldsymbol g(\boldsymbol \theta_{\boldsymbol Z}^*,\boldsymbol \alpha^{\text{mult}})}{\partial \boldsymbol \alpha^{\text{mult}}}\right|_{\boldsymbol\alpha^{\text{mult}}=\boldsymbol \alpha^{\text{mult}*}}=\left[\left.\frac{\partial \boldsymbol g(\boldsymbol \theta_{\boldsymbol Z}^*,\boldsymbol \alpha^{\text{mult}})}{\partial \boldsymbol \alpha_1}\right|_{\boldsymbol\alpha^{\text{mult}}=\boldsymbol \alpha^{\text{mult}*}},...,\left.\frac{\partial \boldsymbol g(\boldsymbol \theta_{\boldsymbol Z}^*,\boldsymbol \alpha^{\text{mult}})}{\partial \boldsymbol \alpha_K}\right|_{\boldsymbol\alpha^{\text{mult}}=\boldsymbol \alpha^{\text{mult}*}}\right]
\end{align*}
Moreover we obtain $\forall$  $k \{1,2,...,K\}$
\begin{align*}
    & \left.\frac{\partial \boldsymbol g(\boldsymbol \theta_{\boldsymbol Z}^*,\boldsymbol \alpha^{\text{mult}})}{\partial \boldsymbol \alpha_k}\right|_{\boldsymbol\alpha^{\text{mult}}=\boldsymbol \alpha^{\text{mult}*}}\\
    & = \left.\frac{\partial}{\partial \boldsymbol\alpha_k}\left[\frac{S^{\text{mult}}}{\pi(\boldsymbol X^{\text{mult}},\boldsymbol \alpha^{\text{mult}})}\left\{D\boldsymbol Z-\frac{e^{\boldsymbol \theta_{\boldsymbol Z}^{*'}\boldsymbol Z}}{(1+e^{\boldsymbol \theta_{\boldsymbol Z}^{*'}\boldsymbol Z})}\cdot \boldsymbol Z\right\}\right]\right|_{\boldsymbol\alpha^{\text{mult}}=\boldsymbol \alpha^{\text{mult}*}}\\
    & = -\frac{S^{\text{mult}}}{\pi(\boldsymbol X^{\text{mult}},\boldsymbol \alpha^{\text{mult}*})^2}\cdot\left\{D\boldsymbol Z-\frac{e^{\boldsymbol \theta_{\boldsymbol Z}^{*'}\boldsymbol Z}}{(1+e^{\boldsymbol \theta_{\boldsymbol Z}^{*'}\boldsymbol Z})}\cdot \boldsymbol Z\right\}\cdot \left.\frac{\partial \pi(\boldsymbol X^{\text{mult}},\boldsymbol \alpha^{\text{mult}})}{\partial \boldsymbol \alpha_k}\right|_{\boldsymbol\alpha^{\text{mult}}=\boldsymbol \alpha^{\text{mult}*}}\\
    & =-\frac{S^{\text{mult}}}{\pi(\boldsymbol X^{\text{mult}},\boldsymbol\alpha^{\text{mult}*})^2}.\{1-\pi_1(\boldsymbol X_1,\boldsymbol \alpha^*_1)\}\cdot \cdot \cdot\{1-\pi_K(\boldsymbol X_K,\boldsymbol \alpha^*_K)\}\boldsymbol Z\cdot \{\pi_k(\boldsymbol X_k,\boldsymbol \alpha^*_k)\cdot \boldsymbol X_k'\}\cdot
   \left\{D-\frac{e^{\boldsymbol \theta_{\boldsymbol Z}^{*'}\boldsymbol Z}}{(1+e^{\boldsymbol \theta_{\boldsymbol Z}^{*'}\boldsymbol Z})}\right\}
\end{align*}
Therefore we obtain 
\begin{align*}
    & G_{\boldsymbol \alpha^{\text{mult}*}}=\mathbb{E}\left[-\frac{S^{\text{mult}}}{\pi(\boldsymbol X^{\text{mult}},\boldsymbol\alpha^{\text{mult}*})^2}.\{1-\pi_1(\boldsymbol X_1,\boldsymbol \alpha^*_1)\}\cdot \cdot \cdot\{1-\pi_K(\boldsymbol X_K,\boldsymbol \alpha^*_K)\}\right.\\
    & \hspace{2cm} \left.\boldsymbol Z\cdot \{\pi_1(\boldsymbol X_1,\boldsymbol \alpha^*_1)\cdot \boldsymbol X_1',\cdot \cdot \cdot,\pi_K(\boldsymbol X_K,\boldsymbol \alpha^*_K)\cdot \boldsymbol X_K'\}\cdot
   \left\{D-\frac{e^{\boldsymbol \theta_{\boldsymbol Z}^{*'}\boldsymbol Z}}{(1+e^{\boldsymbol \theta_{\boldsymbol Z}^{*'}\boldsymbol Z})}\right\}\right]
\end{align*}
Next we calculate $\Psi(\boldsymbol \alpha^{\text{mult}*})$.
\setlength\arraycolsep{-8pt}
\subsubsection*{Calculation for $\boldsymbol\Psi(\boldsymbol \alpha^{\text{mult}*})$}
\begin{align*}
   & \left.\frac{ \partial \boldsymbol h(\boldsymbol \alpha^{\text{mult}})}{\partial \boldsymbol \alpha^{\text{mult}}}\right|_{\boldsymbol \alpha=\boldsymbol \alpha^{\text{mult}*}}=\begin{pmatrix}
     \left.\frac{ \partial \boldsymbol h_1(\boldsymbol \alpha_1)}{\partial \boldsymbol \alpha_1}\right|_{\boldsymbol \alpha_1=\boldsymbol \alpha^*_1} & \hdots & 0 \\
     \vdots & \ddots & \vdots \\
     0 & \hdots &  \left.\frac{ \partial \boldsymbol h_K(\boldsymbol \alpha_K)}{\partial \boldsymbol \alpha_K}\right|_{\boldsymbol \alpha_K=\boldsymbol \alpha^*_K}
         \end{pmatrix}
\end{align*}
Moreover we obtain $\forall$  $k \{1,2,...,K\}$
\begin{align*}
    & \left.\frac{ \partial \boldsymbol h_k(\boldsymbol \alpha_k)}{\partial \boldsymbol \alpha_k}\right|_{\boldsymbol \alpha_k=\boldsymbol \alpha^*_k}= \left.\frac{ \partial}{\partial \boldsymbol \alpha_k}\left\{\frac{S_k\boldsymbol X_k}{\pi_k(\boldsymbol X_k,\boldsymbol \alpha_k)}-\boldsymbol X_k\right\}\right|_{\boldsymbol \alpha_k=\boldsymbol \alpha^*_k}=-\frac{S_k\boldsymbol X_k}{\pi_k(\boldsymbol X_k,\boldsymbol \alpha_k)^2}\left.\frac{ \partial}{\partial \boldsymbol \alpha_k}\pi_k(\boldsymbol X_k,\boldsymbol \alpha_k)\right|_{\boldsymbol \alpha_k=\boldsymbol \alpha^*_k}\\
    & = -\frac{S_k\boldsymbol X_k\boldsymbol X_k'}{\pi_k(\boldsymbol X_k,\boldsymbol \alpha^*_k)^2} \pi_k(\boldsymbol X_k,\boldsymbol \alpha^*_k)\{1- \pi_k(\boldsymbol X_k,\boldsymbol \alpha^*_k)\}\cdot
\end{align*}
\setlength\arraycolsep{-20pt}
This implies
\begin{align*}
    H = - \mathbb{E} \begin{pmatrix}
    \frac{S_1\boldsymbol X_1\boldsymbol X_1'}{\pi_1(\boldsymbol X_1,\boldsymbol \alpha^*_1)^2} \pi_1(\boldsymbol X_1,\boldsymbol \alpha^*_1)\{1- \pi_1(\boldsymbol X_1,\boldsymbol \alpha^*_1)\} & \hdots & 0 \\
     \vdots & \ddots & \vdots \\
     0 & \hdots &   \frac{S_K\boldsymbol X_K\boldsymbol X_K'}{\pi_K(\boldsymbol X_K,\boldsymbol \alpha^*_K)^2}\pi_K(\boldsymbol X_K,\boldsymbol \alpha^*_K)\{1- \pi_K(\boldsymbol X_K,\boldsymbol \alpha^*_K)\}
         \end{pmatrix}
\end{align*}
This gives the asymptotic distribution of $\widehat{\boldsymbol \theta}_{\boldsymbol Z}$ for CL.
\end{proof}
\subsection{Theorem S9}\label{sec:pr9}
\begin{theorem}
    Under all the assumptions of Theorems S7, S8 and \textit{A}6 in Supplementary Section \ref{sec:assumptions}, $\frac{1}{N} \cdot \widehat{G_{\boldsymbol \theta_{\boldsymbol Z}}}^{-1}\cdot \widehat{E}\cdot (\widehat{G_{\boldsymbol \theta_{\boldsymbol Z}}}^{-1})^{'}$ is a consistent estimator of the variance of $\widehat{\boldsymbol \theta}_{\boldsymbol Z}$ for JCL where
\begin{align*}
   & \widehat{G_{\boldsymbol \theta_{\boldsymbol Z}}}=-\frac{1}{N}\sum_{i=1}^N S^{\text{mult}}_i \cdot \frac{1}{\pi(\boldsymbol X^{\text{mult}}_i,\widehat{\boldsymbol \alpha^{\text{mult}}})}\cdot \frac{e^{\widehat{\boldsymbol \theta_{\boldsymbol Z}}'\boldsymbol Z_i}}{(1+e^{\widehat{\boldsymbol \theta_{\boldsymbol Z}}'\boldsymbol Z_i})^2}\cdot\boldsymbol Z_i \boldsymbol Z_i'\\
   & \widehat{H} = -\frac{1}{N}\sum_{i=1}^N \begin{pmatrix}
     \frac{S_{1i}\boldsymbol X_{1i}\boldsymbol X_{1i}'}{\pi_1(\boldsymbol X_{1i},\widehat{\boldsymbol\alpha_1})}\cdot \{1- \pi_1(\boldsymbol X_{1i},\widehat{\boldsymbol\alpha_1})\} & \hdots & 0 \\
     \vdots & \ddots & \vdots \\
     0 & \hdots &   \frac{S_{Ki}\boldsymbol X_{Ki}\boldsymbol X_{Ki}'}{\pi_K(\boldsymbol X_{Ki},\widehat{\boldsymbol\alpha_K})}\cdot \{1- \pi_K(\boldsymbol X_{Ki},\widehat{\boldsymbol\alpha_K})\}
         \end{pmatrix}\\
   & \widehat{G_{\boldsymbol \alpha^{\text{mult}}}}=-\frac{1}{N}\sum_{i=1}^N \frac{S^{\text{mult}}_i}{\pi(\boldsymbol X^{\text{mult}}_i,\widehat{\boldsymbol \alpha^{\text{mult}}})^2}.\{1-\pi_1(\boldsymbol X_{1i},\widehat{\boldsymbol\alpha_1})\}\cdot \cdot \cdot\{1-\pi_K(\boldsymbol X_{Ki},\widehat{\boldsymbol\alpha_K})\}\cdot\\
    & \hspace{2cm} \boldsymbol Z_i\cdot \{\pi_1(\boldsymbol X_{1i},\widehat{\boldsymbol\alpha_1})\cdot \boldsymbol X_{1i}',\cdot \cdot \cdot,\pi_K(\boldsymbol X_{Ki},\widehat{\boldsymbol\alpha_K})\cdot \boldsymbol X_{Ki}'\}\cdot
   \left\{D_i-\frac{e^{\widehat{\boldsymbol \theta_{\boldsymbol Z}}\boldsymbol Z_i}}{(1+e^{\widehat{\boldsymbol \theta_{\boldsymbol Z}}\boldsymbol Z_i})}\right\}
   \end{align*}
\begin{align*}
   & \widehat{E_1}=\frac{1}{N}\sum_{i =1}^N S^{\text{mult}}_i \cdot \left\{\frac{1}{\pi(\boldsymbol X^{\text{mult}}_i,\widehat{\boldsymbol \alpha^{\text{mult}}})}\right\}^2\left\{D_i-\frac{e^{\widehat{\boldsymbol \theta_{\boldsymbol Z}}'\boldsymbol Z_i}}{(1+e^{\widehat{\boldsymbol \theta_{\boldsymbol Z}}' \boldsymbol Z_i})}\right\}^2\cdot \boldsymbol Z_i \boldsymbol Z_i'\\
    & \widehat{E}_2 =  \widehat{E}_3'=\frac{1}{N}\sum_{i=1}^N  \widehat{G_{\boldsymbol \alpha^{\text{mult}}}} \widehat{H}^{-1} \frac{1}{\pi(\boldsymbol X^{\text{mult}}_i,\widehat{\boldsymbol \alpha^{\text{mult}})}}\begin{pmatrix}
     \frac{S_{1i}\boldsymbol X_{1i}}{\pi_1(\boldsymbol X_{1i},\widehat{\boldsymbol \alpha_1)}}\\
     \vdots\\
       \frac{S_{Ki}\boldsymbol X_{Ki}}{\pi_K(\boldsymbol X_{Ki},\widehat{\boldsymbol \alpha_K)}}\\
         \end{pmatrix} \left\{D_i\boldsymbol Z_i'-\frac{e^{\widehat{\boldsymbol \theta_{\boldsymbol Z}}'\boldsymbol Z_i}}{(1+e^{\widehat{\boldsymbol \theta_{\boldsymbol Z}}' \boldsymbol Z_i})}  \boldsymbol Z_i'\right\}\\
         & -\frac{1}{N}\sum_{i=1}^N  \widehat{G_{\boldsymbol \alpha^{\text{mult}}}} \widehat{H}^{-1}  \frac{S^{\text{mult}}_i}{\pi(\boldsymbol X^{\text{mult}}_i,\widehat{\boldsymbol \alpha^{\text{mult}})}} \cdot \begin{pmatrix}
      \boldsymbol X_{1i}\\
     \vdots\\
    \boldsymbol X_{Ki}
         \end{pmatrix}\cdot  \left\{D_i\boldsymbol Z_i'-\frac{e^{\widehat{\boldsymbol \theta_{\boldsymbol Z}}'\boldsymbol Z_i}}{(1+e^{\widehat{\boldsymbol \theta_{\boldsymbol Z}}' \boldsymbol Z_i})} \cdot \boldsymbol Z_i'\right\}
\end{align*}
\begin{align*}
    & \widehat{E_4} = \frac{1}{N}\cdot\widehat{G_{\boldsymbol \alpha^{\text{mult}}}}\cdot \widehat{H}^{-1} \sum_{i=1}^N  \left[\begin{pmatrix}
     \frac{S_{1i}\boldsymbol X_{1i}}{\pi_1(\boldsymbol X_{1i},\widehat{\boldsymbol \alpha_1})}\\
     \vdots\\
       \frac{S_{Ki}\boldsymbol X_{Ki}}{\pi_K(\boldsymbol X_{Ki},\widehat{\boldsymbol \alpha_K})}\\
         \end{pmatrix}\cdot\left(\frac{S_{1i}\boldsymbol X_{1i}'}{\pi_1(\boldsymbol X_{1i},\widehat{\boldsymbol \alpha_1})},\cdot \cdot \cdot, \frac{S_{Ki}\boldsymbol X_{Ki}'}{\pi_K(\boldsymbol X_{Ki},\widehat{\boldsymbol \alpha_K})}\right)\right.\\
        & \left.-2\cdot \begin{pmatrix}
  \boldsymbol X_{1i}\\
  \vdots\\
    \boldsymbol X_{Ki}\\
         \end{pmatrix}\cdot \left(\frac{S_{1i}\boldsymbol X_{1i}'}{\pi_1(\boldsymbol X_{1i},\widehat{\boldsymbol \alpha_1})},\cdot \cdot \cdot, \frac{S_{Ki}\boldsymbol X_{Ki}'}{\pi_K(\boldsymbol X_{Ki},\widehat{\boldsymbol \alpha_K})}\right)+  \frac{S_{i}}{\pi(\boldsymbol X^{\text{mult}}_i,\widehat{\boldsymbol \alpha^{\text{mult}})}} \cdot \begin{pmatrix}
    \boldsymbol X_{1i}\\
     \vdots\\
     \boldsymbol X_{Ki}\\
         \end{pmatrix}(\boldsymbol X_{1i},\cdot \cdot \cdot,\boldsymbol X_{Ki})\right](\widehat{H}^{-1})'\cdot(\widehat{G_{\boldsymbol \alpha^{\text{mult}}}})'\\
          &  \widehat{E}=\widehat{E}_1-\widehat{E}_2-\widehat{E}_3 + \widehat{E}_4
\end{align*}
\end{theorem}
\begin{proof}
    Under all the assumptions of Theorems S7, S8 and \textit{A}6 in Supplementary Section \ref{sec:assumptions} and using ULLN and Continuous Mapping Theorem, the proof of consistency for each of the following sample quantities are exactly same as the approach in Theorem S3. Using the exact same steps on the joint parameters $\boldsymbol \eta$ instead of $\boldsymbol \theta_{\boldsymbol Z}$ (as in Theorem S3) we obtain
\begin{align*}
     & \widehat{G_{\boldsymbol \theta_{\boldsymbol Z}}}=-\frac{1}{N}\sum_{i=1}^N S^{\text{mult}}_i \cdot \frac{1}{\pi(\boldsymbol X^{\text{mult}}_i,\widehat{\boldsymbol \alpha^{\text{mult}}})}\cdot \frac{e^{\widehat{\boldsymbol \theta_{\boldsymbol Z}}'\boldsymbol Z_i}}{(1+e^{\widehat{\boldsymbol \theta_{\boldsymbol Z}}'\boldsymbol Z_i})^2}\cdot\boldsymbol Z_i \boldsymbol Z_i' \\
     & \xrightarrow{p}G_{\boldsymbol \theta_{\boldsymbol Z}^*} = \mathbb{E}\left\{-\frac{S^{\text{mult}}}{\pi(\boldsymbol X^{\text{mult}},\boldsymbol \alpha^{\text{mult}*})}\cdot \frac{e^{\boldsymbol \theta_{\boldsymbol Z}^{*'}\boldsymbol Z}}{(1+e^{\boldsymbol \theta_{\boldsymbol Z}^{*'}\boldsymbol Z})^2}\cdot \boldsymbol Z\boldsymbol Z'\right\}\cdot
\end{align*}
Similarly we obtain 
\begin{align*}
     &  \widehat{G_{\boldsymbol \alpha^{\text{mult}}}}=-\frac{1}{N}\sum_{i=1}^N \frac{S^{\text{mult}}_i}{\pi(\boldsymbol X^{\text{mult}}_i,\widehat{\boldsymbol \alpha^{\text{mult}}})^2}.\{1-\pi_1(\boldsymbol X_{1i},\widehat{\boldsymbol\alpha_1})\}\cdot \cdot \cdot\{1-\pi_K(\boldsymbol X_{Ki},\widehat{\boldsymbol\alpha_K})\}\cdot\\
    & \hspace{2cm} \boldsymbol Z_i\cdot \{\pi_1(\boldsymbol X_{1i},\widehat{\boldsymbol\alpha_1})\cdot \boldsymbol X_{1i}',\cdot \cdot \cdot,\pi_K(\boldsymbol X_{Ki},\widehat{\boldsymbol\alpha_K})\cdot \boldsymbol X_{Ki}'\}\cdot
   \left\{D_i-\frac{e^{\widehat{\boldsymbol \theta_{\boldsymbol Z}}\boldsymbol Z_i}}{(1+e^{\widehat{\boldsymbol \theta_{\boldsymbol Z}}\boldsymbol Z_i})}\right\}\\
     & \xrightarrow{p} G_{\boldsymbol \alpha^{\text{mult}*}}=\mathbb{E}\left[-\frac{S^{\text{mult}}}{\pi(\boldsymbol X^{\text{mult}},\boldsymbol\alpha^{\text{mult}*})^2}.\{1-\pi_1(\boldsymbol X_1,\boldsymbol \alpha^*_1)\}\cdot \cdot \cdot\{1-\pi_K(\boldsymbol X_K,\boldsymbol \alpha^*_K)\}\cdot\right.\\
    & \hspace{2cm} \left.\boldsymbol Z\cdot \{\pi_1(\boldsymbol X_1,\boldsymbol \alpha_1^*)\cdot \boldsymbol X_1',\cdot \cdot \cdot,\pi_K(\boldsymbol X_K,\boldsymbol \alpha^*_K)\cdot \boldsymbol X_K'\}\cdot
   \left\{D-\frac{e^{\boldsymbol \theta_{\boldsymbol Z}^{*'}\boldsymbol Z}}{(1+e^{\boldsymbol \theta_{\boldsymbol Z}^{*'}\boldsymbol Z})}\right\}\right]
\end{align*}
\begin{align*}
     & \widehat{H}= -\frac{1}{N}\sum_{i=1}^N \begin{pmatrix}
     \frac{S_{1i}\boldsymbol X_{1i}\boldsymbol X_{1i}'}{\pi_1(\boldsymbol X_{1i},\widehat{\boldsymbol\alpha_1})}\cdot \{1- \pi_1(\boldsymbol X_{1i},\widehat{\boldsymbol\alpha_1})\} & \hdots & 0 \\
     \vdots & \ddots & \vdots \\
     0 & \hdots &   \frac{S_{Ki}\boldsymbol X_{Ki}\boldsymbol X_{Ki}'}{\pi_K(\boldsymbol X_{Ki},\widehat{\boldsymbol\alpha_K})}\cdot \{1- \pi_K(\boldsymbol X_{Ki},\widehat{\boldsymbol\alpha_K})\}
         \end{pmatrix}\\
     &\xrightarrow{p} H = - \mathbb{E} \begin{pmatrix}
     \frac{S_1\boldsymbol X_1\boldsymbol X_1'}{\pi_1(\boldsymbol X_1,\boldsymbol \alpha_1^*)}\cdot \{1- \pi_1(\boldsymbol X_1,\boldsymbol \alpha^*_1)\} & \hdots & 0 \\
     \vdots & \ddots & \vdots \\
     0 & \hdots &   \frac{S_K\boldsymbol X_K\boldsymbol X_K'}{\pi_K(\boldsymbol X_K,\boldsymbol \alpha_K^*)}\cdot \{1- \pi_1(\boldsymbol X_K,\boldsymbol \alpha^*_K)\}
         \end{pmatrix}
\end{align*}
Similarly we obtain
\begin{align*}
     & \widehat{E_1}=\frac{1}{N}\sum_{i =1}^N S^{\text{mult}}_i \cdot \left\{\frac{1}{\pi(\boldsymbol X^{\text{mult}}_i,\widehat{\boldsymbol \alpha^{\text{mult}}})}\right\}^2\left\{D_i-\frac{e^{\widehat{\boldsymbol \theta_{\boldsymbol Z}}'\boldsymbol Z_i}}{(1+e^{\widehat{\boldsymbol \theta_{\boldsymbol Z}}' \boldsymbol Z_i})}\right\}^2\cdot \boldsymbol Z_i \boldsymbol Z_i'\\
     & \xrightarrow{p}  \mathbb{E}\left[S \cdot \left\{\frac{1}{\pi(\boldsymbol X^{\text{mult}},\boldsymbol \alpha^{\text{mult}*})}\right\}^2\cdot \left\{D-\frac{e^{\boldsymbol \theta_{\boldsymbol Z}^{*'}\boldsymbol Z_i}}{(1+e^{\boldsymbol \theta_{\boldsymbol Z}^{*'} \boldsymbol Z})}\right\}^2\cdot \boldsymbol Z\boldsymbol Z'\right]\cdot
\end{align*}
\begin{align*}
     &  \widehat{E}_2 =  \widehat{E}_3'=\frac{1}{N}\sum_{i=1}^N  \widehat{G_{\boldsymbol \alpha^{\text{mult}}}} \widehat{H}^{-1} \cdot\frac{1}{\pi(\boldsymbol X^{\text{mult}}_i,\widehat{\boldsymbol \alpha^{\text{mult}})}}\begin{pmatrix}
     \frac{S_{1i}\boldsymbol X_{1i}}{\pi_1(\boldsymbol X_{1i},\widehat{\boldsymbol \alpha_1)}}\\
     \vdots\\
       \frac{S_{Ki}\boldsymbol X_{Ki}}{\pi_K(\boldsymbol X_{Ki},\widehat{\boldsymbol \alpha_K)}}\\
         \end{pmatrix}
         \left\{D_i\boldsymbol Z_i'-\frac{e^{\widehat{\boldsymbol \theta_{\boldsymbol Z}}'\boldsymbol Z_i}}{(1+e^{\widehat{\boldsymbol \theta_{\boldsymbol Z}}' \boldsymbol Z_i})}  \boldsymbol Z_i'\right\}\\
         & -\frac{1}{N}\sum_{i=1}^N  \widehat{G_{\boldsymbol \alpha^{\text{mult}}}}\frac{S^{\text{mult}}_i}{\pi(\boldsymbol X^{\text{mult}}_i,\widehat{\boldsymbol \alpha^{\text{mult}})}}\cdot  \widehat{H}^{-1} \cdot \begin{pmatrix}
      \boldsymbol X_{1i}\\
     \vdots\\
    \boldsymbol X_{Ki}
         \end{pmatrix}\left\{D_i\boldsymbol Z_i'-\frac{e^{\widehat{\boldsymbol \theta_{\boldsymbol Z}}'\boldsymbol Z_i}}{(1+e^{\widehat{\boldsymbol \theta_{\boldsymbol Z}}' \boldsymbol Z_i})}  \boldsymbol Z_i'\right\}\xrightarrow{p} \mathbb{E}\{G_{\boldsymbol \alpha^{\text{mult}*}} H^{-1}\boldsymbol h(\boldsymbol \alpha^{\text{mult}*})  \boldsymbol g(\boldsymbol \theta_{\boldsymbol Z}^*,\boldsymbol \alpha^{\text{mult}*})'\}\cdot
\end{align*}
\begin{align*}
    &\widehat{E_4} = \frac{1}{N}\cdot\widehat{G_{\boldsymbol \alpha^{\text{mult}}}}\cdot \widehat{H}^{-1} \sum_{i=1}^N  \left[\begin{pmatrix}
     \frac{S_{1i}\boldsymbol X_{1i}}{\pi_1(\boldsymbol X_{1i},\widehat{\boldsymbol \alpha_1})}\\
     \vdots\\
       \frac{S_{Ki}\boldsymbol X_{Ki}}{\pi_K(\boldsymbol X_{Ki},\widehat{\boldsymbol \alpha_K})}\\
         \end{pmatrix}\cdot\left(\frac{S_{1i}\boldsymbol X_{1i}'}{\pi_1(\boldsymbol X_{1i},\widehat{\boldsymbol \alpha_1})},\cdot \cdot \cdot, \frac{S_{Ki}\boldsymbol X_{Ki}'}{\pi_K(\boldsymbol X_{Ki},\widehat{\boldsymbol \alpha_K})}\right)\right.\\
        & \left.-2\cdot \begin{pmatrix}
  \boldsymbol X_{1i}\\
  \vdots\\
    \boldsymbol X_{Ki}\\
         \end{pmatrix}\cdot \left(\frac{S_{1i}\boldsymbol X_{1i}'}{\pi_1(\boldsymbol X_{1i},\widehat{\boldsymbol \alpha_1})},\cdot \cdot \cdot, \frac{S_{Ki}\boldsymbol X_{Ki}'}{\pi_K(\boldsymbol X_{Ki},\widehat{\boldsymbol \alpha_K})}\right)\right.\\
         & \left.+  \frac{S_{i}}{\pi(\boldsymbol X^{\text{mult}}_i,\widehat{\boldsymbol \alpha^{\text{mult}})}} \cdot \begin{pmatrix}
    \boldsymbol X_{1i}\\
     \vdots\\
     \boldsymbol X_{Ki}\\
         \end{pmatrix}(\boldsymbol X_{1i},\cdot \cdot \cdot,\boldsymbol X_{Ki})\right](\widehat{H}^{-1})'\cdot(\widehat{G_{\boldsymbol \alpha^{\text{mult}}}})'\\
    & \xrightarrow{p} \mathbb{E}\{G_{\boldsymbol \alpha^{\text{mult}*}}H^{-1}\cdot \boldsymbol h(\boldsymbol \alpha^{\text{mult}*})\cdot \boldsymbol h(\boldsymbol \alpha^{\text{mult}*})'\cdot (H^{-1})'\cdot (G_{\boldsymbol\alpha^{\text{mult}*}})'\} 
\end{align*}
Therefore we obtain
\begin{align*}
    &\widehat{E}=\widehat{E}_1-\widehat{E}_2-\widehat{E}_3 + \widehat{E}_4\xrightarrow{p}\mathbb{E}[\{\boldsymbol g(\boldsymbol \theta_{\boldsymbol Z}^*,\boldsymbol \alpha^{\text{mult}*})+G_{\boldsymbol \alpha^{\text{mult}*}}\cdot \boldsymbol \Psi(\boldsymbol \alpha^{\text{mult}*})\}\{\boldsymbol g(\boldsymbol \theta_{\boldsymbol Z}^*,\boldsymbol \alpha^{\text{mult}*})+G_{\boldsymbol \alpha^{\text{mult}*}}.\boldsymbol \Psi(\boldsymbol \alpha^{\text{mult}*})\}']\cdot
\end{align*}
Using all the above results we obtain $\widehat{G_{\boldsymbol \theta_{\boldsymbol Z}}}^{-1}\cdot\widehat{E}\cdot(\widehat{G_{\boldsymbol \theta_{\boldsymbol Z}}}^{-1})^{'}$
\begin{align*}
    &\xrightarrow{p}(G_{\boldsymbol \theta_{\boldsymbol Z}^*})^{-1} \mathbb{E}[\{\boldsymbol g(\boldsymbol \theta_{\boldsymbol Z}^*,\boldsymbol \alpha^{\text{mult}*})+G_{\boldsymbol \alpha^{\text{mult}*}}\boldsymbol \Psi(\boldsymbol \alpha^{\text{mult}*})\}\\
    & \hspace{4cm}\{\boldsymbol g(\boldsymbol \theta_{\boldsymbol Z}^*,\boldsymbol \alpha^{\text{mult}*})+G_{\boldsymbol \alpha^{\text{mult}*}}\boldsymbol \Psi(\boldsymbol \alpha^{\text{mult}*})\}'](G_{\boldsymbol \theta_{\boldsymbol Z}^*}^{-1})^{'}\cdot
\end{align*}
From Theorem S8 using the same approach used in the last step of Theorem S3, we obtain that $\frac{1}{N}\cdot\widehat{G_{\boldsymbol \theta_{\boldsymbol Z}}}^{-1}\cdot\hat{E}\cdot(\widehat{G_{\boldsymbol \theta_{\boldsymbol Z}}}^{-1})^{'}$ is a consistent estimator of the asymptotic variance of $\widehat{\boldsymbol \theta}_{\boldsymbol Z}$.
\end{proof}
\section{JAIPW proofs}\label{sec:jaipwsupp}
\subsection{Proof of Theorem 3.1 : Double Robustness Single Cohort}\label{sec:pr2_10}
\begin{proof}
In this case, the estimating equation consists of the propensity score model, auxiliary score model and disease model parameter estimation. 
We combined the nuisance parameters estimation together, namely $\boldsymbol \alpha$ and $\boldsymbol f(.)$ where $\boldsymbol f(.)$. This is again a two step estimation framework. We need to prove $\mathbb{E}(\phi_N(\boldsymbol \theta_{\boldsymbol Z}^*,\boldsymbol \alpha^{*},\boldsymbol f^*))=\mathbf{0}$ where either $\pi(\boldsymbol{X}_i, \boldsymbol{\alpha})$ is correctly specified or $\boldsymbol f(\boldsymbol{X}_i, \boldsymbol{\theta}_{\boldsymbol{Z}})$ is correctly specified or both and 
$$\phi_N(\boldsymbol \theta_{\boldsymbol Z},\boldsymbol \alpha,\boldsymbol f)=\frac{1}{N} \sum_{i=1}^N \frac{S_i}{\pi(\boldsymbol{X}_i, \boldsymbol{\alpha})}  \left(\mathcal{U}_i(\boldsymbol{\theta}_{\boldsymbol{Z}}) - \boldsymbol f(\boldsymbol{X}_i, \boldsymbol{\theta}_{\boldsymbol{Z}})\right) + \frac{1}{N} \sum_{i=1}^N \frac{S_{\text{ext},i}}{\pi_{\text{ext},i}}\cdot \boldsymbol f(\boldsymbol{X}_i, \boldsymbol{\theta}_{\boldsymbol{Z}})\cdot$$
\noindent
\textbf{Propensity Score Model Correct : } In this case, $P(S=1|\boldsymbol X)=\pi(\boldsymbol X,\boldsymbol \alpha^{*})$ where $\widehat{\boldsymbol \alpha}$ estimated from equation (3.7) of the main text converges in probability to $\boldsymbol \alpha^{*}$. In this case, the auxiliary score model might be incorrectly specified. We define $\boldsymbol f^{*}$ as the value where estimate of equation (3.9) or (3.11) of the main text, namely $\widehat{\boldsymbol f}$ converges in $L_2$ norm. However, one should note $\mathbb{E}(\boldsymbol U(\boldsymbol{\theta}_{\boldsymbol{Z}})|\boldsymbol X)$ might not be equal to $f^*(\boldsymbol X,\boldsymbol \theta_{\boldsymbol Z})$.
\begin{align*}
     \mathbb{E}(\phi_N(\boldsymbol \theta_{\boldsymbol Z}^*,\boldsymbol \alpha^{*},\boldsymbol f^*)) & = \mathbb{E}\left[\frac{1}{N} \sum_{i=1}^N \frac{S_i}{\pi(\boldsymbol X_i,\boldsymbol \alpha^{*})}\cdot(\mathcal{U}_i(\boldsymbol \theta_{\boldsymbol Z}^*)-f^*(\boldsymbol X_i,\boldsymbol \theta_{\boldsymbol Z}^*))+\frac{1}{N}\sum_{i=1}^N \frac{S_{\text{ext},i}}{\pi_{\text{ext},i}(\boldsymbol M_i)}\cdot f^*(\boldsymbol X_i, \boldsymbol \theta_{\boldsymbol Z}^*)\right]
\end{align*}
The above expectation can be rewritten as
\begin{align*}
    \mathbb{E}\left[\frac{1}{N} \sum_{i=1}^N  \frac{S_i}{\pi(\boldsymbol X_i,\boldsymbol \alpha^{*})}\cdot \mathcal{U}_i(\boldsymbol \theta_{\boldsymbol Z}^*)+ \frac{1}{N} \sum_{i=1}^N f^*(\boldsymbol X_i, \boldsymbol \theta_{\boldsymbol Z}^*)\cdot \left(\frac{S_{\text{ext},i}}{\pi_{\text{ext},i}(\boldsymbol M_i)}-\frac{S_i}{\pi(\boldsymbol X_i,\boldsymbol \alpha^{*})}\right)\right]
\end{align*}
The first term of the expectation can be simplified as
\begin{align*}
    \mathbb{E}\left[\frac{1}{N} \sum_{i=1}^N  \frac{S_i}{\pi(\boldsymbol X_i,\boldsymbol \alpha^{*})}\cdot \mathcal{U}_i(\boldsymbol \theta_{\boldsymbol Z}^*)\right]= \frac{1}{N} \sum_{i=1}^N  \mathbb{E}\left[\frac{S_i}{\pi(\boldsymbol X_i,\boldsymbol \alpha^{*})}\cdot \mathcal{U}_i(\boldsymbol \theta_{\boldsymbol Z}^*)\right]
\end{align*}
Since $\pi(\boldsymbol X_i,\boldsymbol \alpha^{*})=\pi(\boldsymbol X_i)$, the above expectation is
\begin{align*}
    \frac{1}{N} \sum_{i=1}^N  \mathbb{E}\left[\frac{S_i}{\pi(\boldsymbol X_i)}\cdot \mathcal{U}_i(\boldsymbol \theta_{\boldsymbol Z}^*)\right]
\end{align*}
Using the same set of arguments as in the proof of Theorem S1 we obtain
\begin{align*}
    \frac{1}{N} \sum_{i=1}^N  \mathbb{E}\left[\frac{S_i}{\pi(\boldsymbol X_i)}\cdot \mathcal{U}_i(\boldsymbol \theta_{\boldsymbol Z}^*)\right]= \frac{1}{N} \sum_{i=1}^N  \mathbb{E}(\mathcal{U}_i(\boldsymbol \theta_{\boldsymbol Z}^*))=0\cdot
\end{align*}
On the other hand the second term of the expectation is
\begin{align*}
    &\mathbb{E}\left[\frac{1}{N} \sum_{i=1}^N f^*(\boldsymbol X_i, \boldsymbol \theta_{\boldsymbol Z}^*)\cdot \left(\frac{S_{\text{ext},i}}{\pi_{\text{ext},i}(\boldsymbol M_i)}-\frac{S_i}{\pi(\boldsymbol X_i,\boldsymbol \alpha^{*})}\right)\right]\\
    = &\frac{1}{N} \sum_{i=1}^N\mathbb{E}\left[f^*(\boldsymbol X_i, \boldsymbol \theta_{\boldsymbol Z}^*)\cdot \left(\frac{S_{\text{ext},i}}{\pi_{\text{ext},i}(M_i)}-\frac{S_i}{\pi(\boldsymbol X_i,\boldsymbol \alpha^{*})}\right)\right]\\
    =& \frac{1}{N} \sum_{i=1}^N \mathbb{E}_{\boldsymbol X_i,\boldsymbol M_i}\left[\mathbb{E}\left\{\left.f^*(\boldsymbol X_i, \boldsymbol \theta_{\boldsymbol Z}^*)\cdot \left(\frac{S_{\text{ext},i}}{\pi_{\text{ext},i}(\boldsymbol M_i)}-\frac{S_i}{\pi(\boldsymbol X_i,\boldsymbol \alpha^{*})}\right)\right|\boldsymbol X_i,\boldsymbol M_i\right\}\right]\\
     = & \frac{1}{N} \sum_{i=1}^N \mathbb{E}_{\boldsymbol X_i,\boldsymbol M_i}\left[f^*(\boldsymbol X_i, \boldsymbol \theta_{\boldsymbol Z}^*) \cdot \left\{\mathbb{E}\left(\left.\frac{S_{\text{ext},i}}{\pi_{\text{ext},i}(\boldsymbol M_i)}\right|\boldsymbol X_i,\boldsymbol M_i\right)-\mathbb{E}\left(\left.\frac{S_i}{\pi(\boldsymbol X_i,\boldsymbol \alpha^{*})}\right|\boldsymbol X_i,\boldsymbol M_i\right)\right\}\right]\\
     = & \frac{1}{N} \sum_{i=1}^N \mathbb{E}_{\boldsymbol X_i,\boldsymbol M_i}\left[f^*(\boldsymbol X_i, \boldsymbol \theta_{\boldsymbol Z}^*) \cdot \left\{\mathbb{E}\left(\left.\frac{S_{\text{ext},i}}{\pi_{\text{ext},i}(\boldsymbol M_i)}\right|\boldsymbol M_i\right)-\mathbb{E}\left(\left.\frac{S_i}{\pi(\boldsymbol X_i,\boldsymbol \alpha^{*})}\right|\boldsymbol X_i\right)\right\}\right]
\end{align*}
Since both $\mathbb{E}\left(\left.\frac{S_{\text{ext},i}}{\pi_{\text{ext},i}(\boldsymbol M_i)}\right|\boldsymbol M_i\right)=1$ and $\mathbb{E}\left(\left.\frac{S_i}{\pi(\boldsymbol X_i,\boldsymbol \alpha^{*})}\right|\boldsymbol X_i\right)=1$, the second term also equals to 0. Hence under correct specification of the selection model, we obtain $\mathbb{E}(\phi_N(\boldsymbol \theta_{\boldsymbol Z}^*,\boldsymbol \alpha^{*},\boldsymbol f^*))=\mathbf{0}$.\\

\noindent
\textbf{Auxiliary Score Model Correct : } Under Condition C2, in this case  $\mathbb{E}(\mathcal{U}(\boldsymbol \theta_{\boldsymbol Z})|\boldsymbol X,S=1)=\mathbb{E}(\mathcal{U}(\boldsymbol \theta_{\boldsymbol Z})|\boldsymbol X)=\boldsymbol f^*(\boldsymbol X,\theta_{\boldsymbol Z})$ where $\widehat{\boldsymbol f}$ estimated from equation (3.9) or (3.11) of the main text converges in $L_2$ norm to $\boldsymbol f^*$. We define $\boldsymbol \alpha^{*}$ as the value where estimate of equation (3.7) of the main text, namely $\widehat{\boldsymbol \alpha}$ converges in probability. However, one should note $\boldsymbol P(S=1|\boldsymbol X)$ might not be equal to $\pi(\boldsymbol X,\boldsymbol \alpha^{*})$.
\begin{align*}
     \mathbb{E}(\phi_N(\boldsymbol \theta_{\boldsymbol Z}^*,\boldsymbol \alpha^{*},\boldsymbol f^*)) & = \mathbb{E}\left[\frac{1}{N} \sum_{i=1}^N \frac{S_i}{\pi(\boldsymbol X_i,\boldsymbol \alpha^{*})}\cdot(\mathcal{U}_i(\boldsymbol \theta_{\boldsymbol Z}^*)-f^*(\boldsymbol X_i, \boldsymbol \theta_{\boldsymbol Z}^*))+\frac{1}{N}\sum_{i=1}^N \frac{S_{\text{ext},i}}{\pi_{\text{ext},i}(\boldsymbol M_i)}\cdot f^*(\boldsymbol X_i, \boldsymbol \theta_{\boldsymbol Z}^*)\right]
\end{align*}
The first term of the expectation can be simplified as
\begin{align*}
     & \mathbb{E}\left[\frac{1}{N} \sum_{i=1}^N \frac{S_i}{\pi(\boldsymbol X_i,\boldsymbol \alpha^{*})}\cdot(\mathcal{U}_i(\boldsymbol \theta_{\boldsymbol Z}^*)-f^*(\boldsymbol X_i, \boldsymbol \theta_{\boldsymbol Z}^*))\right]\\
     & =\frac{1}{N} \sum_{i=1}^N \mathbb{E}\left[\frac{S_i}{\pi(\boldsymbol X_i,\boldsymbol \alpha^{*})}\cdot(\mathcal{U}_i(\boldsymbol \theta_{\boldsymbol Z}^*)-f^*(\boldsymbol X_i, \boldsymbol \theta_{\boldsymbol Z}^*))\right]\\
     &=\frac{1}{N} \sum_{i=1}^N \mathbb{E}_{\boldsymbol X_i}\left[\mathbb{E}\left\{\left.\frac{S_i}{\pi(\boldsymbol X_i,\boldsymbol \alpha^{*})}\cdot(\mathcal{U}_i(\boldsymbol \theta_{\boldsymbol Z}^*)-f^*(\boldsymbol X_i, \boldsymbol \theta_{\boldsymbol Z}^*))\right|\boldsymbol X_i\right\}\right]
\end{align*}
Since $S_i\independent (\mathcal{U}_i(\boldsymbol \theta_{\boldsymbol Z}^*)-f^*(\boldsymbol X_i, \boldsymbol \theta_{\boldsymbol Z}^*))|\boldsymbol X_i $, the above term can be written as
\begin{align*}
    &  \frac{1}{N} \sum_{i=1}^N \mathbb{E}_{\boldsymbol X_i}\left[\mathbb{E}\left\{\left.\frac{S_i}{\pi(\boldsymbol X_i,\boldsymbol \alpha^{*})}\right|\boldsymbol X_i\right\}\cdot \mathbb{E}[(\mathcal{U}_i(\boldsymbol \theta_{\boldsymbol Z}^*)-f^*(\boldsymbol X_i, \boldsymbol \theta_{\boldsymbol Z}^*))|\boldsymbol X_i]\right]\\
    & = \frac{1}{N} \sum_{i=1}^N \mathbb{E}_{\boldsymbol X_i}\left[\frac{\pi(\boldsymbol X_i)}{\pi(\boldsymbol X_i,\boldsymbol \alpha^{*})}\cdot \{ \mathbb{E}[\mathcal{U}_i(\boldsymbol \theta_{\boldsymbol Z}^*)|\boldsymbol X_i]-f^*(\boldsymbol X_i, \boldsymbol \theta_{\boldsymbol Z}^*)\}\right]\\
    & = \frac{1}{N} \sum_{i=1}^N \mathbb{E}_{\boldsymbol X_i}\left[\frac{\pi(\boldsymbol X_i)}{\pi(\boldsymbol X_i,\boldsymbol \alpha^{*})}\cdot \{f^*(\boldsymbol X_i, \boldsymbol \theta_{\boldsymbol Z}^*)-f^*(\boldsymbol X_i, \boldsymbol \theta_{\boldsymbol Z}^*)\}\right]=0\cdot
\end{align*}
This completes the proof for the double robustness property of the proposed estimator.
\end{proof}

\subsection{Proof of Theorem 3.2 : Double Robustness Multi-Cohort}\label{sec:pr2_12}
In this case, under correct specification of either $\pi(\boldsymbol{X}^{\text{mult}}_i, \boldsymbol{\alpha}^{\text{mult}})$ or $\boldsymbol f(\boldsymbol{X}^{\text{mult}}_i, \boldsymbol{\theta}_{\boldsymbol{Z}})$ or both,  $\mathbb{E}(\phi_N(\boldsymbol \theta_{\boldsymbol Z}^*,\boldsymbol \alpha^{\text{mult}*},\boldsymbol f^*))=\mathbf{0}$, where 
$$\phi_N(\boldsymbol \theta_{\boldsymbol Z},\boldsymbol \alpha^{\text{mult}},\boldsymbol f)=\frac{1}{N} \sum_{i=1}^N \frac{S^{\text{mult}}_i}{\pi(\boldsymbol{X}^{\text{mult}}_i,\boldsymbol \alpha^{\text{mult}})}  (\mathcal{U}_i(\boldsymbol{\theta}_{\boldsymbol{Z}}) - \boldsymbol f(\boldsymbol{X}^{\text{mult}}_i, \boldsymbol{\theta}_{\boldsymbol{Z}})) + \frac{1}{N} \sum_{i=1}^N \frac{S_{\text{ext},i}}{\pi_{\text{ext},i}} \boldsymbol f(\boldsymbol{X}^{\text{mult}}_i, \boldsymbol{\theta}_{\boldsymbol{Z}})$$
\noindent
The proof of this theorem is exactly similar to that of Theorem 10, where we just replace the single cohort notations with $S^{\text{mult}},\boldsymbol \alpha^{\text{mult}},\boldsymbol{X}^{\text{mult}}$.

\subsection{Parametric JAIPW: Proof of Theorem 3.3}\label{sec:pr11}

\begin{proof}
Under the assumptions of Theorem 3.2 and assumptions \textit{A}7 in Section S1 and assuming either the propensity score model or the auxiliary score model or both is correctly specified we obtain from Theorem 12, $\widehat{\boldsymbol \theta}_{\boldsymbol Z}$ estimated using equation (3.13) of the main text is consistent for $\boldsymbol \theta^*_{\boldsymbol Z}$ as $N\rightarrow \infty$.
By \citet{tsiatis2006semiparametric}'s arguments on a two step Z-estimation problem, we obtain that 
$$\sqrt{N}(\widehat{\boldsymbol \theta_{\boldsymbol Z}}-\boldsymbol\theta_{\boldsymbol Z}^*)\xrightarrow{d}\mathcal{N}(\mathbf{0},V)\cdot $$
where
\begin{align*}
    & V=(G_{\boldsymbol \theta_{\boldsymbol Z}^*})^{-1}.\mathbb{E}[\{\boldsymbol g(\boldsymbol \theta_{\boldsymbol Z}^*,\boldsymbol \alpha^{\text{mult}*},\boldsymbol \gamma^{\text{mult}*})+G_{\boldsymbol \alpha^{\text{mult}*},\boldsymbol \gamma^{\text{mult}*}}\cdot \boldsymbol \Psi(\boldsymbol \alpha^{\text{mult}*},\boldsymbol \gamma^{\text{mult}*})\}\\
    & \hspace{4cm}\{\boldsymbol g(\boldsymbol \theta_{\boldsymbol Z}^*,\boldsymbol \alpha^{\text{mult}*},\boldsymbol \gamma^{\text{mult}*})+G_{\boldsymbol \alpha^{\text{mult}*},\boldsymbol \gamma^{\text{mult}*}}.\boldsymbol \Psi(\boldsymbol \alpha^{\text{mult}*},\boldsymbol \gamma^{\text{mult}*})\}'].(G_{\boldsymbol \theta_{\boldsymbol Z}^*}^{-1})^{'}
\end{align*}
Next we derive the expression of each of the terms in the above expression.
\begin{align*}
    & G_{\boldsymbol \theta_{\boldsymbol Z}^*}=\left.\mathbb{E}\left\{\frac{\partial g(\boldsymbol \theta_{\boldsymbol Z},\boldsymbol \alpha^{\text{mult}*},\boldsymbol \gamma^{\text{mult}*})}{\partial \boldsymbol \theta_{\boldsymbol Z}}\right\}\right|_{\boldsymbol \theta_{\boldsymbol Z}=\boldsymbol \theta_{\boldsymbol Z}^*}\\
    & G_{\boldsymbol \alpha^{\text{mult}*},\boldsymbol \gamma^{\text{mult}*}}=\left[\left.\mathbb{E}\left\{\frac{\partial g(\boldsymbol \theta_{\boldsymbol Z}^*,\boldsymbol \alpha^{\text{mult}},\boldsymbol \gamma^{\text{mult}*})}{\partial \boldsymbol \alpha^{\text{mult}}}\right|_{\boldsymbol\alpha^{\text{mult}}=\boldsymbol \alpha^{\text{mult}*}}\right\},\left.\mathbb{E}\left\{\frac{\partial g(\boldsymbol \theta_{\boldsymbol Z}^*,\boldsymbol \alpha^{\text{mult}*},\boldsymbol \gamma^{\text{mult}})}{\partial \boldsymbol \alpha^{\text{mult}}}\right|_{\boldsymbol\gamma=\boldsymbol \gamma^{\text{mult}*}}\right\}\right]\\
    & H = \left[\left.\mathbb{E}\left\{\frac{\partial \boldsymbol h(\boldsymbol \alpha^{\text{mult}},\boldsymbol \gamma^{\text{mult}*})}{\partial \boldsymbol \alpha^{\text{mult}}}\right\}\right|_{\boldsymbol\alpha^{\text{mult}}=\boldsymbol \alpha^{\text{mult}*}},\left.\mathbb{E}\left\{\frac{\partial \boldsymbol h(\boldsymbol \alpha^{\text{mult}*},\boldsymbol \gamma^{\text{mult}})}{\partial \boldsymbol \gamma^{\text{mult}}}\right\}\right|_{\boldsymbol\gamma^{\text{mult}}=\boldsymbol\gamma^{\text{mult}*}}\right]\\
    &\boldsymbol \Psi(\boldsymbol \alpha^{\text{mult}*})=-H^{-1}\boldsymbol h(\boldsymbol \alpha^{\text{mult}*},\boldsymbol \gamma^{\text{mult}*}) \cdot
\end{align*}
First we calculate $G_{\boldsymbol \theta_{\boldsymbol Z}^*}$.
\subsubsection*{Calculation for $G_{\boldsymbol \theta_{\boldsymbol Z}^*}$}
\begin{align*}
    & \left.\frac{\partial g(\boldsymbol \theta_{\boldsymbol Z},\boldsymbol \alpha^{\text{mult}*},\boldsymbol \gamma^{\text{mult}*})}{\partial \boldsymbol \theta_{\boldsymbol Z}}\right|_{\boldsymbol \theta_{\boldsymbol Z}=\boldsymbol \theta_{\boldsymbol Z}^*} \\
    & = \frac{\partial}{\partial \boldsymbol \theta_{\boldsymbol Z}}\left[\frac{S^{\text{mult}}}{\pi(\boldsymbol X^{\text{mult}},\boldsymbol \alpha^{\text{mult}*})}\left\{D \boldsymbol Z-\frac{e^{\boldsymbol \theta_{\boldsymbol Z}^{'}\boldsymbol Z}}{(1+e^{\boldsymbol \theta_{\boldsymbol Z}^{'}\boldsymbol Z})}\cdot \boldsymbol Z-\boldsymbol f({\boldsymbol X^{\text{mult}}},\boldsymbol \gamma^{\text{mult}*},\boldsymbol \theta_{\boldsymbol Z})\right\}\left.+\frac{S_{\text{ext}}}{\pi_{\text{ext}}}\cdot f({\boldsymbol X^{\text{mult}}},\boldsymbol \gamma^{\text{mult}*},\boldsymbol \theta_{\boldsymbol Z})\right]\right|_{\boldsymbol \theta_{\boldsymbol Z}=\boldsymbol \theta_{\boldsymbol Z}^*}\\
    & = \frac{S^{\text{mult}}}{\pi(\boldsymbol X^{\text{mult}},\boldsymbol \alpha^{\text{mult}*})}\cdot \left(-\frac{e^{\boldsymbol \theta_{\boldsymbol Z}^{*'}\boldsymbol Z}}{(1+e^{\boldsymbol \theta_{\boldsymbol Z}^{*'}\boldsymbol Z})^2}\cdot \boldsymbol Z\boldsymbol Z'-\left.\frac{\partial \boldsymbol f({\boldsymbol X^{\text{mult}}},\boldsymbol \gamma^{\text{mult}*},\boldsymbol \theta_{\boldsymbol Z})}{\partial \boldsymbol \theta_{\boldsymbol Z}}\right|_{\boldsymbol \theta_{\boldsymbol Z}=\boldsymbol \theta_{\boldsymbol Z}^*}\right)
    +\frac{S_{\text{ext}}}{\pi_{\text{ext}}}\cdot \left.\frac{\partial \boldsymbol f({\boldsymbol X^{\text{mult}}},\boldsymbol \gamma^{\text{mult}*},\boldsymbol \theta_{\boldsymbol Z})}{\partial \boldsymbol \theta_{\boldsymbol Z}}\right|_{\boldsymbol \theta_{\boldsymbol Z}=\boldsymbol \theta_{\boldsymbol Z}^*}
\end{align*}
Therefore we obtain 
\begin{align*}
    &G_{\boldsymbol \theta_{\boldsymbol Z}^*} = \mathbb{E}\left\{\frac{S^{\text{mult}}}{\pi(\boldsymbol X^{\text{mult}},\boldsymbol \alpha^{\text{mult}*})}\cdot \left(-\frac{e^{\boldsymbol \theta_{\boldsymbol Z}^{*'}\boldsymbol Z}}{(1+e^{\boldsymbol \theta_{\boldsymbol Z}^{*'}\boldsymbol Z})^2}\cdot \boldsymbol Z\boldsymbol Z'-\left.\frac{\partial \boldsymbol f({\boldsymbol X^{\text{mult}}},\boldsymbol \gamma^{\text{mult}*},\boldsymbol \theta_{\boldsymbol Z})}{\partial \boldsymbol \theta_{\boldsymbol Z}}\right|_{\boldsymbol \theta_{\boldsymbol Z}=\boldsymbol \theta_{\boldsymbol Z}^*}\right)+\frac{S_{\text{ext}}}{\pi_{\text{ext}}} \left.\frac{\partial \boldsymbol f({\boldsymbol X^{\text{mult}}},\boldsymbol \gamma^{\text{mult}*},\boldsymbol \theta_{\boldsymbol Z})}{\partial \boldsymbol \theta_{\boldsymbol Z}}\right|_{\boldsymbol \theta_{\boldsymbol Z}=\boldsymbol \theta_{\boldsymbol Z}^*}\right\}
\end{align*}
Next we calculate $G_{\boldsymbol \alpha^{\text{mult}*},\boldsymbol \gamma^{\text{mult}*}}$.
\subsubsection*{Calculation for $G_{\boldsymbol \alpha^{\text{mult}*},\boldsymbol \gamma^{\text{mult}*}}$}
\begin{align*}
    & \left.\frac{\partial g(\boldsymbol \theta_{\boldsymbol Z}^*,\boldsymbol \alpha^{\text{mult}},\boldsymbol \gamma^{\text{mult}})}{\partial( \boldsymbol \alpha^{\text{mult}},\boldsymbol \gamma^{\text{mult}})}\right|_{( \boldsymbol \alpha^{\text{mult}},\boldsymbol \gamma^{\text{mult}})(\boldsymbol \alpha^{\text{mult}*},\boldsymbol \gamma^{\text{mult}*})}=\left[\left.\frac{\partial g(\boldsymbol \theta_{\boldsymbol Z}^*,\boldsymbol \alpha^{\text{mult}},\boldsymbol \gamma^{\text{mult}*})}{\partial\boldsymbol \alpha^{\text{mult}}}\right|_{ \boldsymbol \alpha=\boldsymbol \alpha^{\text{mult}*}},\left.\frac{\partial g(\boldsymbol \theta_{\boldsymbol Z}^*,\boldsymbol \alpha^{\text{mult}*},\boldsymbol \gamma^{\text{mult}})}{\partial\boldsymbol \gamma^{\text{mult}}}\right|_{\boldsymbol \gamma^{\text{mult}}=\boldsymbol \gamma^{\text{mult}*}}\right]\\
    & \left.\frac{\partial g(\boldsymbol \theta_{\boldsymbol Z}^*,\boldsymbol \alpha^{\text{mult}},\boldsymbol \gamma^{\text{mult}*})}{\partial\boldsymbol \alpha^{\text{mult}}}\right|_{ \boldsymbol \alpha=\boldsymbol \alpha^{\text{mult}*}}=\left.\frac{\partial}{\partial\boldsymbol \alpha^{\text{mult}}}\left[\frac{S^{\text{mult}}}{\pi(\boldsymbol X^{\text{mult}},\boldsymbol \alpha^{\text{mult}})}\left\{D \boldsymbol Z-\frac{e^{\boldsymbol \theta_{\boldsymbol Z}^{*'}\boldsymbol Z}}{(1+e^{\boldsymbol \theta_{\boldsymbol Z}^{*'}\boldsymbol Z})}\cdot \boldsymbol Z-\boldsymbol f({\boldsymbol X^{\text{mult}}},\boldsymbol \gamma^{\text{mult}*},\boldsymbol \theta_{\boldsymbol Z}^*)\right\}\right.\right.\\
    &\hspace{5cm}\left.\left.+\frac{S_{\text{ext}}}{\pi_{\text{ext}}}\cdot f({\boldsymbol X^{\text{mult}}},\boldsymbol \gamma^{\text{mult}*},\boldsymbol \theta_{\boldsymbol Z}^*)\right]\right|_{\boldsymbol \alpha=\boldsymbol \alpha^{\text{mult}*}}
    \end{align*}
\begin{align*}
    & =-\frac{S^{\text{mult}}}{\pi(\boldsymbol X^{\text{mult}},\boldsymbol \alpha^{\text{mult}*})^2}\left\{D \boldsymbol Z-\frac{e^{\boldsymbol \theta_{\boldsymbol Z}^{*'}\boldsymbol Z}}{(1+e^{\boldsymbol \theta_{\boldsymbol Z}^{*'}\boldsymbol Z})}\cdot \boldsymbol Z-\boldsymbol f({\boldsymbol X^{\text{mult}}},\boldsymbol \gamma^{\text{mult}*},\boldsymbol \theta_{\boldsymbol Z}^*)\right\}\left.\frac{\partial \boldsymbol \pi(\boldsymbol X^{\text{mult}},\boldsymbol \alpha^{\text{mult}*})}{\partial \boldsymbol \alpha^{\text{mult}}}\right|_{\boldsymbol \alpha=\boldsymbol \alpha^{\text{mult}*}}
\end{align*}
\begin{align*}
    & \left.\frac{\partial g(\boldsymbol \theta_{\boldsymbol Z}^*,\boldsymbol \alpha^{\text{mult}*},\boldsymbol \gamma^{\text{mult}})}{\partial\boldsymbol \gamma^{\text{mult}}}\right|_{\boldsymbol \gamma^{\text{mult}}=\boldsymbol \gamma^{\text{mult}*}}\\
    & =\left.\frac{\partial}{\partial\boldsymbol \gamma^{\text{mult}}}\left[\frac{S^{\text{mult}}}{\pi(\boldsymbol X^{\text{mult}},\boldsymbol \alpha^{\text{mult}*})}\left\{D \boldsymbol Z-\frac{e^{\boldsymbol \theta_{\boldsymbol Z}^{*'}\boldsymbol Z}}{(1+e^{\boldsymbol \theta_{\boldsymbol Z}^{*'}\boldsymbol Z})}\cdot \boldsymbol Z-\boldsymbol f({\boldsymbol X^{\text{mult}}},\boldsymbol \gamma^{\text{mult}},\boldsymbol \theta_{\boldsymbol Z}^*)\right\}+\frac{S_{\text{ext}}}{\pi_{\text{ext}}}\cdot f({\boldsymbol X^{\text{mult}}},\boldsymbol \gamma^{\text{mult}},\boldsymbol \theta_{\boldsymbol Z}^*)\right]\right|_{\boldsymbol \gamma^{\text{mult}}=\boldsymbol \gamma^{\text{mult}*}}\\
    &=\left(\frac{S_{\text{ext}}}{\pi_{\text{ext}}}- \frac{S^{\text{mult}}}{\pi(\boldsymbol X^{\text{mult}},\boldsymbol \alpha^{\text{mult}*})}\right)\cdot \left.\frac{\partial \boldsymbol f({\boldsymbol X^{\text{mult}}},\boldsymbol \gamma^{\text{mult}},\boldsymbol \theta_{\boldsymbol Z}^*)}{\partial \boldsymbol \gamma^{\text{mult}}}\right|_{\boldsymbol \gamma^{\text{mult}}=\boldsymbol \gamma^{\text{mult}*}}
\end{align*}
Therefore we obtain 
\begin{align*}
    & G_{\boldsymbol \alpha^{\text{mult}*},\boldsymbol \gamma^{\text{mult}*}}=\mathbb{E}\left[-\frac{S^{\text{mult}}}{\pi(\boldsymbol X^{\text{mult}},\boldsymbol \alpha^{\text{mult}*})^2}\left\{D \boldsymbol Z-\frac{e^{\boldsymbol \theta_{\boldsymbol Z}^{*'}\boldsymbol Z}}{(1+e^{\boldsymbol \theta_{\boldsymbol Z}^{*'}\boldsymbol Z})}\cdot \boldsymbol Z-\boldsymbol f({\boldsymbol X^{\text{mult}}},\boldsymbol \gamma^{\text{mult}*},\boldsymbol \theta_{\boldsymbol Z}^*)\right\}\frac{\partial \boldsymbol \pi(\boldsymbol X^{\text{mult}},\boldsymbol \alpha^{\text{mult}*})}{\partial \boldsymbol \alpha^{\text{mult}}}\right|_{\boldsymbol \alpha^{\text{mult}}=\boldsymbol \alpha^{\text{mult}*}},\\
    & \left.\left(\frac{S_{\text{ext}}}{\pi_{\text{ext}}}- \frac{S^{\text{mult}}}{\pi(\boldsymbol X^{\text{mult}},\boldsymbol \alpha^{\text{mult}*})}\right)\cdot \left.\frac{\partial \boldsymbol f({\boldsymbol X^{\text{mult}}},\boldsymbol \gamma^{\text{mult}},\boldsymbol \theta_{\boldsymbol Z}^*)}{\partial \boldsymbol \gamma^{\text{mult}}}\right|_{\boldsymbol \gamma^{\text{mult}}=\boldsymbol \gamma^{\text{mult}*}}\right]
\end{align*}
Next we calculate $\Psi(\boldsymbol \alpha^{\text{mult}*},\boldsymbol \gamma^{\text{mult}*})$.
\setlength\arraycolsep{-8pt}
\subsubsection*{Calculation for $\boldsymbol\Psi(\boldsymbol \alpha^{\text{mult}*},\boldsymbol \gamma^{\text{mult}*})$}
\begin{align*}
   &  \left.\frac{\partial h(\boldsymbol \alpha^{\text{mult}},\boldsymbol \gamma^{\text{mult}})}{\partial( \boldsymbol \alpha^{\text{mult}},\boldsymbol \gamma^{\text{mult}})}\right|_{( \boldsymbol \alpha^{\text{mult}},\boldsymbol \gamma^{\text{mult}})=(\boldsymbol \alpha^{\text{mult}*},\boldsymbol \gamma^{\text{mult}*})}=\begin{pmatrix}
    \left.\frac{\partial \boldsymbol q(\boldsymbol \alpha^{\text{mult}},\boldsymbol X^{\text{mult}},S,S_{\text{ext}},\pi_{\text{ext}})}{\partial \boldsymbol \alpha^{\text{mult}}}\right|_{\boldsymbol \alpha=\boldsymbol \alpha^{\text{mult}*}} & 0 \\
     0 & \left.\frac{\partial \boldsymbol p(\boldsymbol\gamma^{\text{mult}},\boldsymbol X^{\text{mult}},S,\boldsymbol Z_{1\cap})}{\partial \boldsymbol \alpha^{\text{mult}}}\right|_{\boldsymbol \gamma^{\text{mult}}=\boldsymbol \gamma^{\text{mult}*}}
         \end{pmatrix}
\end{align*}
This implies
\begin{align*}
    H = \mathbb{E} \begin{pmatrix}
    \left.\frac{\partial \boldsymbol q(\boldsymbol \alpha^{\text{mult}},\boldsymbol X^{\text{mult}},S,S_{\text{ext}},\pi_{\text{ext}})}{\partial \boldsymbol \alpha^{\text{mult}}}\right|_{\boldsymbol \alpha=\boldsymbol \alpha^{\text{mult}*}} & 0 \\
     0 & \left.\frac{\partial \boldsymbol p(\boldsymbol\gamma,\boldsymbol X^{\text{mult}},S,\boldsymbol Z_{1\cap})}{\partial \boldsymbol \alpha^{\text{mult}}}\right|_{\boldsymbol \gamma^{\text{mult}}=\boldsymbol \gamma^{\text{mult}*}}
         \end{pmatrix}
\end{align*}
\end{proof}
\subsection{Theorem S10}\label{sec:pr12}
\begin{theorem}
    Under all the assumptions of Theorems 3.2, 3.3 and assumption \textit{A}8 in Supplementary Section \ref{sec:assumptions}, $\frac{1}{N} \cdot \widehat{G_{\boldsymbol \theta_{\boldsymbol Z}}}^{-1}\cdot \widehat{E}\cdot (\widehat{G_{\boldsymbol \theta_{\boldsymbol Z}}}^{-1})^{'}$ is a consistent estimator of the variance of $\widehat{\boldsymbol \theta}_{\boldsymbol Z}$ for JAIPW where
\begin{align*}
   & \widehat{G_{\boldsymbol \theta_{\boldsymbol Z}}}=-\frac{1}{N}\sum_{i=1}^N S^{\text{mult}}_i \cdot \frac{1}{\pi(\boldsymbol X^{\text{mult}}_i,\widehat{\boldsymbol \alpha^{\text{mult}}})}\cdot \left\{\frac{e^{\widehat{\boldsymbol \theta_{\boldsymbol Z}}'\boldsymbol Z_i}}{(1+e^{\widehat{\boldsymbol \theta_{\boldsymbol Z}}'\boldsymbol Z_i})^2}\cdot\boldsymbol Z_i \boldsymbol Z_i' + \left.\frac{\partial \widehat{\boldsymbol f}({\boldsymbol X^{\text{mult}}_i},\boldsymbol \theta_{\boldsymbol Z})}{\partial \boldsymbol \theta_{\boldsymbol Z}}\right|_{\boldsymbol \theta_{\boldsymbol Z}=\widehat{\boldsymbol \theta_{\boldsymbol Z}}}\right\}\\
   & \hspace{3cm} + \frac{1}{N}\sum_{i=1}^N \frac{S_{\text{ext},i}}{\pi_{\text{ext},i}}\cdot  \left.\frac{\partial \widehat{\boldsymbol f}({\boldsymbol X^{\text{mult}}_i},\boldsymbol \theta_{\boldsymbol Z})}{\partial \boldsymbol \theta_{\boldsymbol Z}}\right|_{\boldsymbol \theta_{\boldsymbol Z}=\widehat{\boldsymbol \theta_{\boldsymbol Z}}} \cdot \\
   & \widehat{H} = \frac{1}{N}\sum_{i=1}^N \ \begin{pmatrix}
    \left.\frac{\partial \boldsymbol q(\boldsymbol \alpha^{\text{mult}},\boldsymbol X^{\text{mult}}_i,S^{\text{mult}}_i,S_{\text{ext},i},\pi_{\text{ext}},i)}{\partial \boldsymbol \alpha^{\text{mult}}}\right|_{\boldsymbol \alpha=\widehat{\boldsymbol \alpha^{\text{mult}}}} & 0 \\
     0 & \left.\frac{\partial \boldsymbol p(\boldsymbol\gamma,\boldsymbol X^{\text{mult}}_i,S,\boldsymbol Z_{1\cap,i})}{\partial \boldsymbol \gamma^{\text{mult}}}\right|_{\boldsymbol \gamma^{\text{mult}}=\widehat{\boldsymbol \gamma^{\text{mult}}}}
         \end{pmatrix} \end{align*}
\begin{align*}
   & \widehat{G_{\boldsymbol \alpha^{\text{mult}},\boldsymbol \gamma^{\text{mult}}}}=\left(-\frac{1}{N}\sum_{i=1}^N \frac{S^{\text{mult}}_i}{\pi(\boldsymbol X^{\text{mult}}_i,\widehat{\boldsymbol \alpha^{\text{mult}}})^2}\left\{D_i \boldsymbol Z_i-\frac{e^{\widehat{\boldsymbol \theta}_{\boldsymbol Z}^{'}\boldsymbol Z}}{(1+e^{\widehat{\boldsymbol \theta}_{\boldsymbol Z}^{'}\boldsymbol Z_i})}\boldsymbol Z_i-\boldsymbol f({\boldsymbol X^{\text{mult}}_i},\widehat{\boldsymbol \gamma^{\text{mult}}},\widehat{\boldsymbol \theta_{\boldsymbol Z}})\right\}\cdot \left.\frac{\partial \boldsymbol \pi(\boldsymbol X^{\text{mult}}_i,\boldsymbol \alpha^{\text{mult}*})}{\partial \boldsymbol \alpha^{\text{mult}}}\right|_{\boldsymbol \alpha=\widehat{\boldsymbol \alpha^{\text{mult}}}},\right.\\
   & \left.\hspace{3cm}\frac{1}{N}\sum_{i=1}^N \left(\frac{S_{\text{ext},i}}{\pi_{\text{ext},i}}- \frac{S^{\text{mult}}_i}{\pi(\boldsymbol X^{\text{mult}}_i,\widehat{\boldsymbol \alpha^{\text{mult}}})}\right)\cdot \left.\frac{\partial \boldsymbol f({\boldsymbol X^{\text{mult}}_i},\widehat{\boldsymbol \gamma^{\text{mult}}},\widehat{\boldsymbol \theta_{\boldsymbol Z}})}{\partial \boldsymbol \gamma^{\text{mult}}}\right|_{\boldsymbol \gamma^{\text{mult}}=\widehat{\boldsymbol \gamma^{\text{mult}}}}\right)\\
   & \widehat{E_1}=\frac{1}{N}\sum_{i =1}^N \frac{S^{\text{mult}}_i}{\pi(\boldsymbol X^{\text{mult}}_i,\widehat{\boldsymbol \alpha^{\text{mult}}})^2}\left[D_i\boldsymbol Z_i-\frac{e^{\widehat{\boldsymbol \theta_{\boldsymbol Z}}'\boldsymbol Z_i}}{(1+e^{\widehat{\boldsymbol \theta_{\boldsymbol Z}}' \boldsymbol Z_i})}\cdot \boldsymbol Z_i-f({\boldsymbol X^{\text{mult}}_i},\widehat{\boldsymbol \gamma^{\text{mult}}},\widehat{\boldsymbol \theta_{\boldsymbol Z}})\right]\cdot\\
    & \hspace{3cm}\left[D_i\boldsymbol Z_i-\frac{e^{\widehat{\boldsymbol \theta_{\boldsymbol Z}}'\boldsymbol Z_i}}{(1+e^{\widehat{\boldsymbol \theta_{\boldsymbol Z}}' \boldsymbol Z_i})}\cdot \boldsymbol Z_i-f({\boldsymbol X^{\text{mult}}_i},\widehat{\boldsymbol \gamma^{\text{mult}}},\widehat{\boldsymbol \theta_{\boldsymbol Z}})\right]'\\
    & \hspace{3cm}+ \frac{1}{N}\sum_{i =1}^N \frac{S_{\text{ext},i}}{\pi_{\text{ext},i}^2}\cdot f({\boldsymbol X^{\text{mult}}_i},\widehat{\boldsymbol \gamma^{\text{mult}}},\widehat{\boldsymbol \theta_{\boldsymbol Z}})\cdot f({\boldsymbol X^{\text{mult}}_i},\widehat{\boldsymbol \gamma^{\text{mult}}},\widehat{\boldsymbol \theta_{\boldsymbol Z}})'\\
    & + \frac{1}{N}\sum_{i =1}^N \frac{S^{\text{mult}}_iS_{\text{ext},i}}{\pi(\boldsymbol X^{\text{mult}}_i,\widehat{\boldsymbol \alpha^{\text{mult}}})\pi_{\text{ext},i}}\left[D_i\boldsymbol Z_i-\frac{e^{\widehat{\boldsymbol \theta_{\boldsymbol Z}}'\boldsymbol Z_i}}{(1+e^{\widehat{\boldsymbol \theta_{\boldsymbol Z}}' \boldsymbol Z_i})}\cdot \boldsymbol Z_i-f({\boldsymbol X^{\text{mult}}_i},\widehat{\boldsymbol \gamma^{\text{mult}}},\widehat{\boldsymbol \theta_{\boldsymbol Z}})\right]\\
    &\hspace{1cm}\cdot f({\boldsymbol X^{\text{mult}}_i},\widehat{\boldsymbol \gamma^{\text{mult}}},\widehat{\boldsymbol \theta_{\boldsymbol Z}})'+ \frac{1}{N}\sum_{i =1}^N \frac{S^{\text{mult}}_iS_{\text{ext},i}}{\pi(\boldsymbol X^{\text{mult}}_i,\widehat{\boldsymbol \alpha^{\text{mult}}})\pi_{\text{ext},i}}\cdot f({\boldsymbol X^{\text{mult}}_i},\widehat{\boldsymbol \gamma^{\text{mult}}},\widehat{\boldsymbol \theta_{\boldsymbol Z}})\\
    & \hspace{2cm}\cdot \left[D_i\boldsymbol Z_i-\frac{e^{\widehat{\boldsymbol \theta_{\boldsymbol Z}}'\boldsymbol Z_i}}{(1+e^{\widehat{\boldsymbol \theta_{\boldsymbol Z}}' \boldsymbol Z_i})}\cdot \boldsymbol Z_i-f({\boldsymbol X^{\text{mult}}_i},\widehat{\boldsymbol \gamma^{\text{mult}}},\widehat{\boldsymbol \theta_{\boldsymbol Z}})\right]'\\
    \end{align*}
\begin{align*}
    & \widehat{E}_2 =  \widehat{E}_3'=\frac{1}{N}\sum_{i=1}^N  \widehat{G_{\boldsymbol \alpha^{\text{mult}},\boldsymbol \gamma^{\text{mult}}}}\cdot (-\widehat{H})^{-1} \cdot\frac{1}{\pi(\boldsymbol X^{\text{mult}}_i,\widehat{\boldsymbol \alpha^{\text{mult}})}}\cdot\\
   & \begin{pmatrix}
   q(\widehat{\boldsymbol \alpha^{\text{mult}}},\boldsymbol X^{\text{mult}}_i,S^{\text{mult}}_i,S_{\text{ext},i},\pi_{\text{ext},i})\\
    p(\widehat{\boldsymbol\gamma},\boldsymbol X^{\text{mult}}_i,S^{\text{mult}}_i,\boldsymbol Z_{1\cap,i})
         \end{pmatrix}\left[D_i\boldsymbol Z_i-\frac{e^{\widehat{\boldsymbol \theta_{\boldsymbol Z}}'\boldsymbol Z_i}}{(1+e^{\widehat{\boldsymbol \theta_{\boldsymbol Z}}' \boldsymbol Z_i})}\boldsymbol Z_i-f({\boldsymbol X^{\text{mult}}_i},\widehat{\boldsymbol \gamma^{\text{mult}}},\widehat{\boldsymbol \theta_{\boldsymbol Z}})\right]'\\
         &\hspace{2cm}  +\frac{1}{N}\sum_{i =1}^N \begin{pmatrix}
   q(\widehat{\boldsymbol \alpha^{\text{mult}}},\boldsymbol X^{\text{mult}}_i,S^{\text{mult}}_i,S_{\text{ext},i},\pi_{\text{ext},i})\\
    p(\widehat{\boldsymbol\gamma},\boldsymbol X^{\text{mult}}_i,S^{\text{mult}}_i,\boldsymbol Z_{1\cap,i})
         \end{pmatrix}\frac{S_{\text{ext},i}}{\pi_{\text{ext},i}}\cdot  f({\boldsymbol X^{\text{mult}}_i},\widehat{\boldsymbol \gamma^{\text{mult}}},\widehat{\boldsymbol \theta_{\boldsymbol Z}})'
\end{align*}
\begin{align*}
    & \widehat{E_4} = \frac{1}{N}\cdot\widehat{G_{\boldsymbol \alpha^{\text{mult}}}}\cdot \widehat{H}^{-1} \sum_{i=1}^N  \left[\begin{pmatrix}
   q(\widehat{\boldsymbol \alpha^{\text{mult}}},\boldsymbol X^{\text{mult}}_i,S^{\text{mult}}_i,S_{\text{ext},i},\pi_{\text{ext},i})\\
    p(\widehat{\boldsymbol\gamma},\boldsymbol X^{\text{mult}}_i,S^{\text{mult}}_i,\boldsymbol Z_{1\cap,i})
         \end{pmatrix}\right.\\
         & \left.\cdot (q(\widehat{\boldsymbol \alpha^{\text{mult}}},\boldsymbol X^{\text{mult}}_i,S^{\text{mult}}_i,S_{\text{ext},i},\pi_{\text{ext},i})',p(\widehat{\boldsymbol\gamma},\boldsymbol X^{\text{mult}}_i,S^{\text{mult}}_i,\boldsymbol Z_{1\cap,i}))'\right](\widehat{H}^{-1})'\cdot(\widehat{G_{\boldsymbol \alpha^{\text{mult}}}})'\\
    &  \widehat{E}=\widehat{E}_1-\widehat{E}_2-\widehat{E}_3 + \widehat{E}_4
\end{align*}
\end{theorem}
\begin{proof}
    Under all the assumptions of Theorems 11, 12 and assumption \textit{A}8 in Supplementary Section \ref{sec:assumptions} and using ULLN and Continuous Mapping Theorem, the proof of consistency for each of the following sample quantities are exactly same as the approach in Theorem S3. Using the exact same steps on the joint parameters $\boldsymbol \eta$ instead of $\boldsymbol \theta_{\boldsymbol Z}$ (as in Theorem S3) we obtain
\begin{align*}
    & \widehat{G_{\boldsymbol \theta_{\boldsymbol Z}}}=-\frac{1}{N}\sum_{i=1}^N S^{\text{mult}}_i \frac{1}{\pi(\boldsymbol X^{\text{mult}}_i,\widehat{\boldsymbol \alpha^{\text{mult}}})}\cdot \left\{\frac{e^{\widehat{\boldsymbol \theta_{\boldsymbol Z}}'\boldsymbol Z_i}}{(1+e^{\widehat{\boldsymbol \theta_{\boldsymbol Z}}'\boldsymbol Z_i})^2}\cdot\boldsymbol Z_i \boldsymbol Z_i' \right.\\
    &+\left. \left.\frac{\partial \boldsymbol f({\boldsymbol X^{\text{mult}}_i},\widehat{\boldsymbol \gamma^{\text{mult}}},\boldsymbol \theta_{\boldsymbol Z})}{\partial \boldsymbol \theta_{\boldsymbol Z}}\right|_{\boldsymbol \theta_{\boldsymbol Z}=\widehat{\boldsymbol \theta_{\boldsymbol Z}}}\right\}+ \frac{1}{N}\sum_{i=1}^N \frac{S_{\text{ext},i}}{\pi_{\text{ext},i}}\cdot  \left.\frac{\partial \boldsymbol f({\boldsymbol X^{\text{mult}}_i},\widehat{\boldsymbol \gamma^{\text{mult}}},\boldsymbol \theta_{\boldsymbol Z})}{\partial \boldsymbol \theta_{\boldsymbol Z}}\right|_{\boldsymbol \theta_{\boldsymbol Z}=\widehat{\boldsymbol \theta_{\boldsymbol Z}}} \\
     & \xrightarrow{p}G_{\boldsymbol \theta_{\boldsymbol Z}^*} = \mathbb{E}\left\{\frac{S^{\text{mult}}}{\pi(\boldsymbol X^{\text{mult}},\boldsymbol \alpha^{\text{mult}*})}\cdot \left(-\frac{e^{\boldsymbol \theta_{\boldsymbol Z}^{*'}\boldsymbol Z}}{(1+e^{\boldsymbol \theta_{\boldsymbol Z}^{*'}\boldsymbol Z})^2} \boldsymbol Z\boldsymbol Z'-\left.\frac{\partial \boldsymbol f({\boldsymbol X^{\text{mult}}},\boldsymbol \gamma^{\text{mult}*},\boldsymbol \theta_{\boldsymbol Z})}{\partial \boldsymbol \theta_{\boldsymbol Z}}\right|_{\boldsymbol \theta_{\boldsymbol Z}=\boldsymbol \theta_{\boldsymbol Z}^*}\right)\right.\\
     & \hspace{3cm} \left. +\frac{S_{\text{ext}}}{\pi_{\text{ext}}}\left.\frac{\partial \boldsymbol f({\boldsymbol X^{\text{mult}}},\boldsymbol \gamma^{\text{mult}*},\boldsymbol \theta_{\boldsymbol Z})}{\partial \boldsymbol \theta_{\boldsymbol Z}}\right|_{\boldsymbol \theta_{\boldsymbol Z}=\boldsymbol \theta_{\boldsymbol Z}^*}\right\} \cdot
\end{align*}
Similarly we obtain 
\begin{align*}
      & \widehat{G_{\boldsymbol \alpha^{\text{mult}},\boldsymbol \gamma^{\text{mult}}}}=\left(-\frac{1}{N}\sum_{i=1}^N \frac{S^{\text{mult}}_i}{\pi(\boldsymbol X^{\text{mult}}_i,\widehat{\boldsymbol \alpha^{\text{mult}}})^2}\left\{D_i \boldsymbol Z_i-\frac{e^{\widehat{\boldsymbol \theta}_{\boldsymbol Z}^{'}\boldsymbol Z}}{(1+e^{\widehat{\boldsymbol \theta}_{\boldsymbol Z}^{'}\boldsymbol Z_i})}\boldsymbol Z_i-\boldsymbol f({\boldsymbol X^{\text{mult}}_i},\widehat{\boldsymbol \gamma^{\text{mult}}},\widehat{\boldsymbol \theta_{\boldsymbol Z}})\right\}\right.\\
      & \hspace{4cm}\left.\cdot \left.\frac{\partial \boldsymbol \pi(\boldsymbol X^{\text{mult}}_i,\boldsymbol \alpha^{\text{mult}*})}{\partial \boldsymbol \alpha^{\text{mult}}}\right|_{\boldsymbol \alpha=\widehat{\boldsymbol \alpha^{\text{mult}}}}\right.\\
    &\hspace{1cm} \left.,\frac{1}{N}\sum_{i=1}^N \left(\frac{S_{\text{ext},i}}{\pi_{\text{ext},i}}- \frac{S^{\text{mult}}_i}{\pi(\boldsymbol X^{\text{mult}}_i,\widehat{\boldsymbol \alpha^{\text{mult}}})}\right)\cdot \left.\frac{\partial \boldsymbol f({\boldsymbol X^{\text{mult}}_i},\widehat{\boldsymbol \gamma^{\text{mult}}},\widehat{\boldsymbol \theta_{\boldsymbol Z}})}{\partial \boldsymbol \gamma^{\text{mult}}}\right|_{\boldsymbol \gamma^{\text{mult}}=\widehat{\boldsymbol \gamma^{\text{mult}}}}\right)\\
     & \xrightarrow{p}  G_{\boldsymbol \alpha^{\text{mult}*},\boldsymbol \gamma^{\text{mult}*}}=\mathbb{E}\left[-\frac{S^{\text{mult}}}{\pi(\boldsymbol X^{\text{mult}},\boldsymbol \alpha^{\text{mult}*})^2}\left\{D \boldsymbol Z-\frac{e^{\boldsymbol \theta_{\boldsymbol Z}^{*'}\boldsymbol Z}}{(1+e^{\boldsymbol \theta_{\boldsymbol Z}^{*'}\boldsymbol Z})} \boldsymbol Z-\boldsymbol f({\boldsymbol X^{\text{mult}}},\boldsymbol \gamma^{\text{mult}*},\boldsymbol \theta_{\boldsymbol Z}^*)\right\}\right.\\
      & \left.\cdot \left.\frac{\partial \boldsymbol \pi(\boldsymbol X^{\text{mult}},\boldsymbol \alpha^{\text{mult}*})}{\partial \boldsymbol \alpha^{\text{mult}}}\right|_{\boldsymbol \alpha=\boldsymbol \alpha^{\text{mult}*}},\right.\\
      & \left.\left(\frac{S_{\text{ext}}}{\pi_{\text{ext}}}- \frac{S^{\text{mult}}}{\pi(\boldsymbol X^{\text{mult}},\boldsymbol \alpha^{\text{mult}*})}\right) \left.\frac{\partial \boldsymbol f({\boldsymbol X^{\text{mult}}},\boldsymbol \gamma^{\text{mult}},\boldsymbol \theta_{\boldsymbol Z}^*)}{\partial \boldsymbol \gamma^{\text{mult}}}\right|_{\boldsymbol \gamma^{\text{mult}}=\boldsymbol \gamma^{\text{mult}*}}\right]
\end{align*}
\begin{align*}
  & \widehat{H} = \frac{1}{N}\sum_{i=1}^N \ \begin{pmatrix}
    \left.\frac{\partial \boldsymbol q(\boldsymbol \alpha^{\text{mult}},\boldsymbol X^{\text{mult}}_i,S^{\text{mult}}_i,S_{\text{ext},i},\pi_{\text{ext}},i)}{\partial \boldsymbol \alpha^{\text{mult}}}\right|_{\boldsymbol \alpha=\widehat{\boldsymbol \alpha^{\text{mult}}}} & 0 \\
     0 & \left.\frac{\partial \boldsymbol p(\boldsymbol\gamma,\boldsymbol X^{\text{mult}}_i,S,\boldsymbol Z_{1\cap,i})}{\partial \boldsymbol \gamma^{\text{mult}}}\right|_{\boldsymbol \gamma^{\text{mult}}=\widehat{\boldsymbol \gamma^{\text{mult}}}}
         \end{pmatrix}\\
     &\xrightarrow{p} H = \mathbb{E} \begin{pmatrix}
    \left.\frac{\partial \boldsymbol q(\boldsymbol \alpha^{\text{mult}},\boldsymbol X^{\text{mult}},S,S_{\text{ext}},\pi_{\text{ext}})}{\partial \boldsymbol \alpha^{\text{mult}}}\right|_{\boldsymbol \alpha=\boldsymbol \alpha^{\text{mult}*}} & 0 \\
     0 & \left.\frac{\partial \boldsymbol p(\boldsymbol\gamma,\boldsymbol X^{\text{mult}},S,\boldsymbol Z_{1\cap})}{\partial \boldsymbol \gamma^{\text{mult}}}\right|_{\boldsymbol \gamma^{\text{mult}}=\boldsymbol \gamma^{\text{mult}*}}
         \end{pmatrix}
\end{align*}
Similarly we obtain
\begin{align*}
     & \widehat{E_1}=\frac{1}{N}\sum_{i =1}^N \frac{S^{\text{mult}}_i}{\pi(\boldsymbol X^{\text{mult}}_i,\widehat{\boldsymbol \alpha^{\text{mult}}})^2}\left[D_i\boldsymbol Z_i-\frac{e^{\widehat{\boldsymbol \theta_{\boldsymbol Z}}'\boldsymbol Z_i}}{(1+e^{\widehat{\boldsymbol \theta_{\boldsymbol Z}}' \boldsymbol Z_i})}\cdot \boldsymbol Z_i-f({\boldsymbol X^{\text{mult}}_i},\widehat{\boldsymbol \gamma^{\text{mult}}},\widehat{\boldsymbol \theta_{\boldsymbol Z}})\right]\cdot\\
    & \hspace{3cm}\left[D_i\boldsymbol Z_i-\frac{e^{\widehat{\boldsymbol \theta_{\boldsymbol Z}}'\boldsymbol Z_i}}{(1+e^{\widehat{\boldsymbol \theta_{\boldsymbol Z}}' \boldsymbol Z_i})}\cdot \boldsymbol Z_i-f({\boldsymbol X^{\text{mult}}_i},\widehat{\boldsymbol \gamma^{\text{mult}}},\widehat{\boldsymbol \theta_{\boldsymbol Z}})\right]'\\
    & \hspace{3cm}+ \frac{1}{N}\sum_{i =1}^N \frac{S_{\text{ext},i}}{\pi_{\text{ext},i}^2}\cdot f({\boldsymbol X^{\text{mult}}_i},\widehat{\boldsymbol \gamma^{\text{mult}}},\widehat{\boldsymbol \theta_{\boldsymbol Z}})\cdot f({\boldsymbol X^{\text{mult}}_i},\widehat{\boldsymbol \gamma^{\text{mult}}},\widehat{\boldsymbol \theta_{\boldsymbol Z}})'\\
    & + \frac{1}{N}\sum_{i =1}^N \frac{S^{\text{mult}}_iS_{\text{ext},i}}{\pi(\boldsymbol X^{\text{mult}}_i,\widehat{\boldsymbol \alpha^{\text{mult}}})\pi_{\text{ext},i}}\left[D_i\boldsymbol Z_i-\frac{e^{\widehat{\boldsymbol \theta_{\boldsymbol Z}}'\boldsymbol Z_i}}{(1+e^{\widehat{\boldsymbol \theta_{\boldsymbol Z}}' \boldsymbol Z_i})}\cdot \boldsymbol Z_i-f({\boldsymbol X^{\text{mult}}_i},\widehat{\boldsymbol \gamma^{\text{mult}}},\widehat{\boldsymbol \theta_{\boldsymbol Z}})\right]\\
    &\hspace{1cm}\cdot f({\boldsymbol X^{\text{mult}}_i},\widehat{\boldsymbol \gamma^{\text{mult}}},\widehat{\boldsymbol \theta_{\boldsymbol Z}})'+ \frac{1}{N}\sum_{i =1}^N \frac{S^{\text{mult}}_iS_{\text{ext},i}}{\pi(\boldsymbol X^{\text{mult}}_i,\widehat{\boldsymbol \alpha^{\text{mult}}})\pi_{\text{ext},i}}\cdot f({\boldsymbol X^{\text{mult}}_i},\widehat{\boldsymbol \gamma^{\text{mult}}},\widehat{\boldsymbol \theta_{\boldsymbol Z}})\\
    & \hspace{4cm}\cdot \left[D_i\boldsymbol Z_i-\frac{e^{\widehat{\boldsymbol \theta_{\boldsymbol Z}}'\boldsymbol Z_i}}{(1+e^{\widehat{\boldsymbol \theta_{\boldsymbol Z}}' \boldsymbol Z_i})}\cdot \boldsymbol Z_i-f({\boldsymbol X^{\text{mult}}_i},\widehat{\boldsymbol \gamma^{\text{mult}}},\widehat{\boldsymbol \theta_{\boldsymbol Z}})\right]'\\
    \end{align*}
\begin{align*}
      & \widehat{E}_2 =  \widehat{E}_3'=\frac{1}{N}\sum_{i=1}^N  \widehat{G_{\boldsymbol \alpha^{\text{mult}},\boldsymbol \gamma^{\text{mult}}}}\cdot (-\widehat{H})^{-1}\frac{1}{\pi(\boldsymbol X^{\text{mult}}_i,\widehat{\boldsymbol \alpha^{\text{mult}})}}\cdot\\
      &\begin{pmatrix}
   q(\widehat{\boldsymbol \alpha^{\text{mult}}},\boldsymbol X^{\text{mult}}_i,S^{\text{mult}}_i,S_{\text{ext},i},\pi_{\text{ext},i})\\
    p(\widehat{\boldsymbol\gamma},\boldsymbol X^{\text{mult}}_i,S^{\text{mult}}_i,\boldsymbol Z_{1\cap,i})
         \end{pmatrix}\left[D_i\boldsymbol Z_i-\frac{e^{\widehat{\boldsymbol \theta_{\boldsymbol Z}}'\boldsymbol Z_i}}{(1+e^{\widehat{\boldsymbol \theta_{\boldsymbol Z}}' \boldsymbol Z_i})} \boldsymbol Z_i-f({\boldsymbol X^{\text{mult}}_i},\widehat{\boldsymbol \gamma^{\text{mult}}},\widehat{\boldsymbol \theta_{\boldsymbol Z}})\right]'\\
         &\hspace{1cm}  +\frac{1}{N}\sum_{i =1}^N \begin{pmatrix}
   q(\widehat{\boldsymbol \alpha^{\text{mult}}},\boldsymbol X^{\text{mult}}_i,S^{\text{mult}}_i,S_{\text{ext},i},\pi_{\text{ext},i})\\
    p(\widehat{\boldsymbol\gamma},\boldsymbol X^{\text{mult}}_i,S^{\text{mult}}_i,\boldsymbol Z_{1\cap,i})
         \end{pmatrix}\frac{S_{\text{ext},i}}{\pi_{\text{ext},i}}\cdot  f({\boldsymbol X^{\text{mult}}_i},\widehat{\boldsymbol \gamma^{\text{mult}}},\widehat{\boldsymbol \theta_{\boldsymbol Z}})'\\
    &\xrightarrow{p} \mathbb{E}\{G_{\boldsymbol \alpha^{\text{mult}*},\boldsymbol \gamma^{\text{mult}*}} H^{-1}\cdot \boldsymbol h(\boldsymbol \alpha^{\text{mult}*},\boldsymbol \gamma^{\text{mult}*})\cdot  \boldsymbol g(\boldsymbol \theta_{\boldsymbol Z}^*,\boldsymbol \alpha^{\text{mult}*},\boldsymbol \gamma^{\text{mult}*})'\}\cdot
\end{align*}
\begin{align*}
     & \widehat{E_4} = \frac{1}{N}\cdot\widehat{G_{\boldsymbol \alpha^{\text{mult}}}}\cdot \widehat{H}^{-1} \sum_{i=1}^N  \left[\begin{pmatrix}
   q(\widehat{\boldsymbol \alpha^{\text{mult}}},\boldsymbol X^{\text{mult}}_i,S^{\text{mult}}_i,S_{\text{ext},i},\pi_{\text{ext},i})\\
    p(\widehat{\boldsymbol\gamma},\boldsymbol X^{\text{mult}}_i,S^{\text{mult}}_i,\boldsymbol Z_{1\cap,i})
         \end{pmatrix}\right.\\
         &\left.\cdot (q(\widehat{\boldsymbol \alpha^{\text{mult}}},\boldsymbol X^{\text{mult}}_i,S^{\text{mult}}_i,S_{\text{ext},i},\pi_{\text{ext},i})',p(\widehat{\boldsymbol\gamma},\boldsymbol X^{\text{mult}}_i,S^{\text{mult}}_i,\boldsymbol Z_{1\cap,i}))'\right](\widehat{H}^{-1})'\cdot(\widehat{G_{\boldsymbol \alpha^{\text{mult}}}})'\\
    & \xrightarrow{p} \mathbb{E}\{G_{\boldsymbol \alpha^{\text{mult}*},\boldsymbol \gamma^{\text{mult}*}}H^{-1}\cdot \boldsymbol h(\boldsymbol \alpha^{\text{mult}*},\boldsymbol \gamma^{\text{mult}*})\cdot \boldsymbol h(\boldsymbol \alpha^{\text{mult}*},\boldsymbol \gamma^{\text{mult}*})'\cdot (H^{-1})'\cdot (G_{\boldsymbol\alpha^{\text{mult}*},\boldsymbol \gamma^{\text{mult}*}})'\} 
\end{align*}
Therefore we obtain
\begin{align*}
     & \widehat{E}=\widehat{E}_1-\widehat{E}_2-\widehat{E}_3 + \widehat{E}_4\xrightarrow{p}\\
    & \hspace{2cm} \mathbb{E}[\{\boldsymbol g(\boldsymbol \theta_{\boldsymbol Z}^*,\boldsymbol \alpha^{\text{mult}*},\boldsymbol \gamma^{\text{mult}*})+G_{\boldsymbol \alpha^{\text{mult}*},\boldsymbol \gamma^{\text{mult}*}}\cdot \boldsymbol \Psi(\boldsymbol \alpha^{\text{mult}*},\boldsymbol \gamma^{\text{mult}*})\}\\
    &  \hspace{4cm}\{\boldsymbol g(\boldsymbol \theta_{\boldsymbol Z}^*,\boldsymbol \alpha^{\text{mult}*},\boldsymbol \gamma^{\text{mult}*})+G_{\boldsymbol \alpha^{\text{mult}*},\boldsymbol \gamma^{\text{mult}*}}.\boldsymbol \Psi(\boldsymbol \alpha^{\text{mult}*},\boldsymbol \gamma^{\text{mult}*})\}']
\end{align*}
Using all the above results we obtain $\widehat{G_{\boldsymbol \theta_{\boldsymbol Z}}}^{-1}\cdot\widehat{E}\cdot(\widehat{G_{\boldsymbol \theta_{\boldsymbol Z}}}^{-1})^{'}$
\begin{align*}
    &\xrightarrow{p}(G_{\boldsymbol \theta_{\boldsymbol Z}^*})^{-1}.\mathbb{E}[\{\boldsymbol g(\boldsymbol \theta_{\boldsymbol Z}^*,\boldsymbol \alpha^{\text{mult}*},\boldsymbol \gamma^{\text{mult}*})+G_{\boldsymbol \alpha^{\text{mult}*},\boldsymbol \gamma^{\text{mult}*}}\cdot \boldsymbol \Psi(\boldsymbol \alpha^{\text{mult}*},\boldsymbol \gamma^{\text{mult}*})\}\\
    & \hspace{2cm}\{\boldsymbol g(\boldsymbol \theta_{\boldsymbol Z}^*,\boldsymbol \alpha^{\text{mult}*},\boldsymbol \gamma^{\text{mult}*})+G_{\boldsymbol \alpha^{\text{mult}*},\boldsymbol \gamma^{\text{mult}*}}.\boldsymbol \Psi(\boldsymbol \alpha^{\text{mult}*},\boldsymbol \gamma^{\text{mult}*})\}'].(G_{\boldsymbol \theta_{\boldsymbol Z}^*}^{-1})^{'}
\end{align*}
Using the same approach used in the last step of Theorem S3, we obtain that \\$\frac{1}{N}\cdot\widehat{G_{\boldsymbol \theta_{\boldsymbol Z}}}^{-1}\cdot\hat{E}\cdot(\widehat{G_{\boldsymbol \theta_{\boldsymbol Z}}}^{-1})^{'}$ is a consistent estimator of the asymptotic variance of $\widehat{\boldsymbol \theta}_{\boldsymbol Z}$.
\end{proof}

\subsection{Proof of Theorem 3.4}
\begin{proof}
We establish the asymptotic distribution of $\widehat{\boldsymbol{\theta}}_{\boldsymbol{Z}}$ from JAIPW-NP using Theorem 3.3 from \citet{chernozhukov2018double}, which pertains to non-linear moment functions. We have verified the unbiasedness of the estimating equation (3.13) of the main text and assumed conditions C1 and C3 in the main text and assumptions \textit{A1}, \textit{A2}, \textit{A9} and \textit{A10} in Supplementary Section \ref{sec:assumptions}. To apply Theorem 3.3 from \citet{chernozhukov2018double}, it remains to verify the Neyman orthogonality of the estimating function. The proposed estimating function for JAIPW is given by
\begin{align*}
    \boldsymbol{g}(\boldsymbol{\theta}_{\boldsymbol{Z}}, \pi, \boldsymbol{f}) 
    = \frac{S^{\text{mult}}}{\pi(\boldsymbol{X}^{\text{mult}})} \left( \mathcal{U}(\boldsymbol{\theta}_{\boldsymbol{Z}}) - \boldsymbol{f}(\boldsymbol{X}^{\text{mult}}, \boldsymbol{\theta}_{\boldsymbol{Z}}) \right)
    + \frac{S_{\text{ext}}}{\pi_{\text{ext}}} \boldsymbol{f}(\boldsymbol{X}^{\text{mult}}, \boldsymbol{\theta}_{\boldsymbol{Z}}).
\end{align*}
\noindent
Let \(\boldsymbol{\eta} = (\pi, \boldsymbol{f})\) denote the combined nuisance parameter function. To establish Neyman orthogonality of \(\boldsymbol{g}\), we must show that the expectation of the Gâteaux derivative of \(\boldsymbol{g}\) with respect to \(\boldsymbol{\eta}\), evaluated at the true nuisance parameters \(\boldsymbol{\eta}^* = (\pi^*, \boldsymbol{f}^*)\), is zero. Formally, we aim to prove
\[
\mathbb{E} \left[ \left. \frac{d}{d\epsilon} \, \boldsymbol{g}(\boldsymbol{\theta}_{\boldsymbol{Z}}, \pi_\epsilon, \boldsymbol{f}_\epsilon) \right|_{\epsilon=0} \right] = 0,
\]
where the perturbed nuisance functions are defined as
\begin{align*}
    \pi_\epsilon(\boldsymbol{X}^{\text{mult}}) &= \pi^*(\boldsymbol{X}^{\text{mult}}) + \epsilon \, l_{\pi}(\boldsymbol{X}^{\text{mult}}), \\
    \boldsymbol{f}_\epsilon(\boldsymbol{X}^{\text{mult}}, \boldsymbol{\theta}_{\boldsymbol{Z}}) &= \boldsymbol{f}^*(\boldsymbol{X}^{\text{mult}}, \boldsymbol{\theta}_{\boldsymbol{Z}}) + \epsilon \, l_f(\boldsymbol{X}^{\text{mult}}, \boldsymbol{\theta}_{\boldsymbol{Z}}),
\end{align*}
with \( l_{\pi} \) and \( l_f \) being arbitrary, smooth perturbation functions, and \(\epsilon\) a small scalar.\\

\noindent
We now compute the derivative:
\begin{align*}
    \left. \frac{d}{d\epsilon} \, \boldsymbol{g}(\boldsymbol{\theta}_{\boldsymbol{Z}}, \pi_\epsilon, \boldsymbol{f}_\epsilon) \right|_{\epsilon=0}
    &= -\frac{S^{\text{mult}}}{\pi^*(\boldsymbol{X}^{\text{mult}})^2} \, l_{\pi}(\boldsymbol{X}^{\text{mult}}) \left( \mathcal{U}(\boldsymbol{\theta}_{\boldsymbol{Z}}) - \boldsymbol{f}^*(\boldsymbol{X}^{\text{mult}}, \boldsymbol{\theta}_{\boldsymbol{Z}}) \right) \\
    &\quad + \left( \frac{S_{\text{ext}}}{\pi_{\text{ext}}} - \frac{S^{\text{mult}}}{\pi^*(\boldsymbol{X}^{\text{mult}})} \right) l_f(\boldsymbol{X}^{\text{mult}}, \boldsymbol{\theta}_{\boldsymbol{Z}}).
\end{align*}
\noindent
Taking expectations under the assumption that the nuisance functions are correctly specified, we have
\begin{align*}
    \mathbb{E}\left[ \mathcal{U}(\boldsymbol{\theta}_{\boldsymbol{Z}}) - \boldsymbol{f}^*(\boldsymbol{X}^{\text{mult}}, \boldsymbol{\theta}_{\boldsymbol{Z}}) \right] &= 0, \\
    \mathbb{E}\left[ \frac{S_{\text{ext}}}{\pi_{\text{ext}}} - \frac{S^{\text{mult}}}{\pi^*(\boldsymbol{X}^{\text{mult}})} \right] &= 0.
\end{align*}
\noindent
Therefore,
\[
\mathbb{E} \left[ \left. \frac{d}{d\epsilon} \, \boldsymbol{g}(\boldsymbol{\theta}_{\boldsymbol{Z}}, \pi_\epsilon, \boldsymbol{f}_\epsilon) \right|_{\epsilon=0} \right] = 0,
\]
which proves the Neyman orthogonality of \(\boldsymbol{g}\) with respect to the nuisance parameter function \(\boldsymbol{\eta}\).
\end{proof}

\subsection{Variance Estimation of JAIPW-NP}
We empirically estimate the asymptotic variance of the JAIPW-NP estimator using the result from Theorem 3.4. The estimated variance is given by:
\[
\frac{1}{N} \cdot \widehat{G}_{\boldsymbol{\theta}_{\boldsymbol{Z}}}^{-1} \cdot \widehat{E} \cdot \left(\widehat{G}_{\boldsymbol{\theta}_{\boldsymbol{Z}}}^{-1}\right)',
\]
where \(\widehat{G}_{\boldsymbol{\theta}_{\boldsymbol{Z}}}\) is the empirical gradient matrix of the estimating equation and \(\widehat{E}\) is the empirical variance of the estimating function. Specifically,
\begin{align*}
\widehat{G}_{\boldsymbol{\theta}_{\boldsymbol{Z}}} = -\frac{1}{N} \sum_{i=1}^N & S^{\text{mult}}_i \cdot \frac{1}{\pi(\boldsymbol{X}^{\text{mult}}_i; \widehat{\boldsymbol{\alpha}}^{\text{mult}})} \left\{ \left.\frac{e^{\widehat{\boldsymbol{\theta}}_{\boldsymbol{Z}}' \boldsymbol{Z}_i}}{(1 + e^{\widehat{\boldsymbol{\theta}}_{\boldsymbol{Z}}' \boldsymbol{Z}_i})^2} \cdot \boldsymbol{Z}_i \boldsymbol{Z}_i' + \frac{\partial \widehat{\boldsymbol{f}}(\boldsymbol{X}^{\text{mult}}_i, \boldsymbol{\theta}_{\boldsymbol{Z}})}{\partial \boldsymbol{\theta}_{\boldsymbol{Z}}} \right|_{\boldsymbol{\theta}_{\boldsymbol{Z}} = \widehat{\boldsymbol{\theta}}_{\boldsymbol{Z}}} \right\} \\
& \hspace{4cm}+\frac{1}{N} \sum_{i=1}^N \frac{S_{\text{ext},i}}{\pi_{\text{ext},i}} \cdot \left. \frac{\partial \widehat{\boldsymbol{f}}(\boldsymbol{X}^{\text{mult}}_i, \boldsymbol{\theta}_{\boldsymbol{Z}})}{\partial \boldsymbol{\theta}_{\boldsymbol{Z}}} \right|_{\boldsymbol{\theta}_{\boldsymbol{Z}} = \widehat{\boldsymbol{\theta}}_{\boldsymbol{Z}}}.
\end{align*}
\noindent
The empirical variance component \(\widehat{E}\) is given by:
\begin{align*}
&\widehat{E} = \frac{1}{N} \sum_{i=1}^N  \frac{S^{\text{mult}}_i}{\pi(\boldsymbol{X}^{\text{mult}}_i; \widehat{\boldsymbol{\alpha}}^{\text{mult}})^2} \cdot \left( D_i \boldsymbol{Z}_i - \frac{e^{\widehat{\boldsymbol{\theta}}_{\boldsymbol{Z}}' \boldsymbol{Z}_i}}{1 + e^{\widehat{\boldsymbol{\theta}}_{\boldsymbol{Z}}' \boldsymbol{Z}_i}} \cdot \boldsymbol{Z}_i - \widehat{\boldsymbol{f}}(\boldsymbol{X}^{\text{mult}}_i, \widehat{\boldsymbol{\theta}}_{\boldsymbol{Z}}) \right) \\
& \hspace{3cm}\cdot \left( D_i \boldsymbol{Z}_i - \frac{e^{\widehat{\boldsymbol{\theta}}_{\boldsymbol{Z}}' \boldsymbol{Z}_i}}{1 + e^{\widehat{\boldsymbol{\theta}}_{\boldsymbol{Z}}' \boldsymbol{Z}_i}} \cdot \boldsymbol{Z}_i - \widehat{\boldsymbol{f}}(\boldsymbol{X}^{\text{mult}}_i, \widehat{\boldsymbol{\theta}}_{\boldsymbol{Z}}) \right)' \\
&+ \frac{1}{N} \sum_{i=1}^N \frac{S_{\text{ext},i}}{\pi_{\text{ext},i}^2} \cdot \widehat{\boldsymbol{f}}(\boldsymbol{X}^{\text{mult}}_i, \widehat{\boldsymbol{\theta}}_{\boldsymbol{Z}}) \cdot \widehat{\boldsymbol{f}}(\boldsymbol{X}^{\text{mult}}_i, \widehat{\boldsymbol{\theta}}_{\boldsymbol{Z}})' \\
&+ \frac{1}{N} \sum_{i=1}^N \frac{S^{\text{mult}}_i S_{\text{ext},i}}{\pi(\boldsymbol{X}^{\text{mult}}_i; \widehat{\boldsymbol{\alpha}}^{\text{mult}}) \pi_{\text{ext},i}} \cdot \left( D_i \boldsymbol{Z}_i - \frac{e^{\widehat{\boldsymbol{\theta}}_{\boldsymbol{Z}}' \boldsymbol{Z}_i}}{1 + e^{\widehat{\boldsymbol{\theta}}_{\boldsymbol{Z}}' \boldsymbol{Z}_i}} \cdot \boldsymbol{Z}_i - \widehat{\boldsymbol{f}}(\boldsymbol{X}^{\text{mult}}_i, \widehat{\boldsymbol{\theta}}_{\boldsymbol{Z}}) \right)\cdot \widehat{\boldsymbol{f}}(\boldsymbol{X}^{\text{mult}}_i, \widehat{\boldsymbol{\theta}}_{\boldsymbol{Z}})' \\
&+ \frac{1}{N} \sum_{i=1}^N \frac{S^{\text{mult}}_i S_{\text{ext},i}}{\pi(\boldsymbol{X}^{\text{mult}}_i; \widehat{\boldsymbol{\alpha}}^{\text{mult}}) \pi_{\text{ext},i}} \cdot \widehat{\boldsymbol{f}}(\boldsymbol{X}^{\text{mult}}_i, \widehat{\boldsymbol{\theta}}_{\boldsymbol{Z}}) \cdot \left( D_i \boldsymbol{Z}_i - \frac{e^{\widehat{\boldsymbol{\theta}}_{\boldsymbol{Z}}' \boldsymbol{Z}_i}}{1 + e^{\widehat{\boldsymbol{\theta}}_{\boldsymbol{Z}}' \boldsymbol{Z}_i}} \cdot \boldsymbol{Z}_i - \widehat{\boldsymbol{f}}(\boldsymbol{X}^{\text{mult}}_i, \widehat{\boldsymbol{\theta}}_{\boldsymbol{Z}}) \right)'.
\end{align*}
\noindent
In the above, \(\widehat{\boldsymbol{f}}\) can be any flexible estimator (e.g., Random Forest, XGBoost) trained using cross-fitting. The gradient term
\[
\left. \frac{\partial \widehat{\boldsymbol{f}}(\boldsymbol{X}^{\text{mult}}_i, \boldsymbol{\theta}_{\boldsymbol{Z}})}{\partial \boldsymbol{\theta}_{\boldsymbol{Z}}} \right|_{\boldsymbol{\theta}_{\boldsymbol{Z}} = \widehat{\boldsymbol{\theta}}_{\boldsymbol{Z}}}
\]
is evaluated via numerical differentiation and is also estimated in a cross-fitted manner.
\section{Simulation Details}
\subsection{General Setup}\label{sec:simuset}
\setlength\arraycolsep{2pt}
In all the simulations, we have used the following factors in our simulation design:
\begin{itemize}
    \item \textbf{Target Population Size:} $N=50000$.
    \item \textbf{Number of Cohorts:} $K=3$
    \item \textbf{Number of Replications}: $R=500$
    \item \textbf{Disease Model Covariates} \textbf{$\boldsymbol Z=(Z_1,Z_2,Z_3)$:} The joint distribution of $(Z_1,Z_2,Z_3)$ is specified as,
     \begin{align*}
    \begin{pmatrix}
           Z_1 \\
           Z_2 \\
           Z_3
         \end{pmatrix} \sim \mathcal{N}_3\left(\begin{pmatrix}
           0\\
           0\\
           0
         \end{pmatrix},\begin{pmatrix}
           1 & 0.5 & 0.5\\
           0.5 & 1 & 0.5\\
           0.5 & 0.5 & 1
         \end{pmatrix}\right)\cdot
  \end{align*}
    \item \textbf{Disease outcome $D$}: $D$ is simulated from the conditional distribution specified by,
    $$D|Z_1,Z_2,Z_3\sim\text{Ber}\left(\frac{e^{\theta_0+\theta_1Z_1+\theta_2Z_2+\theta_3Z_3}}{1+e^{\theta_0+\theta_1Z_1+\theta_2Z_2+\theta_3Z_3}}\right)\cdot$$
    where, $\theta_0=-2$, $\theta_1=0.35$, $\theta_2=0.45$ and $\theta_3=0.25$.
    \item \textbf{Selection model covariate for internal non-probability sample:} For $k=1,2,3$, $W_k$ is an univariate random variable simulated from the conditional distribution,
    $$W_k|Z_1,Z_2,Z_3,D \sim \epsilon_{1k} + \mathcal{N}(\gamma_0\cdot D + \gamma_1\cdot Z_1 + \gamma_2\cdot Z_2 + \gamma_3\cdot Z_3,1) +\epsilon_{2k}\cdot$$
    where $\epsilon_{1k}\sim \mathcal{N}(0,1)$ with $\text{Cor}(\epsilon_{1k},Z_1)=\text{Cor}(\epsilon_{1k},Z_2)=\text{Cor}(\epsilon_{1k},Z_3)=0.5$, $\epsilon_{2k}\sim \mathcal{N}(0,1)$ $\forall k \in \{1,2,3\}$. We set $(\gamma_0,\gamma_1,\gamma_2,\gamma_3)=(1,1,0.8,0.6),(1,0.6,0.8,1),(1,1,1,1)$ for the three cohorts respectively.
    \item \textbf{Selection Model for the external probability sample:} For external data, the selection model can take any functional form. Note that these selection probabilities are known to us. In our case, we assumed that the functional form of the external selection model is given by,
    \begin{align*}
    \text{logit}(P(S_{\text{ext}}=1|D,Z_1,Z_2,Z_3))=\nu_{0}+\nu_{1}\cdot D+\nu_{2}\cdot Z_1 + \nu_{3}\cdot Z_2 + \nu_4 \cdot Z_3\cdot
    \end{align*}
    The values of $(\nu_0,\nu_1,\nu_2,\nu_3, \nu_4)$ are given by $(-0.6,1.2,0.4,-0.2,0.5)$. The probabilities $P(S_{\text{ext}}=1|Z_2,W,D)$ from the above equation  were multiplied by a factor of 0.75.
\end{itemize}
In the first simulation setup, the internal selection models are as follows :
\begin{itemize}
     \item \textbf{Cohort 1 :} $\boldsymbol Z_2=(Z_2,Z_3),W=W_1,D$
       \begin{align*}
             \text{logit}(P(S_1=1|Z_2,Z_3,W_1,D))&=\alpha_{01}+\alpha_{11} Z_2+
           \alpha_{21} Z_3 +  \alpha_{31} W_1 + \alpha_{41} D
       \end{align*}
    where $(\alpha_{01},\alpha_{11},\alpha_{21},\alpha_{31},\alpha_{41})=(-1,1.5,0.2,0.8,-0.3)$.
    \item \textbf{Cohort 2 :} $\boldsymbol Z_2=Z_3,W=W_2,D$
        \begin{align*}
             \text{logit}(P(S_2=1|Z_3,W_2,D))=\alpha_{02}+\alpha_{12} Z_3+
           \alpha_{22} W_2 +  \alpha_{32} D 
       \end{align*}
     where $(\alpha_{02},\alpha_{12},\alpha_{22},\alpha_{32})=(-1,1.25,0.4,0.6)$.
     \item \textbf{Cohort 3 :} $\boldsymbol Z_2=Z_2,W=W_3$
       \begin{align*}
             \text{logit}(P(S_3=1|Z_2,W_3))=\alpha_{03}+\alpha_{13} Z_2+
           \alpha_{23} W_3 
       \end{align*}
    where $(\alpha_{03},\alpha_{13},\alpha_{23})=(-3,0.8,0.5)$.
\end{itemize}
In the second simulation setup, we added interaction terms in the internal selection models : 
\begin{itemize}
     \item \textbf{Cohort 1 :} $\boldsymbol Z_2=(Z_2,Z_3),W=W_1,D$
\begin{equation*}
\small
\hspace{-1cm}\text{logit}(P(S_1=1|Z_2,Z_3,W_1,D))=\alpha_{01}+\alpha_{11} Z_2+\alpha_{21} Z_3 +  \alpha_{31} W_1 +\alpha_{41} D + \alpha_{51} DZ_2  + \alpha_{61} DZ_3 + \alpha_{71} DW_1
\end{equation*}
    where $(\alpha_{01},\alpha_{11},\alpha_{21},\alpha_{31},\alpha_{41})=(-1,1.5,0.2,0.8,-0.3)$.
    \item \textbf{Cohort 2 :} $\boldsymbol Z_2=Z_3,W=W_2,D$
        \begin{align*}
             &\text{logit}(P(S_2=1|Z_3,W_2,D))=\alpha_{02}+\alpha_{12} Z_3+
           \alpha_{22} W_2 +  \alpha_{32} D + \alpha_{42} DZ_3 + \alpha_{52} DW_2
       \end{align*}
     where $(\alpha_{02},\alpha_{12},\alpha_{22},\alpha_{32})=(-1,1.25,0.4,0.6)$.
     \item \textbf{Cohort 3 :} $\boldsymbol Z_2=Z_2,W=W_3$
       \begin{align*}
             \text{logit}(P(S_3=1|Z_2,W_3))=\alpha_{03}+\alpha_{13} Z_2+
           \alpha_{23} W_3 +  \alpha_{33} Z_2W_3
       \end{align*}
    where $(\alpha_{03},\alpha_{13},\alpha_{23})=(-3,0.8,0.5)$.
\end{itemize}

\subsection{JPS specifics for simulations}
\subsubsection{Criteria for coarsening variables for PS method}\label{sec:pr2_coarse}
For any continuous random variable say, $L$, a coarsened version $L'$ is defined as, 
\[   
L' = 
     \begin{cases}
      0 &\text{if}\hspace{0.2cm} L<\text{Cutoff}_1\\
       1 &\text{if}\hspace{0.2cm}<\text{Cutoff}_1<=L<=<\text{Cutoff}_2 \\
       \vdots & \vdots\\
      K-1 &\text{if}\hspace{0.2cm} L>\text{Cutoff}_K\\
     \end{cases}
\]
In this simulation for all the selection variables involved, we chose $K=2$ and $\text{Cutoff}_1=\epsilon_{0.15}$, $\text{Cutoff}_2=\epsilon_{0.85}$, where $\epsilon_{0.15}$ and $\epsilon_{0.85}$ are the 15 and 85 percentile quartiles for both the continuous random variables respectively.

\subsubsection{Fitting of JPS method}\label{sec:jpssup}
For both the setups, JPS has been fitted in two ways. In the first scenario, we assumed that the entire joint distribution of the discretized versions of the selection variables are available from the target population, namely $(D,Z_2',Z_3',W_1')$, $(D,Z_3',W_2')$ and $(Z_2',W_3')$ in the three cohorts respectively where $Z_2',Z_3',W_1',W_2',W_3'$ denote the discretized versions of $Z_2,Z_3,W_1,W_2,W_3$. In reality obtaining exact joint distribution of all the selection variables is extremely difficult especially for Cohort 1 and 2. Henceforth in the second scenario, we assumed that we have access of only conditional distributions $D|Z_2'$, $Z_3'|Z_2'$, $W_1'|Z_2'$, $W_2'|Z_2'$ and $W_3'|Z_2'$ from the target population. The joint selection probabilities are approximated using these available conditionals.
\section{Additional specifications for Data Analysis}
\subsection{Phenotype definitions and descriptions}
In this study, phenotypes are defined using the PheWAS R package (Version 0.99.5-4), which maps ICD-9-CM and ICD-10-CM codes to PheWAS codes (PheCodes), resulting in up to 1,817 distinct codes. For our analysis, the PheCodes for cancer include: 145, 145.1, 145.2, 145.3, 145.4, 145.5, 149, 149.1, 149.2, 149.3, 149.4, 149.5, 149.9, 150, 151, 153, 153.2, 153.3, 155, 155.1, 157, 158, 159, 159.2, 159.3, 159.4, 164, 165, 165.1, 170, 170.1, 170.2, 172.1, 172.11, 174, 174.1, 174.11, 174.2, 174.3, 175, 180, 180.1, 180.3, 182, 184, 184.1, 184.11, 184.2, 185, 187, 187.1, 187.2, 187.8, 189, 189.1, 189.11, 189.12, 189.2, 189.21, 189.4, 190, 191, 191.1, 191.11, 193, 194, 195, 195.1, 195.3, 196, 197, 198, 198.1, 198.2, 198.3, 198.4, 198.5, 198.6, 198.7, 199, 199.4, 200, 200.1, 201, 202, 202.2, 202.21, 202.22, 202.23, 202.24, 204, 204.1, 204.11, 204.12, 204.2, 204.21, 204.22, 204.3, 204.4, 209, and 860.\\

\noindent
In addition, other diseases are represented by specific PheCodes: diabetes (250), coronary heart disease (CHD) (411.4), triglycerides (272.12, 272.13), vitamin D deficiency (261.4), depression (296.2), and anxiety (300)
\subsection{Biological Sex as $\boldsymbol{Z}_{1\cap}$}

The JAIPW method relies on conditions C2 and C3 of the main text. We acknowledge these are strong assumptions that cannot be verified with observed data. Its plausibility rests on the hypothesis that sex functions as an indirect correlate, influencing selection primarily through other conditions that are included in $\boldsymbol{X}$. The MHB (mental health) cohort provides a clear illustration: its skewed sex distribution is consistent with our model because sex is a known risk factor for certain mental health diagnoses, and these diagnoses (which are included in $\boldsymbol{X}$) are the primary correlates of selection.\\

\noindent
However, we also note cases where this assumption is likely violated. For instance, the AIPW method's performance in the MIPACT cohort was imperfect, likely due to unmeasured confounders like income, healthcare access and education. These factors are known to be related to both sex and selection into MIPACT, and their omission from $\boldsymbol{X}$ would violate the conditional independence assumption. Despite this violation, the AIPW estimates still offered a substantial partial correction, demonstrating a clear improvement over the unweighted method even when the assumption was not perfectly met. This underscores our point that we are often able to reduce bias and rarely able to remove bias\\

\noindent
Similarly, in our second data analysis, we apply the same core assumptions to the Cancer Polygenic Risk Score (PRS). We assume the Cancer PRS is conditionally independent of selection ($S$) given the observed covariates ($\boldsymbol{X}$). The justification here is that the PRS does not directly influence selection; rather, it acts indirectly by increasing the risk of developing cancer ($D$), and it is the cancer diagnosis itself (the outcome) that is related to selection into the surgery cohort.

\subsection{Results for Age Adjusted Analysis}\label{sec:ageres}
\textbf{MGI-Anesthesiology:} Subpart (C) of Figure \ref{fig:realdata_age} shows the performances of all the methods when implemented using data from MGI-Anesthesiology. Utilizing unweighted logistic regression, the age-adjusted estimate obtained was 0.07 with a 95\% confidence interval (CI) of [0.04, 0.11], which is in the entirely opposite direction from the SEER registry reference range for marginal estimates, illustrated by the grey band in the accompanying figure. Upon incorporating cancer in the selection model, the estimates derived from the four IPW methods — PL, SR, PS, and CL are -0.08 (95\% CI [-0.13, -0.04]), -0.13 (95\% CI [-0.20, -0.06]), -0.08 (95\% CI [-0.13, -0.04]), and -0.05 (95\% CI [-0.09, -0.01]) respectively. These adjustments for selection bias position all of the 95\% CIs from these IPW methods towards the direction of association indicated by SEER reference range. Conversely, when cancer is not included in the selection model, the estimates for all methods exhibit bias in the contrary direction underscoring the importance of considering cancer in the selection modeling. For the AIPW method, estimates yielded -0.09 (95\% CI [-0.16, -0.09]), -0.07 (95\% CI [-0.18, -0.11]) when the selection model included cancer and the auxiliary model did and did not include cancer, respectively. When the selection model excluded cancer with the auxiliary including it, the result was -0.05 (95\% CI [-0.09, -0.01]).\\

\noindent
\textbf{MIPACT:} Subpart (D) of Figure \ref{fig:realdata_age} shows the performances of all the methods when implemented using data from MIPACT. The estimated association between cancer and biological sex using the unweighted logistic regression method is considerably biased, with an estimate of 0.38 and a 95\% confidence interval (CI) of [0.26, 0.51]. The IPW and AIPW methods were able to reduce the bias to some extent, particularly when cancer was included in the selection model or in the auxiliary score model (for AIPW). \\

\noindent
\textbf{MEND:} Subpart (E) of Figure \ref{fig:realdata_age} shows the performance of all methods using data from MEND. The estimated age-adjusted association between cancer and biological sex from the unweighted logistic regression method exhibited bias, with an estimate of 0.21 and a 95\% confidence interval (CI) of [0.05, 0.36]. The IPW and AIPW methods were able to reduce the bias to some extent, particularly when cancer was excluded from the selection model.\\

\noindent
\textbf{MHB:} Subpart (F) of Figure \ref{fig:realdata_age} provides a comparative analysis of different methods when applied to the MHB data. The unweighted logistic regression method results in a highly biased estimate of the age-adjusted association between cancer and biological sex, with an estimated effect size of 0.72 and a 95\% confidence interval of [0.44, 0.99]. This large bias may stem in part from the unique aspect of the MHB cohort, wherein biological sex is a variable that influences selection. IPW methods and JAIPW, when accounting for cancer, demonstrate some capacity to diminish the observed bias. Nonetheless, the estimates across all these methods are characterized by high variance. This is likely due to the modest sample size of the MHB cohort, which can lead to wider confidence intervals and less precise estimates. \\ 

\noindent
\textbf{Combined Data:} Subpart (A) of Figure \ref{fig:realdata_age} shows the performances of all the methods when the cohorts are combined together. Using unweighted logistic regression, the estimate for the age adjusted association between cancer and biological sex was 0.10 [0.07, 0.13], which is in the entirely opposite direction from the SEER registry reference range for marginal estimates. Upon incorporating cancer in the selection model, the estimates derived from the four IPW methods—JPL, JSR, JPS, and JCL—are -0.03 (95\% CI [-0.08, 0.03]), -0.04 (95\% CI [-0.11, 0.02]), -0.08 (95\% CI [-0.11, -0.04]), and 0.00 (95\% CI [-0.04, 0.04]) respectively.For the JAIPW analysis, including cancer in either the selection model or the auxiliary model resulted in estimates of -0.09 (95\% CI [-0.13, -0.06]), -0.09 (95\% CI [-0.13, -0.06]), and -0.02 (95\% CI [-0.06, 0.02]). All these methods corrected the direction of the cancer-biological sex association estimate.

\newpage
\begin{table}[ht]
\renewcommand{\arraystretch}{2}
\centering
\begin{adjustbox}{max width=1\textwidth}
\begin{tabular}{|c|c|c|c|c|c|c|c|c|c|c|}
\hline
\textbf{Method}               & \textbf{\begin{tabular}[c]{@{}c@{}}Selection \\ Model 1\end{tabular}} & \textbf{\begin{tabular}[c]{@{}c@{}}Selection\\  Model 2\end{tabular}} & \textbf{\begin{tabular}[c]{@{}c@{}}Selection \\ Model 3\end{tabular}} & \textbf{\begin{tabular}[c]{@{}c@{}}Auxiliary \\ Model\end{tabular}} & \textbf{MCSE $(\theta_1$)} & \textbf{ESE $(\theta_1$)} & \textbf{MCSE $(\theta_2$)} & \textbf{ESE $(\theta_2$)} & \textbf{MCSE $(\theta_3$)} & \textbf{ESE $(\theta_3$)} \\ \hline
\textbf{Unweighted}           & -                                                                     & -                                                                     & -                                                                     & -                                                                   & 0.017                      & 0.018                     & 0.018                      & 0.018                     & 0.017                      & 0.018                     \\ \hline
\textbf{Unweighted Diff}      & -                                                                     & -                                                                     & -                                                                     & -                                                                   & 0.017                      & 0.018                     & 0.019                      & 0.019                     & 0.017                      & 0.018                     \\ \hline
\multirow{4}{*}{\textbf{JPL}} & Correct                                                               & Correct                                                               & Correct                                                               & -                                                                   & 0.018                      & 0.018                     & 0.019                      & 0.019                     & 0.019                      & 0.018                     \\ \cline{2-11} 
                              & Incorrect                                                             & Correct                                                               & Correct                                                               & -                                                                   & 0.026                      & 0.020                     & 0.020                      & 0.021                     & 0.021                      & 0.020                     \\ \cline{2-11} 
                              & Incorrect                                                             & Incorrect                                                             & Correct                                                               & -                                                                   & 0.024                      & 0.022                     & 0.022                      & 0.024                     & 0.024                      & 0.021                     \\ \cline{2-11} 
                              & Incorrect                                                             & Incorrect                                                             & Incorrect                                                             & -                                                                   & 0.024                      & 0.021                     & 0.021                      & 0.023                     & 0.024                      & 0.021                     \\ \hline
\multirow{4}{*}{\textbf{JSR}} & Correct                                                               & Correct                                                               & Correct                                                               & -                                                                   & 0.021                      & 0.018                     & 0.018                      & 0.019                     & 0.019                      & 0.018                     \\ \cline{2-11} 
                              & Incorrect                                                             & Correct                                                               & Correct                                                               & -                                                                   & 0.023                      & 0.020                     & 0.020                      & 0.020                     & 0.020                      & 0.020                     \\ \cline{2-11} 
                              & Incorrect                                                             & Incorrect                                                             & Correct                                                               & -                                                                   & 0.023                      & 0.021                     & 0.023                      & 0.023                     & 0.024                      & 0.021                     \\ \cline{2-11} 
                              & Incorrect                                                             & Incorrect                                                             & Incorrect                                                             & -                                                                   & 0.023                      & 0.021                     & 0.021                      & 0.023                     & 0.023                      & 0.020                     \\ \hline
\multirow{4}{*}{\textbf{JCL}} & Correct                                                               & Correct                                                               & Correct                                                               & -                                                                   & 0.017                      & 0.018                     & 0.018                      & 0.019                     & 0.019                      & 0.018                     \\ \cline{2-11} 
                              & Incorrect                                                             & Correct                                                               & Correct                                                               & -                                                                   & 0.021                      & 0.020                     & 0.019                      & 0.020                     & 0.020                      & 0.020                     \\ \cline{2-11} 
                              & Incorrect                                                             & Incorrect                                                             & Correct                                                               & -                                                                   & 0.022                      & 0.022                     & 0.021                      & 0.023                     & 0.023                      & 0.021                     \\ \cline{2-11} 
                              & Incorrect                                                             & Incorrect                                                             & Incorrect                                                             & -                                                                   & 0.022                      & 0.021                     & 0.021                      & 0.023                     & 0.023                      & 0.021                     \\ \hline
\textbf{JPS Exact}            & -                                                                     & -                                                                     & -                                                                     &                                                                     & 0.017                      & 0.018                     & 0.018                      & 0.019                     & 0.018                      & 0.018                     \\ \hline
\textbf{JPS Approximate}      & -                                                                     & -                                                                     & -                                                                     &                                                                     & 0.017                      & 0.018                     & 0.018                      & 0.018                     & 0.018                      & 0.017                     \\ \hline
\multirow{4}{*}{\textbf{JAIPW}}              & Correct                                                               & Correct                                                               & Correct                                                               & Correct                                                             & 0.021                      & 0.023                     & 0.025                      & 0.027                     & 0.025                      & 0.027                     \\ \cline{2-11} 
                              & Correct                                                               & Correct                                                               & Correct                                                               & Incorrect                                                           & 0.020                      & 0.022                     & 0.025                      & 0.026                     & 0.024                      & 0.024                     \\ \cline{2-11} 
                              & Incorrect                                                             & Incorrect                                                             & Incorrect                                                             & Correct                                                             & 0.022                      & 0.026                     & 0.024                      & 0.028                     & 0.025                      & 0.028                     \\ \cline{2-11} 
                              & Incorrect                                                             & Incorrect                                                             & Incorrect                                                             & Incorrect                                                           & 0.023                      & 0.023                     & 0.027                      & 0.026                     & 0.024                      & 0.023                     \\ \hline
\end{tabular}
\end{adjustbox}
\vspace{0.5cm}
\caption{Comparison of Estimated Standard Error (ESE) using proposed variance estimators and Monte Carlo Standard Errors (MCSE) across the unweighted, four joint IPW methods (JPL, JSR, JCL, JPS), and JAIPW  under simulation setup 1. The results are obtained using number of simulation replications as $R=500$.\\
\textbf{Abbreviations:} Unweighted = Unweighted Logistic Regression; JSR = Joint Simplex Regression; JPL = Joint Pseudolikelihood; JPS = Joint Post Stratification; JCL = Joint Calibration; JAIPW = Joint Augmented Inverse Probability Weighted.}
\label{tab:table_s1}
\end{table}
\begin{table}[ht]
\renewcommand{\arraystretch}{2}
\centering
\begin{adjustbox}{max width=1\textwidth}
\begin{tabular}{|c|c|c|c|c|c|c|c|c|c|c|}
\hline
\textbf{Method}               & \textbf{\begin{tabular}[c]{@{}c@{}}Selection \\ Model 1\end{tabular}} & \textbf{\begin{tabular}[c]{@{}c@{}}Selection\\  Model 2\end{tabular}} & \textbf{\begin{tabular}[c]{@{}c@{}}Selection \\ Model 3\end{tabular}} & \textbf{\begin{tabular}[c]{@{}c@{}}Auxiliary \\ Model\end{tabular}} & \textbf{MCSE $(\theta_1$)} & \textbf{ESE $(\theta_1$)} & \textbf{MCSE $(\theta_2$)} & \textbf{ESE $(\theta_2$)} & \textbf{MCSE $(\theta_3$)} & \textbf{ESE $(\theta_3$)} \\ \hline
\textbf{Unweighted}           & -                                                                     & -                                                                     & -                                                                     & -                                                                   & 0.018                      & 0.018                     & 0.018                      & 0.017                     & 0.017                      & 0.018                     \\ \hline
\textbf{Unweighted Diff}      & -                                                                     & -                                                                     & -                                                                     & -                                                                   & 0.018                      & 0.020                     & 0.019                      & 0.021                     & 0.018                      & 0.019                     \\ \hline
\multirow{4}{*}{\textbf{JPL}} & Correct                                                               & Correct                                                               & Correct                                                               & -                                                                   & 0.059                      & 0.068                     & 0.091                      & 0.089                     & 0.058                      & 0.076                     \\ \cline{2-11} 
                              & Incorrect                                                             & Correct                                                               & Correct                                                               & -                                                                   & 0.018                      & 0.019                     & 0.019                      & 0.021                     & 0.019                      & 0.023                     \\ \cline{2-11} 
                              & Incorrect                                                             & Incorrect                                                             & Correct                                                               & -                                                                   & 0.018                      & 0.017                     & 0.019                      & 0.019                     & 0.019                      & 0.017                     \\ \cline{2-11} 
                              & Incorrect                                                             & Incorrect                                                             & Incorrect                                                             & -                                                                   & 0.018                      & 0.017                     & 0.019                      & 0.018                     & 0.019                      & 0.017                     \\ \hline
\multirow{4}{*}{\textbf{JSR}} & Correct                                                               & Correct                                                               & Correct                                                               & -                                                                   & 0.029                      & 0.026                     & 0.034                      & 0.029                     & 0.030                      & 0.026                     \\ \cline{2-11} 
                              & Incorrect                                                             & Correct                                                               & Correct                                                               & -                                                                   & 0.018                      & 0.018                     & 0.019                      & 0.019                     & 0.018                      & 0.018                     \\ \cline{2-11} 
                              & Incorrect                                                             & Incorrect                                                             & Correct                                                               & -                                                                   & 0.018                      & 0.018                     & 0.018                      & 0.018                     & 0.018                      & 0.018                     \\ \cline{2-11} 
                              & Incorrect                                                             & Incorrect                                                             & Incorrect                                                             & -                                                                   & 0.018                      & 0.018                     & 0.018                      & 0.018                     & 0.018                      & 0.018                     \\ \hline
\multirow{4}{*}{\textbf{JCL}} & Correct                                                               & Correct                                                               & Correct                                                               & -                                                                   & 0.024                      & 0.022                     & 0.024                      & 0.022                     & 0.020                      & 0.020                     \\ \cline{2-11} 
                              & Incorrect                                                             & Correct                                                               & Correct                                                               & -                                                                   & 0.018                      & 0.018                     & 0.019                      & 0.019                     & 0.018                      & 0.018                     \\ \cline{2-11} 
                              & Incorrect                                                             & Incorrect                                                             & Correct                                                               & -                                                                   & 0.018                      & 0.018                     & 0.019                      & 0.019                     & 0.018                      & 0.018                     \\ \cline{2-11} 
                              & Incorrect                                                             & Incorrect                                                             & Incorrect                                                             & -                                                                   & 0.018                      & 0.018                     & 0.019                      & 0.018                     & 0.018                      & 0.018                     \\ \hline
\textbf{JPS Exact}            & -                                                                     & -                                                                     & -                                                                     & -                                                                   & 0.019                      & 0.019                     & 0.019                      & 0.019                     & 0.018                      & 0.019                     \\ \hline
\textbf{JPS Approximate}      & -                                                                     & -                                                                     & -                                                                     & -                                                                   & 0.018                      & 0.018                     & 0.019                      & 0.019                     & 0.017                      & 0.017                     \\ \hline
\multirow{4}{*}{\textbf{JAIPW}}               & Correct                                                               & Correct                                                               & Correct                                                               & Correct                                                             & 0.032                      & 0.030                     & 0.026                      & 0.029                     & 0.027                      & 0.029                     \\ \cline{2-11}
                              & Correct                                                               & Correct                                                               & Correct                                                               & Incorrect                                                           & 0.035                      & 0.030                     & 0.035                      & 0.033                     & 0.030                      & 0.028                     \\ \cline{2-11}
                              & Incorrect                                                             & Incorrect                                                             & Incorrect                                                             & Correct                                                             & 0.023                      & 0.023                     & 0.025                      & 0.028                     & 0.024                      & 0.028                     \\ \cline{2-11}
                              & Incorrect                                                             & Incorrect                                                             & Incorrect                                                             & Incorrect                                                           & 0.023                      & 0.025                     & 0.022                      & 0.026                     & 0.024                      & 0.028                     \\ \hline
\end{tabular}
\end{adjustbox}
\vspace{0.5cm}
\caption{Comparison of Estimated Standard Error (ESE) using proposed variance estimators and Monte Carlo Standard Errors (MCSE) across the unweighted, four joint IPW methods (JPL, JSR, JCL, JPS), and JAIPW  under simulation setup 2. The results are obtained using number of simulation replications as $R=500$.\\
\textbf{Abbreviations:} Unweighted = Unweighted Logistic Regression; JSR = Joint Simplex Regression; JPL = Joint Pseudolikelihood; JPS = Joint Post Stratification; JCL = Joint Calibration; JAIPW = Joint Augmented Inverse Probability Weighted.}
\label{tab:table_s2}
\end{table}

\newpage
\begin{table}[]
\centering
\begin{adjustbox}{width=1\textwidth}
\begin{tabular}{|ccccc|}
\hline
\textbf{Cohort}                                                                                  & \textbf{Anesthesiology}                                                      & \textbf{MIPACT}                                                              & \textbf{MEND}                                                                & \textbf{MHB}                                                                 \\ \hline
\textbf{Cancer}                                                                                  & \begin{tabular}[c]{@{}c@{}}Yes (52.2\%)\\ No  (47.8\%)\end{tabular}          & \begin{tabular}[c]{@{}c@{}}Yes (24.7\%)\\ No (75.3\%)\end{tabular}           & \begin{tabular}[c]{@{}c@{}}Yes (35.8\%)\\ No (64.2\%)\end{tabular}           & \begin{tabular}[c]{@{}c@{}}Yes (16.8\%)\\ No (83.2\%)\end{tabular}           \\ \hline
\textbf{Gender}                                                                                  & \begin{tabular}[c]{@{}c@{}}Female (53.6\%)\\ Male (46.4\%)\end{tabular}      & \begin{tabular}[c]{@{}c@{}}Female (53.7\%)\\ Male (46.3\%)\end{tabular}      & \begin{tabular}[c]{@{}c@{}}Female (51.1\%)\\ Male (48.9\%)\end{tabular}      & \begin{tabular}[c]{@{}c@{}}Female (62.7\%)\\ Male (37.3\%)\end{tabular}      \\ \hline
\textbf{Age (Last Entry)}                                                                        & \begin{tabular}[c]{@{}c@{}}Mean : 58.8\\ Sd : 16.6\end{tabular}              & \begin{tabular}[c]{@{}c@{}}Mean : 48.8\\ Sd : 16.7\end{tabular}              & \begin{tabular}[c]{@{}c@{}}Mean: 57.4\\ Sd : 16.5\end{tabular}               & \begin{tabular}[c]{@{}c@{}}Mean : 41.1\\ Sd : 14.7\end{tabular}              \\ \hline
\textbf{Race}                                                                                    & \begin{tabular}[c]{@{}c@{}}Caucasian (91.0 \%)\\ Others (9.0\%)\end{tabular} & \begin{tabular}[c]{@{}c@{}}Caucasian (53.4\%)\\ Others (46.5\%)\end{tabular} & \begin{tabular}[c]{@{}c@{}}Caucasian (85.3\%)\\ Others (14.7\%)\end{tabular} & \begin{tabular}[c]{@{}c@{}}Caucasian (88.4\%)\\ Others (11.6\%)\end{tabular} \\ \hline
\textbf{BMI}                                                                                     & \begin{tabular}[c]{@{}c@{}}Mean : 30.0\\ Sd : 7.2\end{tabular}               & \begin{tabular}[c]{@{}c@{}}Mean : 28.5\\ Sd : 7.3\end{tabular}               & \begin{tabular}[c]{@{}c@{}}Mean : 32.6\\ Sd : 7.8\end{tabular}               & \begin{tabular}[c]{@{}c@{}}Mean : 29.0\\ Sd : 7.6\end{tabular}               \\ \hline
\textbf{CHD}                                                                                     & \begin{tabular}[c]{@{}c@{}}Yes (17.0\%)\\ No (83\%)\end{tabular}             & \begin{tabular}[c]{@{}c@{}}Yes (7.9\%)\\ No (92.1\%)\end{tabular}            & \begin{tabular}[c]{@{}c@{}}Yes (29.7\%)\\ No (70.3\%)\end{tabular}           & \begin{tabular}[c]{@{}c@{}}Yes (4.3\%)\\ No (95.7\%)\end{tabular}            \\ \hline
\textbf{Diabetes}                                                                                & \begin{tabular}[c]{@{}c@{}}Yes (30.9\%)\\ No (69.1\%)\end{tabular}           & \begin{tabular}[c]{@{}c@{}}Yes (29.8\%)\\ No (70.2\%)\end{tabular}           & \begin{tabular}[c]{@{}c@{}}Yes (96.2\%)\\ No (3.8\%)\end{tabular}            & \begin{tabular}[c]{@{}c@{}}Yes (39.6\%)\\ No (60.4\%)\end{tabular}           \\ \hline
\textbf{\begin{tabular}[c]{@{}c@{}}High Trigylcerides\\ (\textgreater{}=150 mg/dl)\end{tabular}} & \begin{tabular}[c]{@{}c@{}}Yes (9.4\%)\\ No (90.2\%)\end{tabular}            & \begin{tabular}[c]{@{}c@{}}Yes (12.0\%)\\ No (88.0\%)\end{tabular}           & \begin{tabular}[c]{@{}c@{}}Yes (42.6\%)\\ No (57.4\%)\end{tabular}           & \begin{tabular}[c]{@{}c@{}}Yes (6.6\%)\\ No (93.4\%)\end{tabular}            \\ \hline
\textbf{\begin{tabular}[c]{@{}c@{}}Vitamin D \\ Deficiency\\ (\textless{}=40)\end{tabular}}      & \begin{tabular}[c]{@{}c@{}}Yes (16.9\%)\\ No (83.1\%)\end{tabular}           & \begin{tabular}[c]{@{}c@{}}Yes (19.2\%)\\ No (80.8\%)\end{tabular}           & \begin{tabular}[c]{@{}c@{}}Yes (46.9\%)\\ No (53.1\%)\end{tabular}           & \begin{tabular}[c]{@{}c@{}}Yes (17.2\%)\\ No (82.8\%)\end{tabular}           \\ \hline
\textbf{Depression}                                                                              & \begin{tabular}[c]{@{}c@{}}Yes (32.6\%)\\ No (67.3\%)\end{tabular}           & \begin{tabular}[c]{@{}c@{}}Yes (27.0\%)\\ No (73.0\%)\end{tabular}           & \begin{tabular}[c]{@{}c@{}}Yes (44.0\%)\\ No (56.0\%)\end{tabular}           & \begin{tabular}[c]{@{}c@{}}Yes (81.6\%)\\ No (18.4\%)\end{tabular}           \\ \hline
\textbf{Anxiety}                                                                                 & \begin{tabular}[c]{@{}c@{}}Yes (34.7\%)\\ No (65.3\%)\end{tabular}           & \begin{tabular}[c]{@{}c@{}}Yes (32.5\%)\\ No (67.5\%)\end{tabular}           & \begin{tabular}[c]{@{}c@{}}Yes (40.2\%)\\ No (59.8\%)\end{tabular}           & \begin{tabular}[c]{@{}c@{}}Yes (89.8\%)\\ No (10.2\%)\end{tabular}           \\ \hline
\end{tabular}
\end{adjustbox}
\vspace{0.5cm}
\caption{Descriptive Summaries of the different variables of interest in different cohorts of interest of MGI, namely Anesthesiology, MIPACT (Michigan Predictive Activity and Clinical Trajectories), MEND (Metabolism, Endocrinology, and Diabetes) and MHB (Mental Health Biobank).  CHD stands for Coronary heart disease. For continuous variables we reported Mean (Sd). The phecodes used here have been discussed in details Supplementary Section S12.}
\label{tab:Table 3}
\end{table}

\begin{figure}[ht]
	\centering
    \includegraphics[width =0.9\linewidth]{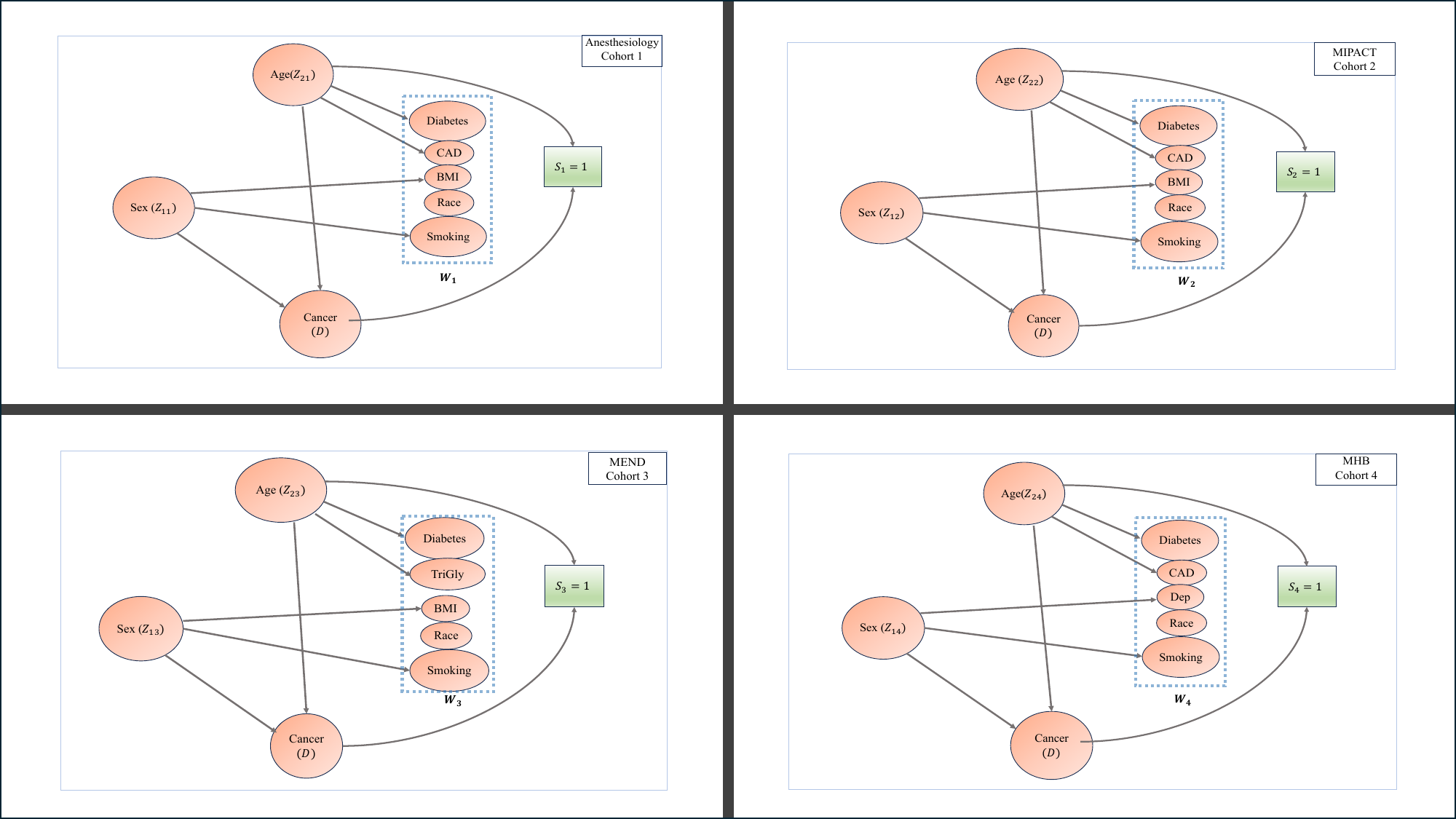}
	\caption{Directed Acyclic Graph Relationships between the disease and selection model variables in analysis of age-adjusted ($Z_2$) association between cancer ($D$) and biological sex ($Z_1$) in the four different cohorts of interest in the Michigan Genomics Initiative, namely MGI Anesthesiology (BB), MIPACT (Michigan Predictive Activity and Clinical Trajectories), MEND (Metabolism, Endocrinology, and Diabetes) and MHB (Mental Health BioBank). CAD, BMI, TriGly and Dep stand for Coronary Artery Disease, Body Mass Index, Triglycerides, and Depression, respectively.}
	\label{fig:dagmgi}
\end{figure}
\begin{figure}[ht]
	\centering
    \includegraphics[width =\linewidth]{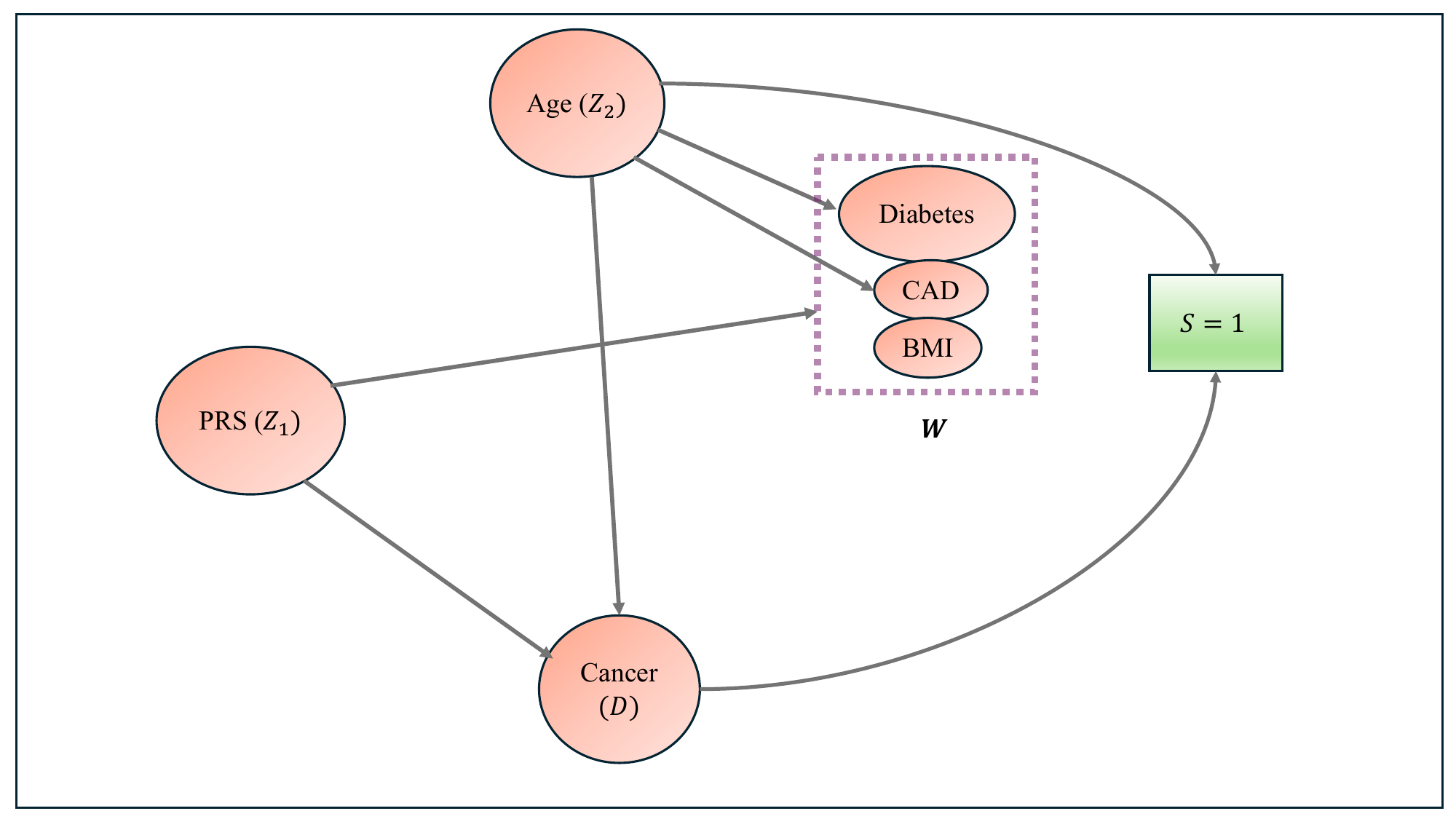}
	\caption{Directed Acyclic Graph Relationships between the disease and selection model variables in analysis of age-adjusted ($Z_2$) association between cancer ($D$) and Polygenic Risk Score (PRS)($Z_1$) for overall cancer. 
    in the Michigan Genomics Initiative. CAD and BMI stand for Coronary Artery Disease and Body Mass Index respectively.}
	\label{fig:dagmgi_PRS}
\end{figure}
\begin{figure}[ht]
    \centering
    \includegraphics[width =0.8\linewidth]{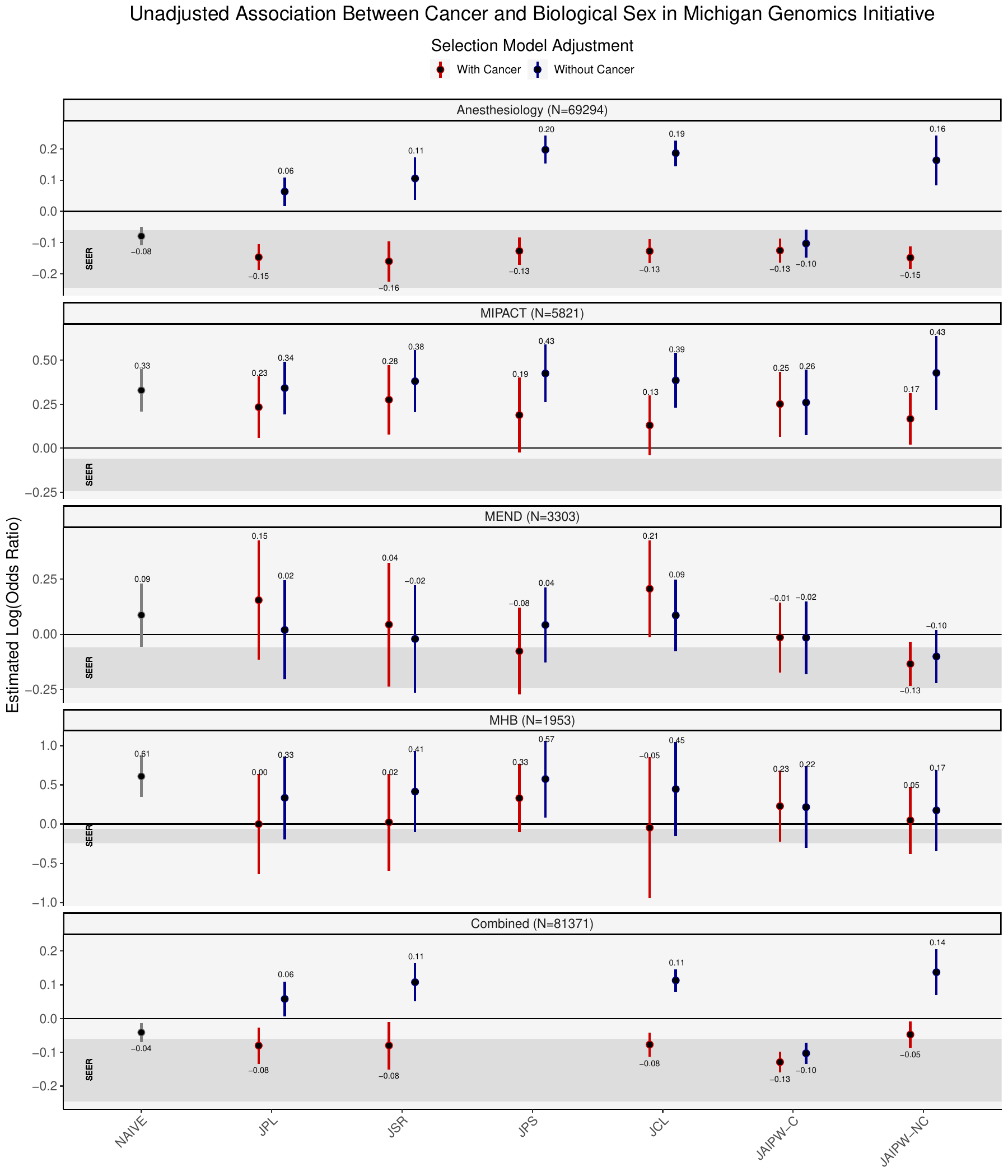}
    \caption{Estimates of the marginal log (Odds Ratio) for the association between cancer and biological sex, with 95\% confidence intervals, across Anesthesiology, MIPACT, MEND, MHB and the combined cohort.
    Comparisons are shown for the unweighted logistic regression (NAIVE, gray) and methods adjusting for selection bias (JPL, JSR, JPS, JCL, JAIPW-C and JAIPW-NC). For the IPW and JAIPW methods, estimates are shown either in red or blue depending on the selection model including cancer status (red, Cancer) or not (blue, No Cancer). JAIPW-C includes cancer in the auxiliary score model, while JAIPW-NC do not. The gray horizontal band represents the estimates of the log(Odds Ratio) from the SEER 2008-2016 registry.}\label{fig:realdata_supp}
\end{figure}
\begin{figure}[ht]
    \centering
    \includegraphics[width =0.8\linewidth]{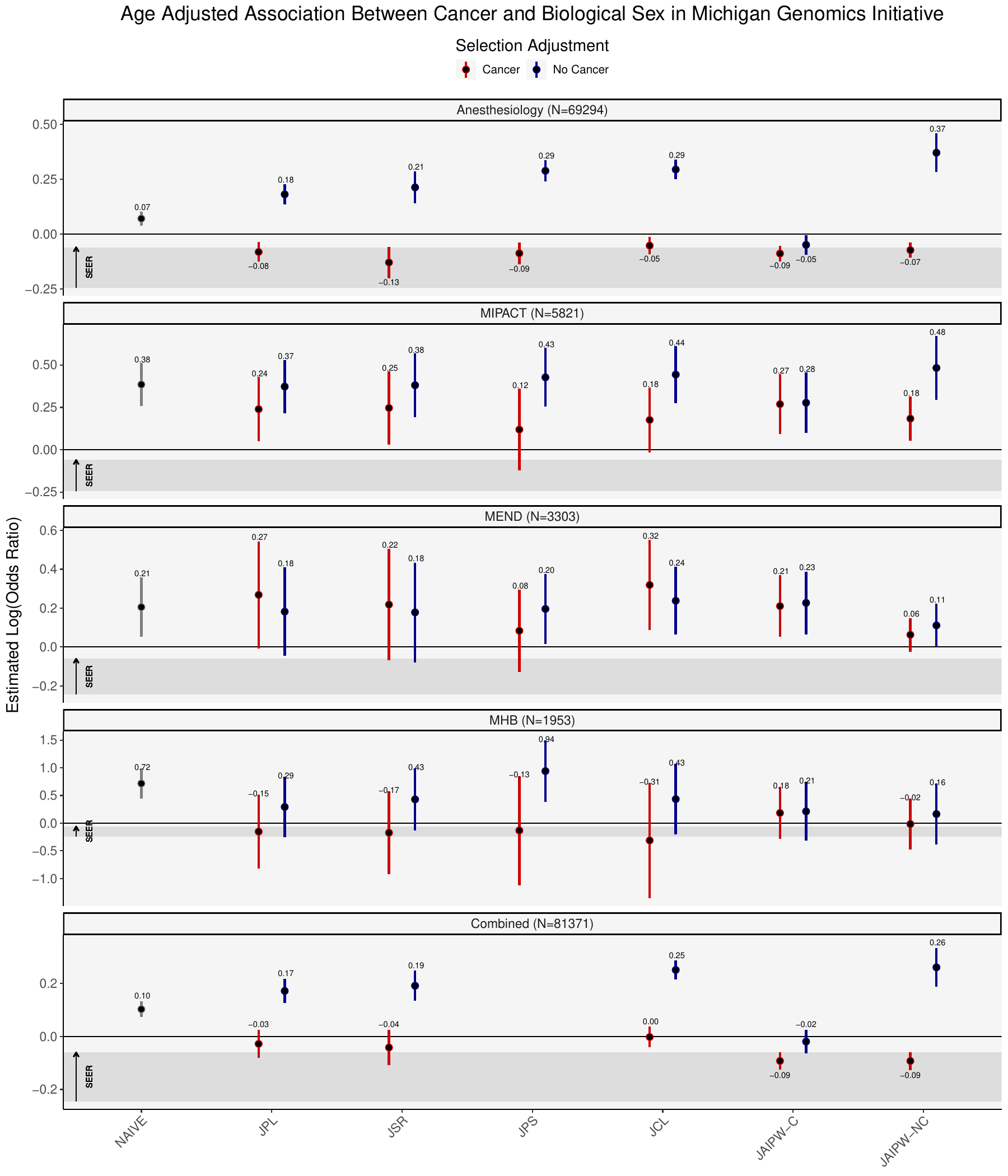}
    \caption{Estimates of the age-adjusted log (Odds Ratio) for the association between cancer and biological sex, with 95\% confidence intervals, across Anesthesiology, MIPACT, MEND, MHB and the combined cohort.
    Comparisons are shown for the unweighted logistic regression (NAIVE, gray) and methods adjusting for selection bias (JPL, JSR, JPS, JCL, JAIPW-C and JAIPW-NC). For the IPW and JAIPW methods, estimates are shown either in red or blue depending on the selection model including cancer status (red, Cancer) or not (blue, No Cancer). JAIPW-C includes cancer in the auxiliary score model, while JAIPW-NC do not. The gray horizontal band represents the estimates of the log(Odds Ratio) from the SEER 2008-2016 registry.}\label{fig:realdata_age}
\end{figure}
\begin{figure}[ht]
    \centering
    \includegraphics[width =\linewidth]{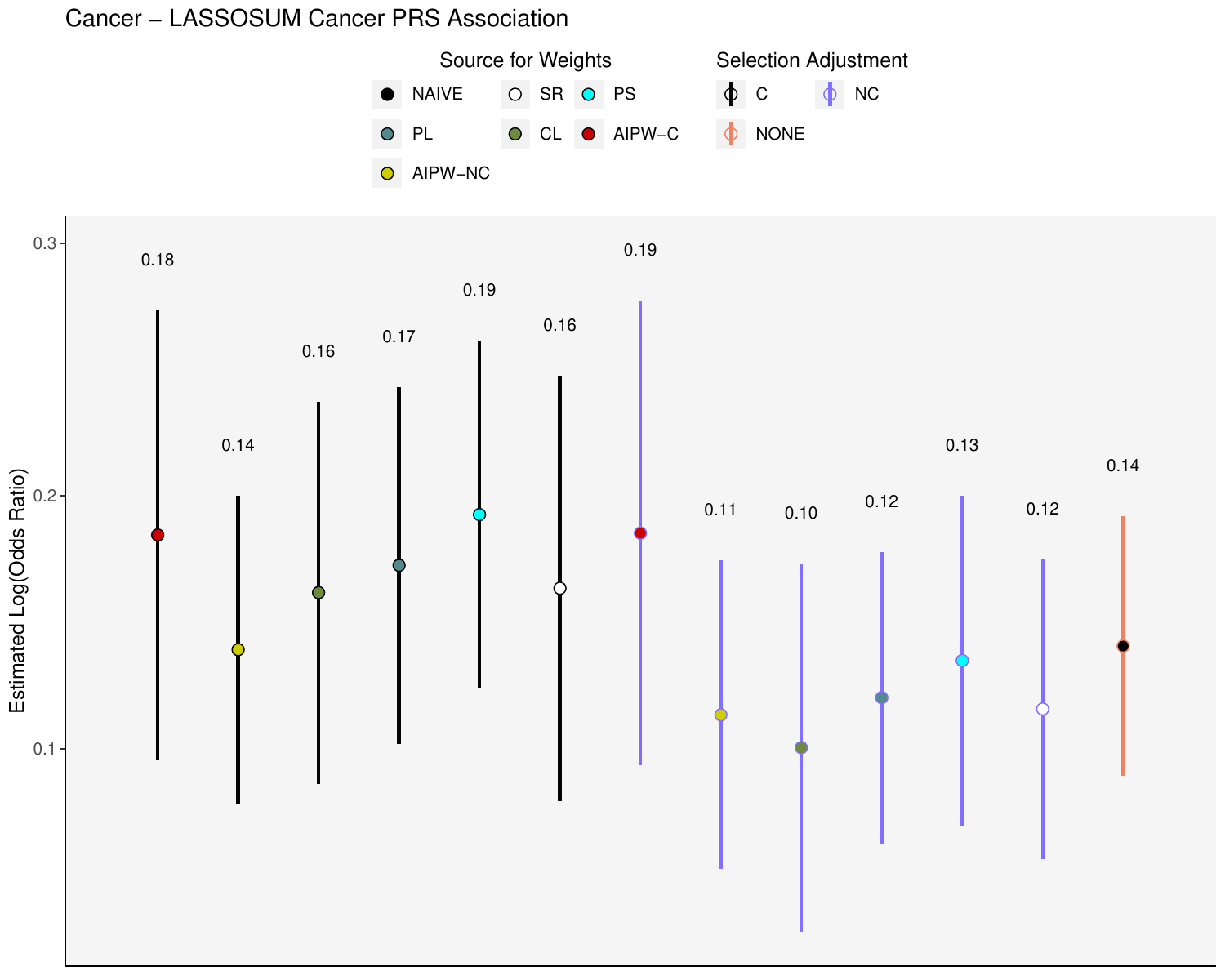}
    \caption{Estimates of the age-adjusted log odds ratio of cancer with one unit change in Interquartile Range (IQR) of  Polygenic Risk Score (PRS) for overall cancer  along with 95\% C.I in  using JAIPW, the four JIPW methods and the unweighted logistic regression with and without including cancer as a selection variable. The black and violet bars correspond to the estimates of the JAIPW method and the four JIPW methods without and with including cancer in the selection model respectively.}\label{fig:prs}
\end{figure}

\end{document}